\documentclass[10pt]{article}

\usepackage{geometry} 
\usepackage{amssymb} 
\usepackage{amsmath, bm} 
\usepackage{amsthm} 
\usepackage{mathtools}
\usepackage{scrextend} 
\usepackage[framemethod=tikz]{mdframed}  
\usepackage{enumitem} 
\usepackage{slashed} %for slashed nabla
\usepackage{comment} 
\usepackage{verbatim} % to use comments. \comment{}
\setlist[enumerate]{itemsep=0.5pt, topsep=4pt}

\usepackage{esint}

\usepackage{multicol} % for multicols
\usepackage{cancel} %for slashing divergence div

\usepackage{float}

\usepackage{mathrsfs}

\usepackage{xcolor}

\usepackage[T1]{fontenc}
\usepackage[utf8]{inputenc}

\usepackage{cite}
\usepackage{hyperref}
\hypersetup{colorlinks=true, linkcolor=blue, citecolor=blue, urlcolor=blue}

\newcommand{\vv}{\vspace{0.5cm}}

\newcommand{\R}{\mathbb{R}}

\newcommand{\red}[1]{\textcolor{red}{#1}}
\newcommand{\ul}[1]{\underline{#1}}

\newcommand{\gtime}{{\mathfrak{g}^{(4)}}}

\newcommand{\mathc}{\mathcal{C}}
\newcommand{\ct}{\mathcal{C}^{2+\alpha, 1+\alpha/2}}
\newcommand{\cz}{\mathcal{C}^{\alpha,\alpha/2}}
\newcommand{\ctS}{\mathcal{S}^{2+\alpha, 1+\alpha/2}}

\newcommand{\tr}{\mathrm{tr}}

\newcommand{\mathcos}{\mathcal{O}_{\mathcal S, r_1}}

\numberwithin{equation}{section}
\newtheorem{thm}{Theorem}[section]
\newtheorem{lem}[thm]{Lemma}
\newtheorem{prop}[thm]{Proposition} 
\newtheorem*{prop*}{Proposition}
\newtheorem {cor}[thm]{Corollary}  

\newtheorem*{question*}{Question}

\newtheorem*{mainthm}{Main Theorem}{\bf}{\it}
{\bf}{\it}

\theoremstyle{definition}
\newtheorem{defn}[thm]{Definition}

\theoremstyle{remark}
\newtheorem{ex}{Example} 
\newtheorem{remark}[thm]{Remark}

\newif\ifshownotes
\shownotesfalse % false hides

\title{The spacetime Penrose inequality under a quasi final state hypothesis}

\author{Ahmed Ellithy \thanks{Department of Mathematics, Uppsala University. Email: \texttt{ahmed.ellithy@math.uu.se}}}
\date{}

\begin{document}

\maketitle

\begin{abstract}

Penrose's original heuristic for his eponymous spacetime inequality---a conjectured lower bound on the ADM mass in terms of the area of a horizon cross-section---relies on the black hole final state conjecture. In this paper we isolate a substantially weaker but
precise late-time condition, which we call the \emph{quasi final state hypothesis} and prove the spacetime Penrose inequality under this hypothesis. More precisely, for an asymptotically flat globally hyperbolic spacetime with a black-hole-type apparent horizon tube $\mathcal{H}_{\mathrm{app}}$  satisfying the dominant energy condition and the quasi final state hypothesis, we show that every asymptotically flat initial data set whose boundary is a MOTS cross-section
of $\mathcal{H}_{\mathrm{app}}$ satisfies the spacetime Penrose inequality. The quasi final state hypothesis requires only a late-time decay condition on the normal component of the shift and the ratio of timelike to spacelike mean curvature, together with convergence of the cross-sectional areas of $\mathcal{H}_{\mathrm{app}}$ to a finite limit. 

Our approach is new and formulated directly in spacetime. The main geometric object is what we call a \emph{tangentially maximal} hypersurface, carrying a foliation by spacelike spheres whose timelike mean curvature vanishes. We show that these hypersurfaces are governed by a quasilinear inward-parabolic PDE, and we develop the corresponding a priori theory and prove global existence. On these hypersurfaces, the spacetime Hawking mass reduces to the Riemannian Hawking mass, and the dominant energy condition gives nonnegative scalar curvature. The Riemannian Penrose inequality, combined with the area laws for dynamical and isolated horizons, then yields the result.

\end{abstract}

\tableofcontents

\section{Introduction}

One of the central open problems in mathematical general relativity is the \emph{Penrose inequality}, a conjectured geometric lower bound on the total mass of a gravitational system in terms of the area of its black holes. The conjecture originates in Penrose's celebrated 1973 heuristic argument~\cite{penrose-naked}, which combines two foundational pillars of black hole physics: the \emph{weak cosmic censorship conjecture} and the \emph{final state conjecture}.

The weak cosmic censorship conjecture, proposed by Penrose~\cite{penrose-naked}, asserts that generic asymptotically flat initial data for the Einstein equations with physically reasonable matter give rise to spacetimes in which singularities are hidden behind event horizons. The final state conjecture asserts that, at late times, the exterior region of such a spacetime settles down to a member of the Kerr family of stationary black hole solutions. Combined with Hawking's area theorem~\cite{hawking-area}, which states that the area of the event horizon is nondecreasing in time under the dominant energy condition, Penrose's heuristic proceeds as follows.

Suppose that the initial data set $(M,g,K)$ contains an outermost marginally outer trapped surface (MOTS) $\Sigma_0$ of area $A$. Weak cosmic censorship predicts that \(\Sigma_0\) lies inside the black
hole region, and hence inside the event horizon. Let \(A_{\mathcal H}(0)\)
denote the area of the corresponding initial cross-section of the event
horizon. The heuristic comparison between the initial apparent horizon
and the event horizon gives $A_{\mathcal H}(0)\ge A$,
while Hawking's area theorem gives $A_f\ge A_{\mathcal H}(0)$,
where \(A_f\) is the area of the final event horizon.

If the exterior
settles down to a Kerr solution of mass \(m_f\) and angular momentum
\(J_f\), then
\[
        A_f
        =
        8\pi\big(m_f^2+\sqrt{m_f^4-J_f^2}\big).
\]
The Bondi mass-loss formula implies that the final Kerr mass cannot exceed
the initial ADM mass:
\[
        m_{\mathrm{ADM}}\ge m_f.
\]
Since $A_f \le 16\pi m_f^2$ for Kerr, this gives
\[
m_{ADM} \ge m_f \ge \sqrt{\frac{A_f}{16\pi}} \ge \sqrt{\frac{A}{16\pi}}.
\]
The resulting inequality
\begin{equation}\label{eq:PI-conjecture}
m_{ADM} \ge \sqrt{\frac{|\Sigma_0|}{16\pi}}
\end{equation}
is the \emph{Penrose inequality}. Its validity would provide strong evidence for the compatibility of the weak cosmic censorship, black-hole area increase, and late-time relaxation to Kerr; a counterexample would force at least one part of the argument to fail.

In fact, the heuristic gives the
stronger spacetime statement that the limiting Bondi mass,
equivalently the mass \(m_f\) of the final Kerr exterior, should already
satisfy
\begin{equation}\label{eq:PI-conjecture2}
        m_{\mathrm{Bondi}}^{+}
        =
        m_f
        \ge
        \sqrt{\frac{|\Sigma_0|}{16\pi}}.
\end{equation}
The ADM version \eqref{eq:PI-conjecture} then follows from Bondi mass
monotonicity, \(m_{\mathrm{ADM}}\ge m_{\mathrm{Bondi}}^{+}\).

\medskip 

A first subtlety is already visible in this argument: it is intrinsically a \emph{spacetime} argument, involving an event horizon and the late-time exterior geometry, whereas the inequality is usually formulated on a single initial data set. In this paper we keep the spacetime point of view. We isolate a precise late-time hypothesis that is substantially weaker than full convergence to Kerr, but still captures the features of the exterior geometry that the proof actually uses. Under this hypothesis we construct a new class of hypersurfaces at late times and use them to prove the spacetime Penrose inequality.

\subsection{Known results}

The Penrose inequality is completely understood in the \emph{time-symmetric} case, that is, on an initial data set for which $K=0$. In that setting the future null expansion of a closed surface is just its scalar mean curvature, so a MOTS is exactly a minimal surface on the initial data set. Thus an outermost MOTS is simply an outermost minimal surface. Furthermore, the dominant energy condition reduces to the Riemannian requirement $R_g\ge 0$. The inequality therefore becomes a purely Riemannian statement. Huisken and Ilmanen~\cite{H-I} proved \eqref{eq:PI-conjecture} for connected horizons using a weak formulation of inverse mean curvature flow (IMCF), and Bray~\cite{bray-conformal} proved the general case, allowing multiple horizon components, by his conformal flow of metrics. His approach was later extended to dimensions $n<8$ in \cite{BrayLeeRPI}. In both approaches the optimal constant is obtained, and equality characterizes the Schwarzschild geometry.

Outside time symmetry, the picture is far less complete. Broadly speaking, there are two main strategies in the literature.

The first strategy seeks to reduce the general case to the Riemannian one by deforming the initial data. The prototype is the classical Jang equation, introduced by Jang and used by Schoen and Yau in their proof of the positive mass theorem~\cite{schoen-yau-jang}. For the Penrose inequality, however, the classical Jang equation is not sufficient: the blow-up behavior at apparent horizons and the resulting cylindrical ends do not by themselves produce a usable reduction to the Riemannian Penrose inequality; see in particular~\cite{malec-jang}. In response, Bray and Khuri introduced a \emph{generalized} Jang equation tailored to the Penrose problem, together with geometrically motivated couplings to inverse mean curvature flow and to Bray's conformal flow~\cite{bray-khuri-jang}. Their work shows that, if the corresponding coupled system can be solved with the right boundary behavior, then the Penrose conjecture follows. This program is successful in spherical symmetry, but in general it leads to a highly nontrivial coupled system whose existence and regularity theory remain open.

The second strategy is to look for monotone quasi-local masses along flows or foliations built directly in spacetime. In the Riemannian case, Huisken--Ilmanen's weak IMCF is the model example. In spacetime, the mixed signature of the normal bundle of a codimension-two surface makes the analogous problem far more delicate. Bray, Hayward, Mars, and Simon identified a class of \emph{uniformly expanding} spacetime flows for which the Hawking mass is monotone~\cite{bray-hayward-mars-simon}. In null directions local existence is available, but in spacelike directions the resulting system leaves a one-parameter freedom and becomes \emph{forward-backward parabolic}, so one cannot expect a general local existence theory of the usual type. Later work of Bray and Jauregui, and of Bray, Jauregui, and Mars, showed that the spacetime Hawking mass is monotone under uniformly area-expanding flows of \emph{time-flat} surfaces~\cite{timeflat,timeflat2}. More recently, Hirsch introduced PDE systems on initial data sets modeling double-null foliations, proved Hawking-mass monotonicity in that framework, and established an accompanying existence theory~\cite{sven}. These developments are promising, but they still fall short of a general existence theory for a globally defined spacelike monotone flow reaching infinity.

Several additional partial results are known. In spherical symmetry the full inequality is classical, going back to work of Malec and \'O Murchadha and Hayward~\cite{malec-omurchadha,hayward-PI}; Bray and Khuri later recovered this regime via the generalized Jang equation~\cite{bray-khuri-jang}. Recently, Allen, Bryden, Kazaras, and Khuri proved a Penrose-type inequality with a universal but nonsharp constant for general three-dimensional asymptotically flat and asymptotically hyperboloidal initial data~\cite{allen-bryden-kazaras-khuri}. Khuri and Kunduri proved the sharp spacetime Penrose inequality under a higher-dimensional cohomogeneity-one symmetry assumption~\cite{khuri-kunduri}. Recent work of An and He proves spacetime Penrose inequalities in perturbative
black-hole-formation and Kerr regimes by hyperbolic evolution methods~\cite{an-he-kerr-black-hole-formation-penrose}. Nevertheless, a proof of the full sharp inequality in the general three-dimensional non-time-symmetric setting remains open.

A line of work closer in spirit to the present paper treats Penrose-type inequalities directly in spacetime, especially along null hypersurfaces. Building on \cite{sauter-null-penrose}, Alexakis proved a perturbative spacetime/null Penrose inequality for black-hole spacetimes close to the Schwarzschild exterior, using a deformation of an initial null hypersurface so that the associated luminosity foliation has monotone Hawking mass and becomes asymptotically round~\cite{alexakis-schwarzschild-penrose} (see also \cite{LeNullPenrose2024} for related results). This is similar to the present argument in that both approaches obtain the inequality by replacing the original hypersurface with a better-adapted comparison hypersurface. The two constructions, however, are quite different: Alexakis deforms a null hypersurface in order to obtain the desired asymptotic structure at null infinity, whereas here we construct spacelike comparison hypersurfaces by solving a prescribed geometric equation on the exterior spacetime. Related null-hypersurface approaches were developed by Mars--Soria, who studied monotonicity and asymptotic behavior of Hawking-type energies along null foliations~\cite{MarsSoria15,mars-soria-null-penrose}. Roesch subsequently introduced a new quasi-local mass, monotone along suitably convex foliations of null cones, and used it to prove a null Penrose inequality under broad convexity hypotheses~\cite{roesch-null-penrose-quasilocal-mass}; in later work he established stability results for quasi-round MOTS and for the Schwarzschild null Penrose inequality~\cite{roesch-quasiround-mots}.

It has also been observed that there are in fact two natural formulations of the Penrose inequality: an \emph{initial data formulation}, which posits the inequality as a statement about a single spacelike hypersurface $(M,g,K)$, and a \emph{spacetime formulation}, which uses the global causal structure of the ambient spacetime. Frauendiener already emphasized that the four-dimensional spacetime structure should be used in an essential way~\cite{frauendiener-PI}, and Mars' survey explains in detail why the initial-data formulation is ambiguous away from time symmetry~\cite{mars-PI}. Specifically, Mars~\cite{mars-PI} provided a comprehensive analysis of the issue, emphasizing that the initial-data formulation requires careful specification of which surface (MOTS, minimal surface, generalized apparent horizon) should appear in the inequality, and that careless choices lead to false statements.

Indeed, Ben--Dov~\cite{ben-dov} constructed a spherically symmetric
counterexample to a commonly stated apparent-horizon version of the Penrose
inequality, showing that the ADM mass need not bound the square root of the
area of an arbitrary outermost MOTS on a single initial data set. This is
particularly relevant for us, since the theorem proved below has a superficially
similar formulation involving an outermost MOTS. The distinction is that our
surface is tied to a future spacetime horizon structure and a quasi final state
hypothesis; we explain in subsection~\ref{sub:intro-bendov} and more in  subsection~\ref{sub:bendov-counterexamples} how this
excludes the mechanism in Ben--Dov's example.

Carrasco and Mars~\cite{carrasco-mars} constructed further counterexamples---
axially symmetric slices of the Kruskal spacetime---to the version of the
Penrose inequality formulated in terms of \emph{generalized apparent horizons}
(surfaces satisfying \(H_\Sigma = |\tr_\Sigma K|\), as proposed
in~\cite{bray-khuri2010}). These counterexamples do not contradict Penrose's
original spacetime heuristic, but they do illustrate the subtlety of purely
initial-data formulations. In this paper we therefore work with a spacetime
formulation, in which the inequality is tied to the future horizon structure
rather than to an arbitrary horizon surface on a single slice.

\subsection{Summary of results and approach}

Our point of departure is that, at late times, one does not need the full statement that the exterior converges to Kerr. In this paper we prove the spacetime Penrose inequality under a precise late-time condition
that is strictly weaker than the final state conjecture. We now describe the hypothesis, the
result, and the method.

\medskip
\noindent
\textbf{The quasi final state hypothesis.}
We work on an asymptotically flat globally hyperbolic spacetime satisfying the dominant energy
condition, with a piecewise smooth outermost apparent horizon tube $\mathcal{H}_{\mathrm{app}}$. We
formulate a \emph{quasi final state hypothesis} (Definition~\ref{def:spi-QFS}), which asserts,
roughly, that the exterior of the last smooth horizon piece $\mathcal{H}_{\mathrm{final}}$ admits a late-time chart 
$$(t,r,p) \in (\ul{T},\infty)\times [r_0,\infty) \times S^2, \qquad \mathcal{H}_{\mathrm{final}} = (\ul{T},\infty) \times \{r_0\} \times S^2,$$
 regular up to the
horizon and asymptotically flat on each fixed exterior tail $r\geq r_1>r_0$, in which: 
\begin{enumerate}[label=\textup{(\roman*)}]
\item the normal component of the shift decays to zero as $t\to\infty$ on every fixed tail;
\item the ratio of the timelike to spacelike mean curvature of the coordinate spheres decay to zero as $t\to\infty$ on every fixed tail;
\item the areas of cross-sections of $\mathcal{H}_{\mathrm{final}}$ converge to a finite limit as
$t\to\infty$.
\end{enumerate}

The connection to the final state conjecture is the following. In the Kerr spacetime with
Boyer--Lindquist coordinates, the shift vector has no radial component, the timelike mean
curvature of the coordinate spheres vanishes identically, and every cross-section of the horizon
has the same area. Therefore, if the exterior of a spacetime settles down to Kerr as the final
state conjecture predicts, all three conditions above are satisfied. The quasi final state
hypothesis asks only for the decay of these specific quantities, without requiring convergence
of the full exterior geometry to any particular stationary solution. In particular, apart from the regularity and asymptotic flatness assumptions, the metric coefficients are unconstrained and not required to have a prescribed late-time limit; they may oscillate, drift, or otherwise evolve nontrivially as $t\to \infty$, with no decay or convergence required.

\medskip
\noindent
\textbf{The Penrose inequality.}
Under the quasi final state hypothesis, we prove that any asymptotically flat initial data set
$(M_\ast, g_\ast, K_\ast)$ embedded in such a spacetime, whose boundary $S_\ast$ is a MOTS
cross-section of $\mathcal{H}_{\mathrm{app}}$, satisfies
\begin{equation}\label{eq:PI-conjecture3}
m_{ADM}(M_\ast, g_\ast, K_\ast) \ge \sqrt{\frac{|S_\ast|}{16\pi}}.
\end{equation}
Here $m_{ADM}$ denotes the invariant ADM mass of the asymptotic end. The horizon tube is allowed to have finitely many jump times, provided the area of the outermost MOTS does not decrease across any jump. 

\medskip
\noindent
\textbf{The method.}
The proof is based on a new type of hypersurfaces in the spacetime, which we call \emph{tangentially maximal},
characterized by the existence of a foliation---the \emph{TMCF foliation} (timelike-mean-curvature-free
foliation)---in which the timelike mean curvature of every leaf vanishes. Spacelike tangentially maximal
hypersurfaces enjoy two key properties under the dominant energy condition:
\begin{enumerate}[label=\textup{(\alph*)}]
\item the spacetime Hawking mass $m_H^{\mathrm{ST}}$ of each leaf coincides with its Riemannian Hawking mass $m_H^{\mathrm{Riem}}$;
\item the induced Riemannian metric has nonnegative scalar curvature.
\end{enumerate}
This is precisely the structure needed to apply the Riemannian Penrose inequality of
Huisken--Ilmanen~\cite{H-I}.

Given a spacetime satisfying the quasi final state hypothesis, we prove the existence of
tangentially maximal slices at late times by solving a quasilinear inward-parabolic
PDE for the graph functions describing the TMCF foliation. Taking a sequence of such slices whose
boundaries approach the late horizon, and combining properties \textup{(a)} and \textup{(b)} above with
Huisken--Ilmanen's theorem and the area law for dynamical and isolated horizons, we obtain the chain of inequalities
\[
m_{ADM}
\;\ge\;
m_H^{\mathrm{Riem}}(\Sigma_j)
\;=\;
m_H^{\mathrm{ST}}(\Sigma_j)
\;\xrightarrow{j\to\infty}\;
\sqrt{\frac{A_\infty}{16\pi}}
\;\ge\;
\sqrt{\frac{|S_\ast|}{16\pi}},
\]
where $\Sigma_j$ are the boundary spheres of the tangentially maximal slices and $A_\infty$ is
the limiting horizon area. Under an additional Bondi-Sachs asymptotic structure assumption, the same construction also yields the stronger inequality in \eqref{eq:PI-conjecture2} with the final Bondi mass $m_B^+$ in place of $m_{ADM}$, allowing for Bondi mass loss; see Appendix~\ref{app:bondi-annuli}.

 We now describe each ingredient in more detail.

\subsection{Tangentially maximal hypersurfaces and the TMCF foliation}

Let $(\mathcal{M}, \gtime)$ be an asymptotically flat spacetime, and let $\Sigma$ be a closed spacelike surface embedded in $\mathcal{M}$. The normal bundle of $\Sigma$ in $\mathcal{M}$ is a rank-two Lorentzian vector bundle, and the \emph{mean curvature vector} ${\bf H}_\Sigma$ of $\Sigma$ in $\mathcal{M}$ is a section of this normal bundle. Given any spacelike hypersurface $M \hookrightarrow \mathcal{M}$ containing $\Sigma$, with induced data $(g, K)$, the mean curvature vector decomposes as
\begin{equation}\label{eq:intro-H-decomp}
{\bf H}_\Sigma = H_\Sigma \, {\bf n} - (\tr_\Sigma K) \, {\bf T},
\end{equation}
where ${\bf T}$ is the future-pointing timelike unit normal to $M$, ${\bf n}$ is the outward unit normal to $\Sigma$ in $M$, $H_\Sigma$ is the mean curvature of $\Sigma$ in $(M,g)$ (the \emph{spacelike mean curvature}), and $\tr_\Sigma K$ is the tangential trace of $K$ along $\Sigma$ (the \emph{timelike mean curvature}). The Lorentzian norm of ${\bf H}_\Sigma$ is
\[
|{\bf H}_\Sigma|^2 = H_\Sigma^2 - (\tr_\Sigma K)^2.
\]
While ${\bf H}_\Sigma$ depends only on the embedding of $\Sigma$ in $\mathcal{M}$, the individual components $H_\Sigma$ and $\tr_\Sigma K$ depend on the choice of ambient hypersurface $M$.

Now let $M \subset \mathcal{M}$ be a hypersurface equipped with a foliation by a one-parameter family of closed spacelike surfaces $\{S_r\}_{r \ge r_0}$.

\begin{defn}\label{def:intro-tmcf}
A foliation $\{S_r\}_{r \ge r_0}$ of a hypersurface $M$ is called a \emph{TMCF foliation} (timelike-mean-curvature-free foliation) if the mean curvature vector ${\bf H}_{S_r}$ of each leaf is tangent to $M$. For any point $p \in M$ in which $M$ is spacelike, the decomposition \eqref{eq:intro-H-decomp} shows that this condition is equivalent to the vanishing of the timelike mean curvature:
\[
 \left. \tr_{S_r} K \right|_p = 0,
\]
where $K$ is the second fundamental form of $M$ in $\mathcal{M}$.

The hypersurface $M$ equipped with a TMCF foliation is called a \emph{tangentially maximal hypersurface}. If the hypersurface is spacelike, then the timelike mean curvature of each leaf vanishes:
\[
\tr_{S_r} K = 0 \qquad \text{for all } r \ge r_0,
\]
and we call $(M,g,K)$ a \emph{tangentially maximal initial data set}.
\end{defn}

\begin{remark}[Topology and weak formulations]
The existence of a smooth TMCF foliation by topological spheres on a spacelike hypersurface forces the underlying manifold to be topologically $\R^3 \setminus B$. However, there is a natural variational formulation associated to the TMCF condition: the first variation of the area functional of $S_r$ in the direction of ${\bf T}$ is proportional to $\int_{S_r} \tr_{S_r} K \, d\mu$, so the TMCF condition can be interpreted as the vanishing of this variation. This suggests the possibility of a weak formulation of the TMCF foliation via variational methods that can allow for more general topologies, which we do not pursue here.
\end{remark}

The tangentially maximal condition has two important consequences that make it the natural condition for the Penrose inequality. For a closed surface $\Sigma \subset M$, let
\[
m_H^{\mathrm{ST}}(\Sigma)
:=
\sqrt{\frac{|\Sigma|}{16\pi}}
\left(
1-\frac{1}{16\pi}\int_\Sigma |{\bf H}_\Sigma|^2\,d\mu_\Sigma
\right)
\]
denote its spacetime Hawking mass, and let
\[
m_H^{\mathrm{Riem}}(\Sigma)
:=
\sqrt{\frac{|\Sigma|}{16\pi}}
\left(
1-\frac{1}{16\pi}\int_\Sigma H_\Sigma^2\,d\mu_\Sigma
\right)
\]
denote its Riemannian Hawking mass computed in $(M,g)$.

\begin{enumerate}[label=\textup{(\roman*)}]
\item \emph{Hawking mass reduction.} Since $\tr_{S_r} K = 0$, the decomposition \eqref{eq:intro-H-decomp} gives $|{\bf H}_{S_r}|^2 = H_{S_r}^2$, so the spacetime Hawking mass of $S_r$ reduces to its Riemannian Hawking mass:
\[
m_H^{\mathrm{ST}}(S_r) = m_H^{\mathrm{Riem}}(S_r).
\]

\item \emph{Nonnegative scalar curvature.} If $(\mathcal{M}, \gtime)$ satisfies the spacetime dominant energy condition and the tangentially maximal hypersurface is spacelike, then the induced Riemannian metric $g$ has nonnegative scalar curvature, $R_{g} \ge 0$. This follows from the Hamiltonian constraint and the algebraic observation that the tangentially maximal condition implies $|K|_{g}^2 - (\tr_{g} K)^2 \ge 0$; see Proposition~\ref{prop:R-nonneg-tm}.
\end{enumerate}

The following consequence is proved in Corollary~\ref{cor:HI-comparison}.

\begin{prop}
Let $(M, g, K)$ be a tangentially maximal initial data set in a spacetime satisfying the dominant energy condition, with TMCF foliation $\{S_r\}_{r \ge r_0}$. Suppose $M$ is asymptotically flat and $S_{r_0} = \partial M$ is connected and outward minimizing. Then
\[
m_{ADM}(M, g) \ge m_H^{\mathrm{ST}}(S_{r_0}).
\]
\end{prop}

This follows by applying Huisken--Ilmanen's IMCF comparison to $(M, g)$ in \cite{H-I}, which has nonnegative scalar curvature and a connected outward-minimizing boundary. The key point is that the right-hand side is a \emph{spacetime} quantity---it depends only on the embedding of $S_{r_0}$ in $\mathcal{M}$, not on the choice of ambient hypersurface.

The TMCF foliation is closely related to, but distinct from, mean curvature flow in spacetime. Mean curvature flow specifies both direction and speed: the surface moves in the direction of its mean curvature vector with speed equal to its norm. The TMCF condition is weaker. It asks only that the leaves be arranged inside a spacelike hypersurface so that the mean curvature vector of each leaf has no timelike component relative to that hypersurface. Equivalently, it fixes the relevant spacetime direction singled out by ${\bf H}$ but leaves the radial parametrization free. This additional freedom is exactly what leads to a tractable parabolic graph equation.

\begin{ex}[Kerr spacetimes]\label{ex:kerr-intro}
In the Kerr spacetime with Boyer--Lindquist coordinates $(t,r,\vartheta,\varphi)$, the constant-$t$ slices of the exterior region are tangentially maximal, and the coordinate spheres $\{r = \mathrm{const}\}$ form a TMCF foliation (Proposition~\ref{prop:kerr}). This is consistent with the well-known fact that the constant-$t$ slices in Kerr are \emph{not} maximal ($\tr_g K \ne 0$) but do have vanishing tangential trace of $K$ along the Boyer--Lindquist spheres.
\end{ex}

\begin{ex}[Time-symmetric initial data sets]\label{ex:timesym-intro}
Any asymptotically flat time-symmetric initial data set $(M, g, 0)$ with $M \cong \R^3 \setminus B$ is automatically tangentially maximal, since $K = 0$ implies $\tr_{S_r} K = 0$ for every surface. Thus the TMCF foliation framework contains the time-symmetric setting as a special case. This includes, for example, the standard time-symmetric slices of Schwarzschild and Reissner--Nordstr\"om, as well as more general asymptotically flat time-symmetric black-hole data admitting a spherical foliation.
\end{ex}

\subsection{The TMCF equation and the existence result}

We now describe the PDE governing the TMCF foliation when the tangentially maximal hypersurface is spacelike and realized as a graph. In a certain ``good gauge'', we also describe the PDE for general tangentially maximal hypersurfaces that are not necessarily spacelike.

 We work in the late-time exterior region of the spacetime, which we identify with
\[
\mathcal{M} = (\underline{T}, \infty) \times (r_0, \infty) \times S^2,
\]
where the inner boundary corresponds to the last smooth piece of the horizon. The spacetime metric is written in the general ADM form
\[
\gtime = -(N^2 - |\beta|_{g_t}^2)\,dt^2 + 2\beta \odot dt + g_t,
\qquad
 g_t=(\lambda^2+|b|_\gamma^2)\,dr^2+2b\odot dr+\gamma,
\]
with $\beta$ having no $dt$-component and $b$ having no $dt$- or $dr$-component.

A spacelike tangentially maximal hypersurface with TMCF foliation $\{S_r\}$ is then realized as the graph $M_f = \{t = f(r,p)\}$ of a function $f: (r_0, \infty) \times S^2 \to \mathbb{R}$, provided $f$ satisfies the \emph{TMCF equation}:
\begin{equation}\label{eq:intro-tmcf}
\cancel{div}_{\gamma_f}\big(\beta_f^T\big) + H_{f,r}\,\beta_f({\bf n}_f) - \frac{1}{2}\tr_{\gamma_f}(\partial_t \gamma_f) = 0.
\end{equation}
Here $\gamma_f$ is the induced metric on the graph spheres $S_r\subset M_f$, $H_{f,r}$ is their mean curvature in $(M_f,g_f)$, ${\bf n}_f$ is the outward unit normal in $(M_f,g_f)$, and $\beta_f^T$ is the tangential part of the shifted spacetime $1$-form $\beta_f:= \beta - (N^2-|\beta|^2_{g_t})df$.

Equation~\eqref{eq:intro-tmcf} is valid in arbitrary coordinates. In general it has a geometric quasilinear parabolic structure. On a \emph{good gauge slab}, namely in a chart for which
\[
b\equiv 0,
\qquad
\beta_r\equiv 0,
\]
it becomes an explicit quasilinear backward-parabolic equation of the form
\begin{equation}\label{eq:intro-parabolic}
f_r + \mathcal{A}(r, p, f, \slashed{\nabla} f)\,:\,\slashed{\nabla}^2_\gamma f = \mathcal{B}(r, p, f, \slashed{\nabla} f),
\end{equation}
where $\mathcal{A}$ is a positive definite $(2,0)$-tensor, $\mathcal{B}$ is a lower-order term, and ``$:$'' denotes the contraction of $(2,0)$ tensor against a $(0,2)$ tensor. Thus the TMCF equation is quasilinear backward-parabolic in the outward radial variable $r$, and forward-parabolic in the inward direction. Section~2 derives the equation in general form and its explicit good-gauge version. The appendix then shows that, under a late-time decay assumption on the geometric normal shift, such a good gauge can always be achieved on sufficiently late slabs away from the horizon.

Furthermore, in the good gauge the explicit equation
\eqref{eq:intro-parabolic} has a geometric meaning beyond the spacelike
regime. Indeed, if the graph spheres $S_r$ remain spacelike and have nonvanishing spacelike mean curvature (i.e. ${H_{f,r}} \neq0$), then \eqref{eq:intro-parabolic} is equivalent to the coordinate-free condition
\[
{\bf H}_{S_{f,r}}\in TM_f.
\]
Thus the good-gauge equation continues to describe tangentially maximal graph
hypersurfaces even when \(M_f\) is timelike or null (see proposition \ref{prop:good-gauge-geometric-tmcf}). In particular, in the good gauge, a hypersurface $M_f$ that is not necessarily spacelike is tangentially maximal if and only if equation \eqref{eq:intro-parabolic} is satisfied. On the other hand, equation \eqref{eq:intro-tmcf} makes sense only if $M_f$ is spacelike. 

\vv
In the good gauge, the lower-order term admits a natural decomposition $\mathcal{B} = \Xi|_{t=f} + \mathcal{B}_2$, where $\mathcal{B}_2$ vanishes when $\slashed{d}f = 0$. The leading forcing term is
\begin{equation}\label{eq:intro-Xi}
\Xi := \frac{\lambda}{N}\,\frac{\tr_{S_{t,r}} K_t}{H_{t,r}}.
\end{equation}
This is the ratio of the timelike to the spacelike mean curvature of the coordinate spheres, rescaled by the factor $\lambda/N$ natural for the parabolic structure of the equation. In Kerr, and likewise in the basic time-symmetric model examples, one has $\Xi \equiv 0$, and the TMCF equation admits the trivial solution $f = \mathrm{const}$. 

We measure the forcing by the tail-size functions
\begin{align*}
\mathfrak{X}_0(r, T) &:= r \sup_{\substack{t \ge T \\ p \in S^2}} |\Xi(t, r, p)|, \\
\mathfrak{X}_1(r, T) &:= r \sup_{\substack{t \ge T \\ p \in S^2}} |\slashed{d}\Xi(t, r, p)|_\gamma, \\
\mathfrak{X}_2(r, T) &:= r^2 \sup_{\substack{t \ge T \\ p \in S^2}} |\slashed{\nabla}^2 \Xi(t, r, p)|_\gamma,
\end{align*}
as well as the derived quantities on each fixed tail $r\ge r_1>r_0$:
\[
G_{r_1}(T):=\int_{r_1}^\infty \frac{\mathfrak X_0(\sigma,T)}{\sigma}\,d\sigma,
\qquad
\Phi_{j,r_1}(T):=\int_{r_1}^\infty \mathfrak X_j(\sigma,T)\,d\sigma \quad (j=1,2),
\]
\[
\Omega_{r_1}(T):=\sup_{r\ge r_1}\mathfrak X_0(r,T),
\qquad
\Psi_{r_1}(T):=\sup_{t\ge T}\|r\Xi(t,\cdot)\|_{\cz_{-\tau}(M_{r_1,\infty})}.
\]
On fixed late slabs we use the corresponding slab versions obtained by restricting the time supremum to the slab. In addition, the late-time gauge reduction uses the normal-shift tail
\[
\mathfrak B^\sharp(T;r_1)
:=
\sup_{t\ge T}\|\beta^\perp(t,\cdot)\|_{{\ct}^{\sharp}_{-\tau}(M_{r_1,\infty})},
\qquad
\beta^\perp:=\beta({\bf n}_t),
\]
on each fixed tail $r\ge r_1>r_0$.

Our existence theorem may be summarized informally as follows.

\vv

\begin{mainthm}[Informal; see Theorem~\ref{thm:global-existence}]
Let $(\mathcal M,\gtime)$ be an asymptotically flat exterior spacetime of order $\tau\in(1/2,1)$ on
\[
\mathcal M=(\underline T,\infty)\times (r_0,\infty)\times S^2,
\]
with coefficient tuple $(N,\lambda, \beta, b,\gamma)$ in the class $\ctS_{-\tau}$ introduced in Section~3. Then:
\begin{enumerate}[label=\textup{(\roman*)}]
\item for every $T_0>\underline T$, there exists a unique tangentially maximal hypersurface on a sufficiently far-out exterior region whose graph converges to the time slice $t=T_0$ at infinity;

\item if the normal-shift tail $\mathfrak B^\sharp(T;r_1)$ and the basic forcing tails $G_{r_1}(T)$ and $\Phi_{1,r_1}(T)$ tend to $0$ on every fixed tail, then for every prescribed tail radius $r_1>r_0$ and all sufficiently large $T_0$ there exists a unique tangentially maximal graph on $(r_1,\infty)\times S^2$ with $\lim_{r\to\infty}f=T_0$;

\item if, in addition, either the pair $(\Phi_{2,r_1}(T),\Omega_{r_1}(T))$ tends to $(0,0)$ on every fixed tail or $\Psi_{r_1}(T)\to0$ on every fixed tail, then the graph in \textup{(ii)} is spacelike, hence defines a tangentially maximal initial data set, and it satisfies quantitative $C^0$, gradient, Hessian, and radial-slope bounds on the corresponding good-gauge slab.
\end{enumerate}
\end{mainthm}

The proof combines four ingredients: the late-slab gauge reduction from a general ADM chart to the good gauge, local-in-$r$ existence from a large outer sphere via standard quasilinear parabolic theory, sharp a priori estimates on truncated annuli, and a continuation argument that closes the bootstrap and then passes to the limit as the outer radius tends to infinity.

\subsection{The spacetime Penrose inequality}

We now describe the spacetime setting for our Penrose inequality. Let $(\widehat{\mathcal{M}}, \gtime)$ be a globally hyperbolic asymptotically flat spacetime satisfying Einstein's equations with the spacetime dominant energy condition. We fix an asymptotic rest frame and a time function $t$ whose level sets $\Sigma_t$ are Cauchy hypersurfaces. Assume that each slice $\Sigma_t$ contains an outermost MOTS $\mathcal{S}_t$ (not necessarily connected), and that the resulting apparent horizon evolution
\[
\mathcal{H}_{\mathrm{app}} = \bigcup_t \mathcal{S}_t
\]
is piecewise smooth with finitely many jumps, each smooth piece being a MOTT. Let $\mathcal{H}_{\mathrm{final}}$ denote the last smooth piece of the horizon.

Our Penrose inequality requires a \emph{quasi final state hypothesis} (Definition~\ref{def:spi-QFS}), which is a precise late-time formulation of the exterior settling needed for the argument. Informally, it consists of three parts:
\begin{enumerate}[label=\textup{(Q\arabic*)}]
\item the closure of the exterior of $\mathcal H_{\mathrm{final}}$ admits a late-time $C^2$ chart
\[
(t,r,p)\in(\underline T,\infty)\times [r_0,\infty)\times S^2
\]
with $\mathcal H_{\mathrm{final}}=\{r=r_0\}$, with coefficient tuple $(N,\lambda, \beta, b,\gamma)$ in the class $\ctS_{-\tau}$ introduced in Section~3;
\item on every fixed tail, the normal-shift tail $\mathfrak B^\sharp(T;r_1)$ and the late-time forcing quantities $G_{r_1}(T)$ and $\Phi_{1,r_1}(T)$ tend to $0$, and either the pair $(\Phi_{2,r_1}(T),\Omega_{r_1}(T))$ tends to $(0,0)$ or $\Psi_{r_1}(T)\to0$;
\item the horizon area stabilizes: if $A(t)=|\mathcal S_t|$, then $A(t)\to A_\infty\in(0,\infty)$ as $t\to\infty$.
\end{enumerate}

These conditions are much weaker than asking for the entire late-time exterior to converge to a specific stationary model. In particular, the lapse $N$, the radial lapse $\lambda$, the tangential part of the shift $\beta^T$, and the spatial metric $g_t$ of the chart are entirely unconstrained at late times: they may oscillate, drift, or otherwise evolve nontrivially as $t\to\infty$, with no decay or convergence required. On the other hand, they are exactly what is needed for the comparison-tail construction and the final passage to the horizon. In particular, the strong final state hypothesis implies the quasi final state hypothesis; see Proposition~\ref{prop:strong-implies-weak}. More generally, Proposition~\ref{prop:stability-implies-QFS} gives a stability criterion: any spacetime whose late-time exterior converges, in the required weighted sense, to a stationary black-hole exterior with $\beta^\perp=0$ and $\Xi=0$ satisfies the quasi final state hypothesis, provided the final-horizon regularity and area stabilization assumptions hold. Thus the known and expected nonlinear stability results for Schwarzschild and slowly rotating Kerr give natural classes of perturbative spacetimes to which the hypothesis holds.

\begin{remark}[Horizon-penetrating final-state charts]
The standard Boyer--Lindquist coordinates in Kerr (and Schwarzschild coordinates in Schwarzschild) do not penetrate the future horizon, since their constant-$t$ slices limit to the bifurcation sphere. The chart required by the final-state hypotheses is therefore not the standard stationary chart, but a horizon-penetrating modification that agrees with the standard chart on $\{r\ge r_{\mathrm{cut}}(t)\}$ with $r_{\mathrm{cut}}(t)\downarrow r_+$ as $t\to\infty$, and whose constant-$t$ slices foliate the future horizon by MOTS cross-sections. Appendix~\ref{app:horizon-penetrating-schwarzschild-model} carries out this construction explicitly for Schwarzschild; the same strategy applies to Kerr in ingoing Kerr coordinates. Both spacetimes therefore satisfy the quasi final state hypothesis; see Remark~\ref{rem:horizon-penetrating-final-state-gauge} for details.
\end{remark}

\begin{remark}[Fixed-tail decay versus eventual smallness]
\label{rem:intro-eventual-smallness}
The decay-to-zero formulation above is a fixed-tail condition.  For each fixed
\(r_1>r_0\), the TMCF construction only requires the normal-shift tail and the forcing quantities
to be sufficiently small on \(r\ge r_1\) for all sufficiently late times.  The required smallness,
however, depends on the nondegenerate coefficient bounds on that fixed tail.  As \(r_1\downarrow
r_0\), these bounds may deteriorate. Thus the small-data thresholds need not be uniform as the tail approaches the horizon.

It may be possible to replace the fixed-tail decay assumptions by a pure eventual-smallness
condition, but doing so would require a more refined, horizon-adapted analysis of the TMCF
equation, likely in a degenerate parabolic regime where the redshift structure is used rather
than estimated away.  We do not pursue that theory here.  The present decay assumptions are used
as a robust way to ensure that the tail-dependent smallness thresholds, and the smallness needed
for the late-slab gauge reduction, are eventually satisfied on every fixed exterior tail.  No
uniform convergence rate as \(r_1\downarrow r_0\) is imposed.
\end{remark}

\begin{thm}[Spacetime Penrose inequality; informal version of Theorem~\ref{thm:SPI}]
Let $(\widehat{\mathcal{M}}, \gtime)$ be a globally hyperbolic asymptotically flat spacetime satisfying Einstein's equations and the spacetime dominant energy condition. Assume that for every $t\ge 0$ under consideration the slice $\Sigma_t$ contains an outermost MOTS $\mathcal S_t$, that the union
\[
\mathcal H_{\mathrm{app}}=\bigcup_{t\ge 0}\mathcal S_t
\]
is piecewise smooth, and that its last smooth piece $\mathcal{H}_{\mathrm{final}}$ satisfies the quasi final state hypothesis. 

Let $(M_\ast,g_\ast,K_\ast)\hookrightarrow(\widehat{\mathcal M},\gtime)$ be an
asymptotically flat initial data set. Assume that its
boundary $S_\ast:=\partial M_\ast$
is a MOTS and a smooth cross-section of $\mathcal H_{\mathrm{app}}$. Assume that the portion of $\mathcal H_{\mathrm{app}}$ to the future of $S_\ast$
has only finitely many jump times, and that
the outermost-horizon area does not decrease across any jump. Then
\begin{equation}
m_{ADM}(M_\ast,g_\ast,K_\ast)
\ge
\sqrt{\frac{|S_\ast|}{16\pi}},
\end{equation}
where $m_{ADM}$ denotes the invariant ADM mass of the asymptotic end.
\end{thm}

The assumption on jumps is natural in this framework. Along each smooth dynamical-horizon segment the area is nondecreasing, while along isolated-horizon segments it is constant; thus the only extra input is that no area is lost at the finitely many jump times of the apparent-horizon evolution.

\subsection{Sketch of the proof}

The proof of the spacetime Penrose inequality proceeds in two main steps.

\medskip
\noindent
\textbf{Step 1: Existence of tangentially maximal comparison tails near the horizon.}
Fix a sequence of tail radii $r_j \downarrow r_0$. For each $j$, the quasi final state hypothesis and the appendix allow us, after choosing a sufficiently large late time $T_j$, to work in a good gauge on a slab over the tail $r\ge r_j$. The existence theorem then gives a spacelike TMCF graph $f_j$ on $M_{r_j,\infty}$ with $\lim_{r\to\infty}f_j=T_j$. Each graph $M_{f_j}$ is a spacelike, tangentially maximal hypersurface with nonnegative scalar curvature, and its boundary sphere $\Sigma_j$ is connected and outward minimizing.

By Huisken--Ilmanen's IMCF comparison~\cite{H-I}, applied to the Riemannian manifold $(M_{f_j},g_{f_j})$, we obtain
\[
m_{ADM} \ge m_H^{\mathrm{Riem}}(\Sigma_j) = m_H^{\mathrm{ST}}(\Sigma_j),
\]
where the equality uses the tangentially maximal condition $\tr_{\Sigma_j}K_{f_j}=0$.

\medskip
\noindent
\textbf{Step 2: Passage to the horizon using the horizon area law.}
The a priori estimates show that, after first choosing $r_j$ and only then pushing $T_j$ sufficiently far into the future, the boundary spheres $\Sigma_j$ become small $C^2$-graphs over the late horizon sections $\mathcal S_{T_j}$. The late-horizon continuity lemma therefore gives
\[
m_H^{\mathrm{ST}}(\Sigma_j)\longrightarrow \sqrt{\frac{A_\infty}{16\pi}}.
\]
Passing to the limit in the mass bound yields
\[
m_{ADM} \ge \sqrt{\frac{A_\infty}{16\pi}}.
\]
Finally, by the area law on smooth horizon pieces together with the assumption that the outermost-horizon area does not decrease across jumps, one has $A_\infty \ge |S_\ast|$, and hence
\[
m_{ADM} \ge \sqrt{\frac{|S_\ast|}{16\pi}}.
\]

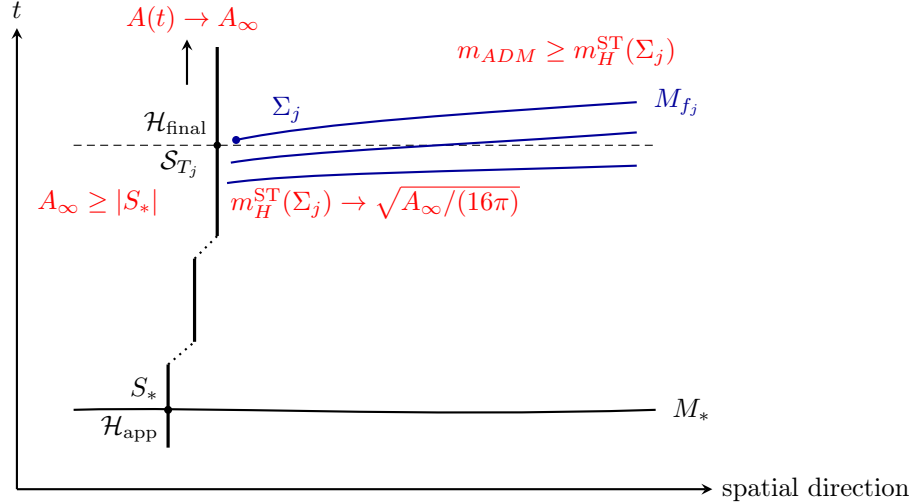
\begin{figure}[H]
\centering
\begin{tikzpicture}[x=1cm,y=1cm,>=stealth]

% axes
\draw[->,thick] (0,0) -- (9.2,0) node[right] {spatial direction};
\draw[->,thick] (0,0) -- (0,6.1) node[above] {$t$};

% apparent horizon tube, with jumps indicated
\draw[very thick] (2.0,0.55) -- (2.0,1.65);
\draw[very thick] (2.35,1.95) -- (2.35,3.05);
\draw[very thick] (2.65,3.35) -- (2.65,5.85);
\draw[dotted,thick] (2.0,1.65) -- (2.35,1.95);
\draw[dotted,thick] (2.35,3.05) -- (2.65,3.35);
\node[left] at (2.0,0.8) {$\mathcal H_{\mathrm{app}}$};
\node[left] at (2.65,4.85) {$\mathcal H_{\mathrm{final}}$};

% initial data set M_*
\draw[thick] (0.75,1.05) .. controls (2.5,1.1) and (5.2,0.95) .. (8.45,1.05);
\node[right] at (8.55,1.05) {$M_\ast$};
\fill (2.0,1.05) circle (1.5pt);
\node[above left] at (2.0,1.05) {$S_\ast$};

% late rest-frame slice
\draw[densely dashed] (0.75,4.55) -- (8.45,4.55);
\node[right] at (8.55,4.55) {$$};

% TMCF comparison tails
\draw[thick,blue!60!black] (2.78,4.05) .. controls (3.7,4.18) and (5.7,4.20) .. (8.2,4.28);
\draw[thick,blue!60!black] (2.83,4.32) .. controls (3.8,4.48) and (5.9,4.55) .. (8.2,4.72);
\draw[thick,blue!60!black] (2.90,4.62) .. controls (4.0,4.85) and (6.1,4.98) .. (8.2,5.12);
\node[right,blue!60!black] at (8.3,5.12) {$M_{f_j}$};

% boundary sphere
\fill[blue!60!black] (2.90,4.62) circle (1.5pt);
\node[above left,blue!60!black] at (3.90,4.72) {$\Sigma_j$};

% area and Hawking mass annotations
\draw[->,thick] (2.25,5.35) -- (2.25,5.95);

\node[right] at (5.7,5.8) {$\red{m_{ADM} \ge m_H^{\mathrm{ST}}(\Sigma_j)}$};

\node[right] at (1.32,6.22) {$\red{A(t)\to A_\infty}$};

%\draw[->,thick] (4.75,4.72) -- (4.05,4.98);
\node[right] at (2.7,3.8) {$\red{m_H^{\mathrm{ST}}(\Sigma_j)\to\sqrt{A_\infty/(16\pi)}}$};

% small label for late horizon section
\fill (2.65,4.55) circle (1.4pt);
\node[left] at (2.55,4.3) {$\mathcal S_{T_j}$};
\node[left] at (2,3.8) {$\red{A_{\infty} \geq |S_*|}$};
\end{tikzpicture}
\caption{Schematic of the proof. The apparent-horizon evolution has an initial cross-section $S_\ast$ and a last smooth piece $\mathcal H_{\mathrm{final}}$. The late-time tangentially maximal comparison tails $M_{f_j}$ meet the exterior near $\mathcal H_{\mathrm{final}}$ along spheres $\Sigma_j$ whose spacetime Hawking mass approaches $\sqrt{A_\infty/(16\pi)}$. The area law implies that $A_{\infty} \geq |S_*|$.}
\label{fig:proof-sketch}
\end{figure}

\subsection{The Ben--Dov counterexamples and the future black-hole condition}
\label{sub:intro-bendov}

A natural concern about the formulation of Theorem~\ref{thm:SPI} is that it
appeals to the area of an outermost MOTS, since Ben--Dov \cite{ben-dov} has
constructed asymptotically flat, spherically symmetric initial data
satisfying the dominant energy condition for which the outermost MOTS
$\mathcal S$ has area exceeding $16\pi m_{ADM}^2$, so that
\[
m_{ADM}<\sqrt{\frac{|\mathcal S|}{16\pi}}.
\]
This shows that the Penrose inequality is false if formulated as a purely
initial-data inequality for the raw area of an arbitrary outermost MOTS. As
already emphasized by Frauendiener~\cite{frauendiener-PI} and
Mars~\cite{mars-PI}, the resolution is that the inequality is a
\emph{spacetime} statement, and the surface appearing in it must be
distinguished by the global causal structure, not merely by being an
outermost MOTS on some chosen Cauchy slice.

The version of the theorem proved here uses precisely such a spacetime
distinction. Each smooth leaf of the apparent-horizon tube
$\mathcal H_{\mathrm{app}}$ is required to be a \emph{future} marginally
outer trapped surface, meaning that not only does the future-directed
outgoing null expansion vanish, $\theta_{(\ell)}=0$, but also the
future-directed ingoing null expansion is strictly negative,
$\theta_{(\underline\ell)}<0$. This is the condition that makes
$\mathcal H_{\mathrm{app}}$ behave like a future black-hole horizon: it is
spacelike or null, never timelike, and supports the dynamical-horizon area
law on which the proof relies.

The exceptional outermost MOTS in Ben--Dov's construction does not satisfy
this condition. We verify in
section~\ref{sub:bendov-counterexamples} that the marginally outer tube
giving rise to it lies on a white-hole-type branch, with
$\theta_{(\underline\ell)}>0$; the smooth piece of that tube is moreover
timelike rather than spacelike or null. There is also an independent global
obstruction: the asymptotic exterior in Ben--Dov's spacetime is a
Schwarzschild region of mass strictly smaller than what the exceptional
MOTS area would require, so no global completion of his construction can
satisfy the future black-hole condition together with the area-monotonicity
across jumps and the quasi final state hypothesis used here. Both
obstructions are made precise in section~\ref{sub:bendov-counterexamples}.

Thus Ben--Dov's counterexamples and the present theorem are not in
conflict: they pertain to two different objects. Ben--Dov rules out a
naive initial-data formulation; the theorem here proves the inequality for
the genuine future apparent-horizon tube, in line with Penrose's original
spacetime heuristic.

\subsection{The equality case}

We also record what our method says in the equality case. Since the proof uses a
limiting family of tangentially maximal comparison tails, equality does not
immediately give a fixed comparison hypersurface on which the Riemannian
Penrose inequality is saturated. Nevertheless, equality has strong consequences.
It forces the limiting horizon area to agree with the initial horizon area,
\[
A_\infty=|S_*|=16\pi m_{ADM}^2,
\]
and hence there can be no area increase along the future apparent-horizon tube
from \(S_*\) to the final horizon. In particular, all allowed jumps have zero
area jump, and every smooth future black-hole-type horizon piece between
\(S_*\) and the final horizon is null rather than spacelike.

Equivalently, equality rules out the dynamical-horizon behavior responsible for
strict area growth. On smooth null pieces, the Raychaudhuri equation then gives
vanishing null shear and vanishing null matter flux. The comparison tails
produced by the TMCF construction also become asymptotically sharp for the
Riemannian Penrose inequality:
\[
m_{ADM}-m_H^{\mathrm{Riem}}(\Sigma_j)\to0.
\]
Thus the equality case is rigid at the level of the horizon and of the
comparison sequence.

We do not claim a full Schwarzschild rigidity theorem from the main hypotheses
alone. Such a statement would require an additional compactness or quantitative
stability input upgrading the asymptotically sharp sequence of comparison tails
to an exact limiting equality case for the Riemannian Penrose inequality. Under
a strong Kerr final-state assumption, however, equality forces the limiting Kerr
state to be Schwarzschild. The precise statements are given in
Subsection~\ref{sub:equality-case}.

\subsection{Organization of the paper}

The remainder of the paper is organized as follows.

Section~2 derives the TMCF equation, analyzes its geometric content, and
identifies its parabolic structure. We first compute the graph equation in a
general spacetime and then obtain the explicit backward-parabolic form in
the good gauge $b=0$, $\beta_r=0$.

Section~3 proves the existence theorem for tangentially maximal tails. We
begin by introducing the weighted parabolic H\"older spaces and the
coefficient class $\ctS_{-\tau}$. We then establish the a priori estimates
controlling the $C^0$-norm, tangential gradient, tangential Hessian, and
radial slope of solutions, prove local-in-$r$ existence from an initial
sphere, and finally carry out the global continuation and limiting argument
on late good-gauge slabs.

Section~4 formulates the quasi final state hypothesis and proves the
spacetime Penrose inequality by combining the tangentially maximal
comparison tails with the Riemannian Penrose inequality and the late-time
horizon asymptotics. We also explain, in more details than in \ref{sub:intro-bendov}, why the counterexamples of Ben--Dov to
the raw outermost-MOTS Penrose inequality are not in conflict with our
result, due to the future black-hole condition imposed on
$\mathcal H_{\mathrm{app}}$.

Appendix~\ref{app:late-slab-gauge} proves the late-slab gauge-reduction
result that allows one to pass from a general ADM chart to the good gauge
$b=0$, $\beta_r=0$ away from the horizon.

Appendix~\ref{app:horizon-penetrating-schwarzschild-model} constructs an
explicit horizon-penetrating spacelike foliation of the Schwarzschild
spacetime that agrees with the standard Schwarzschild time slices on each
fixed exterior tail and foliates the future horizon by MOTS cross-sections.
This serves as a model for the type of late-time chart envisioned in the
quasi final state hypothesis.

Appendix~\ref{app:bondi-annuli} gives a localized version of the comparison
argument on finite TMCF annuli near timelike infinity. Under an additional
Bondi--Sachs asymptotic structure assumption, this yields a Penrose bound
with the final Bondi mass in place of the ADM mass, showing in particular
that the method is compatible with Bondi mass loss.

Finally, the section ``Notation, terminology, and definitions'' collects the principal notation, terminology, and
definitions used in the paper, organized by category, for ease of reference.

\subsection*{Acknowledgements}

I would like to thank Spyros Alexakis, Stephen McCormick, Anna Sakovich, Alejandro Pe{\~n}uela D{\'i}az, and Marcus Khuri for their interest in this work, and for many helpful conversations, comments, and suggestions. Also, the author acknowledges support from Knut and Alice Wallenberg Foundation under grant KAW 2022.0285.

\section{TMCF Foliation}

\subsection{Deriving the TMCF Equation}

In this section we fix a spacetime $(\mathcal{M},\gtime)$ and derive the basic formulas
needed to formulate the existence of a TMCF foliation for a spacelike hypersurface as a PDE for a graph function
$f$. Throughout, we work on the exterior region of the spacetime whose inner boundary is the final smooth portion of the horizon; in particular,
\[
M := \R^3\setminus B_{r_0}\cong (r_0,\infty)\times S^2,
\qquad \mathcal{M}:=(\ul{T},\infty)\times M.
\]
We denote the spheres $\{r\}\times S^2$ in $M$ by $S_r$ and the spheres $\{t\}\times \{r\}\times S^2$ in $\mathcal{M}$ by $S_{t,r}$. We write the spacetime metric in coordinates
$(t,r,p) \in (\ul{T},\infty) \times (r_0,\infty) \times S^2$ as
\begin{equation}\label{eq:metric-nonstat}
\gtime
=
-\big(N^2-|\beta|_{g_t}^2\big)\,dt^2
+2\,\beta\odot dt
+\underbrace{\big(\lambda^2+|b|_\gamma^2\big)\,dr^2+2b\odot dr+\gamma}_{g_t}.
\end{equation}

Here $N$ and $\lambda$ are positive functions on $\mathcal M$, $\beta$ is a $1$-form on $\mathcal M$ with no $dt$-component, $b$ is a $1$-form on $\mathcal M$ with no $dt$- or $dr$-component, and $\gamma=\gamma(t,r)$ is a two-parameter family of metrics on $S^2$ (tangent to $S_{t,r}$). We use $\odot$ to denote the symmetrized tensor product:
$\alpha\odot\omega := \frac12(\alpha\otimes\omega+\omega\otimes\alpha)$.

\begin{remark}[Coordinate degeneration at the horizon]\label{rem:coord-degeneration}
For the purposes of deriving the TMCF equation in this section, the coefficients $N,\lambda,\beta,b,\gamma$ are only required to be defined on the open region $r>r_0$, and are allowed to degenerate as $r\to r_0^+$; we do not require the chart to extend smoothly to the inner cylinder $\{r=r_0\}$, which corresponds to the last smooth portion of the horizon. This flexibility is essential for our main examples. In the Kerr spacetime, the Boyer--Lindquist coordinates realize each constant-$t$ slice of the exterior as a tangentially maximal hypersurface (Proposition~\ref{prop:kerr}), yet they are precisely the coordinates in which the lapse $N$ vanishes and the radial lapse $\lambda$ blows up at the horizon $r=r_+$. They moreover do not penetrate the future event horizon---every constant-$t$ slice of the exterior limits to the bifurcation sphere as $r\to r_+^+$.

The situation changes when we later formulate the quasi final state hypothesis in Section~\ref{sub:spacetime-final-state}: there we will require a late-time chart that is regular up to and across the future horizon $\{r=r_0\}$, with the horizon sections appearing as coordinate spheres of the chart. We will comment on this further later on (see remark~\ref{rem:horizon-penetrating-final-state-gauge}).
\end{remark}

\medskip

The level sets $\{t=\mathrm{const}\}$ are spacelike since
\[
{\gtime}^{-1}(dt,dt)=-\frac{1}{N^2}<0,
\]
but $\partial_t$ need not be timelike since
\[
\gtime(\partial_t,\partial_t)=-(N^2-|\beta|_{g_t}^2),
\]
which may change sign.

\medskip

For each $r$, let $\slashed\nabla$, $\slashed d$, $\cancel{div}$, and $\slashed\Delta$
denote the Levi--Civita connection, exterior derivative, divergence, and Laplacian on $(S_{t,r},\gamma(t,r))$.

Whenever a $1$-form on $\mathcal{M}$ with no $dt$-component is restricted to $S_{t,r}$, we decompose it into its $dr$-component
and its tangential component. In particular we write
\begin{equation}\label{eq:theta-decomp}
\beta=\beta_r\,dr+\beta^T,
\end{equation}
where $\beta^T$ is tangential to $S_{t,r}$ (i.e.\ $\beta^T(\partial_r)=0$). Since
\[
g_t^{-1}(dr,dr)=\lambda^{-2},
\qquad
\nabla^{g_t}r=\lambda^{-2}\big(\partial_r-b^{\sharp_\gamma}\big),
\]
the outward unit normal to $S_{t,r}\subset (M_t,g_t)$ is
\begin{equation}\label{eq:nt-general}
{\bf n}_t
=
\frac{\nabla^{g_t} r}{|\nabla^{g_t} r|}
=
\frac1\lambda\big(\partial_r-b^{\sharp_\gamma}\big).
\end{equation}
It is convenient to denote the normal component of the shift by
\begin{equation}\label{eq:beta-perp}
\beta^\perp:=\beta({\bf n}_t)
=
\frac{1}{\lambda}\Big(\beta_r-\langle b,\beta^T\rangle_\gamma\Big).
\end{equation}
We also define
\[
\alpha:=N^2-|\beta|_{g_t}^2
      =N^2-|\beta^T|_\gamma^2-(\beta^\perp)^2
      =N^2-|\beta^T|_\gamma^2-\lambda^{-2}\Big(\beta_r-\langle b,\beta^T\rangle_\gamma\Big)^2.
\]

For each $t>\ul T$, let
\[
M_t:=\{t\}\times M,
\]
and let ${\bf T}_t$ be the future-pointing timelike unit normal to $M_t$ in $(\mathcal{M}, \gtime)$, given by
\[
{\bf T}_t
=
-\frac{\nabla^{(4)} t}{|\nabla^{(4)} t|}
=
\frac1N\Big(\partial_t-\beta^{\sharp_{g_t}}\Big).
\]
Let $K_t$ denote the second fundamental form of $M_t\hookrightarrow (\mathcal M,\gtime)$. We adopt the convention that for tangent vectors $X,Y$ to $M_t$,
\[
K_t(X,Y) := -\gtime(\nabla^{(4)}_X {\bf T}_t, Y),
\]
so that $K_t$ is symmetric and agrees with the usual ADM sign convention.

Let ${\bf H}_{t,r}$ be the mean curvature vector of $S_{t,r}$ in $(\mathcal{M}, \gtime)$. At each point, ${\bf T}_t$ and ${\bf n}_t$ span the normal bundle of $S_{t,r}$ in $\mathcal{M}$, so
\begin{equation}
{\bf H}_{t,r} = H_{t,r} \,{\bf n}_t - \tr_{S_{t,r}} K_t \,{\bf T}_t,
\end{equation}
where $H_{t,r}$ is the mean curvature of $S_{t,r}$ in $(M_t,g_t)$ and $\tr_{S_{t,r}} K_t$ is the $S_{t,r}$-trace of the second fundamental form of $M_t$ in $\mathcal{M}$. We call $H_{t,r}$ and $\tr_{S_{t,r}} K_t$ the spacelike and timelike mean curvature of $S_{t,r}$ with respect to the initial data set $(M_t, g_t, K_t)$. The $\gtime$ norm is
\begin{equation}
	|{\bf H}_{t,r}|^2 = \gtime({\bf H}_{t,r}, {\bf H}_{t,r}) = H_{t,r}^2 - (\tr_{S_{t,r}} K_t)^2.
\end{equation}
Note that $|{\bf H}_{t,r}|^2$ depends only on the embedding of $S_{t,r}$ in $\mathcal{M}$.

\vv

Let $f\in C^2(M)$.
The graph hypersurface
\[
M_f:=\{(t,x)\in\mathcal{M}: t=f(x)\}
\]
is spacelike provided the induced metric is Riemannian. We denote by $(g_f,K_f)$ the induced metric and second fundamental form of $M_f\hookrightarrow (\mathcal{M},\gtime)$. Equating $M_f$ with $M$, we have that $(M,g_f,K_f)$ is another initial data set on $M$. By an abuse of notation, we denote the spheres
\[
M_f \cap \Big( (\ul{T},\infty) \times \{r\}\times S^2 \Big)
\]
by $S_{f,r}$, and when working on $M_f$ we will usually suppress the subscript and simply write $S_r$.

\medskip
\noindent
\textbf{Notation for Evaluating on the Graph}.  For any tensor field $Q(t,x)$ on $\mathcal{M}$, we write $Q\big|_{t=f}$ for its pullback to $M_f$ under the graph embedding $\iota_f:M_f\to\mathcal M$. For instance, $N|_{t=f} (r,p) := N(f(r,p),r,p)$. By an abuse of notation, we will typically omit the pullback symbol $``|_{t=f}''$ and simply write $Q$ when working on $M_f$; in particular, in the formulas below, $N,\lambda,\beta,b,\gamma$ should be understood as evaluated at $t=f(x)$ unless explicitly stated otherwise.

\vv

Let $\tau:=t-f(x)$. Since $M_f=\{\tau=0\}$, the graph is spacelike if and only if
$d\tau$ is timelike, i.e.
\begin{equation}\label{eq:graph-spacelike}
{\gtime}^{-1}(d\tau,d\tau)<0.
\end{equation}
We compute
\[
{\gtime}^{-1}(d\tau,d\tau)
=
|df|_{g_t}^2
-\frac{\big(1+\langle \beta, df\rangle_{g_t}\big)^2}{N^2}.
\]
Hence $M_f$ is spacelike if and only if
\begin{equation}\label{eq:Wf}
\mu_f
:=
\big(1+\langle \beta, df\rangle_{g_t}\big)^2-N^2|df|_{g_t}^2
>0.
\end{equation}
In that case, the lapse associated to the time function $\tau=t-f(x)$ is
\begin{equation}\label{eq:Nf}
N_f
:=
\frac{1}{\sqrt{-{\gtime}^{-1}(d\tau,d\tau)}}
=
\frac{N}{\sqrt{\mu_f}}.
\end{equation}

\medskip

We define the following quantities on $M_f$:

\begin{itemize}
\item $\gamma_f(r)$ is the induced metric of $S_{r}$ in $(M_f,g_f)$.
\item $H_{f,r}$ is the mean curvature of $S_{r}$ in $(M_f, g_f)$ (i.e.\ the spacelike mean curvature of $S_r$ in $(M_f,g_f,K_f)$).
\item $\tr_{S_r} K_f$ is the $S_r$-tangential trace of $K_f$ on $(S_r, \gamma_f(r))$ (i.e.\ the timelike mean curvature of $S_r$ in $(M_f,g_f,K_f)$).
\item ${\bf n}_f$ is the unit normal vector field of $S_r$ in $(M_f,g_f)$, given by
\[
{\bf n}_f = \frac{\nabla^{g_f} r}{|\nabla^{g_f}r|}.
\]
\item ${\bf T}_f$ is the future-pointing timelike unit normal vector field of $(M_f,g_f)$ in $(\mathcal{M}, \gtime)$, given by
\[
{\bf T}_f = -\frac{\nabla^{(4)} (t-f(x))}{|\nabla^{(4)} (t-f(x))|}.
\]
\item $\beta_f$ is the $1$-form on $M_f$ defined by
\[
\beta_f:= \beta -\alpha\, df.
\]
Decompose
\[
\beta =\beta_r \,dr + \beta^T,\qquad
\beta_f = (\beta_r -\alpha \partial_r f)\,dr + (\beta^T -\alpha \slashed d f)=: (\beta_f)_r\,dr+\beta_f^T.
\]
\end{itemize}

It will be convenient to write $g_f$ in a form adapted to the foliation.
Since $dt=d\tau+df$, we obtain
\begin{align}
g_f
&=
g_t|_{t=f}
+2\,\beta\odot df
-\alpha\,df^2
\label{eq:g_f-def}\\
&=
(\lambda_f^2+|b_f|_{\gamma_f}^2)\,dr^2
+2b_f\odot dr
+\gamma_f(r),
\label{eq:g_f-split}
\end{align}
where
\begin{equation}\label{eq:bf-lambdaf}
b_f:=b+\beta_r\,\slashed d f+\partial_r f\,\beta_f^T,
\qquad
\lambda_f^2
:=
\lambda^2+|b|_\gamma^2+2\beta_r\partial_r f-\alpha(\partial_r f)^2-|b_f|_{\gamma_f}^2,
\end{equation}
and, in particular, the induced metric $\gamma_f(r)$ on $S_r$ is
\begin{equation}\label{eq:gammaf}
\gamma_f(r)
=
\gamma(t,r)|_{t=f}
+2(\beta^T\odot \slashed d f)
-\alpha\,(\slashed d f)^2.
\end{equation}
Note that, by construction,
\[
\lambda_f=|\nabla^{g_f}r|_{g_f}^{-1}.
\]

A direct computation in a $\gamma$-orthonormal frame with $\slashed d f$ parallel to the first basis vector gives
\begin{equation}\label{eq:muTf-general}
\mu_f^T
:=
\frac{\det \gamma_f}{\det \gamma}
=
\big(1+\langle \beta^T,\slashed d f\rangle_\gamma\big)^2
-
\big(N^2-(\beta^\perp)^2\big)\,|\slashed d f|_\gamma^2.
\end{equation}
It follows that
\[
\gamma_f \text{ is Riemannian }
\iff
\mu_f^T>0.
\]

By the block decomposition \eqref{eq:g_f-split}, the induced metric $g_f$ is Riemannian
if and only if $\gamma_f$ is positive definite and $\lambda_f^2>0$. Equivalently, by \eqref{eq:Wf}
\[
g_f \text{ is Riemannian }
\iff
\mu_f>0,
\]
which is just another way of saying that the graph $M_f$ is spacelike.

\vv

We compute some important quantities in the next lemma.

\begin{lem}[Basic geometric formulas]\label{lem:basic-geom}
Let $f\in C^2(M)$ be such that $M_f$ is spacelike, and define $\gamma_f$, $b_f$, and
$\lambda_f$ by \eqref{eq:bf-lambdaf}--\eqref{eq:gammaf}. Then:
\begin{enumerate}[label=(\arabic*)]
\item The unit normal to $S_r$ in $(M_f,g_f)$ is
\begin{equation}\label{eq:nf}
{\bf n}_f
=
\frac{1}{\lambda_f}\Big(\partial_r-b_f^{\sharp_{\gamma_f}}\Big),
\end{equation}
where $b_f^{\sharp_{\gamma_f}}$ denotes the $\gamma_f$-dual vector field to the
$1$-form $b_f$.

\item The spacelike mean curvature of $S_r$ in $(M_f,g_f)$ is
\begin{equation}\label{eq:Hf}
H_{f,r}
=
\frac{1}{2\lambda_f}\,
\tr_{\gamma_f}\!\Big(\partial_r\gamma_f-\mathcal L_{b_f^{\sharp_{\gamma_f}}}\gamma_f\Big).
\end{equation}

\item The timelike mean curvature of $S_r$ in $(M_f,g_f,K_f)$
(i.e.\ the tangential trace of $K_f$ along $S_r$) is
\begin{equation}\label{eq:trKf-nonstat}
\tr_{S_r}K_f
=
\frac{1}{N_f}\left(
\cancel{div}_{\gamma_f}(\beta_f^T)
+H_{f,r}\,\beta_f({\bf n}_f)
-\frac12\,\tr_{\gamma_f}\big(\partial_t\gamma_f\big)
\right),
\end{equation}
where $N_f$ is given by \eqref{eq:Nf}, and $\partial_t\gamma_f$ denotes the partial
$t$-derivative (at fixed $(r,p)$, with $f$ held fixed) of the tensor field
\[
\gamma(t,r)+2(\beta^T(t,r)\odot \slashed d f)-\alpha(t,r)(\slashed d f)^2
\]
evaluated on the graph $t=f(x)$.
\end{enumerate}
\end{lem}

\begin{proof}
Items (1) and (2) are standard formulas for level-set geometry in a foliation-adapted
decomposition of a Riemannian metric; see e.g.\ \cite{eric-3+1}.

We now prove (3). Introduce the shifted time function $\tau:=t-f(x)$ so that
$M_f=\{\tau=0\}$. Since $dt=d\tau+df$, the spacetime metric becomes
\[
\gtime
=
-\alpha\,d\tau^2 +2\,\beta_f\odot d\tau + g_f
=
-\big(N_f^2-|\beta_f|_{g_f}^2\big)\,d\tau^2+2\,\beta_f\odot d\tau+g_f.
\]
Thus, in the coordinates $(\tau,r,p)$, the slices $\{\tau=\mathrm{const}\}$ are written in
standard ADM form with induced metric $g_f$, shift $1$-form $\beta_f$, and lapse $N_f$.

The second fundamental form of the
$\tau=\mathrm{const}$ hypersurfaces is given by the standard ADM identity
\begin{equation}\label{eq:ADM-K-formula}
(K_f)_{ij}
=
\frac{1}{2N_f}
\Big(
(\mathcal L_{\beta_f^{\sharp_{g_f}}}g_f)_{ij}
-(\partial_\tau g_f)_{ij}
\Big),
\end{equation}
see e.g.\ \cite[Sec.~4.2--4.3]{eric-3+1}. Restricting \eqref{eq:ADM-K-formula} to tangential
indices $A,B$ on $S_r$ and tracing with $\gamma_f^{AB}$ yields
\[
\tr_{S_r}K_f
=
\frac{1}{2N_f}
\Big(
\tr_{\gamma_f}(\mathcal L_{\beta_f^{\sharp_{g_f}}}\gamma_f)
-
\tr_{\gamma_f}(\partial_\tau\gamma_f)
\Big).
\]
Write
\[
\beta_f^{\sharp_{g_f}}
=
X^T+\beta_f({\bf n}_f)\,{\bf n}_f,
\]
where $X^T$ is tangential to $S_r$. The $\gamma_f$-dual $1$-form of $X^T$ is exactly
$\beta_f^T$, so the usual tangential/normal decomposition gives
\[
\frac12\,\tr_{\gamma_f}(\mathcal L_{\beta_f^{\sharp_{g_f}}}\gamma_f)
=
\cancel{div}_{\gamma_f}(\beta_f^T)
+
H_{f,r}\,\beta_f({\bf n}_f).
\]
Finally, since $\partial_\tau=\partial_t$ at fixed $(r,p)$, we have
$\partial_\tau\gamma_f=\partial_t\gamma_f$. Substituting these identities yields \eqref{eq:trKf-nonstat}.
\end{proof}

\medskip

It follows from \eqref{eq:trKf-nonstat} that $(M_f,g_f, K_f)$ is tangentially maximal in $(\mathcal{M}, \gtime)$ with TMCF foliation $\{S_{f,r}\}_{r> r_0}$ if and only if $f$ satisfies the following PDE:
\begin{equation}\label{eq:TMCF-equation}
\boxed{
\cancel{div}_{\gamma_f}\Big( \,\beta_f^T\Big) \;+\; H_{f,r}\, \,\beta_f({\bf n}_f)\;-\;\frac12 \tr_{\gamma_f}\big(\partial_t\gamma_f\big)=0.
	}
\end{equation}

We have reduced the search of a tangentially maximal initial data set $(M,g_f,K_f)$ in $(\mathcal{M}, \gtime)$ to a PDE for the function $f$, which we call the \emph{TMCF equation}.

In the next section we show that, in the gauge $b=0$ and $\beta_r=0$ that we call {\it the good gauge}, and under reasonable conditions such as mean-convexity and spacelikeness, \eqref{eq:TMCF-equation} can be rewritten as a quasilinear backward-parabolic evolution equation in the radial variable $r$.

\medskip

\begin{remark}
Suppose that in the coordinates $(t,r,p)$ the spacetime is \emph{stationary}, i.e.\ the coefficients
$N,\lambda,\beta,b,\gamma$ are independent of $t$. Then $\partial_t\gamma_f\equiv 0$, and the TMCF equation reduces to
\[
{\cancel{div}}_{\gamma_f}\big( \,\beta_f^T\big)
\;+\; H_{f,r}\,\,\beta_f({\bf n}_f)=0.
\]

Moreover, the equation does not depend on $f$ itself but only on its derivatives, since the spacetime coefficients evaluated on the graph are independent of $f$, i.e.
\[
(N,\lambda,\beta,b,\gamma)|_{t=f}=(N,\lambda,\beta,b,\gamma),
\]
\end{remark}

\medskip

A nontrivial example of tangentially maximal initial data sets appears naturally in Kerr spacetimes:

\begin{prop}\label{prop:kerr}
Fix parameters $m>0$ and $a\in\R$ with $|a|<m$, and consider the subextremal Kerr spacetime
$(\mathcal M_{m,a},\gtime_{m,a})$ in Boyer--Lindquist coordinates
$(t,r,\vartheta,\varphi)$ on the exterior region $r>r_+$, where
\[
r_+ := m+\sqrt{m^2-a^2}.
\]
Then each hypersurface $\{t=\mathrm{constant},\,r>r_+\}$ is tangentially maximal. More precisely,
the coordinate spheres $\{S_r\}_{r>r_+}$ form a TMCF foliation of each constant-$t$ exterior slice.
Equivalently, every constant graph $f\equiv T_0$, and in particular $f\equiv 0$, solves the TMCF
equation on the Kerr exterior.
\end{prop}

\begin{proof}
In Boyer--Lindquist coordinates
\[
(t,r,\vartheta,\varphi)
\in
\R\times(r_+,\infty)\times(0,\pi)\times(0,2\pi),
\qquad
r_+:=m+\sqrt{m^2-a^2},
\]
the Kerr metric is
\[
\gtime_{m,a}
=
-\left(1-\frac{2mr}{\rho^2}\right)dt^2
-\frac{4mar\sin^2\vartheta}{\rho^2}\,dt\,d\varphi
+\frac{\rho^2}{\Delta}\,dr^2
+\rho^2\,d\vartheta^2
+\frac{\Lambda}{\rho^2}\sin^2\vartheta\,d\varphi^2,
\]
where
\[
\rho^2:=r^2+a^2\cos^2\vartheta,
\qquad
\Delta:=r^2-2mr+a^2,
\qquad
\Lambda:=(r^2+a^2)^2-a^2\Delta\sin^2\vartheta.
\]
See, for example, \cite{BoyerLindquist1967,ONeillKerr}.

Let $(M,g_{m,a},K_{m,a})$ be a constant-$t$ exterior slice $\{t=\mathrm{constant},\,r>r_+\}$. The
metric can be written as the special case $b_{m,a}=0$ of the ADM form used in this subsection:
\[
\gtime_{m,a}
=
-\big(N_{m,a}^2-|\beta_{m,a}|_{g_{m,a}}^2\big)\,dt^2
+2\,\beta_{m,a}\odot dt
+
\underbrace{\lambda_{m,a}^2\,dr^2+\gamma_{m,a}(r)}_{=:g_{m,a}},
\]
where
\begin{align}
N_{m,a}^2(r,\vartheta)
&:=
\frac{\rho^2\Delta}{\Lambda},
\\
\beta_{m,a}
&:=
-\frac{2mar\sin^2\vartheta}{\rho^2}\,d\varphi,
\\
\lambda_{m,a}^2(r,\vartheta)
&:=
\frac{\rho^2}{\Delta},
\\
\gamma_{m,a}(r)
&:=
\rho^2\,d\vartheta^2
+
\frac{\Lambda}{\rho^2}\sin^2\vartheta\,d\varphi^2.
\end{align}
One checks that
\[
N_{m,a}^2-|\beta_{m,a}|_{g_{m,a}}^2
=
1-\frac{2mr}{\rho^2},
\]
so this reproduces the usual Boyer--Lindquist form.

We now verify the TMCF condition. Since the spacetime is stationary,
\[
\partial_t\gamma_{m,a}=0.
\]
Moreover, $\beta_{m,a}$ has no $dr$-component, so it is tangential to the coordinate spheres and
\[
(\beta_{m,a})_r=0,
\qquad
\beta_{m,a}({\bf n})=0,
\]
where ${\bf n}$ is the outward unit normal to $S_r\subset(M,g_{m,a})$. Therefore
\eqref{eq:trKf-nonstat}, applied with $f\equiv T_0$, reduces to
\[
\tr_{S_r}K_{m,a}
=
\frac{1}{N_{m,a}}\,
\cancel{div}_{\gamma_{m,a}(r)}(\beta_{m,a}).
\]
Now
\[
(\beta_{m,a})^{\sharp_{\gamma_{m,a}}}
=
-\frac{2mar}{\Lambda}\,\partial_\varphi.
\]
Since the coefficients of $\gamma_{m,a}(r)$ are independent of $\varphi$, we have
\[
\cancel{div}_{\gamma_{m,a}(r)}(\beta_{m,a})
=
\frac{1}{\sqrt{\det\gamma_{m,a}(r)}}
\partial_\varphi\left(
\sqrt{\det\gamma_{m,a}(r)}
\left(-\frac{2mar}{\Lambda}\right)
\right)
=0.
\]
Hence $\tr_{S_r}K_{m,a}=0$ for every $r>r_+$. Therefore the coordinate spheres
$\{S_r\}_{r>r_+}$ form a TMCF foliation, and $(M,g_{m,a},K_{m,a})$ is tangentially maximal.
Equivalently, every constant graph solves \eqref{eq:TMCF-equation}.
\end{proof}

\begin{remark}[Further examples]\label{rem:other-examples}
Kerr is by no means the only spacetime whose exterior admits a natural tangentially maximal foliation. The same argument applies verbatim to the Kerr--Newman family of charged rotating black holes in Boyer--Lindquist-type coordinates: the shift has no radial component, the spatial metric is already of the form $\lambda^2 dr^2 + \gamma$, and the coefficients of $\gamma$ are independent of the azimuthal coordinate, so the divergence of $\beta^{\sharp_\gamma}$ along each coordinate sphere vanishes identically. More generally, any stationary axisymmetric exterior that admits coordinates $(t,r,\vartheta,\varphi)$ in which the spatial slices are orthogonal to the rotation Killing field and the shift is purely azimuthal falls into this framework. In the spherically symmetric subcase, such as for Schwarzschild and Reissner--Nordstr\"om in standard coordinates, the shift vanishes altogether and $K\equiv 0$ on the standard slices, so the tangentially maximal condition holds trivially; this is a special case of Example~\ref{ex:timesym-intro} in the introduction.
\end{remark}

\subsection{Parabolic Structure of the TMCF Equation and the Good Gauge}

In this subsection we now analyze the TMCF equation as a nonlinear PDE for a graph function
\[
f=f(r,p),
\qquad
(r,p)\in M\cong (r_0,\infty)\times S^2.
\]
We first identify its principal part and show that, under natural admissibility assumptions, the
linearized operator at a solution has the parabolic sign corresponding to forward parabolicity in
the inward radial direction. We then impose the good gauge
\[
b=0,
\qquad
\beta_r=0,
\]
in which the equation simplifies substantially and can be rewritten explicitly in quasilinear
parabolic form. Later, this good gauge will be produced on sufficiently late slabs away from the
horizon by the coordinate reduction proved in appendix~\ref{app:late-slab-gauge}.

\vv

Let $(\mathcal{M},\gtime)$ be an asymptotically flat spacetime determined by the parameters
\[
\mathcal{S}=(N,\lambda,\beta,b,\gamma),
\]
so that the spacetime metric takes the general ADM form
\[
\gtime
=
-\big( \underbrace{N^2-|\beta|_{g_t}^2}_{=:\alpha}\big)\,dt^2
+2\,\beta\odot dt
+\underbrace{\big(\lambda^2+|b|_\gamma^2\big)\,dr^2+2b\odot dr+\gamma}_{=:g_t}.
\]
We treat the background coefficients $\mathcal S$ as fixed, and study the \emph{TMCF equation} as
a nonlinear PDE for $f$.

\vv

Given a function $f:M\to \R$ for which the graph
\[
M_f=\{t=f(x)\}
\]
is spacelike, we evaluate the parameters $\mathcal S$ on $M_f$, which we denote by
\[
(N|_{t=f},\lambda|_{t=f},\beta|_{t=f},b|_{t=f},\gamma|_{t=f}).
\]
We will drop the pullback notation $|_{t=f}$ whenever it is clear from context.

Recall from the previous subsection the following quantities on $M_f$:
\begin{align}
\alpha &:=N^2-|\beta|_{g_t}^2
      =N^2-|\beta^T|_\gamma^2-(\beta^\perp)^2,
\\
\beta_f&:=\beta-\alpha\,df,
\qquad
\beta_f^T=\beta^T-\alpha\,\slashed d f,
\\
b_f&:=b+\beta_r\,\slashed d f+f_r\,\beta_f^T,
\\
\gamma_f(r)&=\gamma+2\,\beta^T\odot \slashed d f-\alpha\,(\slashed d f)^2,
\\
\mu_f^T
&:=\frac{\det\gamma_f}{\det\gamma}
=
\big(1+\langle \beta^T,\slashed d f\rangle_\gamma\big)^2
-
\big(N^2-(\beta^\perp)^2\big)\,|\slashed d f|_\gamma^2,
\\
\lambda_f^2&=\lambda^2+|b|_\gamma^2+2\beta_r f_r-\alpha f_r^2-|b_f|_{\gamma_f}^2,
\\
\mu_f
&:=
\big(1+\langle \beta, df\rangle_{g_t}\big)^2-N^2|df|_{g_t}^2,
\qquad
N_f:=\frac{N}{\sqrt{\mu_f}},
\\
g_f&=(\lambda_f^2+|b_f|_{\gamma_f}^2)\,dr^2+2b_f\odot dr+\gamma_f(r),
\\
H_{f,r}&=\frac{1}{2\lambda_f}\,
\tr_{\gamma_f}\!\Big(\partial_r\gamma_f-\mathcal L_{b_f^{\sharp_{\gamma_f}}}\gamma_f\Big),
\\
{\bf n}_f&=\frac{1}{\lambda_f}\Big(\partial_r-b_f^{\sharp_{\gamma_f}}\Big).
\end{align}
Here $\beta^\perp=\beta({\bf n}_t)$ denotes the normal component of the shift, as in
\eqref{eq:beta-perp}.

In particular, $g_f$ is the induced metric on $M_f$, $\gamma_f$ is the induced metric on
$S_r\subset M_f$, $H_{f,r}$ is the mean curvature of $S_r\subset (M_f,g_f)$, and ${\bf n}_f$ is the
outward unit normal to $S_r\subset (M_f,g_f)$.

\medskip

The TMCF equation is
\begin{equation}\label{eq:TMCF-nonlinear}
\mathcal F[f]:=
\cancel{div}_{\gamma_f}(\beta_f^T)
+
H_{f,r}\,\beta_f({\bf n}_f)
-\frac12\,\tr_{\gamma_f}(\partial_t\gamma_f)
=0
\qquad\text{on } M.
\end{equation}
Observe that the geometric quantities above are defined only when $M_f$ is spacelike; in
particular, $\lambda_f$, $\tr_{\gamma_f}$, $\sharp_{\gamma_f}$, and
$\cancel{div}_{\gamma_f}$ require $g_f$ to be Riemannian. Analytically, however, the explicit
good-gauge equation derived below only requires the leaf metric $\gamma_f$ to be Riemannian, not
the full graph metric $g_f$.

\vv

We now characterize the PDE-type and the principal part of $\mathcal F$. At first sight, since
$H_{f,r}$ involves $\partial_r\gamma_f$, one might expect mixed derivatives
$\slashed\nabla(f_r)$ to appear. A crucial cancellation occurs between $\partial_r\gamma_f$ and
$\mathcal L_{b_f^{\sharp_{\gamma_f}}}\gamma_f$: after expansion, all terms containing
$\slashed\nabla(f_r)$ cancel identically.

\begin{lem}[Principal coefficients and linearized parabolic structure]\label{lem:PDEtype}
The TMCF equation \eqref{eq:TMCF-nonlinear} has the schematic form
\begin{equation}\label{eq:Fschematic}
\mathcal F\Big(r,p,\ f,\ f_r,\ \slashed\nabla f,\ \slashed\nabla^2 f;\
\mathcal S|_{t=f},\ \partial_t\mathcal S|_{t=f},\ \partial_r\mathcal S|_{t=f},\
\slashed\nabla\mathcal S|_{t=f}\Big)=0,
\end{equation}
where the background coefficients and their first derivatives are evaluated on the graph $t=f(x)$.

Furthermore, letting $N_f$ denote the lapse associated to $\tau=t-f(x)$, the principal coefficients
are
\begin{align}
\mathcal F_{\slashed\nabla^2 f}&:= \frac{\partial \mathcal F}{\partial \slashed \nabla^2 f}
=
-\,N_f^2\,\gamma_f^{-1},
\label{eq:F-Hess}
\\
\mathcal F_{f_r}&:= \frac{\partial \mathcal F}{\partial f_r}
=
-\frac{H_{f,r}}{\lambda_f}\,N_f^2
-\frac{\beta_f({\bf n}_f)}{\lambda_f}\,\mathcal F[f].
\label{eq:F-fr}
\end{align}
Here $\mathcal F_{\slashed\nabla^2 f}$ is viewed as a $(2,0)$-tensor. In particular, if $f$ is a
solution of \eqref{eq:TMCF-nonlinear}, then
\begin{equation}\label{eq:F-fr-onshell}
\left.\mathcal F_{f_r}\right|_{\mathcal F=0}
=
-\frac{H_{f,r}}{\lambda_f}\,N_f^2.
\end{equation}
Consequently, if $f$ is a solution of \eqref{eq:TMCF-nonlinear}, then the linearized operator
$D\mathcal F[f]$ acting on a variation $u$ has the form
\begin{equation}\label{eq:linF}
D\mathcal F[f](u)
=
(\mathcal F_{f_r})\,u_r
+
\mathcal F_{\slashed\nabla^2 f}:\slashed\nabla^2_{\gamma_f}u
=
-\frac{H_{f,r}}{\lambda_f}\,N_f^2\,u_r
-
N_f^2\,\gamma_f^{-1}:\slashed\nabla^2_{\gamma_f}u
+
\text{lower order terms}.
\end{equation}
Here $:$ denotes full contraction; for example,
$\gamma_f^{-1}:\slashed\nabla^2_{\gamma_f}u=\slashed\Delta_{\gamma_f}u$.
\end{lem}

\begin{proof}
We first show that the mixed derivatives $\slashed\nabla(f_r)$ cancel. We observe that only the
divergence term and the mean-curvature term can contribute second derivatives of $f$. Since
\[
\gamma_f=\gamma+2\beta^T\odot \slashed d f-\alpha\,(\slashed d f)^2,
\]
the $\slashed\nabla(f_r)$-part of $\partial_r\gamma_f$ is
\[
2\,\beta_f^T\odot \slashed\nabla f_r.
\]
On the other hand, since
\[
b_f=b+\beta_r\,\slashed d f+f_r\,\beta_f^T,
\]
the $\slashed\nabla(f_r)$-part of $\mathcal L_{b_f^{\sharp_{\gamma_f}}}\gamma_f$ is the same tensor
\[
2\,\beta_f^T\odot \slashed\nabla f_r.
\]
Hence the mixed derivatives cancel identically in
\[
\partial_r\gamma_f-\mathcal L_{b_f^{\sharp_{\gamma_f}}}\gamma_f,
\]
and therefore in $H_{f,r}$.

Next, we compute the tangential Hessian coefficient. Viewing
$\mathcal F_{\slashed\nabla^2 f}$ as the coefficient tensor obtained by differentiating with
respect to the tangential second derivatives of $f$, the divergence term contributes
\[
\big(\cancel{div}_{\gamma_f}(\beta_f^T)\big)_{\slashed\nabla^2 f}
=
-\big(\alpha+|\beta_f^T|_{\gamma_f}^2\big)\,\gamma_f^{-1}.
\]
The mean-curvature term contributes
\[
\big(H_{f,r}\,\beta_f({\bf n}_f)\big)_{\slashed\nabla^2 f}
=
-\beta_f({\bf n}_f)^2\,\gamma_f^{-1}.
\]
Therefore
\[
\mathcal F_{\slashed\nabla^2 f}
=
-\Big(\alpha+|\beta_f^T|_{\gamma_f}^2+\beta_f({\bf n}_f)^2\Big)\gamma_f^{-1}.
\]

Now let $\tau=t-f(x)$. In the coordinates $(\tau,r,p)$, the metric can be written as
\[
\gtime
=
-\alpha_\tau\,d\tau^2+2\,\beta_\tau\odot d\tau+g_\tau
=
-\big(N_\tau^2-|\beta_\tau|_{g_\tau}^2\big)\,d\tau^2
+2\,\beta_\tau\odot d\tau
+g_\tau,
\]
where
\[
\alpha_\tau:=\alpha|_{t=\tau+f},
\qquad
\beta_\tau:= (\beta-\alpha\,df)\big|_{t=\tau+f},
\qquad
g_\tau:= \big(g_t+2\,\beta\odot df-\alpha\,df^2\big)\big|_{t=\tau+f}.
\]
In particular, at $\tau=0$,
\[
\alpha_0=\alpha,
\qquad
\beta_0=\beta_f,
\qquad
g_0=g_f,
\qquad
N_0=N_f.
\]
Hence, on the graph,
\[
N_f^2=\alpha+|\beta_f|_{g_f}^2
      =\alpha+|\beta_f^T|_{\gamma_f}^2+\beta_f({\bf n}_f)^2.
\]
This proves \eqref{eq:F-Hess}.

For the $f_r$-coefficient, set
\[
H:=H_{f,r},
\qquad
B:=\beta_f({\bf n}_f),
\qquad
D:=\cancel{div}_{\gamma_f}(\beta_f^T),
\qquad
T:=\frac12\tr_{\gamma_f}(\partial_t\gamma_f),
\]
so that $\mathcal F=D+HB-T$. Since $\gamma_f$ and $\beta_f^T$ are independent of $f_r$, only
$H$ and $B$ contribute:
\[
\mathcal F_{f_r}=(\partial_{f_r}H)\,B+H\,(\partial_{f_r}B).
\]
A direct differentiation of the formulas above gives
\[
\partial_{f_r}\lambda_f=B,
\qquad
\partial_{f_r}B=-\frac{N_f^2}{\lambda_f},
\qquad
\partial_{f_r}H=\frac{T-D-HB}{\lambda_f}.
\]
Substituting these identities yields
\[
\mathcal F_{f_r}
=
\frac{B}{\lambda_f}(T-D-HB)-\frac{H}{\lambda_f}N_f^2
=
-\frac{H}{\lambda_f}N_f^2-\frac{B}{\lambda_f}\,\mathcal F[f],
\]
which is \eqref{eq:F-fr}. The on-shell formula \eqref{eq:F-fr-onshell} follows immediately, and
\eqref{eq:linF} then follows from \eqref{eq:F-Hess} and \eqref{eq:F-fr-onshell}.
\end{proof}

\begin{cor}\label{cor:general-parabolicity}
Suppose $f$ is a solution of \eqref{eq:TMCF-nonlinear} for which $M_f$ is spacelike and
$H_{f,r}>0$. Then, locally, the implicit function theorem rewrites the equation in the form
\begin{equation}\label{eq:local-quasilinear-form}
f_r = \Phi \Big(r,p,\ f,\ \slashed\nabla f,\ \slashed\nabla^2 f;\
\mathcal S|_{t=f},\ \partial_t\mathcal S|_{t=f},\ \partial_r\mathcal S|_{t=f},\
\slashed\nabla\mathcal S|_{t=f}\Big),
\end{equation}
where $\Phi$ is smooth in its arguments near the given solution.
\end{cor}

\begin{proof}
This is immediate from Lemma~\ref{lem:PDEtype} and the implicit function theorem, since
\eqref{eq:F-fr-onshell} shows that
\[
\mathcal F_{f_r}=-\frac{H_{f,r}}{\lambda_f}\,N_f^2<0
\]
along the solution.
\end{proof}

We emphasize, however, that the function $\Phi$ is only known to exist via the implicit function
theorem and is not known explicitly. We therefore now impose the good gauge
\[
b=0,
\qquad
\beta_r=0,
\]
in which the algebra simplifies and one obtains an explicit expression of the TMCF equation in the
form
\[
f_r + \mathcal{A}(r,p,f,\slashed\nabla f):\slashed\nabla_\gamma^2 f
=
\mathcal{B}(r,p,f,\slashed\nabla f),
\]
where $:$ denotes the full contraction of a $(2,0)$-tensor with a $(0,2)$-tensor. This expression will not require $M_f$ to be spacelike and will in fact describe the tangentially maximal condition in general without restricting its causal structure.

\vv

\subsubsection*{The good gauge}

From now on, we impose
\begin{equation}\label{eq:good-gauge}
b\equiv 0,
\qquad
\beta_r\equiv 0
\quad \text{on } \mathcal{M}.
\end{equation}
Then $g_t=\lambda^2dr^2+\gamma$, $\beta=\beta^T$ is tangential to the spheres $S_{t,r}$, and we
have the following simplifications:
\begin{equation}
\alpha=N^2-|\beta^T|_\gamma^2,
\qquad
b_f=f_r\,\beta_f^T,
\qquad
\lambda_f^2=\lambda^2-\Big(\alpha+|\beta_f^T|_{\gamma_f}^2\Big)f_r^2.
\end{equation}
Moreover,
\begin{equation}
\mu_f^T
=
\big(1+\langle \beta^T,\slashed d f\rangle_\gamma\big)^2
-
N^2|\slashed d f|_\gamma^2,
\qquad
\mu_f
=
\mu_f^T-\frac{N^2}{\lambda^2}f_r^2.
\end{equation}

\vv

\begin{lem}[Elementary identities in the good gauge]
\label{lem:beta-r-gauge-identities}
Assume the good gauge in \eqref{eq:good-gauge}. Then the following identities hold for every
spacelike graph $M_f$:
\begin{enumerate}[label=\textup{(\arabic*)}]
\item
\[
\det\gamma_f=\mu_f^T\,\det\gamma,
\qquad
\det g_f=\lambda^2\,\mu_f\,\det\gamma.
\]
In particular,
\[
\gamma_f \text{ is Riemannian }\iff \mu_f^T>0,
\qquad
g_f \text{ is Riemannian }\iff \mu_f>0.
\]

\item The pointwise bound
\begin{equation}\label{eq:muT-sufficient}
|\slashed d f|_\gamma<\frac{1}{N+|\beta^T|_\gamma}
\end{equation}
is sufficient to ensure that $\mu_f^T>0$, and hence that $\gamma_f$ is Riemannian.

\item
\begin{equation}\label{eq:gammaf-inverse}
\gamma_f^{-1}
=
\gamma^{-1}
+\frac{1}{\mu_f^T}\Big(
N^2\big((\slashed d f)^{\sharp_\gamma}\big)^2
-2\big(1+\langle \beta^T,\slashed d f\rangle_\gamma\big)\,
(\slashed d f)^{\sharp_\gamma}\odot(\beta^T)^{\sharp_\gamma}
+|\slashed d f|_\gamma^2\big((\beta^T)^{\sharp_\gamma}\big)^2
\Big),
\end{equation}
where $X^2:=X\otimes X$.

\item
\begin{equation}\label{eq:betafT-sharp-gammaf}
(\beta_f^T)^{\sharp_{\gamma_f}}
=
\frac{
\big(1+\langle \beta^T,\slashed d f\rangle_\gamma\big)(\beta^T)^{\sharp_\gamma}
-
N^2(\slashed d f)^{\sharp_\gamma}
}{\mu_f^T}.
\end{equation}
Consequently,
\begin{equation}\label{eq:betafT-norm-id}
\alpha+|\beta_f^T|_{\gamma_f}^2=\frac{N^2}{\mu_f^T},
\qquad
N_f^2=\frac{N^2}{\mu_f}.
\end{equation}

\item
\begin{equation}\label{eq:lambdaf-mu-muT}
\lambda_f^2=\lambda^2\,\frac{\mu_f}{\mu_f^T}.
\end{equation}

\item
\begin{equation}\label{eq:betaf-nf-simplified}
\beta_f({\bf n}_f)
=
-\frac{N^2}{\mu_f^T}\,\frac{f_r}{\lambda_f}.
\end{equation}

\item Given any $1$-form $\omega$ on $S_r$,
\begin{equation}\label{eq:div-gammab}
\cancel{div}_{\gamma_f}(\omega)
=
\frac{1}{\sqrt{\mu_f^T}}\,
\cancel{div}_{\gamma}\!\Big(\sqrt{\mu_f^T}\,\omega^{\sharp_{\gamma_f}}\Big).
\end{equation}
In particular,
\begin{equation}\label{eq:div-betafT-explicit}
\cancel{div}_{\gamma_f}(\beta_f^T)
=
\frac{1}{\sqrt{\mu_f^T}}\,
\cancel{div}_{\gamma}\!\left(
\frac{
\big(1+\langle \beta^T,\slashed d f\rangle_\gamma\big)(\beta^T)^{\sharp_\gamma}
-
N^2(\slashed d f)^{\sharp_\gamma}
}{\sqrt{\mu_f^T}}
\right).
\end{equation}
\end{enumerate}
\end{lem}

\begin{proof}
Fix a point and choose a $\gamma$-orthonormal coframe $(e^1,e^2)$ such that
\[
\slashed d f=s\,e^1
\]
at that point. Write
\[
\beta^T=\beta_1\,e^1+\beta_2\,e^2.
\]
Then
\[
\gamma_f=
\begin{pmatrix}
1+2\beta_1s-\alpha s^2 & \beta_2 s\\
\beta_2 s & 1
\end{pmatrix}.
\]
Using $\alpha=N^2-\beta_1^2-\beta_2^2$, we compute
\[
\det\gamma_f
=
(1+\beta_1 s)^2-N^2 s^2
=
\big(1+\langle \beta^T,\slashed d f\rangle_\gamma\big)^2
-
N^2|\slashed d f|_\gamma^2
=
\mu_f^T.
\]
Since the lower-right entry is $1$, this also shows that $\gamma_f$ is positive definite if and
only if $\mu_f^T>0$.

Moreover, by Cauchy--Schwarz,
\[
\mu_f^T
=
\big(1+\langle \beta^T,\slashed d f\rangle_\gamma\big)^2
-
N^2|\slashed d f|_\gamma^2
\ge
\big(1-|\beta^T|_\gamma\,|\slashed d f|_\gamma\big)^2
-
N^2|\slashed d f|_\gamma^2.
\]
Thus \eqref{eq:muT-sufficient} implies
\[
1-|\beta^T|_\gamma\,|\slashed d f|_\gamma
>
N\,|\slashed d f|_\gamma,
\]
hence $\mu_f^T>0$. This proves \textup{(2)}.

Direct inversion of the above $2\times2$ matrix yields \eqref{eq:gammaf-inverse}. Applying this
inverse to
\[
\beta_f^T=\beta^T-\alpha\,\slashed d f
\]
gives \eqref{eq:betafT-sharp-gammaf}. Taking the $\gamma_f$-norm of $\beta_f^T$ then yields
\[
\alpha+|\beta_f^T|_{\gamma_f}^2=\frac{N^2}{\mu_f^T}.
\]
Next,
\[
\lambda_f^2
=
\lambda^2-\Big(\alpha+|\beta_f^T|_{\gamma_f}^2\Big)f_r^2
=
\lambda^2-\frac{N^2}{\mu_f^T}f_r^2
=
\lambda^2\,\frac{\mu_f}{\mu_f^T},
\]
which proves \eqref{eq:lambdaf-mu-muT}. Since $g_f$ has the block form
\[
g_f=(\lambda_f^2+|b_f|_{\gamma_f}^2)\,dr^2+2b_f\odot dr+\gamma_f,
\]
we have $\det g_f=\lambda_f^2\det\gamma_f$, and therefore
\[
\det g_f=\lambda^2\mu_f\det\gamma.
\]
This proves \textup{(1)}.

Also,
\[
\beta_f({\bf n}_f)
=
\frac{1}{\lambda_f}\Big((\beta_f)_r-\langle \beta_f^T,b_f\rangle_{\gamma_f}\Big)
=
-\frac{f_r}{\lambda_f}\Big(\alpha+|\beta_f^T|_{\gamma_f}^2\Big),
\]
and therefore \eqref{eq:betaf-nf-simplified} follows from \eqref{eq:betafT-norm-id}.

Finally, from \eqref{eq:Wf} and \eqref{eq:good-gauge},
\[
\mu_f
=
\big(1+\langle \beta^T,\slashed d f\rangle_\gamma\big)^2
-
N^2|\slashed d f|_\gamma^2
-
\frac{N^2}{\lambda^2}f_r^2,
\]
so $N_f^2=N^2/\mu_f$.

The divergence identity \eqref{eq:div-gammab} follows from the ratio of volume forms
\[
d\mu_{\gamma_f}=\sqrt{\mu_f^T}\,d\mu_\gamma.
\]
Substituting \eqref{eq:betafT-sharp-gammaf} into \eqref{eq:div-gammab} gives
\eqref{eq:div-betafT-explicit}.
\end{proof}

\vv

We next rewrite $H_{f,r}/\lambda_f$ in a form adapted to the good gauge.
In a general gauge this quantity can depend on the evolution derivative
$f_r$.  The good-gauge identities remove this dependence: after using
the TMCF equation, $H_{f,r}/\lambda_f$ becomes a function only of the background coefficients, $f$, and the angular first derivatives
$\slashed d f$.  This is crucial because this scalar is the coefficient
which controls the parabolicity of the explicit TMCF equation.

 For any $1$-form $\omega$ on $S_r$, define
\begin{equation}\label{eq:QrQt-gauge}
Q_r(\omega):=\partial_r\gamma+2\,\partial_r\beta^T\odot \omega-\partial_r\alpha\,\omega^2,
\qquad
Q_t(\omega):=\partial_t\gamma+2\,\partial_t\beta^T\odot \omega-\partial_t\alpha\,\omega^2.
\end{equation}

\begin{lem}[Off-shell and on-shell formulas for $H_{f,r}/\lambda_f$]
\label{lem:H-over-lambda}
Assume the good gauge in \eqref{eq:good-gauge}. Then for every $f\in C^2(M)$ for which $g_f$ is
Riemannian,
\begin{equation}\label{eq:H-over-lambda-offshell}
\frac{H_{f,r}}{\lambda_f}
=
\frac{1}{2\lambda_f^2}
\left(
\tr_{\gamma_f}Q_r(\slashed d f)
+
f_r\,\tr_{\gamma_f}Q_t(\slashed d f)
-
2f_r\,\cancel{div}_{\gamma_f}(\beta_f^T)
\right).
\end{equation}
If $f$ moreover solves the TMCF equation \eqref{eq:TMCF-nonlinear}, then
\begin{equation}\label{eq:H-over-lambda-onshell}
\frac{H_{f,r}}{\lambda_f}
=
\frac{1}{2\lambda^2}\,\tr_{\gamma_f}Q_r(\slashed d f).
\end{equation}
\end{lem}

\begin{proof}
Recall that $\gamma_f$ satisfies
\[
\gamma_f=\gamma+2\beta^T\odot \slashed d f-\alpha(\slashed d f)^2.
\]
Differentiating along the graph $t=f(r,p)$ gives
\[
\partial_r\gamma_f
=
Q_r(\slashed d f)+f_r\,Q_t(\slashed d f)+2\,\beta_f^T\odot \slashed\nabla f_r.
\]
On the other hand, in the good gauge $b_f=f_r\beta_f^T$, so
\[
\mathcal L_{b_f^{\sharp_{\gamma_f}}}\gamma_f
=
\mathcal L_{f_r(\beta_f^T)^{\sharp_{\gamma_f}}}\gamma_f
=
f_r\,\mathcal L_{(\beta_f^T)^{\sharp_{\gamma_f}}}\gamma_f
+
2\,\beta_f^T\odot \slashed\nabla f_r.
\]
Subtracting, the mixed derivatives cancel:
\[
\partial_r\gamma_f-\mathcal L_{b_f^{\sharp_{\gamma_f}}}\gamma_f
=
Q_r(\slashed d f)+f_r\,Q_t(\slashed d f)-f_r\,\mathcal L_{(\beta_f^T)^{\sharp_{\gamma_f}}}\gamma_f.
\]
Taking the $\gamma_f$-trace and using
\[
\tr_{\gamma_f}\!\Big(\mathcal L_{(\beta_f^T)^{\sharp_{\gamma_f}}}\gamma_f\Big)
=
2\,\cancel{div}_{\gamma_f}(\beta_f^T)
\]
gives \eqref{eq:H-over-lambda-offshell}.

Now assume $f$ solves the TMCF equation. Set
\[
D:=\cancel{div}_{\gamma_f}(\beta_f^T),
\qquad
T:=\frac12\tr_{\gamma_f}Q_t(\slashed d f),
\qquad
B:=\beta_f({\bf n}_f).
\]
Then the TMCF equation is $D+H_{f,r}B-T=0$, that is,
\[
D=T-H_{f,r}B.
\]
Substituting this into \eqref{eq:H-over-lambda-offshell} yields
\[
2\lambda_f H_{f,r}
=
\tr_{\gamma_f}Q_r(\slashed d f)+2f_rH_{f,r}B.
\]
Hence
\[
2H_{f,r}\big(\lambda_f-f_rB\big)=\tr_{\gamma_f}Q_r(\slashed d f).
\]
Using \eqref{eq:betaf-nf-simplified} and \eqref{eq:lambdaf-mu-muT},
\[
\lambda_f-f_rB
=
\lambda_f+\frac{N^2}{\mu_f^T}\frac{f_r^2}{\lambda_f}
=
\frac{\lambda_f^2+\frac{N^2}{\mu_f^T}f_r^2}{\lambda_f}
=
\frac{\lambda^2}{\lambda_f}.
\]
Substituting this identity gives \eqref{eq:H-over-lambda-onshell}.
\end{proof}

\vv

It is convenient to introduce the scalar
\begin{equation}\label{eq:hhf-def}
\mathfrak a_f
:=
\frac{1}{2\lambda^2}\,\tr_{\gamma_f}Q_r(\slashed d f).
\end{equation}
By Lemma~\ref{lem:H-over-lambda}, this agrees with $H_{f,r}/\lambda_f$ along any solution of the
TMCF equation.

We now write the TMCF equation explicitly in the good gauge.

\begin{prop}[Explicit quasilinear form in the good gauge]
\label{prop:gammab-quasilinear}
Assume the good gauge in \eqref{eq:good-gauge}, and let $f\in C^2(M)$ satisfy $\mu_f>0$. On any
region where $\mathfrak a_f\neq 0$, the TMCF equation \eqref{eq:TMCF-nonlinear} is equivalent to
\begin{equation}\label{eq:TMCF-gammab}
f_r
+
\mathcal A(r,p,f,\slashed\nabla f):\slashed\nabla_\gamma^2 f
=
\mathcal B(r,p,f,\slashed\nabla f),
\end{equation}
where
\begin{equation}\label{eq:A-explicit-gauge}
\mathcal A
=
\frac{1}{\mathfrak a_f}\,\gamma_f^{-1}
=
\frac{1}{\mathfrak a_f}\left[
\gamma^{-1}
+\frac{1}{\mu_f^T}\Big(
N^2\big((\slashed d f)^{\sharp_\gamma}\big)^2
-2\big(1+\langle \beta^T,\slashed d f\rangle_\gamma\big)\,
(\slashed d f)^{\sharp_\gamma}\odot(\beta^T)^{\sharp_\gamma}
+|\slashed d f|_\gamma^2\big((\beta^T)^{\sharp_\gamma}\big)^2
\Big)
\right],
\end{equation}
and
\begin{equation}\label{eq:B-explicit-gauge}
\begin{aligned}
\mathcal B
:=
&\frac{1}{N^2\mathfrak a_f}
\Big(
\big(1+\langle \beta^T,\slashed d f\rangle_\gamma\big)\,\cancel{div}_{\gamma}(\beta^T)
+\langle \slashed\nabla_{(\beta^T)^{\sharp_\gamma}}\beta^T,\slashed d f\rangle_\gamma
-\langle \slashed d N^2,\slashed d f\rangle_\gamma
\Big)
\\
&\quad
-\frac{1}{N^2\mathfrak a_f\,\mu_f^T}
\Bigg(
\big(1+\langle \beta^T,\slashed d f\rangle_\gamma\big)
\left\langle
\slashed\nabla_{
\big(1+\langle \beta^T,\slashed d f\rangle_\gamma\big)(\beta^T)^{\sharp_\gamma}
-N^2(\slashed d f)^{\sharp_\gamma}
}
\beta^T,\slashed d f
\right\rangle_\gamma
\\
&\hspace{6em}
-\frac12 |\slashed d f|_\gamma^2
\left\langle
\slashed d N^2,
\big(1+\langle \beta^T,\slashed d f\rangle_\gamma\big)(\beta^T)^{\sharp_\gamma}
-N^2(\slashed d f)^{\sharp_\gamma}
\right\rangle_\gamma
\Bigg)
\\
&\quad
-\frac{\mu_f^T}{2N^2\mathfrak a_f}\,\tr_{\gamma_f}Q_t(\slashed d f).
\end{aligned}
\end{equation}
Here $:$ denotes full contraction.
\end{prop}

\begin{proof}
Set
\[
D:=\cancel{div}_{\gamma_f}(\beta_f^T),
\qquad
T:=\frac12\tr_{\gamma_f}Q_t(\slashed d f).
\]
By \eqref{eq:betaf-nf-simplified}, the TMCF equation reads
\[
D-\frac{N^2}{\mu_f^T}\,\frac{H_{f,r}}{\lambda_f}\,f_r-T=0.
\]
If $f$ solves the TMCF equation, then by \eqref{eq:H-over-lambda-onshell},
\[
\frac{H_{f,r}}{\lambda_f}=\mathfrak a_f.
\]
Hence
\begin{equation}\label{eq:divform-proof}
f_r=\frac{\mu_f^T}{N^2\mathfrak a_f}\big(D-T\big).
\end{equation}

Now define the tangential vector field
\[
V_f
:=
\big(1+\langle \beta^T,\slashed d f\rangle_\gamma\big)(\beta^T)^{\sharp_\gamma}
-
N^2(\slashed d f)^{\sharp_\gamma}.
\]
By \eqref{eq:div-betafT-explicit},
\[
D
=
\frac{1}{\sqrt{\mu_f^T}}\,
\cancel{div}_{\gamma}\!\left(\frac{V_f}{\sqrt{\mu_f^T}}\right)
=
\frac{1}{\mu_f^T}\,\cancel{div}_{\gamma}(V_f)
-
\frac{1}{2(\mu_f^T)^2}\,\slashed d\mu_f^T(V_f).
\]
A direct computation gives
\[
\cancel{div}_{\gamma}(V_f)
=
\big(1+\langle \beta^T,\slashed d f\rangle_\gamma\big)\,\cancel{div}_{\gamma}(\beta^T)
+\langle \slashed\nabla_{(\beta^T)^{\sharp_\gamma}}\beta^T,\slashed d f\rangle_\gamma
-\langle \slashed d N^2,\slashed d f\rangle_\gamma
+\slashed\nabla_\gamma^2 f\big((\beta^T)^{\sharp_\gamma},(\beta^T)^{\sharp_\gamma}\big)
-
N^2\slashed\Delta_\gamma f,
\]
and
\[
\slashed d\mu_f^T(V_f)
=
2\big(1+\langle \beta^T,\slashed d f\rangle_\gamma\big)
\left\langle
\slashed\nabla_{V_f}\beta^T,\slashed d f
\right\rangle_\gamma
-
|\slashed d f|_\gamma^2\langle \slashed d N^2,V_f\rangle_\gamma
+
2\slashed\nabla_\gamma^2 f(V_f,V_f).
\]
Substituting these identities into \eqref{eq:divform-proof} and collecting the second-order terms
gives
\[
-\frac{1}{\mathfrak a_f}\,\gamma_f^{-1}:\slashed\nabla_\gamma^2 f,
\]
by \eqref{eq:gammaf-inverse}. The remaining lower-order terms are precisely
\eqref{eq:B-explicit-gauge}. This proves that every solution of the TMCF equation satisfies
\eqref{eq:TMCF-gammab}.

Conversely, if $f$ satisfies \eqref{eq:TMCF-gammab}, then reversing the same calculation gives
\[
D-T=\frac{N^2}{\mu_f^T}\,\mathfrak a_f\,f_r.
\]
Substituting this into the off-shell identity \eqref{eq:H-over-lambda-offshell} yields
\[
\frac{H_{f,r}}{\lambda_f}
=
\frac{1}{2\lambda_f^2}
\left(
\tr_{\gamma_f}Q_r(\slashed d f)-2\,\frac{N^2}{\mu_f^T}\,\mathfrak a_f\,f_r^2
\right).
\]
Since
\[
\tr_{\gamma_f}Q_r(\slashed d f)=2\lambda^2\mathfrak a_f
\qquad\text{and}\qquad
\lambda_f^2=\lambda^2-\frac{N^2}{\mu_f^T}f_r^2,
\]
it follows that
\[
\frac{H_{f,r}}{\lambda_f}=\mathfrak a_f.
\]
Hence
\[
D+H_{f,r}\beta_f({\bf n}_f)-T
=
D-\frac{N^2}{\mu_f^T}\mathfrak a_f f_r-T
=0,
\]
so $f$ satisfies the TMCF equation.
\end{proof}

\begin{remark}
In the explicit formula for $\mathcal B$, all background coefficients are first pulled back to the graph $t=f(r,p)$. The operators $\slashed d$, $\slashed\nabla$, and $\cancel{div}$ applied to such coefficients are the intrinsic operators of the pulled-back metric $\gamma|_{t=f}$. Thus, for instance,
\[
\slashed d(N^2|_{t=f})
=
(\slashed dN^2)|_{t=f}+(\partial_tN^2)|_{t=f}\,\slashed df.
\]	
Furthermore, $\slashed\nabla_\gamma^2 f$ in equation \eqref{eq:TMCF-gammab} is computed using the pulled-back metric $\gamma(t=f(r,p),r)$.
\end{remark}

\begin{remark}[The forcing decomposition]\label{rem:c1c2}
The lower-order term $\mathcal B$ in \eqref{eq:TMCF-gammab} admits a natural splitting
\[
\mathcal B=\mathcal B_1+\mathcal B_2,
\]
where $\mathcal B_1$ is obtained by setting $\slashed d f=0$, and $\mathcal B_2$ is the remaining
term. Since
\[
\mu_f^T\big|_{\slashed d f=0}=1,
\qquad
\mathfrak a_f\big|_{\slashed d f=0}=\frac{H_{t,r}}{\lambda},
\]
and, under the good gauge \eqref{eq:good-gauge},
\[
\tr_{S_{t,r}}K_t
=
\frac1N\left(
\cancel{div}_{\gamma}(\beta^T)-\frac12\tr_{\gamma}(\partial_t\gamma)
\right),
\]
we obtain
\[
\mathcal B_1
=
\Xi|_{t=f},
\qquad
\Xi:=\frac{\lambda}{N}\,\frac{\tr_{S_{t,r}}K_t}{H_{t,r}}.
\]
Accordingly,
\[
\mathcal B_2:=\mathcal B-\Xi|_{t=f},
\qquad
\mathcal B_2\big|_{\slashed d f=0}=0.
\]
Thus $\Xi$ is the genuine background forcing coming from the ratio of the timelike and spacelike
mean curvature of the coordinate spheres.
\end{remark}

\begin{remark}[Analytic versus geometric admissibility]
\label{rem:analytic-vs-geometric}

The general geometric graph equation \eqref{eq:TMCF-nonlinear}, written using
\(g_f\), \(K_f\), \(H_{f,r}\), and \({\bf n}_f\), is an initial-data equation and
therefore requires $M_f$ to be spacelike (i.e.  \(\mu_f>0\) ). By contrast, in the good gauge the explicit
equation \eqref{eq:TMCF-gammab} continues to make sense under the weaker
conditions that $S_{f,r}$ are spacelike (i.e. \(\mu_f^T>0\)) and $H_{f,r}$ is nonvanishing (i.e. \(\mathfrak a_f\neq0\)), and Proposition~\ref{prop:good-gauge-geometric-tmcf}
shows that it is equivalent to the coordinate-free tangency condition
\({\bf H}_{S_{f,r}}\in TM_f\) appearing in definition \ref{def:intro-tmcf}	.
\end{remark}

\begin{prop}[Geometric meaning of the good-gauge equation] \label{prop:good-gauge-geometric-tmcf}
Assume the good gauge
\[
b=0,
\qquad
\beta_r=0.
\]
Let \(f\in C^2\big(M\big)\), and assume that $\mu_f^T>0$ 
so that the graph spheres \(S_{f,r}\) are spacelike. 

Then the spacetime mean curvature vector of \(S_{f,r}\) is tangent to \(M_f\),
\[
{\bf H}_{S_{f,r}}\in TM_f,
\]
if and only if
\begin{equation}\label{eq:good-gauge-geometric-tmcf}
\cancel{div}_{\gamma_f}(\beta_f^T)-\frac{N^2}{\mu_f^T}\,\mathfrak a_f\,f_r-\frac12\tr_{\gamma_f}Q_t(\slashed d f)=0.
\end{equation}
In particular, on any region where \(\mathfrak a_f\neq0\) (equivalently, ${H_{f,r}\neq0}$), this condition is
equivalent to the explicit good-gauge equation
\begin{equation}\label{eq:good-gauge-geometric-parabolic}
f_r+\mathcal A(r,p,f,\slashed\nabla f):\slashed\nabla_\gamma^2 f
=
\mathcal B(r,p,f,\slashed\nabla f),
\end{equation}
with \(\mathcal A\) and \(\mathcal B\) as in
Proposition~\ref{prop:gammab-quasilinear}.
\end{prop}

\begin{proof}
Write
\[
D_f:=\cancel{div}_{\gamma_f}(\beta_f^T),
\qquad
T_f:=\frac12\tr_{\gamma_f}Q_t(\slashed d f).
\]
Let
\[
\tau:=t-f(r,p),
\]
so that \(M_f=\{\tau=0\}\). Since the spacetime mean-curvature vector
\({\bf H}_{S_{f,r}}\) is normal to \(S_{f,r}\), the condition
\({\bf H}_{S_{f,r}}\in TM_f\) is equivalent to
\[
d\tau({\bf H}_{S_{f,r}})=0.
\]

We use the standard first-variation identity
\[
\langle {\bf H}_{S_{f,r}},X\rangle
=
\cancel{div}_{\gamma_f}(X^{\top})
-
\frac12\tr_{\gamma_f}(\mathcal L_X\gtime),
\]
where \(X^\top\) denotes the tangential projection of \(X\) to \(S_{f,r}\).
In the good gauge \(b=0\), \(\beta_r=0\), the radial vector \(\partial_r\)
is normal to \(S_{f,r}\). Moreover, the tangential projection of
\(\partial_t\) to \(S_{f,r}\) is \((\beta_f^T)^{\sharp_{\gamma_f}}\).
Therefore
\begin{equation}\label{eq:H-pair-dt}
\langle {\bf H}_{S_{f,r}},\partial_t\rangle
=
D_f-T_f,
\end{equation}
while, using the definition of \(\mathfrak a_f\),
\begin{equation}\label{eq:H-pair-dr}
\langle {\bf H}_{S_{f,r}},\partial_r\rangle
=
-\frac12\tr_{\gamma_f}Q_r(\slashed d f)
=
-\lambda^2\mathfrak a_f .
\end{equation}

It remains to express \(d\tau\) metrically in the normal bundle of
\(S_{f,r}\). Let
\[
U_f:=\partial_t-(\beta_f^T)^{\sharp_{\gamma_f}}
\]
be the normal projection of \(\partial_t\) to \(S_{f,r}\). The graph
algebra gives
\[
\langle U_f,U_f\rangle=-\frac{N^2}{\mu_f^T}.
\]
Also \(d\tau(U_f)=1\), since \(d\tau(\partial_t)=1\) and \(d\tau\) vanishes
on \(TS_{f,r}\). In the good gauge, \(\partial_r\) is orthogonal to
\(U_f\), satisfies \(\langle\partial_r,\partial_r\rangle=\lambda^2\), and
\[
d\tau(\partial_r)=-f_r.
\]
Hence the metric dual of \(d\tau\), modulo tangential vectors to
\(S_{f,r}\), is
\[
(d\tau)^{\sharp}
\equiv
-\frac{\mu_f^T}{N^2}\partial_t
-
\frac{f_r}{\lambda^2}\partial_r
\qquad \mod TS_{f,r}.
\]
Pairing this identity with \({\bf H}_{S_{f,r}}\), and using
\eqref{eq:H-pair-dt} and \eqref{eq:H-pair-dr}, gives
\[
\begin{aligned}
d\tau({\bf H}_{S_{f,r}})
&=
-\frac{\mu_f^T}{N^2}
\langle{\bf H}_{S_{f,r}},\partial_t\rangle
-
\frac{f_r}{\lambda^2}
\langle{\bf H}_{S_{f,r}},\partial_r\rangle
\\
&=
-\frac{\mu_f^T}{N^2}(D_f-T_f)
+
f_r\mathfrak a_f .
\end{aligned}
\]
Thus \(d\tau({\bf H}_{S_{f,r}})=0\) is equivalent to
\[
D_f-T_f-\frac{N^2}{\mu_f^T}\mathfrak a_f f_r=0,
\]
which is exactly \eqref{eq:good-gauge-geometric-tmcf}.

If \(\mathfrak a_f\neq0\), then \eqref{eq:good-gauge-geometric-tmcf} can
be solved for \(f_r\):
\[
f_r
=
\frac{\mu_f^T}{N^2\mathfrak a_f}(D_f-T_f).
\]
Using the good-gauge divergence identity
\[
D_f
=
\frac{1}{\sqrt{\mu_f^T}}
\cancel{div}_{\gamma}
\left(
\frac{
\big(1+\langle\beta^T,\slashed d f\rangle_\gamma\big)
(\beta^T)^{\sharp_\gamma}
-
N^2(\slashed d f)^{\sharp_\gamma}
}
{\sqrt{\mu_f^T}}
\right),
\]
and collecting the tangential Hessian terms, the second-order part becomes
\[
-\frac{1}{\mathfrak a_f}
\gamma_f^{-1}:\slashed\nabla_\gamma^2 f.
\]
Moving this term to the left-hand side gives
\[
f_r+\frac{1}{\mathfrak a_f}
\gamma_f^{-1}:\slashed\nabla_\gamma^2 f
=
\mathcal B(r,p,f,\slashed\nabla f),
\]
which is \eqref{eq:good-gauge-geometric-parabolic}, with
\[
\mathcal A=\frac{1}{\mathfrak a_f}\gamma_f^{-1}.
\]
\end{proof}

\begin{defn}[Parabolicity conditions]\label{def:parabolicity-conditions}
Let $r_0<r_1<r_2\le \infty$. A function $f$ is said to be parabolically admissible on
$(r_1,r_2]\times S^2\subset M$ if
\begin{equation}\label{eq:parabolicity-conditions}
\mu_f^T>0,
\qquad
r\,\mathfrak a_f>0
\qquad\text{on }(r_1,r_2]\times S^2.
\end{equation}
It is called admissible if, in addition,
\begin{equation}
\mu_f>0
\qquad\text{on }(r_1,r_2]\times S^2.
\end{equation}
For solutions of the spacelike TMCF equation, admissibility is equivalent to
\(M_f\) being spacelike and the leaves satisfying \(H_{f,r}>0\). Parabolic
admissibility alone only asserts that the graph spheres are spacelike and that
the good-gauge equation is parabolic.
\end{defn}

\begin{remark}
Under the conditions \eqref{eq:parabolicity-conditions}, the principal tensor in
\eqref{eq:TMCF-gammab} is
\[
\mathcal A=\frac{1}{\mathfrak a_f}\,\gamma_f^{-1},
\]
which is positive definite. Thus \eqref{eq:TMCF-gammab} is quasilinear parabolic with the forward
direction given by the inward radial variable $\rho=-r$.
\end{remark}

\section{Existence of tangentially maximal slices} 

\subsection{Preliminaries and Function Spaces}

In this section, we work on the open exterior region of the spacetime. Thus throughout this section
we fix
\[
M:=\R^3\setminus \overline{B_{r_0}}\cong (r_0,\infty)\times S^2,
\qquad
\mathcal M:=(\underline T,\infty)\times M,
\qquad
r_0>0.
\]
We also denote the formal limiting inner cylinder by
\[
\mathcal H:=(\underline T,\infty)\times \{r_0\}\times S^2.
\]
In this section, $\mathcal H$ is only a limiting cylinder and is not assumed to be part of
$\mathcal M$. In particular, the coefficients below are only required to be defined for $r>r_0$,
and the tangentially maximal slices constructed later lie entirely in this open exterior region.

When the TMCF equation is written in the inward radial direction, it becomes uniformly forward
parabolic in the logarithmic radius rather than the variable $r$ itself. In particular, after
passing from $r$ to
\[
s:=-\log r,
\]
the parabolicity becomes uniform on truncated cylinders. We therefore introduce the logarithmic
radial variable
\begin{equation}\label{eq:sdef}
s:=-\log r,
\qquad
r=e^{-s},
\qquad
s\in (-\infty,s_0),
\qquad
s_0:=-\log r_0.
\end{equation}
In particular,
\begin{equation}\label{eq:rsder}
\partial_s = -r\,\partial_r,
\qquad
\partial_r = -r^{-1}\partial_s.
\end{equation}

Given $r_0\le r_*<R\le \infty$, define
\[
M_{r_*,R}:=(r_*,R]\times S^2,
\]
with the convention that $M_{r_0,\infty}=M$. We also define
\[
s_*:=-\log r_*,
\qquad
S:=-\log R\in[-\infty,s_0),
\]
with the convention $-\log \infty=-\infty$, and identify
\[
M_{r_*,R}\cong [S,s_*)\times S^2=:M_{S,s_*}.
\]
For a fixed tail radius $r_1>r_0$, we write
\[
\mathcal M_{r_1}:=(\underline T,\infty)\times M_{r_1,\infty}.
\]

\medskip

We now introduce the parabolic H\"older spaces used in the existence argument. These are the
standard anisotropic H\"older spaces for parabolic equations, written here in the logarithmic
radial variable $s$; see, for example, \cite{parabolic1,parabolic2,parabolic3}.
Fix once and for all a smooth reference metric on $S^2$, say the round metric
$\gamma_{S^2}$, and let $\slashed D$ denote its Levi--Civita connection. We use the associated
parabolic distance in the $s$-variable:
\begin{equation}\label{eq:dp}
d_p\big((s,x),(s',y)\big)
:=
d_{\gamma_{S^2}}(x,y)+|s-s'|^{1/2}.
\end{equation}
For $\alpha\in(0,1)$ and a function $v$ on $M_{r_*,R}$, define the parabolic H\"older seminorm by
\begin{equation}\label{eq:holdersemi}
[v]_{\alpha,\alpha/2;M_{r_*,R}}
:=
\sup_{\substack{(s,x)\neq(s',y)\\ (s,x),(s',y)\in M_{r_*,R}}}
\frac{|v(s,x)-v(s',y)|}{d_p((s,x),(s',y))^\alpha}.
\end{equation}

\begin{defn}[Parabolic H\"older spaces in the $s$-variable]\label{def:cts}
Let $r_0\le r_*<R\le \infty$. We define $\cz(M_{r_*,R})$ and $\ct(M_{r_*,R})$ by
\begin{align*}
\|v\|_{\cz(M_{r_*,R})}
&:=
\|v\|_{\mathc^0(M_{r_*,R})}
+[v]_{\alpha,\alpha/2;M_{r_*,R}},
\\
\|v\|_{\ct(M_{r_*,R})}
&:=
\|v\|_{\mathc^0(M_{r_*,R})}
+\|\slashed D v\|_{\mathc^0(M_{r_*,R})}
+\|\slashed D^2 v\|_{\mathc^0(M_{r_*,R})}
+\|v_s\|_{\mathc^0(M_{r_*,R})}
\\
&\qquad
+[\slashed D^2 v]_{\alpha,\alpha/2;M_{r_*,R}}
+[v_s]_{\alpha,\alpha/2;M_{r_*,R}}.
\end{align*}
When $R<\infty$ and $v$ extends continuously to $\overline{M_{r_*,R}}$ with finite $\ct$-norm, we
write $v\in \ct(\overline{M_{r_*,R}})$.
\end{defn}

\begin{defn}[Weighted spaces in the $s$-variable]\label{def:weightedlog}
For $\sigma\in \R$ and $r_0\le r_*<R\le \infty$, define the weighted spaces
$\ct_{-\sigma}(M_{r_*,R})$ and $\cz_{-\sigma}(M_{r_*,R})$ by
\[
\|v\|_{\ct_{-\sigma}(M_{r_*,R})}
:=
\|e^{-\sigma s}v\|_{\ct(M_{r_*,R})},
\qquad
\|v\|_{\cz_{-\sigma}(M_{r_*,R})}
:=
\|e^{-\sigma s}v\|_{\cz(M_{r_*,R})}.
\]
Since $e^{-\sigma s}=r^\sigma$, this is equivalent to
\[
v\in \ct_{-\sigma}(M_{r_*,R})
\Longleftrightarrow
r^\sigma v\in \ct(M_{r_*,R}),
\qquad
v\in \cz_{-\sigma}(M_{r_*,R})
\Longleftrightarrow
r^\sigma v\in \cz(M_{r_*,R}).
\]
Thus the notation $\ct_{-\sigma}$ and $\cz_{-\sigma}$ corresponds to decay like $r^{-\sigma}$.
\end{defn}

\begin{lem}[Derivative decay from the logarithmic definition]\label{lem:logweights-derivs}
Let $\sigma\in\R$ and $v\in \ct_{-\sigma}(M)$. Then
\[
r^\sigma v,\quad r^\sigma \slashed D v,\quad r^\sigma \slashed D^2 v\in \mathc^0(M),
\qquad
r^{\sigma+1}v_r\in \mathc^0(M),
\]
and likewise at the H\"older level,
\[
r^\sigma \slashed D^2 v\in \cz(M),
\qquad
r^{\sigma+1}v_r\in \cz(M).
\]
\end{lem}

\begin{proof}
By definition, $r^\sigma v\in \ct(M)$, so $(r^\sigma v)_s\in \mathc^0(M)$ and
\[
[(r^\sigma v)_s]_{\alpha,\alpha/2;M}<\infty.
\]
Using
\[
(r^\sigma v)_s=-\sigma r^\sigma v+r^\sigma v_s,
\qquad
v_s=-r\,v_r,
\]
we obtain
\[
-r^{\sigma+1}v_r
=
r^\sigma v_s
=
(r^\sigma v)_s+\sigma r^\sigma v,
\]
which is bounded and H\"older whenever $r^\sigma v$ and $(r^\sigma v)_s$ are. The tangential
statements follow from
\[
r^\sigma \slashed D^k v = \slashed D^k(r^\sigma v),
\qquad k=1,2,
\]
because $r^\sigma$ is constant on each sphere $S_r$.
\end{proof}

\begin{defn}[A strengthened weighted tail space]\label{def:ctsharp-main}
For $\sigma\in\R$ and $r_0<r_*<R\le \infty$, define the space
${\ct}^{\sharp}_{-\sigma}(M_{r_*,R})$ by
\begin{equation}\label{eq:ctsharp-main}
\begin{split}
\|v\|_{{\ct}^{\sharp}_{-\sigma}(M_{r_*,R})}
:=
&\ \|v\|_{\ct_{-\sigma}(M_{r_*,R})}
+\sum_{j=3}^{4}\|r^\sigma\slashed D^j v\|_{\mathc^0(M_{r_*,R})}
\\
&\ +\sum_{\ell=1}^{3}\sum_{j=0}^{3-\ell}
\|r^{\sigma+\ell}\partial_r^\ell\slashed D^j v\|_{\cz(M_{r_*,R})}.
\end{split}
\end{equation}
\end{defn}

\begin{remark}
We keep the shorthand ${\ct}^{\sharp}_{-\sigma}$ for this auxiliary space, rather than introducing
a more elaborate multi-index notation, since it is used only in the gauge-reduction argument of
Appendix~\ref{app:late-slab-gauge}. The point of the notation is simply to record the extra
tangential and mixed radial/tangential regularity needed there.
\end{remark}

\medskip
\noindent
\textbf{Tensor-valued spaces.}
By abuse of notation we use the same symbols
\[
\ct,\ \cz,\ \ct_{-\sigma},\ \cz_{-\sigma},\ {\ct}^{\sharp}_{-\sigma}
\]
for tensor fields whose components, in any fixed $\gamma_{S^2}$-orthonormal frame, belong to the
corresponding spaces. Different frame choices give equivalent norms.

\medskip

We now introduce the coefficient class used in the existence theorem. We write the metric in the
fully general ADM form
\begin{equation}\label{eq:metric-general-prelim}
\gtime
=
-\big(N^2-|\beta|_{g(t)}^2\big)\,dt^2
+2\,\beta\odot dt
+g(t),
\qquad
g(t):=\big(\lambda^2+|b|_\gamma^2\big)\,dr^2+2\,b\odot dr+\gamma,
\end{equation}
where $N$ and $\lambda$ are functions on $\mathcal M$, $\beta$ is a $1$-form on $\mathcal M$ with
no $dt$-component, $b$ is a $1$-form on $\mathcal M$ tangent to the spacetime spheres
\[
S_{t,r}:=\{t\}\times\{r\}\times S^2,
\]
and $\gamma=\gamma(t,r)$ is a two-parameter family of metrics on $S^2$, viewed as a tensor field
on $\mathcal M$ tangent to $S_{t,r}$. In particular, $\lambda$ is the geometric foliation lapse of
the $r$-level sets in $(M_t,g(t))$, namely
\[
\lambda=|\nabla^{g(t)}r|^{-1}.
\]

For each $t>\underline T$, let
\[
M_t:=\{t\}\times M,
\qquad
{\bf n}_t:=\frac{1}{\lambda}\Big(\partial_r-b^{\sharp_\gamma}\Big),
\]
so that ${\bf n}_t$ is the outward unit normal to $S_{t,r}\subset (M_t,g(t))$. Let $H_{t,r}$
denote the mean curvature of $S_{t,r}\subset (M_t,g(t))$, equivalently
\begin{equation}\label{eq:Htr-general}
H_{t,r}
=
\frac{1}{2\lambda}\,
\tr_{\gamma(t,r)}\!\Big(\partial_r\gamma(t,r)-\mathcal L_{b^{\sharp_\gamma}}\gamma(t,r)\Big),
\end{equation}
and let $K_t$ denote the second fundamental form of $M_t\hookrightarrow(\mathcal M,\gtime)$.

The geometric normal component of the shift is
\begin{equation}\label{eq:beta-perp-def-main}
\beta^\perp:=\beta({\bf n}_t).
\end{equation}

For later use, we define the background forcing
\begin{equation}\label{eq:Xi-background-def}
\Xi(t,r,p)
:=
\frac{\lambda}{N}\,
\frac{\tr_{S_{t,r}}K_t}{H_{t,r}},
\qquad
(t,r,p)\in \mathcal M,
\end{equation}
whenever this quantity is finite. In the good gauge $b=0$, $\beta_r=0$, this is equal to the
forcing introduced in the previous subsection.

For $r>r_0$ and $T\in\R$, we define the associated tail-size functions
\begin{align}
\mathfrak X_0(r,T)
&:=
r\sup_{\substack{t\ge T\\ p\in S^2}}
|\Xi(t,r,p)|,
\label{eq:X0-def}
\\
\mathfrak X_1(r,T)
&:=
r\sup_{\substack{t\ge T\\ p\in S^2}}
|\slashed d_\gamma\Xi(t,r,p)|_{\gamma(t,r)},
\label{eq:X1-def}
\\
\mathfrak X_2(r,T)
&:=
r^2\sup_{\substack{t\ge T\\ p\in S^2}}
|\slashed\nabla_{\gamma}^2\Xi(t,r,p)|_{\gamma(t,r)}.
\label{eq:X2-def}
\end{align}
Here, and below, $\sup_{t\ge T}$ is understood as $\sup_{t\ge \max\{T,\underline T\}}$.

More generally, if $I\subset (\underline T,\infty)$ is a compact interval, we define the
corresponding slab quantities
\begin{align}
\mathfrak X_0(r;I)
&:=
r\sup_{\substack{t\in I\\ p\in S^2}}
|\Xi(t,r,p)|,
\label{eq:X0-slab-def}
\\
\mathfrak X_1(r;I)
&:=
r\sup_{\substack{t\in I\\ p\in S^2}}
|\slashed d_\gamma\Xi(t,r,p)|_{\gamma(t,r)},
\label{eq:X1-slab-def}
\\
\mathfrak X_2(r;I)
&:=
r^2\sup_{\substack{t\in I\\ p\in S^2}}
|\slashed\nabla_{\gamma}^2\Xi(t,r,p)|_{\gamma(t,r)}.
\label{eq:X2-slab-def}
\end{align}
These are the versions that will be used after passing to a good gauge on a late time slab.

\begin{defn}[Coefficient class on exterior tails]
\label{def:ctS-nonstat}
Fix $\tau\in(1/2,1)$. We denote by $\ctS_{-\tau}(\mathcal M)$ the set of tuples
\[
\mathcal S=(N,\lambda,\beta,b,\gamma)
\]
on $\mathcal M=(\underline T,\infty)\times M$ satisfying the following conditions.

\begin{enumerate}
\item[(i)] \emph{Asymptotic flatness with two time derivatives on every fixed tail.}
For every $r_1>r_0$,
\begin{equation}\label{eq:ctS-nonstat-decay}
\begin{split}
\mathcal N_\tau(r_1)
:=
\sum_{j=0}^2\sup_{t>\underline T}
&\Big(
\|\partial_t^j(N-1)(t)\|_{\ct_{-\tau-j}(M_{r_1,\infty})}
+\|\partial_t^j(\lambda-1)(t)\|_{\ct_{-\tau-j}(M_{r_1,\infty})}
\\
&\quad
+\|\partial_t^j b(t)\|_{\ct_{-\tau-j}(M_{r_1,\infty})}
+\|\partial_t^j\beta(t)\|_{\ct_{-\tau-j}(M_{r_1,\infty})}
\\
&\quad
+\|\partial_t^j(r^{-2}\gamma(t)-\gamma_{S^2})\|_{\ct_{-\tau-j}(M_{r_1,\infty})}
\Big)
<\infty.
\end{split}
\end{equation}

\item[(ii)] \emph{Tail nondegeneracy and mean-curvature bounds.}
For every $r_1>r_0$, there exist positive constants
\[
\delta_*(r_1),
\qquad
\vartheta_*(r_1),
\qquad
c_*(r_1),
\]
independent of $t$, such that on $(\underline T,\infty)\times M_{r_1,\infty}$ one has
\begin{equation}\label{eq:tail-nondeg-global}
\delta_*(r_1)^{-1}\geq \lambda\ge \delta_*(r_1),
\qquad
\vartheta_*(r_1)^{-1} \geq N+|\beta^T|^2_{\gamma}\geq N\ge \vartheta_*(r_1),
\qquad 
c_*(r_1)^{-1}\geq r\,H_{t,r}\ge c_*(r_1).
\end{equation}

\item[(iii)] \emph{Finite forcing tails on every fixed tail.}
For every $r_1>r_0$ and every $T\in\R$, the tail-size functions introduced above are finite and
satisfy
\begin{equation}\label{eq:finite-forcing-tail}
\int_{r_1}^\infty \frac{\mathfrak X_0(\sigma,T)}{\sigma}\,d\sigma<\infty,
\qquad
\int_{r_1}^\infty \mathfrak X_1(\sigma,T)\,d\sigma<\infty,
\qquad
\int_{r_1}^\infty \mathfrak X_2(\sigma,T)\,d\sigma<\infty.
\end{equation}
\end{enumerate}

For a fixed tail $r_1>r_0$, we write
\[
\|\mathcal S\|_{\ctS_{-\tau}(\mathcal M_{r_1})}
:=
\mathcal N_\tau(r_1).
\]
Thus $\ctS_{-\tau}(\mathcal M)$ is controlled by a family of tail norms, one for each
$r_1>r_0$, rather than by a single norm up to the limiting cylinder $r=r_0$.
\end{defn}

\begin{remark}[Uniformly nondegenerate horizon-penetrating gauges]
\label{rem:uniform-horizon-penetrating-gauge}
The definition above is intentionally formulated on each fixed tail $r\ge r_1>r_0$, and the
constants $\delta_*(r_1)$, $\vartheta_*(r_1)$, and $c_*(r_1)$ are allowed to deteriorate as
$r_1\downarrow r_0$. This allows for exterior coordinate systems such as Boyer--Lindquist-type
coordinates, which may degenerate at the horizon. In horizon-penetrating gauges with uniform
control up to $r=r_0$, one can instead arrange the bounds for $\lambda$, $N$, and
$rH_{t,r}$ to be independent of $r_1$ (but possibly dependent on $t$). A
model construction of such a horizon-penetrating foliation in Schwarzschild, agreeing with the
standard Schwarzschild slices outside a radius tending to the horizon, is given in
Appendix~\ref{app:horizon-penetrating-schwarzschild-model}.
\end{remark}

\begin{defn}[Strengthened tail coefficient space and normal-shift tail]\label{def:sharp-tail-data-main}
Let ${\ctS}^\sharp_{-\tau}(\mathcal M)$ denote the space of coefficient tuples
\[
\mathcal S=(N,\lambda,\beta,b,\gamma) \in \ctS_{-\tau}(\mathcal{M})
\]
for which: for every $r_1>r_0$,
\begin{equation}\label{eq:Nsharp-def-main}
\begin{split}
\mathcal N_\tau^\sharp(r_1)
:=
\sum_{k=0}^{3}\sup_{t>\underline T}
&\Big(
\|\partial_t^k(N-1)(t)\|_{{\ct}^{\sharp}_{-\tau-k}(M_{r_1,\infty})}
+\|\partial_t^k(\lambda-1)(t)\|_{{\ct}^{\sharp}_{-\tau-k}(M_{r_1,\infty})}
\\
&\quad
+\|\partial_t^k b(t)\|_{{\ct}^{\sharp}_{-\tau-k}(M_{r_1,\infty})}
+\|\partial_t^k\beta(t)\|_{{\ct}^{\sharp}_{-\tau-k}(M_{r_1,\infty})}
\\
&\quad
+\|\partial_t^k(r^{-2}\gamma(t)-\gamma_{S^2})\|_{{\ct}^{\sharp}_{-\tau-k}(M_{r_1,\infty})}
\Big) <\infty
\end{split}
\end{equation}
 Also, for $T\in\R$, define the geometric normal-shift tail by
\begin{equation}\label{eq:Bsharp-def-main}
\mathfrak B^\sharp(T;r_1)
:=
\sup_{t\ge T}
\|\beta^\perp(t,\cdot)\|_{{\ct}^{\sharp}_{-\tau}(M_{r_1,\infty})}.
\end{equation}
which is also finite by the finiteness of $\mathcal N_\tau^\sharp(r_1)$. 
\end{defn}

\begin{remark}
The finiteness of $\mathcal N_\tau^\sharp(r_1)$ is the strengthened tail regularity needed for the
coordinate changes in Appendix~\ref{app:late-slab-gauge} that allows one to pass from a
general ADM chart to the good gauge $b=0$, $\beta_r=0$ on slabs away from horizon. To make this coordinate change, we will in addition need $\mathfrak B^\sharp(T;r_1)$ to decay as $T$ goes to infinity; this will be part of the quasi final state hypothesis introduced later.
\end{remark}

\begin{comment}
\begin{remark}
The lower bounds in \eqref{eq:tail-nondeg-global} are imposed only on fixed tails
$\{r\ge r_1\}$, $r_1>r_0$. This is deliberate: we do not require positive lower bounds for
$\lambda$, $N$, or $H_{t,r}$ all the way down to $r=r_0$. In Kerr, for example, Boyer--Lindquist
coordinates degenerate at the horizon, with $N$ vanishing and $\lambda$ diverging there. On the
other hand, no lower bound is imposed on
\[
\alpha=N^2-|\beta|_{g(t)}^2,
\]
which may vanish or change sign near the limiting inner cylinder.
\end{remark}
\end{comment}

\begin{remark}[On the finite forcing-tail assumption]
\label{rem:finite-forcing-tail-assumption}
We explain that the finite forcing-tail assumption \eqref{eq:finite-forcing-tail} is not a restriction on the spacetime for the purposes of proving the spacetime Penrose inequality. 

 By the Corvino--Schoen gluing theory and its variants, one may replace an
asymptotically flat initial data set, outside a sufficiently large
compact set and without changing the data near the horizon, by data which
agree exactly with a Schwarzschild or Kerr end.  Moreover, the ADM
charges of the glued data can be arranged to be arbitrarily close to the
original ADM charges; see, for example,
\cite{Corvino2000,CorvinoSchoen2006,ChruscielDelay2003}.  Thus, for
purposes of proving a Penrose inequality, one may first prove the desired
estimate for such glued data and then let the gluing radius tend to
infinity.

For such data which are exactly Kerr on the far end, the corresponding
development agrees with the Kerr development in the domain of dependence
of that end. Hence, on
each sufficiently far tail, the forcing is compactly supported in \(r\) as $\Xi$ vanishes in Kerr,
and the integral condition \eqref{eq:finite-forcing-tail} is automatic.
The finite forcing-tail assumption is therefore not an additional
restriction on the compact black-hole region, but as a convenient way of
excluding irrelevant non-integrable forcing at spatial infinity.
\end{remark}

\begin{defn}[Asymptotically flat exterior spacetime of order $\tau$]
\label{def:AFspacetime}
Fix $\tau\in(1/2,1)$. We say that a spacetime $(\mathcal M,\gtime)$, with
\[
\mathcal M=(\underline T,\infty)\times M,
\]
is an \emph{asymptotically flat exterior spacetime of order $\tau$} if $\gtime$ can be written in
the general ADM form \eqref{eq:metric-general-prelim} with coefficient tuple
\[
\mathcal S=(N,\lambda,\beta,b,\gamma)\in \ctS_{-\tau}(\mathcal M).
\]
\end{defn}

For later use, it is convenient to record the natural lower and upper bounds of the background
mean-curvature coefficient on fixed tails. Define
\begin{equation}\label{eq:mu0T-b0-def}
\mathfrak a_0
:=
\frac{H_{t,r}}{\lambda}
=
\frac{1}{2\lambda^2}\,
\tr_{\gamma(t,r)}\!\Big(\partial_r\gamma(t,r)-\mathcal L_{b^{\sharp_\gamma}}\gamma(t,r)\Big).
\end{equation}
For each $r_1>r_0$, set
\begin{equation}\label{eq:tail-constants-def}
\underline{\mathfrak a_0}(r_1)
:=
\inf_{\substack{t>\underline T\\ r\ge r_1\\ p\in S^2}}
r\,\mathfrak a_0(t,r,p),
\qquad
\overline{\mathfrak a_0}(r_1)
:=
\sup_{\substack{t>\underline T\\ r\ge r_1\\ p\in S^2}}
r\,\mathfrak a_0(t,r,p).
\end{equation}
By Definition~\ref{def:ctS-nonstat},
\[
0<\underline{\mathfrak a_0}(r_1)\le \overline{\mathfrak a_0}(r_1)<\infty
\qquad\text{for every }r_1>r_0.
\]

\subsection{A Priori Estimates} \label{sub:apriori}

In this subsection we now derive the pointwise and Schauder estimates needed for the bootstrap
argument carried out later. Throughout we fix
\[
r_0<r_1\le r_*<R<\infty,
\qquad
T_0\in \R,
\]
and let $(\mathcal M,\gtime)$ be an asymptotically flat exterior spacetime with coefficients
\[
\mathcal S=(N,\lambda,\beta,b,\gamma)\in \ctS_{-\tau}(\mathcal M),
\]
satisfying the good gauge
\[
b\equiv 0,
\qquad
\beta_r\equiv 0.
\]
All estimates in this subsection take place on the fixed tail $M_{r_1,\infty}$. We also fix the
constants
\[
\delta_*(r_1),
\qquad
\vartheta_*(r_1),
\qquad
\underline{\mathfrak a_0}(r_1),
\]
defined in Definition~\ref{def:ctS-nonstat}, so that all constants below are uniform in $r_*$ and
$R$. To simplify notation, we write
\[
\delta_*:=\delta_*(r_1),
\qquad
\vartheta_*:=\vartheta_*(r_1),
\qquad
\underline{\mathfrak a_0}:=\underline{\mathfrak a_0}(r_1).
\]

We write
\[
\mu:=\mu_f,
\qquad
\mu^T:=\mu_f^T,
\qquad
\mathfrak a:=\mathfrak a_f
\]
for the graph quantities associated to a function $f$, and we recall from
\eqref{eq:mu0T-b0-def} that, in the good gauge,
\[
\mathfrak a_0
=
\frac{H_{t,r}}{\lambda}
=
\frac{1}{2\lambda^2}\tr_\gamma(\partial_r\gamma)
\]
denotes the corresponding background coefficient. We also recall the tail bounds
\[
\underline{\mathfrak a_0} 
\le
r\,\mathfrak a_0
\le
\overline{\mathfrak a_0}
\qquad\text{on }(\underline T,\infty)\times M_{r_1,\infty}.
\]

\medskip

We also recall 
\[
N+|\beta^T|_\gamma\le \vartheta_*^{-1}
\qquad\text{on }\mathcal M_{r_1}.
\]

%(Here we are reusing the constant $\vartheta_*$ in a slightly stronger role than in Definition~\ref{def:ctS-nonstat} to avoid introducing additional notation; this is harmless since both quantities are tail-controlled and we may take the smaller of the two.) 

Hence, if
\[
|\slashed d f|_\gamma< \frac{\vartheta_*}{2},
\]
then, by \eqref{eq:muT-sufficient},
\begin{equation}\label{eq:varthetastar-implies-muT}
\mu^T
=
\big(1+\langle \beta^T,\slashed d f\rangle_\gamma\big)^2
-
N^2|\slashed d f|_\gamma^2
\ge
\Big(1-\big(N+|\beta^T|_\gamma\big)|\slashed d f|_\gamma\Big)^2
>
\frac14.
\end{equation}

For the $C^0$-estimate we use the tail functions
$\mathfrak X_0(\cdot,T)$, $\mathfrak X_1(\cdot,T)$, and $\mathfrak X_2(\cdot,T)$ introduced in
\eqref{eq:X0-def}--\eqref{eq:X2-def}, always restricted to the fixed tail $r\ge r_1$. After
Proposition~\ref{prop:C0-f} constructs the threshold $T_*(T_0;r_1)$, we set
\[
\Delta_*(T_0;r_1):=T_0-T_*(T_0;r_1),
\]
and
\[
I_*:=I_*(T_0;r_1)
:=
[T_0-\Delta_*(T_0;r_1),\,T_0+\Delta_*(T_0;r_1)]
=
[T_*(T_0;r_1),\,2T_0-T_*(T_0;r_1)].
\]
We use the corresponding slab quantities
$\mathfrak X_0(\cdot;I_*)$, $\mathfrak X_1(\cdot;I_*)$, and $\mathfrak X_2(\cdot;I_*)$ from
\eqref{eq:X0-slab-def}--\eqref{eq:X2-slab-def}.

\vv

The estimates in this subsection are proved under two quantitative bootstrap assumptions. These assumptions are
stronger than parabolic admissibility/admissibility defined in definition \ref{def:parabolicity-conditions}, and are designed to give uniform constants in
the maximum-principle and Schauder arguments.

\begin{defn}[Bootstrap assumptions]\label{def:BS1-BS2}
Let $f$ be defined on $M_{r_*,R}$. We say that $f$ satisfies
\textbf{BS-1} on $M_{r_*,R}$ if
\begin{equation}\label{eq:bootstrap-assumption-apriori}
|\slashed d f|_\gamma< \frac{\vartheta_*}{2},
\qquad
r\,\mathfrak a> \frac12\,\underline{\mathfrak a_0}, 
\qquad\text{on }M_{r_*,R}.
\end{equation}
We say that $f$ satisfies \textbf{BS-2} on $M_{r_*,R}$ if
\begin{equation}\label{eq:bootstrap-assumption-apriori-2}
\mu>0 ,
\qquad
r^2|\slashed \nabla^2_{\gamma} f|_{\gamma} < 1
\qquad\text{on }M_{r_*,R}.
\end{equation}
\end{defn}

\begin{thm}[A priori estimates on the truncated annulus]
\label{thm:apriori-estimates}
Let $f\in C^4(M_{r_*,R})\cap C^0(\overline{M_{r_*,R}})$ solve the TMCF equation with an outer boundary condition
\begin{equation}\label{eq:tmcf-truncated-bvp}
\begin{cases}
f_r
+
\mathcal A(r,p,f,\slashed\nabla f):\slashed\nabla_\gamma^2 f = \mathcal B(r,p,f,\slashed\nabla f),
&\text{on }M_{r_*,R},\\
f=T_0,&\text{on }S_R.
\end{cases}
\end{equation}
Then the following statements hold.

\begin{enumerate}[label=\textup{(\roman*)}]
\item
Suppose $f$ satisfies \textbf{BS-1} on $M_{r_*,R}$. There exists a threshold
\[
T_*=T_*(T_0;r_1)\le T_0,
\]
depending only on $T_0$, $r_1$, and the background coefficients on $\mathcal M_{r_1}$, such that,
if we set
\[
\Delta_*:=T_0-T_*(T_0;r_1),
\qquad
I_*:=[T_0-\Delta_*,\,T_0+\Delta_*]=[T_*(T_0;r_1),\,2T_0-T_*(T_0;r_1)],
\]
then
\begin{equation}\label{eq:thm-C0}
\|f-T_0\|_{\mathc^0(M_{r_*,R})}
\le
\Delta_*,
\qquad
f(M_{r_*,R})\subset I_*.
\end{equation}
Moreover,
\[
\|f-T_0\|_{\mathc^0(M_{r_*,R})}
\le
\int_{r_1}^\infty \frac{\mathfrak X_0(\sigma;I_*)}{\sigma}\,d\sigma
\le
\Delta_*
=
\int_{r_1}^\infty \frac{\mathfrak X_0(\sigma,T_*)}{\sigma}\,d\sigma.
\]

\item
Suppose $f$ satisfies \textbf{BS-1} on $M_{r_*,R}$. There exists a constant
\[
C=C\!\left(
\tau,r_1,\underline{\mathfrak a_0}^{-1},\delta_*^{-1},\vartheta_*^{-1},
\|\mathcal S\|_{\ctS_{-\tau}(\mathcal M_{r_1})}
\right)
\]
such that for every $(r,p)\in M_{r_*,R}$,
\begin{equation}\label{eq:thm-grad}
|\slashed d f|_{\gamma}(r,p)
\le
\frac{C}{r}\int_r^R \mathfrak X_1(\sigma;I_*)\,d\sigma.
\end{equation}

\item
Suppose $f$ satisfies \textbf{BS-1} and \textbf{BS-2} on $M_{r_*,R}$. There exists a constant
\[
C=C\!\left(
\alpha,\tau,r_1,\underline{\mathfrak a_0}^{-1},\delta_*^{-1},\vartheta_*^{-1},
\|\mathcal S\|_{\ctS_{-\tau}(\mathcal M_{r_1})}
\right)
\]
such that
\begin{equation}\label{eq:thm-schauder}
\|f-T_0\|_{\ct_{-\tau}(M_{r_*,R})}
\le
C\,
\sup_{t\in I_*}\|r\Xi(t,\cdot)\|_{\cz_{-\tau}(M_{r_*,R})}.
\end{equation}

\item
Suppose $f$ satisfies \textbf{BS-1} and \textbf{BS-2} on $M_{r_*,R}$. There exists a constant
\[
C=C\!\left(
\tau,r_1,\underline{\mathfrak a_0}^{-1},\delta_*^{-1},\vartheta_*^{-1},
\|\mathcal S\|_{\ctS_{-\tau}(\mathcal M_{r_1})}
\right)
\]
such that for every $(r,p)\in M_{r_*,R}$,
\begin{align}
|\slashed\nabla_{\gamma}^2 f|_{\gamma}(r,p)
&\le
\frac{C}{r^2}\int_r^R \mathfrak X_2(\sigma;I_*)\,d\sigma,
\label{eq:thm-hessian}
\\
\frac{N}{\lambda}|f_r|(r,p)
&\le
\frac{C}{r}\,\mathfrak X_0(r;I_*)
+
\frac{C}{r}\int_r^R \mathfrak X_2(\sigma;I_*)\,d\sigma.
\label{eq:thm-fr}
\end{align}
\end{enumerate}
\end{thm}

\begin{remark}[Uniform constants near the horizon]
\label{rem:apriori-uniform-horizon-gauge}
The dependence of the constant $C$ on the tail radius $r_1$ comes only through the fixed-tail
coefficient bounds and the lower bounds for $\lambda$, $N$, and $rH_{t,r}$. If the exterior is
written in a horizon-penetrating gauge with uniform nondegeneracy up to the final horizon, then
these lower bounds and coefficient norms can be controlled independently of $r_1$ as
$r_1\downarrow r_0$. In such a gauge the constants in the a priori estimates can be chosen
uniformly near the horizon. In the Schwarzschild model of
Appendix~\ref{app:horizon-penetrating-schwarzschild-model}, the slices agree with the standard
Schwarzschild time slices on $r\ge r_{\mathrm{cut}}(T)$ with $r_{\mathrm{cut}}(T)\downarrow2m$, and hence the forcing $\Xi$
vanishes on the exterior portion where the slices are exactly Schwarzschild, but not on $r_0<r<r_{\mathrm{cut}}(T)$.
\end{remark}

The remainder of this subsection is devoted to the proof of Theorem~\ref{thm:apriori-estimates}.
We fix a function $
f \in C^4(M_{r_*,R})\cap C^0(\overline{M_{r_*,R}})$ 
satisfying \textbf{BS-1}; we will in addition assume \textbf{BS-2} later when we prove
\textup{(iii)} and \textup{(iv)} in the statement of the theorem. Recall that we write
\[
T_*:=T_*(T_0;r_1),
\qquad
\Delta_*:=T_0-T_*,
\qquad
I_*:=[T_0-\Delta_*,\,T_0+\Delta_*]=[T_*,\,2T_0-T_*].
\]

\vv

For the pointwise maximum-principle arguments it is convenient to use the inward radial variable
\[
\rho:=R-r\in [0,R-r_*),
\]
so that $\partial_\rho=-\partial_r$. We also keep the notation from
Proposition~\ref{prop:gammab-quasilinear}:
\[
\mathcal A=\frac1{\mathfrak a}\,\gamma_f^{-1},
\qquad
\mathcal B=\Xi|_{t=f}+\mathcal B_2,
\qquad
\Xi:= \frac{\lambda}{N}\,\frac{\tr_{S_{t,r}}K_t}{H_{t,r}}.
\]
In terms of $\rho$, the equation \eqref{eq:TMCF-gammab} reads
\begin{equation}\label{eq:tmcf-rho-form}
f_\rho-\mathcal A : \slashed\nabla_\gamma^2 f=-\mathcal B.
\end{equation}

\medskip
\noindent
\textbf{The notation $\mathcos$.}
Throughout the rest of this subsection, we introduce the notation
\[
\mathcal G=\mathcos(r^\delta)
\]
to mean that there exists a constant
\[
C=C\!\left(
\tau,r_1,\underline{\mathfrak a_0}^{-1},\delta_*^{-1},\vartheta_*^{-1},
\|\mathcal S\|_{\ctS_{-\tau}(\mathcal M_{r_1})}
\right)>0
\]
such that
\[
|\mathcal G|\le C\,r^\delta
\qquad\text{on }M_{r_*,R}.
\]
For tangential tensors, the norm $|\cdot|$ is taken with respect to $\gamma$; for scalars it is the
ordinary absolute value.

\begin{prop}[Self-improving $C^0$-estimate and tail threshold]\label{prop:C0-f}
There exists a number
\[
T_*=T_*(T_0;r_1)\le T_0
\]
such that
\begin{equation}\label{eq:Tstar-fixed}
T_0-T_*(T_0;r_1)
=
\int_{r_1}^{\infty}
\sup_{\substack{t\ge T_*(T_0;r_1)\\ p\in S^2}}
|\Xi(t,\sigma,p)|\,d\sigma .
\end{equation}
Moreover,
\begin{equation}\label{eq:C0-f-global}
\|f-T_0\|_{\mathc^0(M_{r_*,R})}
\le
T_0-T_*(T_0;r_1)
=
\int_{r_1}^{\infty}
\sup_{\substack{t\ge T_*(T_0;r_1)\\ p\in S^2}}
|\Xi(t,\sigma,p)|\,d\sigma.
\end{equation}
In particular,
\begin{equation}\label{eq:f-range-Tstar}
T_*(T_0;r_1)\le f(r,p)\le 2T_0-T_*(T_0;r_1)
\qquad\text{for all }(r,p)\in M_{r_*,R}.
\end{equation}
Finally,
\begin{equation}\label{eq:Tstar-limit}
\lim_{T_0\to\infty} T_*(T_0;r_1)=\infty.
\end{equation}
\end{prop}

\begin{proof}
We now define the tail function
\begin{equation} \label{eq:G-via-X0}
G_{r_1}(T)
:=
\int_{r_1}^\infty \frac{\mathfrak X_0(\sigma,T)}{\sigma}\,d\sigma
=
\int_{r_1}^{\infty}
\sup_{\substack{t\ge T\\ p\in S^2}}
|\Xi(t,\sigma,p)|\,d\sigma .
\end{equation}

Set
\[
u:=f-T_0.
\]
We now prove the following bootstrap claim.

\medskip
\noindent
\emph{Bootstrap Claim.}
Let $T\le T_0$. If
\begin{equation}\label{eq:bootstrap-threshold}
f\ge T
\qquad\text{on }M_{r_*,R},
\end{equation}
then
\begin{equation}\label{eq:C0-claim}
\|f-T_0\|_{C^0(M_{r_*,R})}
=
\|u\|_{C^0(M_{r_*,R})}
\le
G_{r_1}(T).
\end{equation}

\smallskip
\noindent
To prove the claim, let
\[
m(\rho):=\sup_{S_{R-\rho}} u(\rho,\cdot).
\]
Choose a point on $S_{R-\rho}$ where $u(\rho,\cdot)$ attains its maximum. Since $T_0$ is constant
on each sphere,
\[
\slashed d f=\slashed d u=0,
\qquad
\slashed\nabla_{\gamma}^2 f=\slashed\nabla_{\gamma}^2 u\le 0,
\qquad
\mathcal B_2=0
\]
at that point. Evaluating \eqref{eq:tmcf-rho-form} there gives
\[
u_\rho=f_\rho
=
\mathcal A:\slashed\nabla_{\gamma}^2 f-\Xi|_{t=f}
\le
-\Xi|_{t=f}
\le
|\Xi|_{t=f}.
\]
By \eqref{eq:bootstrap-threshold}, we have $f\ge T$, hence
\[
|\Xi|_{t=f}
\le
\sup_{\substack{t\ge T\\ p\in S^2}}
|\Xi(t,R-\rho,p)|.
\]
Therefore the upper Dini derivative of $m$ satisfies
\[
D^+m(\rho)
\le
\sup_{\substack{t\ge T\\ p\in S^2}}
|\Xi(t,R-\rho,p)|.
\]
Since $u=0$ on $S_R$, we have $m(0)=0$, and integrating from $0$ to $\rho$ yields
\[
m(\rho)
\le
\int_0^\rho
\sup_{\substack{t\ge T\\ p\in S^2}}
|\Xi(t,R-\sigma,p)|\,d\sigma
=
\int_{R-\rho}^{R}
\sup_{\substack{t\ge T\\ p\in S^2}}
|\Xi(t,\sigma,p)|\,d\sigma
\le
G_{r_1}(T).
\]

Apply the same argument to $-u$. Let
\[
\widetilde m(\rho):=\sup_{S_{R-\rho}}(-u)(\rho,\cdot).
\]
At a point where $-u$ attains its maximum, we have again $\slashed d f=0$, $\mathcal B_2=0$, and
now
\[
\slashed\nabla_{\gamma}^2 f\ge 0.
\]
Hence
\[
(-u)_\rho=-f_\rho
=
-\mathcal A:\slashed\nabla_{\gamma}^2 f+\Xi|_{t=f}
\le
\Xi|_{t=f}
\le
|\Xi|_{t=f}.
\]
The same Dini-derivative argument yields
\[
\widetilde m(\rho)
\le
\int_{R-\rho}^{R}
\sup_{\substack{t\ge T\\ p\in S^2}}
|\Xi(t,\sigma,p)|\,d\sigma
\le
G_{r_1}(T).
\]
Combining the estimates for $u$ and $-u$ proves \eqref{eq:C0-claim}.

\medskip

We now define the threshold $T_*=T_*(T_0;r_1)$ by monotone iteration. Set
\[
G_{r_1,-\infty}
:=
\int_{r_1}^{\infty}
\sup_{\substack{t>\underline T\\ p\in S^2}}
|\Xi(t,\sigma,p)|\,d\sigma
<\infty,
\]
and define a sequence $\{T^{(n)}\}_{n\ge 0}$ by
\begin{equation}\label{eq:Tn-iteration}
T^{(0)}:=T_0-G_{r_1,-\infty},
\qquad
T^{(n+1)}:=T_0-G_{r_1}\big(T^{(n)}\big).
\end{equation}
Since $G_{r_1}$ is nonincreasing, the sequence $\{T^{(n)}\}$ is nondecreasing. Moreover,
$T^{(n)}\le T_0$ for every $n$, so $\{T^{(n)}\}$ converges to a limit
\[
T_*:=\lim_{n\to\infty}T^{(n)}\le T_0.
\]

We claim that
\begin{equation}\label{eq:f-above-Tn}
f\ge T^{(n)}
\qquad\text{on }M_{r_*,R}\quad\text{for every }n\ge 0.
\end{equation}
For $n=0$, apply the bootstrap claim with the coarse bound
\[
\|f-T_0\|_{C^0(M_{r_*,R})}\le G_{r_1,-\infty}.
\]
This gives $f\ge T_0-G_{r_1,-\infty}=T^{(0)}$. Now suppose $f\ge T^{(n)}$. Then the bootstrap
claim \eqref{eq:C0-claim} gives
\[
\|f-T_0\|_{C^0(M_{r_*,R})}
\le
G_{r_1}\big(T^{(n)}\big)
=
T_0-T^{(n+1)},
\]
hence $f\ge T^{(n+1)}$. This proves \eqref{eq:f-above-Tn}.

Passing to the limit $n\to\infty$ in \eqref{eq:f-above-Tn} gives
\[
f\ge T_*
\qquad\text{on }M_{r_*,R}.
\]
Applying the bootstrap claim one last time with $T=T_*$ yields
\[
\|f-T_0\|_{C^0(M_{r_*,R})}\le G_{r_1}(T_*).
\]
It remains to identify $T_*$. For each $\sigma$, the function
\[
T\longmapsto \sup_{\substack{t\ge T\\ p\in S^2}} |\Xi(t,\sigma,p)|
\]
is continuous and nonincreasing because $\Xi$ is continuous in $t$, and it is dominated by
\[
\sup_{\substack{t>\underline T\\ p\in S^2}} |\Xi(t,\sigma,p)|\in L^1([r_1,\infty)).
\]
Hence $G_{r_1}(T)$ is continuous by dominated convergence. Passing to the limit in \eqref{eq:Tn-iteration} therefore gives
\[
T_*=T_0-G_{r_1}(T_*),
\]
which is exactly \eqref{eq:Tstar-fixed}. Combining this with the previous bound proves
\eqref{eq:C0-f-global}. The range bound \eqref{eq:f-range-Tstar} is immediate from
\eqref{eq:C0-f-global}.

Finally, since
\[
T_0-T_*=G_{r_1}(T_*)\le G_{r_1,-\infty},
\]
we have
\[
T_*\ge T_0-G_{r_1,-\infty}.
\]
The quantity $G_{r_1,-\infty}$ depends only on the background spacetime on the fixed tail, so
$T_*(T_0;r_1)\to\infty$ as $T_0\to\infty$. This proves \eqref{eq:Tstar-limit}.

Finally, since $f(M_{r_*,R})\subset I_*$, the same maximum-principle argument with
$\sup_{t\in I_*}$ in place of $\sup_{t\ge T_*}$ yields
\[
\|f-T_0\|_{\mathc^0(M_{r_*,R})}
\le
\int_{r_1}^\infty \frac{\mathfrak X_0(\sigma;I_*)}{\sigma}\,d\sigma.
\]
Combined with \eqref{eq:Tstar-fixed}, this gives the slab version of \eqref{eq:thm-C0}.
\end{proof}

\vv

\begin{lem}[Basic pointwise estimates]\label{lem:coeff-structure}
Assume $f$ satisfies \textbf{BS-1}. Then the following estimates hold on $M_{r_*,R}$:
\begin{enumerate}[label=\textup{(\roman*)}]
\item
\begin{equation}\label{eq:basic-muT-est}
\mu^T> \frac14.
\end{equation}

\item
\begin{equation}\label{eq:basic-a-est}
\mathfrak a-\mathfrak a_0
=
\mathcos(r^{-1-\tau})\,|\slashed d f|_{\gamma}
+
\mathcos(r^{-1})\,|\slashed d f|_{\gamma}^2.
\end{equation}
In particular,
\begin{equation}\label{eq:basic-a-twosided}
r\,\mathfrak a=\mathcos(1),
\qquad
\frac1{\mathfrak a}\le \frac{2r}{\underline{\mathfrak a_0}}.
\end{equation}

\item
The principal tensor is uniformly comparable to $r\,\gamma^{-1}$:
\begin{equation}\label{eq:A-elliptic-basic}
\frac{1}{C}\,r\,|\omega|_\gamma^2
\le
\mathcal A(\omega,\omega)
\le
C\,r\,|\omega|_\gamma^2
\end{equation}
for every tangential $1$-form $\omega$.

\item
\begin{equation}\label{eq:basic-B2-size-est}
|\mathcal B_2|
=
\mathcos(r^{-\tau})\,|\slashed d f|_{\gamma}.
\end{equation}

\item
If $q:=|\slashed d f|_{\gamma}^2$, then
\begin{equation}\label{eq:basic-B2-est}
\big\langle \slashed d f,\slashed d\mathcal B_2\big\rangle_{\gamma}
=
\mathcos(r^{-1-\tau})\,|\slashed d f|_{\gamma}^2
+
\mathcos(r^{-\tau})\,|\slashed d q|_{\gamma}.
\end{equation}
In particular, at every point where $\slashed d q=0$,
\begin{equation}\label{eq:basic-B2-est-max}
\big\langle \slashed d f,\slashed d\mathcal B_2\big\rangle_{\gamma}
=
\mathcos(r^{-1-\tau})\,|\slashed d f|_{\gamma}^2.
\end{equation}
\end{enumerate}
\end{lem}

\begin{proof}
For \textup{(i)}, \eqref{eq:bootstrap-assumption-apriori} and
\eqref{eq:varthetastar-implies-muT} give
\[
\mu^T>\frac14.
\]

For \textup{(ii)}, recall
\[
\mathfrak a
=
\frac{1}{2\lambda^2}\,\tr_{\gamma_f}Q_r(\slashed d f),
\qquad
\mathfrak a_0
=
\frac{1}{2\lambda^2}\,\tr_{\gamma}(\partial_r\gamma).
\]
Thus
\[
2\lambda^2(\mathfrak a-\mathfrak a_0)
=
(\gamma_f^{-1}-\gamma^{-1}):\partial_r\gamma
+
\gamma_f^{-1}:\big(Q_r(\slashed d f)-\partial_r\gamma\big).
\]
We estimate the two terms separately. First,
\[
Q_r(\slashed d f)-\partial_r\gamma
=
2\,\partial_r\beta^T\odot \slashed d f-\partial_r\alpha\,(\slashed d f)^2.
\]
By asymptotic flatness on the fixed tail,
\[
|\partial_r\beta^T|_\gamma=\mathcos(r^{-1-\tau}),
\qquad
|\partial_r\alpha|=\mathcos(r^{-1-\tau}),
\]
and $\gamma_f^{-1}$ is uniformly comparable to $\gamma^{-1}$ by
\eqref{eq:basic-muT-est}. Hence
\[
\gamma_f^{-1}:\big(Q_r(\slashed d f)-\partial_r\gamma\big)
=
\mathcos(r^{-1-\tau})|\slashed d f|_\gamma
+
\mathcos(r^{-1-\tau})|\slashed d f|_\gamma^2.
\]

Second, the explicit formula \eqref{eq:gammaf-inverse}, together with $\mu^T>1/4$, gives
\[
\gamma_f^{-1}-\gamma^{-1}
=
\mathcos(r^{-\tau})|\slashed d f|_\gamma
+
\mathcos(1)|\slashed d f|_\gamma^2
\]
as a $(2,0)$-tensor. Here the term linear in $\slashed d f$ contains one factor of $\beta^T$ and
therefore gains the decay of the shift. Since
\[
\partial_r\gamma=\frac{2}{r}\gamma+\mathcos(r^{-1-\tau}),
\]
contracting with $\gamma_f^{-1}-\gamma^{-1}$ gives
\[
(\gamma_f^{-1}-\gamma^{-1}):\partial_r\gamma
=
\mathcos(r^{-1-\tau})|\slashed d f|_\gamma
+
\mathcos(r^{-1})|\slashed d f|_\gamma^2.
\]
Combining the two estimates and using $\lambda=\mathcos(1)$ on the fixed tail proves
\eqref{eq:basic-a-est}. The estimate $r\,\mathfrak a=\mathcos(1)$ follows immediately, while
\eqref{eq:basic-a-twosided} follows from \eqref{eq:bootstrap-assumption-apriori}.

For \textup{(iii)}, the explicit formula \eqref{eq:gammaf-inverse}, together with
$|\slashed d f|_\gamma<\vartheta_*/2$ and $\mu^T>1/4$, implies that $\gamma_f^{-1}$ is uniformly
comparable to $\gamma^{-1}$. Since $\mathcal A=\mathfrak a^{-1}\gamma_f^{-1}$ and
$\mathfrak a^{-1}=\mathcos(r)$, we obtain \eqref{eq:A-elliptic-basic}.

For \textup{(iv)} and \textup{(v)}, write
\[
\mathcal B_2
=
\mathfrak B_2(r,p,t,\omega)\Big|_{t=f,\;\omega=\slashed d f},
\qquad \omega\in T^*S_r.
\]
By Proposition~\ref{prop:gammab-quasilinear} and Remark~\ref{rem:c1c2}, the map
$\mathfrak B_2$ is smooth on the region determined by \eqref{eq:bootstrap-assumption-apriori}, and
\[
\mathfrak B_2(r,p,t,0)\equiv 0.
\]
A direct inspection of the explicit formula for $\mathcal B$ in \eqref{eq:B-explicit-gauge} shows
that, on this region,
\begin{equation}\label{eq:B2-structure-proof-new}
|\partial_\omega\mathfrak B_2|\le C r^{-\tau},
\qquad
|\partial_x\mathfrak B_2|+|\partial_t\mathfrak B_2|
\le C r^{-1-\tau}|\omega|,
\end{equation}
where $\partial_x$ denotes any tangential derivative in the $S^2$-variables of the background
coefficients. Since $\mathfrak B_2(\cdot,\cdot,\cdot,0)=0$, the mean value theorem and
\eqref{eq:B2-structure-proof-new} give
\[
|\mathcal B_2|
\le
C r^{-\tau}|\slashed d f|_\gamma,
\]
which is \eqref{eq:basic-B2-size-est}.

Now let $X$ be a tangential vector field on $S_r$. By the chain rule,
\[
X(\mathcal B_2)
=
(\partial_x\mathfrak B_2)(X)
+
(\partial_t\mathfrak B_2)\,X(f)
+
(\partial_\omega\mathfrak B_2)\cdot \slashed\nabla_X(\slashed d f).
\]
Using \eqref{eq:B2-structure-proof-new}, we obtain
\[
|X(\mathcal B_2)|
\le
C r^{-1-\tau}|\slashed d f|_\gamma\,|X|
+
C r^{-\tau}|\slashed\nabla_X\slashed d f|.
\]
Evaluating this on $X=(\slashed d f)^{\sharp_\gamma}$ yields
\[
\big|\langle \slashed d f,\slashed d\mathcal B_2\rangle_{\gamma}\big|
\le
C r^{-1-\tau}|\slashed d f|_\gamma^2
+
C r^{-\tau}\big|\slashed\nabla_{\gamma}^2 f\big((\slashed d f)^{\sharp_\gamma},\cdot\big)\big|_{\gamma}.
\]
Finally, since
\[
\frac12\,\slashed d q
=
\slashed\nabla_{\gamma}^2 f\big((\slashed d f)^{\sharp_\gamma},\cdot\big),
\]
we obtain \eqref{eq:basic-B2-est}, and \eqref{eq:basic-B2-est-max} follows when $\slashed d q=0$.
\end{proof}

\begin{lem}[Exact equation for $q=|\slashed d f|_\gamma^2$]\label{lem:q-exact}
Let
\[
V:=(\slashed d f)^{\sharp_\gamma},
\qquad
q:=|\slashed d f|_\gamma^2.
\]
Then $q$ satisfies the exact identity
\begin{align}
\frac12\Big(
q_\rho-\mathcal A:\slashed\nabla_\gamma^2 q
\Big)
&=
-\frac12\,\mathcal D_\rho\gamma(V,V)
-
\mathcal A^{ab}\,\big\langle \slashed\nabla_a\slashed d f,\slashed\nabla_b\slashed d f\big\rangle_\gamma
\label{eq:q-exact-new}
\\
&\quad
-K_\gamma\Big(
(\tr_\gamma\mathcal A)\,q-\mathcal A(\slashed d f,\slashed d f)
\Big)
+
\big(\slashed\nabla_V\mathcal A\big):\slashed\nabla_\gamma^2 f
-
\langle \slashed d f,\slashed d\mathcal B\rangle_\gamma.
\nonumber
\end{align}
Here $K_\gamma$ is the Gauss curvature of $(S_r,\gamma)$, and
\[
\mathcal D_\rho\gamma:= \partial_{\rho}(\gamma|_{t=f})=(-\partial_r+f_\rho\partial_t)\gamma.
\]
\end{lem}

\begin{proof}
From \eqref{eq:tmcf-rho-form},
\[
f_\rho=\mathcal A:\slashed\nabla_\gamma^2 f-\mathcal B.
\]
Differentiating tangentially and taking the $\gamma$-inner product with $\slashed d f$ gives
\[
\big\langle \slashed d f,\slashed d f_\rho\big\rangle_\gamma
=
\big\langle \slashed d f,\slashed d(\mathcal A:\slashed\nabla_\gamma^2 f)\big\rangle_\gamma
-
\langle \slashed d f,\slashed d\mathcal B\rangle_\gamma.
\]
Since $q=|\slashed d f|_\gamma^2$, we also have
\[
\frac12 q_\rho
=
\big\langle \slashed d f,\slashed d f_\rho\big\rangle_\gamma
-\frac12\,\mathcal D_\rho\gamma(V,V).
\]

It therefore remains to expand the middle term. A standard Bochner-type identity on the two-dimensional
manifold $(S_r,\gamma)$ gives
\begin{align*}
\frac12\,\mathcal A:\slashed\nabla_\gamma^2 q
&=
\mathcal A^{ab}\,\big\langle \slashed\nabla_a\slashed d f,\slashed\nabla_b\slashed d f\big\rangle_\gamma
+
\big\langle \slashed d f,\slashed d(\mathcal A:\slashed\nabla_\gamma^2 f)\big\rangle_\gamma
\\
&\quad
+K_\gamma\Big(
(\tr_\gamma\mathcal A)\,q-\mathcal A(\slashed d f,\slashed d f)
\Big)
-
\big(\slashed\nabla_V\mathcal A\big):\slashed\nabla_\gamma^2 f.
\end{align*}
Substituting this into the previous identity yields \eqref{eq:q-exact-new}.
\end{proof}

\begin{prop}[Tangential gradient estimate]\label{prop:grad-f}
Let $T_*=T_*(T_0;r_1)$ be the threshold from Proposition~\ref{prop:C0-f}. Then there exists a
constant
\[
C=C\!\left(
\tau,r_1,\underline{\mathfrak a_0}^{-1},\delta_*^{-1},\vartheta_*^{-1},
\|\mathcal S\|_{\ctS_{-\tau}(\mathcal M_{r_1})}
\right)
\]
such that, for every $(r,p)\in M_{r_*,R}$,
\begin{equation}\label{eq:grad-local}
|\slashed d f|_{\gamma}(r,p)
\le
\frac{C}{r}
\int_r^R
\mathfrak X_1(\sigma;I_*)\,d\sigma.
\end{equation}
\end{prop}

\begin{proof}
By Proposition~\ref{prop:C0-f}, we have $f(r,p)\in I_*$ on $M_{r_*,R}$. Hence
\begin{equation}\label{eq:X1-control-new}
\big|(\slashed d_\gamma\Xi)|_{t=f}\big|_{\gamma}
\le
\frac{1}{r}\,\mathfrak X_1(r;I_*).
\end{equation}

Set
\[
q:=|\slashed d f|_\gamma^2,
\qquad
Q(\rho):=\sup_{S_{R-\rho}} q.
\]
Fix $\rho\in[0,R-r_*)$ and suppose $Q(\rho)>0$. Choose a point $p_\rho\in S_{R-\rho}$ such that
\[
q(\rho,p_\rho)=Q(\rho).
\]
At the point $(\rho,p_\rho)$, choose a $\gamma$-orthonormal frame $\{e_1,e_2\}$ such that
\[
e_1=\widehat V:=\frac{(\slashed d f)^{\sharp_{\gamma}}}{|\slashed d f|_{\gamma}}.
\]
Since $\slashed d q=0$ at a spatial maximum of $q$, we have
\[
0=\frac12(\slashed d q)(e_i)
=
\slashed\nabla_{\gamma}^2 f\big(e_i,(\slashed d f)^{\sharp_{\gamma}}\big),
\qquad i=1,2.
\]
Hence
\[
\slashed\nabla_{\gamma}^2 f(e_1,e_1)=0,
\qquad
\slashed\nabla_{\gamma}^2 f(e_1,e_2)=0.
\]
Therefore the only possibly nonzero Hessian component is
\[
m:=\slashed\nabla_{\gamma}^2 f(e_2,e_2)=\slashed\Delta_{\gamma}f.
\]
Letting $\{e^1,e^2\}$ be the coframe, set
\[
a_{22}:=\mathcal A(e^2,e^2).
\]
By \eqref{eq:A-elliptic-basic}, $a_{22}>0$ and $a_{22}\sim r$. Then
\[
\mathcal A^{ab}\,\big\langle \slashed\nabla_a\slashed d f,\slashed\nabla_b\slashed d f\big\rangle_\gamma
=
a_{22}\,m^2,
\]
and
\[
(\tr_\gamma\mathcal A)\,q-\mathcal A(\slashed d f,\slashed d f)
=
a_{22}\,q.
\]

Evaluating \eqref{eq:q-exact-new} at $(\rho,p_\rho)$, and using that
\[
\mathcal A:\slashed\nabla_\gamma^2 q\le 0
\]
there, gives
\begin{align}
q_\rho
&\le
-\mathcal D_\rho\gamma(\widehat V,\widehat V)\,q
-2a_{22}m^2
-2K_\gamma\,a_{22}\,q
+2\sqrt q\,\widehat V(a_{22})\,m
-2\langle \slashed d f,\slashed d\mathcal B\rangle_\gamma.
\label{eq:q-max-pre-est}
\end{align}

We now estimate the coefficient terms. At the
maximum point, the equation gives
\[
f_\rho
=
\mathcal A:\slashed\nabla_\gamma^2 f-\mathcal B
=
a_{22}m-\Xi|_{t=f}-\mathcal B_2.
\]
Using $a_{22}=\mathcos(r)$, $|\mathcal B_2|=\mathcos(r^{-\tau})\sqrt q$, and
$|\Xi|=\mathcos(r^{-\tau})$ on the fixed tail, we obtain
\[
|f_\rho|
\le
C r |m|+\mathcos(r^{-\tau})+\mathcos(r^{-\tau})\sqrt q.
\]
Since
\[
\partial_r\gamma(\widehat V,\widehat V)=\frac2r+\mathcos(r^{-1-\tau}),
\qquad
\partial_t\gamma(\widehat V,\widehat V)=\mathcos(r^{-1-\tau}),
\]
we have
\[
-\mathcal D_\rho\gamma(\widehat V,\widehat V)\,q
=
\partial_r\gamma(\widehat V,\widehat V)\,q
-
f_\rho\,\partial_t\gamma(\widehat V,\widehat V)\,q.
\]
Therefore
\[
-\mathcal D_\rho\gamma(\widehat V,\widehat V)\,q
\le
\left(\frac2r+\mathcos(r^{-1-\tau})\right)q
+
\mathcos(r^{-\tau})|m|q.
\]
Since $q$ is uniformly bounded by \textbf{BS-1}, Young's inequality and $a_{22}\sim r$ give
\[
C r^{-\tau}|m|q
\le
\frac12 a_{22}m^2+\mathcos(r^{-1-\tau})q.
\]
Thus the contribution of the $\mathcal D_\rho\gamma$ term may be estimated by
\[
-\mathcal D_\rho\gamma(\widehat V,\widehat V)\,q
\le
\left(\frac2r+\mathcos(r^{-1-\tau})\right)q
+
\frac12 a_{22}m^2.
\]

We next estimate $\widehat V(a_{22})$. At $(\rho,p_\rho)$ we have
\[
a_{22}
=
\mathfrak a^{-1}\gamma_f^{-1}(e^2,e^2)
=
\mathfrak a^{-1}
\left(1+\frac{q(\beta^T(e_2))^2}{\mu^T}\right),
\]
where we used $e_2\perp(\slashed d f)^{\sharp_\gamma}$. Since $q$ attains a spatial maximum at
$(\rho,p_\rho)$, we also have
\[
\slashed\nabla_{\widehat V}\slashed d f=0.
\]
Using the above identity together with asymptotic flatness and choosing the frame so that
$\nabla_{\widehat V}e_2=0$ at the point, we obtain
\[
\widehat V\big(1+\langle \beta^T,\slashed d f\rangle_\gamma\big)
=
\langle \slashed\nabla_{\widehat V}\beta^T,\slashed d f\rangle_\gamma
+
\langle \beta^T,\slashed\nabla_{\widehat V}\slashed d f\rangle_\gamma
=
\mathcos(r^{-1-\tau}),
\]
\[
\widehat V(\mu^T)
=
2\big(1+\langle \beta^T,\slashed d f\rangle_\gamma\big)
\widehat V\big(1+\langle \beta^T,\slashed d f\rangle_\gamma\big)
-\widehat V(N^2)\,q
-N^2\widehat V(q)
=
\mathcos(r^{-1-\tau}),
\]
and
\[
\widehat V(\beta^T(e_2))
=
(\slashed\nabla_{\widehat V}\beta^T)(e_2)+\beta^T(\nabla_{\widehat V}e_2)
=
\mathcos(r^{-1-\tau}).
\]
Estimating $\widehat V(\mathfrak a)$, we get
\begin{align*}
\widehat V(\mathfrak a)
&=
\widehat V\!\left(\frac{1}{2\lambda^2}\right)\,
\gamma_f^{-1}:Q_r(\slashed d f)
+
\frac{1}{2\lambda^2}\,
\widehat V(\gamma_f^{-1}):Q_r(\slashed d f)
+
\frac{1}{2\lambda^2}\,
\gamma_f^{-1}:\widehat V\big(Q_r(\slashed d f)\big)
\\
&=
\mathcos(r^{-2-\tau}).
\end{align*}
Since $r\,\mathfrak a=\mathcos(1)$, we have $\mathfrak a^{-1}=\mathcos(r)$, and hence
\begin{equation}\label{eq:e1-ainv}
\widehat V(\mathfrak a^{-1})
=
-\mathfrak a^{-2} \widehat V(\mathfrak a)
=
\mathcos(r^{-\tau}).
\end{equation}
Putting these estimates together gives
\[
\widehat V(a_{22})=\mathcos(r^{-\tau}).
\]

Since $a_{22}>0$, completing the square in $m$ gives
\[
-\frac32 a_{22}m^2+2\sqrt q\,\widehat V(a_{22})\,m
\le
\mathcos(r^{-1-\tau})\,q.
\]
The curvature term satisfies
\[
-2K_{\gamma}\,a_{22}\,q
=
-\frac{2}{r^2}\,a_{22}\,q+\mathcos(r^{-2-\tau})a_{22}q
\le
\mathcos(r^{-1-\tau})q.
\]
Substituting these estimates into \eqref{eq:q-max-pre-est}, and discarding the remaining negative
part of the $m^2$ term, gives
\begin{equation}\label{eq:q-max-pre-est2}
q_\rho
\le
\left(\frac{2}{r}+\mathcos(r^{-1-\tau})\right)q
-2\langle \slashed d f,\slashed d\mathcal B\rangle_{\gamma}.
\end{equation}

We now split
\[
\mathcal B=\Xi|_{t=f}+\mathcal B_2.
\]
Since
\[
\slashed d(\Xi|_{t=f})
=
(\slashed d_\gamma\Xi)|_{t=f}
+
(\partial_t\Xi)|_{t=f}\,\slashed d f,
\]
we have
\[
-\langle \slashed d f,\slashed d\mathcal B\rangle_\gamma
=
-\big\langle \slashed d f,(\slashed d_\gamma\Xi)|_{t=f}\big\rangle_\gamma
-(\partial_t\Xi)|_{t=f}\,q
-\langle \slashed d f,\slashed d\mathcal B_2\rangle_\gamma.
\]
By asymptotic flatness,
\[
(\partial_t\Xi)|_{t=f}=\mathcos(r^{-1-\tau}),
\]
and by \eqref{eq:basic-B2-est-max},
\[
\langle \slashed d f,\slashed d\mathcal B_2\rangle_\gamma
=
\mathcos(r^{-1-\tau})\,q
\]
at the maximum point. Substituting these estimates into \eqref{eq:q-max-pre-est2} and using
\eqref{eq:X1-control-new} gives
\[
q_\rho(\rho,p_\rho)
\le
2\big|(\slashed d_\gamma\Xi)|_{t=f}\big|_{\gamma}\sqrt{q(\rho,p_\rho)}
+
\left(\frac{2}{r}+\mathcos(r^{-1-\tau})\right)q(\rho,p_\rho).
\]
This implies
\[
D^+Q(\rho)
\le
\left(
\frac{2}{r}+\mathcos(r^{-1-\tau})
\right)Q(\rho)
+
\frac{2}{r}\,\mathfrak X_1(r;I_*)\,\sqrt{Q(\rho)},
\qquad r=R-\rho.
\]

Set
\[
Y(\rho):=\sqrt{Q(\rho)}.
\]
Then, in the Dini sense,
\[
D^+Y(\rho)
\le
\left(
\frac{1}{r}+\mathcos(r^{-1-\tau})
\right)Y(\rho)
+
\frac1r\,\mathfrak X_1(r;I_*).
\]
Choose a constant $C_*>0$, depending only on the fixed tail data, such that
\[
\mathcos(r^{-1-\tau})\le C_* r^{-1-\tau}.
\]
Define the integrating factor
\[
I(\rho)
:=
r\,\exp\!\left(
-C_*\int_{R-\rho}^{R}\sigma^{-1-\tau}\,d\sigma
\right).
\]
Since $\partial_\rho r=-1$, we have
\[
\frac{D^+I}{I}
=
-\frac1r-C_*r^{-1-\tau}.
\]
Multiplying the differential inequality for $Y$ by $I(\rho)$, we obtain
\[
D^+\big(I(\rho)Y(\rho)\big)
\le
\exp\!\left(
-C_*\int_{R-\rho}^{R}\sigma^{-1-\tau}\,d\sigma
\right)\mathfrak X_1(r;I_*).
\]
Now $f|_{S_R}=T_0$, so $\slashed d f=0$ on $S_R$, hence $Q(0)=0$ and $Y(0)=0$. Integrating from
$0$ to $\rho$ gives
\[
I(\rho)Y(\rho)
\le
\int_0^\rho
\exp\!\left(
-C_*\int_{R-\sigma}^{R}s^{-1-\tau}\,ds
\right)\mathfrak X_1(R-\sigma;I_*)\,d\sigma.
\]
Changing variables gives
\[
Y(\rho)
\le
\frac{C}{r}
\int_r^R \mathfrak X_1(\sigma;I_*)\,d\sigma.
\]
Recalling that $Y(\rho)=\sqrt{Q(\rho)}$, this proves \eqref{eq:grad-local}.
\end{proof}

\begin{prop}[Uniform weighted truncated Schauder estimate]
\label{prop:schauder-truncated}
Assume that \(f\) satisfies \textbf{BS-1} and \textbf{BS-2}. Let
\(T_*=T_*(T_0;r_1)\) be the threshold from
Proposition~\ref{prop:C0-f}. Then
\[
        f-T_0\in \ct_{-\tau}(\overline{M_{r_*,R}})
\]
and
\begin{equation}\label{eq:schauder-truncated}
\|f-T_0\|_{\ct_{-\tau}(M_{r_*,R})}
\le
C\,
\sup_{t\in I_*}\|r\Xi(t,\cdot)\|_{\cz_{-\tau}(M_{r_*,R})},
\end{equation}
where
\[
C=
C\!\left(
\alpha,\tau,r_1,\underline{\mathfrak a_0}^{-1},\delta_*^{-1},\vartheta_*^{-1},
\|\mathcal S\|_{\ctS_{-\tau}(\mathcal M_{r_1})}
\right).
\]
\end{prop}

\begin{proof}
Set
\[
        v:=f-T_0.
\]
We introduce the logarithmic radial variable \(s=-\log r\), so that
\[
        r=e^{-s}.
\]
Let
\[
        S_R:=-\log R,
        \qquad
        s_*:=-\log r_*,
\]
and write
\[
        M_{S_R,s_*}:=[S_R,s_*)\times S^2 .
\]
Since \(v=0\) on \(r=R\), the outer boundary condition becomes
\[
        v=0
        \qquad\text{on }\{S_R\}\times S^2 .
\]

Multiplying \eqref{eq:tmcf-rho-form} by \(r=e^{-s}\), we obtain
\begin{equation}\label{eq:tmcf-s-form-gamma}
        v_s-r\,\mathcal A:\slashed\nabla_{\gamma}^2 v
        =
        -r\,\mathcal B .
\end{equation}
We regard this as a quasilinear parabolic equation on the \(s\)-cylinder.
Since \(M_{S_R,s_*}\) is open at \(s=s_*\), we first work on compact
cylinders
\[
        M_{S_R,s_*-\varepsilon}:=[S_R,s_*-\varepsilon]\times S^2,
        \qquad
        0<\varepsilon<s_*-S_R,
\]
and prove estimates independent of \(\varepsilon\).  We then let
\(\varepsilon\downarrow0\).

Choose a finite atlas on \(S^2\).  In each chart,
\eqref{eq:tmcf-s-form-gamma} has the local form
\begin{equation}\label{eq:tmcf-s-form}
        v_s-a^{ij}(s,x,v,\partial_xv)\,\partial_{ij}v
        =
        F(s,x,v,\partial_xv)+G_0(s,x),
\end{equation}
where
\begin{equation}\label{eq:G0-def}
        G_0(s,x):=-e^{-s}\,\Xi(T_0,e^{-s},x).
\end{equation}
Here \(F\) contains all lower-order terms except for the reference forcing
\(G_0\).  More explicitly, using
\[
        \mathcal B=\Xi|_{t=f}+\mathcal B_2,
\]
we may write schematically
\[
        F(s,x,z,p)
        =
        F_0(s,x,z,p)
        -
        e^{-s}\big(\Xi(T_0+z,e^{-s},x)-\Xi(T_0,e^{-s},x)\big),
\]
where \(F_0\) contains the contribution of \(\mathcal B_2\) and the
lower-order terms arising from writing
\(\slashed\nabla_\gamma^2 v\) in fixed local coordinates.  The important
structural point is
\begin{equation}\label{eq:F-zero-structure}
        F(s,x,0,0)=0.
\end{equation}

We first record the uniform parabolic structure.  By
\eqref{eq:A-elliptic-basic},
\[
        \mathcal A(\omega,\omega)\sim r\,|\omega|_\gamma^2 .
\]
Equivalently, as a coefficient tensor in fixed sphere coordinates,
\[
        \mathcal A\sim r\,\gamma^{-1}.
\]
Since \(r^{-2}\gamma\) is uniformly equivalent to \(\gamma_{S^2}\) on the
fixed tail \(M_{r_1,\infty}\), the tensor \(r^2\gamma^{-1}\) is uniformly
equivalent to \(\gamma_{S^2}^{-1}\).  Hence the matrix \(a^{ij}\) in
\eqref{eq:tmcf-s-form} is uniformly elliptic: there are constants
\(0<\lambda_*\le \Lambda_*<\infty\), depending only on
\[
\alpha,\tau,r_1,\underline{\mathfrak a_0}^{-1},\delta_*^{-1},\vartheta_*^{-1},
\|\mathcal S\|_{\ctS_{-\tau}(\mathcal M_{r_1})},
\]
such that
\[
        \lambda_*|\xi|^2
        \le
        a^{ij}(s,x,v,\partial_xv)\xi_i\xi_j
        \le
        \Lambda_*|\xi|^2
\]
on the bootstrap region, for all \(\xi\in\mathbb R^2\).

We next record the compactness of the lower-order arguments.  By the
\(C^0\) estimate,
\[
        |v|\le \Delta_* .
\]
For the gradient variable, we use \textbf{BS-2}.  Set
\[
        \bar\gamma(t,r):=r^{-2}\gamma(t,r).
\]
On the fixed tail \(M_{r_1,\infty}\), the metrics \(\bar\gamma(t,r)\) are
uniformly equivalent to \(\gamma_{S^2}\), with uniformly controlled
coefficients.  The condition
\[
        r^2|\slashed\nabla_\gamma^2 f|_\gamma<1
\]
is equivalent, up to constants depending only on the fixed tail, to a
uniform bound for $|\slashed\nabla_{\bar\gamma}^2 f|_{\bar\gamma}$
on each sphere \(S_r\).  On each sphere \(S_r\), the function
\(f|_{S_r}\) has a critical point.  Integrating
\(\slashed\nabla_{\bar\gamma}^2 f\) along \(\bar\gamma\)-geodesics from
such a critical point gives a uniform bound for
\(|\slashed d f|_{\bar\gamma}\).  Equivalently, in every fixed coordinate
chart on \(S^2\), \(\partial_x v=\partial_x f\) is uniformly bounded.
Thus the gradient variable remains in a compact set depending only on
\[
\tau,r_1,\underline{\mathfrak a_0}^{-1},\delta_*^{-1},\vartheta_*^{-1},
\|\mathcal S\|_{\ctS_{-\tau}(\mathcal M_{r_1})}.
\]
The coefficient bounds are uniform in the time variable, so the constants
below do not depend on \(T_0\), \(r_*\), \(R\), or \(\varepsilon\).

We now apply the standard local quasilinear Schauder estimate on unit
parabolic cylinders in the \(s\)-variable.  On each interior unit cylinder
and on the corresponding initial half-cylinder at \(s=S_R\), the estimate
for equations of the form \eqref{eq:tmcf-s-form} gives
\begin{equation}\label{eq:local-quasilinear-schauder}
        \|v\|_{\ct(Q_{1/2})}
        \le
        C\Big(
        \|v\|_{\mathc^0(Q_1)}
        +
        \|G_0\|_{\cz(Q_1)}
        \Big).
\end{equation}
The constant depends only on the ellipticity constants, the fixed-tail
coefficient bounds, and the bootstrap bounds for \((v,\partial_xv)\).  It
does not depend on the location of the cylinder, on \(r_*\), \(R\), or
\(\varepsilon\).

Let us indicate why \eqref{eq:local-quasilinear-schauder} has no
additional source term involving \(F\).  One may either quote the standard
quasilinear Schauder estimate in this form, or derive it from the linear
estimate by interpolation.  Indeed, by \eqref{eq:F-zero-structure} and the
uniform coefficient bounds,
\[
        \|F(\cdot,\cdot,v,\partial_xv)\|_{\cz(Q_1)}
        \le
        C\Big(
        \|v\|_{C^{\alpha,\alpha/2}(Q_1)}
        +
        \|\partial_xv\|_{C^{\alpha,\alpha/2}(Q_1)}
        \Big).
\]
The interpolation inequality gives, for every \(\eta>0\),
\[
        \|v\|_{C^{1+\alpha,(1+\alpha)/2}(Q_1)}
        \le
        \eta\|v\|_{\ct(Q_1)}
        +
        C_\eta\|v\|_{\mathc^0(Q_1)}.
\]
Choosing \(\eta\) sufficiently small absorbs the \(\ct\) term into the
left-hand side of the local linear estimate.  This is the standard
absorption step in the local quasilinear Schauder theory; see, for
example, \cite{parabolic1,parabolic2,parabolic3}.

We now make the estimate weighted.  Let \(Q_1=Q_1(s_0,x_0)\) be one of the
unit parabolic cylinders in the logarithmic variable, and set
\[
        r_0:=e^{-s_0}.
\]
On \(Q_1\), the radius \(r=e^{-s}\) is comparable to \(r_0\), with
constants depending only on the fixed cylinder size.  Therefore
multiplying \eqref{eq:local-quasilinear-schauder} by \(r_0^\tau\) gives
\begin{equation}\label{eq:local-weighted-quasilinear-schauder}
        \|v\|_{\ct_{-\tau}(Q_{1/2})}
        \le
        C\Big(
        \|v\|_{\mathc^0_{-\tau}(Q_1)}
        +
        \|G_0\|_{\cz_{-\tau}(Q_1)}
        \Big),
\end{equation}
with the same type of uniform constant.  The interpolation step above is
unchanged, because the weight \(r^\tau\) is comparable to the constant
weight \(r_0^\tau\) on each unit cylinder.

We cover \(M_{S_R,s_*-\varepsilon}\) by such unit cylinders and initial
half-cylinders.  The constants in
\eqref{eq:local-weighted-quasilinear-schauder} are uniform, and the number
of overlaps is uniformly bounded.  The global weighted parabolic Hölder
seminorm is controlled by the local weighted seminorms together with the
weighted \(C^0\) norm: pairs of points with parabolic distance bounded
below are controlled directly by the weighted \(C^0\) estimate, while
pairs with small parabolic distance lie in a uniformly bounded number of
coordinate cylinders.  Hence
\begin{equation}\label{eq:schauder-local-to-global}
        \|v\|_{\ct_{-\tau}(M_{S_R,s_*-\varepsilon})}
        \le
        C\Big(
        \|v\|_{\mathc^0_{-\tau}(M_{r_*,R})}
        +
        \|G_0\|_{\cz_{-\tau}(M_{S_R,s_*-\varepsilon})}
        \Big),
\end{equation}
with \(C\) independent of \(\varepsilon\), \(r_*\), and \(R\).

We estimate the source \(G_0\).  Since
\[
        G_0=-r\Xi(T_0,r,\cdot)
\]
and \(T_0\in I_*\), the definition of the weighted norm gives
\begin{equation}\label{eq:G0-bound}
        \|G_0\|_{\cz_{-\tau}(M_{S_R,s_*-\varepsilon})}
        \le
        C
        \sup_{t\in I_*}
        \|r\Xi(t,\cdot)\|_{\cz_{-\tau}(M_{r_*,R})}.
\end{equation}
The constant here depends only on \(r_1\) and \(\tau\).

It remains to prove the weighted \(C^0\) bound for \(v\).  Since
\(f(M_{r_*,R})\subset I_*\), the maximum-principle argument from
Proposition~\ref{prop:C0-f}, applied from the outer boundary \(R\) down to
a given radius \(r\in[r_*,R]\), gives
\[
        \sup_{S_r}|v|
        \le
        \int_r^R
        \sup_{\substack{t\in I_*\\ p\in S^2}}
        |\Xi(t,\sigma,p)|\,d\sigma .
\]
Using
\[
        |\Xi(t,\sigma,p)|
        \le
        \sigma^{-1-\tau}
        \sup_{t\in I_*}
        \|r\Xi(t,\cdot)\|_{C^0_{-\tau}(M_{r_*,R})},
\]
we obtain
\[
\begin{aligned}
        \sup_{S_r}|v|
        &\le
        \left(\int_r^R \sigma^{-1-\tau}\,d\sigma\right)
        \sup_{t\in I_*}
        \|r\Xi(t,\cdot)\|_{C^0_{-\tau}(M_{r_*,R})}
        \\
        &\le
        \frac{1}{\tau}\,r^{-\tau}
        \sup_{t\in I_*}
        \|r\Xi(t,\cdot)\|_{C^0_{-\tau}(M_{r_*,R})}.
\end{aligned}
\]
Therefore
\begin{equation}\label{eq:weighted-C0-v}
        \|v\|_{\mathc^0_{-\tau}(M_{r_*,R})}
        \le
        C
        \sup_{t\in I_*}
        \|r\Xi(t,\cdot)\|_{\cz_{-\tau}(M_{r_*,R})}.
\end{equation}

Combining \eqref{eq:schauder-local-to-global},
\eqref{eq:G0-bound}, and \eqref{eq:weighted-C0-v}, we obtain
\[
        \|v\|_{\ct_{-\tau}(M_{S_R,s_*-\varepsilon})}
        \le
        C
        \sup_{t\in I_*}
        \|r\Xi(t,\cdot)\|_{\cz_{-\tau}(M_{r_*,R})}.
\]
The constant is independent of \(\varepsilon\).  Letting
\(\varepsilon\downarrow0\), we obtain the same estimate on
\(M_{S_R,s_*}\), equivalently on \(M_{r_*,R}\).  This proves
\eqref{eq:schauder-truncated}.
\end{proof}

\begin{remark}\label{rem:mu-needs-fr}
The parabolicity of the TMCF equation is controlled by tangential data:
a lower bound for \(\mu_f^T\) controls the angular spacelikeness, and a
lower bound for \(r\mathfrak a_f\) controls the ellipticity of the
equation.  This does not imply spacelikeness of the full graph.  The
full spacelikeness factor is
\[
        \mu_f
        =
        \mu_f^T-\frac{N^2}{\lambda^2}f_r^2 .
\]
Thus even after controlling the angular slope and the parabolicity
coefficient, one still needs a separate estimate for the dimensionless
radial slope $\frac{N}{\lambda}|f_r|$ in order to conclude that \(\mu_f>0\).
\end{remark}

We now write
\[
\slashed\nabla:=\slashed\nabla_{\gamma},
\qquad
U:=\slashed\nabla_{\gamma}^2 f.
\]
All norms $|\cdot|$ below are taken with respect to $\gamma$. We also abbreviate
\begin{equation}\label{eq:X2-global}
\mathfrak X_2(r)
:=
\mathfrak X_2(r;I_*)
=
r^2\sup_{\substack{t\in I_*\\ \omega\in S^2}}
|\slashed\nabla_{\gamma}^2\Xi(t,r,\omega)|_{\gamma(t,r)},
\qquad r\in[r_*,R].
\end{equation}

Finally, for any coefficient $F(t,r,p)$ evaluated on the graph $t=f(r,p)$, we denote by
\[
\mathcal D_\rho F
:=
\partial_{\rho}\big(F(f(r,p),r,p)\big)
=
(-\partial_r+f_\rho\partial_t)F
\]
its $\rho$-derivative along the graph.

\begin{lem}[Basic coefficient estimates for the Hessian analysis]
\label{lem:hessian-coeff}
 There exists a constant
\[
C=
C\!\left(
\tau,r_1,\underline{\mathfrak a_0}^{-1},\delta_*^{-1},\vartheta_*^{-1},
\|\mathcal S\|_{\ctS_{-\tau}(\mathcal M_{r_1})}
\right)
\]
such that the following estimates hold on $M_{r_*,R}$:
\begin{align}
|\slashed\nabla \mu^T|
&\le
C\big(r^{-1-\tau}+|U|\big),
\label{eq:mu1-final}
\\
|\slashed\nabla^2 \mu^T|
&\le
C\big(r^{-2-\tau}+|\slashed\nabla U|+|U|^2\big),
\label{eq:mu2-final}
\\
|\slashed\nabla \mathfrak a|
&\le
C\big(r^{-2-\tau}+r^{-1}|U|\big),
\label{eq:a1-final}
\\
|\slashed\nabla^2 \mathfrak a|
&\le
C\big(r^{-3-\tau}+r^{-1}|\slashed\nabla U|+r^{-1}|U|^2\big),
\label{eq:a2-final}
\\
|\slashed\nabla \mathcal A|
&\le
C\big(r^{-\tau}+r|U|\big),
\label{eq:A1-final}
\\
|\slashed\nabla^2 \mathcal A|
&\le
C\big(r^{-1-\tau}+r|\slashed\nabla U|+r|U|^2\big).
\label{eq:A2-final}
\end{align}
Moreover,
\begin{align}
|\slashed\nabla^2 \mathcal B|
&\le
\frac{C}{r^2}\,\mathfrak X_2(r)
+
C r^{-\tau}|\slashed\nabla U|
+
C r^{-1-\tau}|U|
+
C r^{1-\tau}|U|^2
\notag\\
&\quad
+
C r^{-2-\tau}\big(|\slashed d f|_\gamma+|\slashed d f|_\gamma^2\big).
\label{eq:B2-final-corrected}
\end{align}
Consequently, using proposition~\ref{prop:grad-f},
\begin{align}
|\slashed\nabla^2 \mathcal B|
&\le
\frac{C}{r^2}\,\mathfrak X_2(r)
+
C r^{-\tau}|\slashed\nabla U|
+
C r^{-1-\tau}|U|
+
C r^{1-\tau}|U|^2
\notag\\
&\quad
+
C r^{-3-\tau} \int_r^R \mathfrak X_1(\sigma;I_*)\,d\sigma
+
C r^{-4-\tau} \left( \int_r^R \mathfrak{X}_1(\sigma;I_*)\,d\sigma\right)^2.
\label{eq:B2-final-P1}
\end{align}
\end{lem}

\begin{proof}
Throughout the proof, $C$ denotes a positive constant depending only on the fixed tail data, and
may change from line to line.

The estimates for $\mu^T$ follow directly from
\[
\mu^T
=
\big(1+\langle \beta^T,\slashed d f\rangle_\gamma\big)^2
-
N^2|\slashed d f|_\gamma^2.
\]
Indeed, differentiating once gives terms involving tangential derivatives of the background
coefficients and one factor of $U$, and differentiating twice gives terms involving two tangential
derivatives of the background coefficients, one factor of $\slashed\nabla U$, and quadratic terms
in $U$. Using the fixed-tail decay of the coefficients gives
\eqref{eq:mu1-final}--\eqref{eq:mu2-final}.

For $\mathfrak a$, recall that
\[
\mathfrak a
=
\frac{1}{2\lambda^2}\,\tr_{\gamma_f}Q_r(\slashed d f).
\]
The tensor $Q_r(\slashed d f)$ is affine-quadratic in $\slashed d f$, with coefficients
$\partial_r\gamma$, $\partial_r\beta^T$, and $\partial_r\alpha$. Its first tangential derivative
therefore contains background terms of size $r^{-2-\tau}$ and terms of size $r^{-1}|U|$; its
second tangential derivative contains terms of size
\[
r^{-3-\tau},
\qquad
r^{-1}|\slashed\nabla U|,
\qquad
r^{-1}|U|^2.
\]
The same schematic bounds hold for the tangential derivatives of $\gamma_f^{-1}$, using
\eqref{eq:gammaf-inverse} and $\mu^T>1/4$. This proves
\eqref{eq:a1-final}--\eqref{eq:a2-final}. Differentiating
\[
\mathcal A=\mathfrak a^{-1}\gamma_f^{-1}
\]
then gives \eqref{eq:A1-final}--\eqref{eq:A2-final}, since $\mathfrak a^{-1}=\mathcos(r)$.

It remains to estimate $\slashed\nabla^2\mathcal B$. We split
\[
\mathcal B=\Xi|_{t=f}+\mathcal B_2.
\]
For the forcing term, the chain rule gives
\[
\slashed\nabla^2(\Xi|_{t=f})
=
(\slashed\nabla_\gamma^2\Xi)|_{t=f}
+
(\partial_t\Xi)|_{t=f}\,U
+
2(\slashed d_\gamma\partial_t\Xi)|_{t=f}\odot\slashed d f
+
(\partial_t^2\Xi)|_{t=f}(\slashed d f)^2.
\]
Therefore,
\[
|\slashed\nabla^2(\Xi|_{t=f})|
\le
\frac{C}{r^2}\mathfrak X_2(r)
+
C r^{-1-\tau}|U|
+
C r^{-2-\tau}\big(|\slashed d f|_\gamma+|\slashed d f|_\gamma^2\big).
\]

For $\mathcal B_2$, write
\[
\mathcal B_2
=
\mathfrak B_2(r,p,t,\omega)\Big|_{t=f,\omega=\slashed d f}.
\]
Since $\mathfrak B_2(r,p,t,0)\equiv0$, its pure $(r,p,t)$-derivatives vanish at $\omega=0$.
Differentiating the explicit expression \eqref{eq:B-explicit-gauge}, one obtains the schematic
bounds
\[
|\partial_\omega\mathfrak B_2|\le C r^{-\tau},
\qquad
|\partial_{\omega\omega}^2\mathfrak B_2|\le C r^{1-\tau},
\]
\[
|\partial_x\partial_\omega\mathfrak B_2|+|\partial_t\partial_\omega\mathfrak B_2|
\le C r^{-1-\tau},
\]
and
\[
|\partial_x^2\mathfrak B_2|+
|\partial_t\partial_x\mathfrak B_2|
+
|\partial_t^2\mathfrak B_2|
\le
C r^{-2-\tau}|\omega|.
\]
Applying the chain rule twice gives
\[
|\slashed\nabla^2\mathcal B_2|
\le
C r^{-\tau}|\slashed\nabla U|
+
C r^{-1-\tau}|U|
+
C r^{1-\tau}|U|^2
+
C r^{-2-\tau}\big(|\slashed d f|_\gamma+|\slashed d f|_\gamma^2\big).
\]
Combining the estimates for $\Xi|_{t=f}$ and $\mathcal B_2$ proves
\eqref{eq:B2-final-corrected}.
\end{proof}

\begin{lem}[Maximum-point inequality for $\mathfrak q=|U|^2$]
\label{lem:hessian-max-ode}
Assume that $f$ satisfies \textbf{BS-2}. Let
\[
\mathfrak q:=|U|^2,
\qquad
M(\rho):=\sup_{S_{R-\rho}}\mathfrak q, \qquad \mathfrak P_1(r) 
:=
\int_r^R \mathfrak X_1(\sigma;I_*)\,d\sigma, 
\qquad
r:=R-\rho.
\]
Then the following differential inequality holds in the upper Dini sense:
\begin{align}
D^+M(\rho)
&\le
\left(
\frac{4}{r}
+
C r^{-1-\tau}
+
C r^{-2-\tau}\mathfrak P_1(r)
+
C r^{-3}\mathfrak P_1(r)^2
\right)M(\rho)
\notag\\
&\quad
+
C r\,M(\rho)^2
+
C r^{1-\tau}M(\rho)^{3/2}
+
\frac{C}{r^2}\mathfrak X_2(r)\sqrt{M(\rho)}
\notag\\
&\quad
+
C r^{-3-\tau}\mathfrak P_1(r)\sqrt{M(\rho)}
+
C r^{-4-\tau}\mathfrak P_1(r)^2\sqrt{M(\rho)}.
\label{eq:M-ODE-final-corrected}
\end{align}
\end{lem}

\begin{proof}
Fix $\rho\in[0,R-r_*)$, let $r=R-\rho$, and choose a point $p_\rho\in S_r$ such that
\[
\mathfrak q(\rho,p_\rho)=M(\rho)=\sup_{S_r}\mathfrak q.
\]
At $p_\rho$, choose $\gamma$-normal coordinates. Then
\[
\gamma_{AB}=\delta_{AB},
\qquad
\Gamma_{AB}^C(\gamma)=0
\qquad\text{at }p_\rho,
\]
and
\[
\slashed\nabla_A\mathfrak q=0,
\qquad
\mathcal A:\slashed\nabla_\gamma^2\mathfrak q\le0
\qquad\text{at }p_\rho.
\]

We now differentiate the equation
\[
f_\rho-\mathcal A^{ab}U_{ab}=-\mathcal B
\]
twice tangentially. Since $U_{ij}=\slashed\nabla_i\slashed\nabla_jf$, differentiating in $\rho$
also differentiates the connection of $\gamma$. Thus
\[
(U_{ij})_\rho
=
\slashed\nabla_i\slashed\nabla_j f_\rho
-
(\mathcal D_\rho\Gamma)^k_{ij}\,\slashed\nabla_k f.
\]
Substituting $f_\rho=\mathcal A^{ab}U_{ab}-\mathcal B$ gives
\begin{equation}\label{eq:Ueq-before-curvature-split}
(U_{ij})_\rho-\mathcal A^{ab}U_{ij;ab}
=
-\mathcal B_{;ij}
+
\mathcal C_{ij}
+
\mathcal T_{ij}.
\end{equation}
Here $\mathcal C$ is the leading constant-curvature commutator term, and $\mathcal T$ contains the
remaining perturbative terms. We now describe these two pieces.

On the fixed tail,
\[
K_\gamma=\frac1{r^2}+\mathcos(r^{-2-\tau}),
\]
and
\[
\mathcal A=\frac r2\,\gamma^{-1}+\mathcal E,
\]
where
\[
|\mathcal E|_\gamma
\le
C r^{1-\tau}
+
C r^{1-\tau}|\slashed d f|_\gamma
+
C r|\slashed d f|_\gamma^2.
\]
Using Proposition~\ref{prop:grad-f}, this gives
\[
|\mathcal E|_\gamma
\le
C r^{1-\tau}
+
C r^{-\tau}\mathfrak P_1(r)
+
C r^{-1}\mathfrak P_1(r)^2.
\]

The standard commutator identity for the Hessian of a scalar on a two-dimensional surface is
\[
\slashed\nabla_i\slashed\nabla_j\slashed\Delta f
-
\slashed\Delta U_{ij}
=
-4K_\gamma U_{ij}
+
2K_\gamma\gamma_{ij}\tr_\gamma U
+
(\slashed\nabla K_\gamma)*\slashed d f.
\]
Keeping only the leading piece $K_\gamma=r^{-2}$ and
$\mathcal A=(r/2)\gamma^{-1}$ gives
\[
\mathcal C_{ij}
=
-\frac{2}{r}U_{ij}
+
\frac1r\gamma_{ij}\tr_\gamma U.
\]
All other curvature-commutator contributions are placed in $\mathcal T$. From the estimates above,
these perturbative curvature terms are bounded by
\[
C r^{-1-\tau}|U|
+
C r^{-2-\tau}\mathfrak P_1(r)|U|
+
C r^{-3}\mathfrak P_1(r)^2|U|
+
C r^{-3-\tau}\mathfrak P_1(r).
\]

The remaining terms in $\mathcal T$ come from derivatives of $\mathcal A$ and from
$\mathcal D_\rho\Gamma$. Using Lemma~\ref{lem:hessian-coeff}, Proposition~\ref{prop:grad-f}, and the
equation for $f_\rho$, they satisfy
\begin{align}
|\mathcal T|
&\le
C\big(r^{-\tau}+r|U|\big)|\slashed\nabla U|
+
C\big(r^{-1-\tau}+r|U|^2\big)|U|
\notag\\
&\quad
+
C r^{-2-\tau}\mathfrak P_1(r)|U|
+
C r^{-3}\mathfrak P_1(r)^2|U|
+
C r^{-3-\tau}\mathfrak P_1(r)
+
C r^{-4-\tau}\mathfrak P_1(r)^2.
\label{eq:T-final-corrected-no-Cr}
\end{align}
Thus \eqref{eq:Ueq-before-curvature-split} becomes
\begin{equation}\label{eq:Ueq-final-corrected}
(U_{ij})_\rho-\mathcal A^{ab}U_{ij;ab}
=
-\mathcal B_{;ij}
+
\left(
-\frac{2}{r}U_{ij}
+
\frac1r\gamma_{ij}\tr_\gamma U
\right)
+
\mathcal T_{ij}.
\end{equation}

We next compute the equation for $\mathfrak q=|U|^2$. Since
\[
(\gamma^{-1})_\rho
=
\frac2r\gamma^{-1}
+
\mathcos(r^{-1-\tau})\gamma^{-1},
\]
we have
\[
\mathfrak q_\rho
=
\left(\frac{4}{r}+\mathcos(r^{-1-\tau})\right)\mathfrak q
+
2\langle U,U_\rho\rangle.
\]
At the normal-coordinate point,
\[
\mathcal A:\slashed\nabla_\gamma^2\mathfrak q
=
2\mathcal A^{ab}\langle U_{;a},U_{;b}\rangle
+
2\langle U,\mathcal A^{ab}U_{;ab}\rangle.
\]
Combining this identity with \eqref{eq:Ueq-final-corrected}, we get
\begin{align}
\frac12\big(\mathfrak q_\rho-\mathcal A:\slashed\nabla_\gamma^2\mathfrak q\big)
&=
\left(\frac{2}{r}+\mathcos(r^{-1-\tau})\right)\mathfrak q
-
\mathcal A^{ab}\langle U_{;a},U_{;b}\rangle
\notag\\
&\quad
-\langle U,\slashed\nabla^2\mathcal B\rangle
+
\left\langle U,
-\frac{2}{r}U+\frac1r\gamma\,\tr_\gamma U
\right\rangle
+
\langle U,\mathcal T\rangle.
\label{eq:Qmain-final-corrected}
\end{align}
The leading curvature contribution is therefore
\[
\left(\frac{2}{r}\right)|U|^2
+
\left\langle U,
-\frac{2}{r}U+\frac1r\gamma\,\tr_\gamma U
\right\rangle
=
\frac1r(\tr_\gamma U)^2.
\]
Since $(\tr_\gamma U)^2\le 2|U|^2$ in dimension two, this is bounded by
\[
\frac{2}{r}\mathfrak q.
\]
Thus the leading curvature term does not produce any additional $C/r$ error beyond the existing
good term $\frac{4}{r}M$ in the final inequality.

The ellipticity bound \eqref{eq:A-elliptic-basic} gives
\[
\mathcal A^{ab}\langle U_{;a},U_{;b}\rangle
\ge
\frac{r}{C}|\slashed\nabla U|^2.
\]
Using \eqref{eq:B2-final-P1} and \eqref{eq:T-final-corrected-no-Cr} in
\eqref{eq:Qmain-final-corrected}, and then applying Young's inequality to the terms involving
$|\slashed\nabla U|$, gives
\begin{align*}
\frac12\big(\mathfrak q_\rho-\mathcal A:\slashed\nabla_\gamma^2\mathfrak q\big)
&\le
\left(
\frac{2}{r}
+
C r^{-1-\tau}
+
C r^{-2-\tau}\mathfrak P_1(r)
+
C r^{-3}\mathfrak P_1(r)^2
\right)\mathfrak q
\\
&\quad
+
C r\,\mathfrak q^2
+
C r^{1-\tau}\mathfrak q^{3/2}
+
\frac{C}{r^2}\mathfrak X_2(r)\mathfrak q^{1/2}
\\
&\quad
+
C r^{-3-\tau}\mathfrak P_1(r)\mathfrak q^{1/2}
+
C r^{-4-\tau}\mathfrak P_1(r)^2\mathfrak q^{1/2}.
\end{align*}
Since $\mathcal A:\slashed\nabla_\gamma^2\mathfrak q\le0$ at the maximum point, this implies
\eqref{eq:M-ODE-final-corrected}.
\end{proof}

\begin{prop}[Tail estimate for the tangential Hessian]
\label{prop:hessian-tail-final}
Assume that $f$ satisfies \textbf{BS-2}. Then there exists a constant
\[
C
=
C\!\left(
\tau,r_1,\underline{\mathfrak a_0}^{-1},\delta_*^{-1},\vartheta_*^{-1},
\|\mathcal S\|_{\ctS_{-\tau}(\mathcal M_{r_1})}
\right)
\]
such that for every $(r,p)\in M_{r_*,R}$,
\begin{equation}\label{eq:hessian-tail-final}
|\slashed\nabla^2 f|_{\gamma}(r,p)
\le
\frac{C}{r^2}
\int_r^R
\mathfrak X_2(\sigma;I_*)\,d\sigma.
\end{equation}
\end{prop}

\begin{proof}
Let
\[
M(\rho):=\sup_{S_r}|U|^2,
\qquad
Z(\rho):=r^2\sqrt{M(\rho)},
\qquad
r=R-\rho.
\]
By \textbf{BS-2}, we have
\begin{equation}\label{eq:Z-bootstrap-less-one}
0\le Z(\rho)<1
\qquad\text{for all }\rho\in[0,R-r_*).
\end{equation}

We now use the maximum-point inequality from Lemma~\ref{lem:hessian-max-ode}.
In the
notation above, it gives
\begin{align}
D^+M(\rho)
&\le
\left(
\frac{4}{r}
+
C r^{-1-\tau}
+
C r^{-2-\tau}\mathfrak P_1(r)
+
C r^{-3}\mathfrak P_1(r)^2
\right)M(\rho)
\notag\\
&\quad
+
C r\,M(\rho)^2
+
C r^{1-\tau}M(\rho)^{3/2}
+
\frac{C}{r^2}\mathfrak X_2(r;I_*)\sqrt{M(\rho)}
\notag\\
&\quad
+
C r^{-3-\tau}\mathfrak P_1(r)\sqrt{M(\rho)}
+
C r^{-4-\tau}\mathfrak P_1(r)^2\sqrt{M(\rho)}.
\label{eq:M-ODE-for-Z-proof}
\end{align}
Set
\[
Y(\rho):=r^4M(\rho)=Z(\rho)^2.
\]
Since $r_\rho=-1$, we compute
\[
D^+Y(\rho)
=
-4r^3M(\rho)+r^4D^+M(\rho).
\]
Substituting \eqref{eq:M-ODE-for-Z-proof}, and using
\[
M=\frac{Z^2}{r^4},
\qquad
\sqrt M=\frac{Z}{r^2},
\]
the leading $4r^{-1}M$ term cancels with $-4r^3M$. We obtain
\begin{align}
D^+Y
&\le
C\Big(
r^{-1-\tau}
+
r^{-2-\tau}\mathfrak P_1(r)
+
r^{-3}\mathfrak P_1(r)^2
\Big)Z^2
\notag\\
&\quad
+
C r^{-3}Z^4
+
C r^{-1-\tau}Z^3
+
C\mathfrak X_2(r;I_*)Z
\notag\\
&\quad
+
C r^{-1-\tau}\mathfrak P_1(r)Z
+
C r^{-2-\tau}\mathfrak P_1(r)^2Z.
\label{eq:Y-ODE-with-P}
\end{align}
If $Z(\rho)>0$, dividing \eqref{eq:Y-ODE-with-P} by $2Z(\rho)$ gives
\begin{align}
D^+Z(\rho)
&\le
C\Big(
r^{-1-\tau}
+
r^{-2-\tau}\mathfrak P_1(r)
+
r^{-3}\mathfrak P_1(r)^2
\Big)Z(\rho)
\notag\\
&\quad
+
C r^{-3}Z(\rho)^3
+
C r^{-1-\tau}Z(\rho)^2
+
C\mathfrak X_2(r;I_*)
\notag\\
&\quad
+
C r^{-1-\tau}\mathfrak P_1(r)
+
C r^{-2-\tau}\mathfrak P_1(r)^2.
\label{eq:Z-ODE-with-P}
\end{align}
If $Z(\rho)=0$, the same inequality follows by the standard upper-Dini argument. Using
\eqref{eq:Z-bootstrap-less-one}, we absorb the nonlinear terms $Z^2$ and $Z^3$ into the linear
coefficient and obtain
\begin{align}
D^+Z(\rho)
&\le
a(r)\,Z(\rho)+b(r),
\label{eq:Z-linear-with-P}
\end{align}
where
\begin{align}
a(r)
&:=
C\Big(
r^{-1-\tau}
+
r^{-3}
+
r^{-2-\tau}\mathfrak P_1(r)
+
r^{-3}\mathfrak P_1(r)^2
\Big),
\label{eq:a-r-hess-proof}
\\
b(r)
&:=
C\Big(
\mathfrak X_2(r;I_*)
+
r^{-1-\tau}\mathfrak P_1(r)
+
r^{-2-\tau}\mathfrak P_1(r)^2
\Big).
\label{eq:b-r-hess-proof}
\end{align}

We now relate $\mathfrak P_1$ to $\mathfrak X_2$. We use the following scale-invariant estimate on
the spheres. For every smooth function $h$ on $S_r$,
\begin{equation}\label{eq:sphere-X1-X2}
r\,\|\slashed d h\|_{L^\infty(S_r,\gamma)}
\le
C\,r^2\|\slashed\nabla_\gamma^2 h\|_{L^\infty(S_r,\gamma)}.
\end{equation}
Indeed, after rescaling $\gamma$ to $r^{-2}\gamma$, the family of metrics is uniformly controlled
on the fixed tail. If \eqref{eq:sphere-X1-X2} failed, a compactness argument would produce a
nonzero parallel vector field on $S^2$, which is impossible. Applying \eqref{eq:sphere-X1-X2} to
$h=\Xi(t,r,\cdot)$ and taking the relevant suprema gives
\begin{equation}\label{eq:X1-by-X2}
\mathfrak X_1(r;I_*)\le C\,\mathfrak X_2(r;I_*).
\end{equation}
Consequently, setting 
\[\mathfrak P_2(r) := \int_r^R \mathfrak{X}_2(\sigma;I_*)d\sigma,\]
we have
\begin{equation}\label{eq:P1-by-I2}
\mathfrak P_1(r)
\le
C\,\mathfrak{P}_2(r).
\end{equation}

We now prove the desired bound for $Z$. First note that
\[
Z(0)=0,
\]
because $f=T_0$ on $S_R$, and hence all tangential derivatives of $f$ vanish on $S_R$.

If $\mathfrak{P}_2(r)\ge 1$, then \eqref{eq:Z-bootstrap-less-one} immediately gives
\[
Z(\rho)<1\le \mathfrak{P}_2(r),
\]
so the desired estimate follows in this case after increasing $C$.

It remains to treat the case $\mathfrak{P}_2(r)<1$. Since $\mathfrak{P}_2$ is nonincreasing as $r$ decreases and
$\mathfrak P_1(\sigma)\le C \mathfrak{P}_2(\sigma)\le C \mathfrak{P}_2(r)$ for $\sigma\in[r,R]$, we have
\[
\int_r^R a(\sigma)\,d\sigma
\le
C.
\]
Indeed, the terms $\sigma^{-1-\tau}$ and $\sigma^{-3}$ are integrable on $[r_1,\infty)$, while the
terms involving $\mathfrak P_1$ are bounded by the same integrable weights multiplied by powers of
$\mathfrak{P}_2(r)<1$.

Let
\[
\mathcal I(\rho)
:=
\exp\left(
-\int_0^\rho a(R-s)\,ds
\right).
\]
Multiplying \eqref{eq:Z-linear-with-P} by $\mathcal I$ and integrating from $0$ to $\rho$ gives
\[
Z(\rho)
\le
\exp\left(\int_r^R a(\sigma)\,d\sigma\right)
\int_r^R b(\sigma)\,d\sigma
\le
C\int_r^R b(\sigma)\,d\sigma.
\]
Using \eqref{eq:P1-by-I2}, we estimate
\[
\int_r^R b(\sigma)\,d\sigma
\le
C \mathfrak{P}_2(r)
+
C\int_r^R \sigma^{-1-\tau}\mathfrak{P}_2(\sigma)\,d\sigma
+
C\int_r^R \sigma^{-2-\tau}\mathfrak{P}_2(\sigma)^2\,d\sigma.
\]
Since $\mathfrak{P}_2(\sigma)\le \mathfrak{P}_2(r)<1$ for $\sigma\in[r,R]$, and since the weights are integrable on
$[r_1,\infty)$, this gives
\[
\int_r^R b(\sigma)\,d\sigma
\le
C \mathfrak{P}_2(r).
\]
Therefore
\[
Z(\rho)\le C \mathfrak{P}_2(r).
\]

Combining the two cases, we have shown that for every $r\in[r_*,R]$,
\[
Z(\rho)
\le
C\int_r^R \mathfrak X_2(\sigma;I_*)\,d\sigma.
\]
Recalling that
\[
Z(\rho)=r^2\sup_{S_r}|\slashed\nabla^2 f|_\gamma,
\]
we conclude
\[
\sup_{S_r}|\slashed\nabla^2 f|_\gamma
\le
\frac{C}{r^2}
\int_r^R \mathfrak X_2(\sigma;I_*)\,d\sigma.
\]
This proves \eqref{eq:hessian-tail-final}.
\end{proof}

\begin{prop}[Pointwise bound for the radial slope]
\label{prop:fr-tail-final}
Assume that $f$ satisfies \textbf{BS-2}. Then there exists a constant
\[
C
=
C\!\left(
\tau,r_1,\underline{\mathfrak a_0}^{-1},\delta_*^{-1},\vartheta_*^{-1},
\|\mathcal S\|_{\ctS_{-\tau}(\mathcal M_{r_1})}
\right)
\]
such that for every $(r,p)\in M_{r_*,R}$,
\begin{equation}\label{eq:fr-tail-final}
\frac{N}{\lambda}|f_r|(r,p)
\le
\frac{C}{r}\,\mathfrak X_0(r;I_*)
+
\frac{C}{r}\int_r^R \mathfrak X_2(\sigma;I_*)\,d\sigma.
\end{equation}
Equivalently,
\begin{equation}\label{eq:fr-tail-final-expanded}
\frac{N}{\lambda}|f_r|(r,p)
\le
C\left(
\sup_{\substack{t\in I_*\\ \omega\in S^2}}
|\Xi(t,r,\omega)|
+
\frac1r
\int_r^R
\sigma^2
\sup_{\substack{t\in I_*\\ \omega\in S^2}}
|\slashed\nabla^2_\gamma\Xi(t,\sigma,\omega)|_{\gamma(t,\sigma)}
\,d\sigma
\right).
\end{equation}
\end{prop}

\begin{proof}
Recall the $\rho$-equation
\[
f_\rho-\mathcal A : \slashed\nabla_\gamma^2 f=-\mathcal B,
\qquad
\mathcal B=\Xi|_{t=f}+\mathcal B_2.
\]
Since $f_\rho=-f_r$, this becomes
\[
f_r
=
\mathcal B-\mathcal A : \slashed\nabla_\gamma^2 f
=
\Xi|_{t=f}
+
\mathcal B_2
-
\mathcal A : U.
\]
Therefore
\begin{equation}\label{eq:fr-split-final}
\frac{N}{\lambda}|f_r|
\le
\frac{N}{\lambda}|\Xi|_{t=f}
+
\frac{N}{\lambda}|\mathcal B_2|
+
\frac{N}{\lambda}|\mathcal A:U|.
\end{equation}

We now estimate the three terms on the right-hand side of \eqref{eq:fr-split-final}. By
Proposition~\ref{prop:C0-f}, we know that $f(r,p)\in I_*$ on $M_{r_*,R}$. Hence
\[
\frac{N}{\lambda}|\Xi|_{t=f}(r,p)
\le
\mathcos(1)\,
\sup_{\substack{t\in I_*\\ \omega\in S^2}}
|\Xi(t,r,\omega)|
=
\frac{\mathcos(1)}{r}\,\mathfrak X_0(r;I_*).
\]

We next estimate the principal Hessian term. By \eqref{eq:A-elliptic-basic}, the operator norm of
$\mathcal A$ with respect to $\gamma$ satisfies
\[
|\mathcal A|_\gamma\le \mathcos(r).
\]
Therefore
\[
\frac{N}{\lambda}|\mathcal A:U|
\le
\mathcos(r)\,|U|.
\]
Applying Proposition~\ref{prop:hessian-tail-final}, we obtain
\[
|U|(r,p)
\le
\frac{C}{r^2}
\int_r^R
\mathfrak X_2(\sigma;I_*)\,d\sigma.
\]
Thus
\begin{equation}\label{eq:A-U-final}
\frac{N}{\lambda}|\mathcal A:U|
\le
\frac{C}{r}
\int_r^R
\mathfrak X_2(\sigma;I_*)\,d\sigma.
\end{equation}

We now estimate the lower-order remainder $\mathcal B_2$. From
Lemma~\ref{lem:coeff-structure}\textup{(iv)}, we have
\[
|\mathcal B_2|
\le
\mathcos(r^{-\tau})\,|\slashed d f|_\gamma.
\]
Since $N/\lambda=\mathcos(1)$ on the fixed tail, Proposition~\ref{prop:grad-f} gives
\[
\frac{N}{\lambda}|\mathcal B_2|
\le
C r^{-\tau}|\slashed d f|_\gamma
\le
C r^{-\tau}\frac{1}{r}
\int_r^R \mathfrak X_1(\sigma;I_*)\,d\sigma.
\]
Using the scale-invariant sphere estimate \eqref{eq:X1-by-X2}, justified in the proof of
Proposition~\ref{prop:hessian-tail-final}, we have
\[
\mathfrak X_1(\sigma;I_*)\le C\,\mathfrak X_2(\sigma;I_*).
\]
Since $r\ge r_1$, the factor $r^{-\tau}$ is bounded by a fixed-tail constant, and hence
\begin{equation}\label{eq:B2-fr-bound-final}
\frac{N}{\lambda}|\mathcal B_2|
\le
\frac{C}{r}
\int_r^R \mathfrak X_2(\sigma;I_*)\,d\sigma.
\end{equation}

Combining \eqref{eq:fr-split-final}, \eqref{eq:A-U-final}, and
\eqref{eq:B2-fr-bound-final} yields
\[
\frac{N}{\lambda}|f_r|(r,p)
\le
\frac{C}{r}\,\mathfrak X_0(r;I_*)
+
\frac{C}{r}\int_r^R \mathfrak X_2(\sigma;I_*)\,d\sigma,
\]
which proves \eqref{eq:fr-tail-final}. The expanded form \eqref{eq:fr-tail-final-expanded}
follows from the definitions of $\mathfrak X_0$ and $\mathfrak X_2$.
\end{proof}

\subsection{Local-in-$r$ Existence} \label{sub:localinr}

In this subsection we prove a local existence theorem in the inward radial variable. We fix a tail
radius $r_1>r_0$ and work in a coordinate chart on $(\mathcal M,\gtime)$ whose coefficients
\[
\mathcal S=(N,\lambda,\beta,b,\gamma)
\]
belong to $\ctS_{-\tau}(\mathcal M)$. We impose the good gauge
\[
b\equiv 0,
\qquad
\beta_r\equiv 0.
\]

This is precisely the setting in which Proposition~\ref{prop:gammab-quasilinear} yields an
explicit quasilinear parabolic form of the TMCF equation. In the next subsection this good gauge
will be produced on late slabs away from the horizon by the coordinate reduction proved in
Appendix~\ref{app:late-slab-gauge}.

As in the previous subsection, we write
\[
\vartheta_*:=\vartheta_*(r_1),
\qquad
\underline{\mathfrak a_0}:=\underline{\mathfrak a_0}(r_1).
\]

Fix $R_1>r_1$, and let $\varphi\in C^{2+\alpha}(S^2)$ be a function with
\[
\varphi(S^2)\subset(\underline T,\infty).
\]
We identify $\varphi$ with a function on the sphere $S_{R_1}$. Evaluating the background
coefficients on the graph $t=\varphi$ at $r=R_1$, we obtain the boundary quantities
\[
\mu^T[\varphi],
\qquad
\mathfrak a[\varphi],
\qquad
\mathcal A[\varphi],
\qquad
\mathcal B[\varphi],
\]
all defined on $S_{R_1}$. Explicitly,
\begin{equation}\label{eq:muT-phi-def}
\mu^T[\varphi]
=
\Big(1+\langle \beta[\varphi]^T,\slashed d\varphi\rangle_{\gamma[\varphi]}\Big)^2
-
N[\varphi]^2\,|\slashed d\varphi|_{\gamma[\varphi]}^2,
\end{equation}
and
\begin{equation}\label{eq:a-phi-def}
\mathfrak a[\varphi]
=
\frac{1}{2\lambda[\varphi]^2}\,
\tr_{\gamma_f[\varphi]}Q_r(\slashed d\varphi),
\end{equation}
where
\[
\gamma_f[\varphi]
=
\gamma[\varphi]
+2\,\beta[\varphi]^T\odot \slashed d\varphi
-\alpha[\varphi]\,(\slashed d\varphi)^2.
\]
Here, and below, $[\,\cdot\,]$ means that the quantity is evaluated at
\[
(t,r,p)=\big(\varphi(p),R_1,p\big).
\]
The corresponding principal tensor and lower-order term are
\[
\mathcal A[\varphi]
=
\frac{1}{\mathfrak a[\varphi]}\,\gamma_f[\varphi]^{-1},
\qquad
\mathcal B[\varphi]
=
\mathcal B\big(R_1,p,\varphi,\slashed d\varphi\big).
\]

We also define the radial derivative forced by the TMCF equation on the initial sphere. Since
\[
f_\rho-\mathcal A:\slashed\nabla_\gamma^2 f=-\mathcal B,
\qquad
f_\rho=-f_r,
\]
any $C^2$-solution whose trace on $S_{R_1}$ is $\varphi$ must satisfy
\begin{equation}\label{eq:fr-phi-def}
f_r[\varphi]
:=
\mathcal B[\varphi]-\mathcal A[\varphi]:(\slashed\nabla^2_{\gamma[\varphi]}\varphi).
\end{equation}
Thus $f_r[\varphi]$ is determined entirely by the boundary profile $\varphi$ and the background
spacetime.

Finally, we define the corresponding spacelikeness quantity by
\begin{equation}\label{eq:mu-phi-def}
\mu[\varphi]
:=
\mu^T[\varphi]-\frac{N[\varphi]^2}{\lambda[\varphi]^2}\,f_r[\varphi]^2.
\end{equation}
Geometrically, $\mu[\varphi]$ is exactly the value of the full graph spacelikeness indicator
$\mu$ on the initial sphere that any $C^2$-solution with trace $\varphi$ would have.

The bootstrap assumptions \textbf{BS-1} and \textbf{BS-2} were defined for functions on annuli in
Definition~\ref{def:BS1-BS2}. We now use the corresponding boundary versions for profiles on a
single sphere.

\begin{defn}[Boundary versions of \textbf{BS-1} and \textbf{BS-2}]
\label{def:admissible-boundary-profile}
Let $\varphi\in C^{2+\alpha}(S^2)$ be a profile on $S_{R_1}$.

We say that $\varphi$ satisfies \textbf{BS-1} on $S_{R_1}$ if
\begin{equation}\label{eq:parabolic-profile-admissible}
|\slashed d\varphi|_{\gamma[\varphi]}<\frac{\vartheta_*}{2},
\qquad
R_1\,\mathfrak a[\varphi]>\frac12\,\underline{\mathfrak a_0}
\quad\text{on }S_{R_1}.
\end{equation}
We say that $\varphi$ satisfies \textbf{BS-2} on $S_{R_1}$ if
\begin{equation}\label{eq:spacelike-profile-admissible}
\mu[\varphi]>0,
\qquad
R_1^2\,\big|\slashed\nabla^2_{\gamma[\varphi]}\varphi\big|_{\gamma[\varphi]}<1
\quad\text{on }S_{R_1}.
\end{equation}
\end{defn}

Thus \textbf{BS-1} is the quantitative condition needed to start the inward parabolic evolution, while \textbf{BS-2} is
the additional quantitative condition that will be propagated when we need spacelikeness and the
second-derivative bootstrap.

\begin{prop}[Local-in-$r$ existence from an initial sphere]
\label{prop:local-r-extension}
Fix $R_1>r_1$ and $\alpha\in(0,1)$. Let $\varphi\in C^{2+\alpha}(S^2)$ satisfy \textbf{BS-1} on
$S_{R_1}$. Then there exists
\[
\epsilon
=
\epsilon\!\left(
\varphi,\alpha,R_1,r_1,
\|\mathcal S\|_{\ctS_{-\tau}(\mathcal M_{r_1})}
\right)>0
\]
with $\epsilon<R_1-r_1$, and a unique solution
\[
f\in \ct(\overline{M_{R_1-\epsilon,R_1}})
\]
of the TMCF equation on $M_{R_1-\epsilon,R_1}$ such that
\begin{equation}\label{eq:local-initial-data}
f(R_1,\cdot)=\varphi .
\end{equation}
Moreover, $\epsilon$ may be chosen so that $f$ satisfies \textbf{BS-1} on
$M_{R_1-\epsilon,R_1}$, namely
\begin{equation}\label{eq:local-persistence-parabolic}
|\slashed d f|_{\gamma}<\frac{\vartheta_*}{2},
\qquad
r\,\mathfrak a>\frac12\,\underline{\mathfrak a_0}
\qquad\text{on }M_{R_1-\epsilon,R_1}.
\end{equation}
If $\varphi$ also satisfies \textbf{BS-2} on $S_{R_1}$, then $\epsilon$ may be chosen so that $f$
satisfies \textbf{BS-2} on $M_{R_1-\epsilon,R_1}$, namely
\begin{equation}\label{eq:local-persistence-spacelike}
\mu>0,
\qquad
r^2\,|\slashed\nabla_{\gamma}^2 f|_{\gamma}<1
\qquad\text{on }M_{R_1-\epsilon,R_1}.
\end{equation}
\end{prop}

\begin{proof}
We now introduce the inward time variable
\[
\sigma:=\log\frac{R_1}{r}.
\]
Thus $\sigma=0$ corresponds to $r=R_1$, and for any $\epsilon\in(0,R_1-r_1)$ the annulus
$M_{R_1-\epsilon,R_1}$ corresponds to
\[
0\le \sigma<\log\frac{R_1}{R_1-\epsilon}.
\]
Moreover,
\[
\partial_\sigma=-r\partial_r=r\partial_\rho.
\]

In the $(\sigma,p)$-coordinates, the TMCF equation takes the quasilinear form
\begin{equation}\label{eq:local-parabolic-form}
f_\sigma-a^{AB}(\sigma,p,f,\slashed D f)\,(\slashed D^2 f)_{AB}
=
F(\sigma,p,f,\slashed D f),
\end{equation}
where $\slashed D$ is the Levi--Civita connection of the fixed round metric $\gamma_{S^2}$. The
matrix $a^{AB}$ and the nonlinearity $F$ are smooth in the $(f,\slashed D f)$-variables and
$C^\alpha$ in $(\sigma,p)$ on the fixed tail $r\ge r_1$, because
$\mathcal S\in \ctS_{-\tau}(\mathcal M)$.

At $\sigma=0$, the principal matrix in \eqref{eq:local-parabolic-form} is the coordinate expression
of
\[
R_1\,\mathcal A[\varphi].
\]
Since $\varphi$ satisfies \textbf{BS-1} on $S_{R_1}$, we have
\[
|\slashed d\varphi|_{\gamma[\varphi]}<\frac{\vartheta_*}{2},
\qquad
R_1\,\mathfrak a[\varphi]>\frac12\,\underline{\mathfrak a_0}.
\]
Hence $\mu^T[\varphi]>0$, $\gamma_f[\varphi]$ is positive definite, and the principal matrix at
$\sigma=0$ is uniformly positive definite on $S^2$. By continuity, after restricting to a small
$\sigma$-interval, the equation remains uniformly parabolic.

We may therefore apply the standard short-time existence theorem for quasilinear uniformly
parabolic equations on the closed manifold $S^2$ to the initial-value problem
\[
\begin{cases}
f_\sigma-a^{AB}(\sigma,p,f,\slashed D f)(\slashed D^2 f)_{AB}=F(\sigma,p,f,\slashed D f),\\
f(0,\cdot)=\varphi.
\end{cases}
\]
It yields a unique solution
\[
f\in C^{2+\alpha,1+\alpha/2}\big([0,\sigma_0]\times S^2\big)
\]
for some $\sigma_0>0$. Decreasing $\sigma_0$ if necessary, we may assume
\[
0<\sigma_0<\log\frac{R_1}{r_1}.
\]
Set
\[
\epsilon:=R_1(1-e^{-\sigma_0}).
\]
Then $0<\epsilon<R_1-r_1$, and returning to the $r$-variable gives
\[
f\in\ct(\overline{M_{R_1-\epsilon,R_1}}).
\]

Because $f\in C^{2+\alpha,1+\alpha/2}$, the quantities
\[
|\slashed d f|_\gamma,
\qquad
r\,\mathfrak a,
\qquad
\mu,
\qquad
r^2|\slashed\nabla_{\gamma}^2 f|_{\gamma}
\]
depend continuously on $(\sigma,p)$. At $\sigma=0$, the first two agree with
\[
|\slashed d\varphi|_{\gamma[\varphi]},
\qquad
R_1\mathfrak a[\varphi],
\]
which satisfy the strict inequalities in \eqref{eq:parabolic-profile-admissible}. Therefore, after
possibly decreasing $\epsilon$, the solution satisfies \eqref{eq:local-persistence-parabolic}.

If $\varphi$ satisfies \textbf{BS-2} on $S_{R_1}$, then
\[
\mu[\varphi]>0,
\qquad
R_1^2|\slashed\nabla^2_{\gamma[\varphi]}\varphi|_{\gamma[\varphi]}<1
\qquad\text{on }S_{R_1}.
\]
We now check that the spacelikeness quantity of the local solution agrees initially with
$\mu[\varphi]$. The equation in the $\sigma$-variable gives
\[
f_\sigma(0,\cdot)
=
R_1\,\mathcal A[\varphi]:(\slashed\nabla^2_{\gamma[\varphi]}\varphi)
-
R_1\,\mathcal B[\varphi],
\]
and since $f_r=-R_1^{-1}f_\sigma$ at $\sigma=0$, we obtain exactly
\[
f_r(0,\cdot)=f_r[\varphi].
\]
Thus $\mu$ agrees with $\mu[\varphi]$ on the initial sphere. Shrinking $\epsilon$ further if
necessary gives \eqref{eq:local-persistence-spacelike}. Uniqueness is the standard uniqueness for
quasilinear parabolic initial-value problems.
\end{proof}

\begin{remark}[Constant outer data]
\label{rem:constant-profile-admissible}
The special case $\varphi\equiv T_0$ will be used later. In this case
\[
\mu^T[\varphi]=1,
\qquad
R_1\,\mathfrak a[\varphi]=R_1\,\mathfrak a_0(T_0,R_1,\cdot),
\qquad
R_1^2|\slashed\nabla^2_{\gamma[\varphi]}\varphi|_{\gamma[\varphi]}=0.
\]
Moreover, since $\slashed d\varphi=0$, the TMCF equation at the initial sphere gives
\[
f_r[\varphi]=\Xi(T_0,R_1,\cdot),
\]
and therefore
\[
\mu[\varphi]
=
1-\frac{N^2}{\lambda^2}(T_0,R_1,\cdot)\,\Xi(T_0,R_1,\cdot)^2
=
1-\left(\frac{\tr_{S_{T_0,R_1}}K_{T_0}}{H_{T_0,R_1}}\right)^2
=
\frac{|{\bf H}_{T_0,R_1}|^2}{H_{T_0,R_1}^2}.
\]
In particular, for every fixed $T_0>\underline T$, asymptotic flatness implies that there exists
$R_\infty(T_0)>r_1$ such that the constant profile $\varphi\equiv T_0$ satisfies \textbf{BS-1} and
\textbf{BS-2} on every sphere $S_R$ with $R\ge R_\infty(T_0)$. Proposition~\ref{prop:local-r-extension}
therefore gives a local solution of the TMCF equation satisfying \textbf{BS-1} and \textbf{BS-2} on
a short annulus adjacent to every sufficiently large sphere. In the next subsection this is
upgraded to a solution on an entire exterior region or, after a late-time gauge reduction, on a
prescribed tail.
\end{remark}

\subsection{Statement and proof of the Existence theorem}

In this subsection we work from the outset in a general ADM chart. Using
Appendix~\ref{app:late-slab-gauge}, we will pass to the good gauge on a suitable late slab in order to
invoke the results of Sections~\ref{sub:apriori} and \ref{sub:localinr}.

We work on the spacetime
\[
\mathcal M=(\underline T,\infty)\times M,
\qquad
M=(r_0,\infty)\times S^2,
\]
with coefficient tuple
\[
\mathcal S=(N,\lambda,\beta,b,\gamma)\in \ctS_{-\tau}(\mathcal M),
\]
where the metric is written in the general ADM form \eqref{eq:metric-general-prelim}.

For a tail radius $r_1>r_0$ and $T\ge \underline T$, define
\begin{align}
G_{r_1}(T)
&:=
\int_{r_1}^\infty \frac{\mathfrak X_0(r,T)}{r}\,dr,
\label{eq:G-tail-def}
\\
\Phi_{1,r_1}(T)
&:=
\int_{r_1}^\infty \mathfrak X_1(r,T)\,dr,
\label{eq:Phi1-def}
\\
\Phi_{2,r_1}(T)
&:=
\int_{r_1}^\infty \mathfrak X_2(r,T)\,dr,
\label{eq:Phi2-def}
\\
\Omega_{r_1}(T)
&:=
\sup_{r\ge r_1}\mathfrak X_0(r,T),
\label{eq:Omega-def}
\\
\Psi_{r_1}(T)
&:=
\sup_{t\ge T}\|r\Xi(t,\cdot)\|_{\cz_{-\tau}(M_{r_1,\infty})}.
\label{eq:Psi-def}
\end{align}
By Definition~\ref{def:ctS-nonstat}, all these quantities are finite for every $r_1>r_0$ and every
$T\ge\underline T$.

Similarly, if $I\subset(\underline T,\infty)$ is a compact interval, define the corresponding slab
quantities on the tail $M_{r_1,\infty}$ by
\begin{align}
G_{r_1,I}
&:=
\int_{r_1}^\infty \frac{\mathfrak X_0(r;I)}{r}\,dr,
\label{eq:GI-def}
\\
\Phi_{1,r_1,I}
&:=
\int_{r_1}^\infty \mathfrak X_1(r;I)\,dr,
\label{eq:Phi1I-def}
\\
\Phi_{2,r_1,I}
&:=
\int_{r_1}^\infty \mathfrak X_2(r;I)\,dr,
\label{eq:Phi2I-def}
\\
\Omega_{r_1,I}
&:=
\sup_{r\ge r_1}\mathfrak X_0(r;I),
\label{eq:OmegaI-def}
\\
\Psi_{r_1,I}
&:=
\sup_{t\in I}\|r\Xi(t,\cdot)\|_{\cz_{-\tau}(M_{r_1,\infty})}.
\label{eq:PsiI-def}
\end{align}

\begin{defn}[Late-time forcing decay]
\label{def:tail-assumptions}
We define the following late-time forcing decay assumptions:
\begin{enumerate}[label=\textup{(A\arabic*)}]
\item\label{it:A1} For every $r_1>r_0$,
\[
G_{r_1}(T)\to0
\qquad\text{and}\qquad
\Phi_{1,r_1}(T)\to0
\qquad\text{as }T\to\infty.
\]

\item\label{it:A2} For every $r_1>r_0$,
\[
\Phi_{2,r_1}(T)\to0
\quad\text{and}\quad
\Omega_{r_1}(T)\to0
\qquad\text{as }T\to\infty.
\]

\item\label{it:A3} For every $r_1>r_0$,
\[
\Psi_{r_1}(T)\to0
\qquad\text{as }T\to\infty.
\]
\end{enumerate}
\end{defn}

\begin{defn}[Late-time gauge reducibility]
\label{def:late-gauge-decay}
We say that $\mathcal S$ is \emph{late-time gauge reducible} if $\mathcal{S} \in {\ctS}^\sharp_{-\tau}(\mathcal{M})$, and for every $r_1>r_0$,
\[
\mathfrak B^\sharp(T;r_1)\to 0
\qquad\text{as }T\to\infty.
\]
\end{defn}

\begin{remark}
\label{rem:tail-assumption-roles}
\begin{itemize}
\item Assumption~\ref{it:A1} is the basic continuation assumption. The smallness of $G_{r_1}$
keeps the graph inside the late slab, while the smallness of $\Phi_{1,r_1}$ controls
$|\slashed d f|$ and hence keeps $\mu^T$ and $r\mathfrak a$ away from $0$.

\item Assumption~\ref{it:A2} is the direct pointwise assumption used to control
$r^2|\slashed\nabla^2 f|$ and $\frac{N}{\lambda}|f_r|$ on solutions satisfying \textbf{BS-1} and
\textbf{BS-2}.

\item Assumption~\ref{it:A3} plays the same upgrading role as~\ref{it:A2}, but through the truncated
Schauder estimate rather than the tail integrals $\Phi_2$ and $\Omega$.

\item The late-time gauge reducibility assumption is the condition that allows us to pass from the
given general ADM chart to the good gauge on late slabs away from the horizon.
\end{itemize}
\end{remark}

Recall from Definition~\ref{def:parabolicity-conditions} that parabolically admissible means
$\mu_f^T>0$ and $\mathfrak a_f>0$, while admissible means, in addition, $\mu_f>0$. The quantitative
bootstrap assumptions \textbf{BS-1} and \textbf{BS-2} were introduced in
Definition~\ref{def:BS1-BS2}, and their boundary versions were introduced in
Definition~\ref{def:admissible-boundary-profile}.

We now state the main existence theorem of this subsection.

\begin{thm}[Existence of tangentially maximal hypersurfaces]
\label{thm:global-existence}
Let $(\mathcal M,\gtime)$ be an asymptotically flat exterior spacetime of order $\tau$ with
coefficient tuple
\[
\mathcal S=(N,\lambda,\beta,b,\gamma)\in {\ctS}^\sharp_{-\tau}(\mathcal M).
\]

Then the following hold.
\begin{enumerate}[label=\textup{(\roman*)}]
\item For every $T_0>\underline T$ there exists $R_\infty(T_0)>r_0$ and a unique admissible solution
\[
f_\infty\in \ct(M_{R_\infty(T_0),\infty})
\]
of the TMCF equation on
\[
M_{R_\infty(T_0),\infty}:=(R_\infty(T_0),\infty)\times S^2
\]
such that
\[
\lim_{r\to\infty} f_\infty(r,p)=T_0
\qquad\text{uniformly in }p\in S^2.
\]
In particular, the graph $t=f_\infty(x)$ is a spacelike tangentially maximal hypersurface.

\item Assume that $\mathcal S$ is late-time gauge reducible and satisfies~\ref{it:A1}. Fix
$r_1>r_0$. Then there exist constants
\[
h_{\mathrm{slab}}=h_{\mathrm{slab}}(r_1)>0,
\qquad
T_{\mathrm{large}}(r_1)>\underline T,
\]
such that, for every $T_0\ge T_{\mathrm{large}}(r_1)$, there is a unique parabolically admissible solution
\[
f^{(r_1)}\in \ct(M_{r_1,\infty})
\]
of the TMCF equation satisfying
\[
f^{(r_1)}(M_{r_1,\infty})
\subset
[T_0-h_{\mathrm{slab}},T_0+h_{\mathrm{slab}}],
\qquad
\lim_{r\to\infty} f^{(r_1)}(r,p)=T_0
\]
uniformly in $p\in S^2$. In particular, the graph $t=f^{(r_1)}$ is a tangentially maximal hypersurface.

\item Assume that $\mathcal S$ is late-time gauge reducible and satisfies~\ref{it:A1}. Assume
further that it satisfies~\ref{it:A2} or~\ref{it:A3}. Then, after possibly enlarging
$T_{\mathrm{large}}(r_1)$, the solution $f^{(r_1)}$ in \textup{(ii)} is admissible. In particular, the graph
$t=f^{(r_1)}(x)$ is a spacelike tangentially maximal hypersurface.

\end{enumerate}

Moreover, the construction gives the following quantitative estimates. There is a slab interval
\[
I(T_0):=[T_0-h_{\mathrm{slab}},T_0+h_{\mathrm{slab}}],
\]
and the solution in \textup{(ii)} satisfies
\begin{align}
\|f^{(r_1)}-T_0\|_{\mathc^0(M_{r_1,\infty})}
&\le
\int_{r_1}^{\infty} \frac{\mathfrak X_0(\sigma;I(T_0))}{\sigma}\,d\sigma
\le h_{\mathrm{slab}},
\label{eq:global-C0-est}
\\
|\slashed d f^{(r_1)}|_{\gamma}(r,p)
&\le
\frac{C}{r}\int_r^\infty \mathfrak X_1(\sigma;I(T_0))\,d\sigma,
\qquad (r,p) \in M_{r_1,\infty}.
\label{eq:global-grad-est}
\end{align}
If the assumptions of \textup{(iii)} hold, then
\begin{align}
|\slashed\nabla_{\gamma}^2 f^{(r_1)}|_{\gamma}(r,p)
&\le
\frac{C}{r^2}\int_r^\infty \mathfrak X_2(\sigma;I(T_0))\,d\sigma,
\qquad (r,p) \in M_{r_1,\infty},
\label{eq:global-Hess-est}
\\
\frac{N}{\lambda}|f_r^{(r_1)}|(r,p)
&\le
\frac{C}{r}\,\mathfrak X_0(r;I(T_0))
+
\frac{C}{r}\int_r^\infty \mathfrak X_2(\sigma;I(T_0))\,d\sigma,
\qquad (r,p) \in M_{r_1,\infty}.
\label{eq:global-fr-est}
\end{align}
Here $C$ depends only on $r_1$, $\tau$, and the fixed-tail bounds for the spacetime coefficients.
\end{thm}

The remainder of this subsection is devoted to the proof of Theorem~\ref{thm:global-existence}.
The proof of \textup{(i)} is the same continuation argument on a sufficiently far-out tail, where
asymptotic flatness replaces late-time decay. We therefore first prove \textup{(ii)} and
\textup{(iii)}, and return to \textup{(i)} at the end.

For the rest of the proof of \textup{(ii)} and \textup{(iii)}, we fix $r_1>r_0$ and assume that
$\mathcal S$ lies in $\ctS_{-\tau}(\mathcal{M})$, is late-time gauge reducible and satisfies~\ref{it:A1}. When proving
\textup{(iii)}, we will additionally assume~\ref{it:A2} or~\ref{it:A3}.

We first record the good-gauge consequence of Appendix~\ref{app:late-slab-gauge} that will be used
below.

\begin{prop}[Late good-gauge slabs with small forcing]\label{prop:late-good-gauge-small}
There exists a number
\[
h_{\mathrm{slab}}\in(0,1]
\]
with the following property.

For every $\varepsilon>0$ there exists $T_{\mathrm{gauge}}=T_{\mathrm{gauge}}(\varepsilon)$ such
that, for every $T_0$ satisfying
\[
T_0-h_{\mathrm{slab}}\ge T_{\mathrm{gauge}},
\]
there is a time-preserving chart on
\[
I(T_0)\times M_{r_1,\infty},
\qquad
I(T_0):=[T_0-h_{\mathrm{slab}},T_0+h_{\mathrm{slab}}],
\]
with transformed coefficient tuple
\[
\widetilde{\mathcal S}
=
(\tilde N,\tilde\lambda,\tilde\beta,\tilde b,\tilde\gamma)
\]
satisfying the slab version of the fixed-tail $\ctS_{-\tau}$ bounds on
$I(T_0)\times M_{r_1,\infty}$, and with the good gauge
\[
\tilde b\equiv0,
\qquad
\tilde\beta_{\tilde r}\equiv0.
\]
Let
\[
\widetilde G_{r_1,I(T_0)},\quad
\widetilde\Phi_{1,r_1,I(T_0)},\quad
\widetilde\Phi_{2,r_1,I(T_0)},\quad
\widetilde\Omega_{r_1,I(T_0)},\quad
\widetilde\Psi_{r_1,I(T_0)}
\]
denote the slab quantities in this good-gauge chart, computed on the new tail $\tilde r\ge r_1$
using suprema over $t\in I(T_0)$. Then
\[
\widetilde G_{r_1,I(T_0)}
+
\widetilde\Phi_{1,r_1,I(T_0)}
<\varepsilon.
\]
If, in addition, assumption~\ref{it:A2} holds, then $T_{\mathrm{gauge}}$ may be chosen so that
\[
\widetilde\Phi_{2,r_1,I(T_0)}
+
\widetilde\Omega_{r_1,I(T_0)}
<\varepsilon.
\]
If instead assumption~\ref{it:A3} holds, then $T_{\mathrm{gauge}}$ may be chosen so that
\[
\widetilde\Psi_{r_1,I(T_0)}
<\varepsilon.
\]
\end{prop}

\begin{proof}
Apply Proposition~\ref{prop:app-late-slab-gauge} on the fixed tail $r\ge r_1$, with
$r_-=(r_0+r_1)/2$. The transformed slab quantities on the new tail $r\ge r_1$ are bounded by the
corresponding original slab quantities on the old tail $r\ge r_-$ plus
\[
h_{\mathrm{slab}}\,\mathfrak B^\sharp(T_0-h_{\mathrm{slab}};r_-).
\]
For example,
\[
\widetilde G_{r_1,I(T_0)}
\le
C\Big(
G_{r_-,I(T_0)}
+
h_{\mathrm{slab}}\mathfrak B^\sharp(T_0-h_{\mathrm{slab}};r_-)
\Big),
\]
and similarly for the other quantities. Since
\[
G_{r_-,I(T_0)}
\le
G_{r_-}(T_0-h_{\mathrm{slab}}),
\qquad
\Phi_{1,r_-,I(T_0)}
\le
\Phi_{1,r_-}(T_0-h_{\mathrm{slab}}),
\]
assumption~\ref{it:A1} and late-time gauge reducibility make the right-hand side arbitrarily small
as $T_0\to\infty$. The same argument gives the additional conclusions under~\ref{it:A2} or
\ref{it:A3}.
\end{proof}

In the continuation argument below, the good-gauge slab is fixed before any solution is
constructed. There is no circularity: we first choose the coordinate system and make the slab
quantities small; only then do we solve the TMCF equation. Once the slab has been chosen, we drop
the tildes and denote the good-gauge coefficients and forcing again by
\[
N,\lambda,\beta,\gamma,\Xi.
\]

\begin{defn}[$\textbf{BS-1}$ and $\textbf{BS-2}$ solutions on truncated annuli]
\label{def:admissible-truncated}
Let $\varepsilon>0$, let $T_{\mathrm{gauge}}(\varepsilon)$ and $h_{\mathrm{slab}}$ be as in
Proposition~\ref{prop:late-good-gauge-small}, and let $T_0$ satisfy
\[
T_0-h_{\mathrm{slab}}\ge T_{\mathrm{gauge}}(\varepsilon).
\]
Set
\[
I(T_0):=[T_0-h_{\mathrm{slab}},T_0+h_{\mathrm{slab}}],
\]
and fix one of the good-gauge charts from Proposition~\ref{prop:late-good-gauge-small}. For
$r_1\le r<R<\infty$, a function
\[
f\in \ct(\overline{M_{r,R}})
\]
is called an \emph{$(\varepsilon,T_0;r,R)$-\textbf{BS-1} solution} if
\begin{itemize}
\item $f(M_{r,R})\subset I(T_0)$;
\item $f$ satisfies the TMCF equation on $M_{r,R}$ in the chosen good-gauge chart;
\item $f(R,\cdot)=T_0$;
\item $f$ satisfies \textbf{BS-1} on $M_{r,R}$.
\end{itemize}
It is called an \emph{$(\varepsilon,T_0;r,R)$-\textbf{BS-2} solution} if, in addition, $f$ satisfies
\textbf{BS-2} on $M_{r,R}$.
\end{defn}

\begin{remark}\label{rem:apriori-parabolic-class}
Theorem~\ref{thm:apriori-estimates}\textup{(i)}--\textup{(ii)} applies to
\textbf{BS-1} solutions and gives the $C^0$ and tangential gradient estimates used in the
continuation argument. Theorem~\ref{thm:apriori-estimates}\textup{(iii)}--\textup{(iv)} applies to
\textbf{BS-2} solutions and gives the Schauder, Hessian, and radial-slope estimates.
\end{remark}

\begin{lem}[Uniqueness on truncated annuli]
\label{lem:uniqueness-truncated}
Let $\varepsilon>0$ and $T_0$ be as in Definition~\ref{def:admissible-truncated}, and let
$r_1\le r<R<\infty$. Then there is at most one $(\varepsilon,T_0;r,R)$-\textbf{BS-1} solution.
\end{lem}

\begin{proof}
Let $f_1$ and $f_2$ be two \textbf{BS-1} solutions on $M_{r,R}$, and set
\[
w:=f_1-f_2.
\]
In the $s=-\log r$ variable, both $f_1$ and $f_2$ solve the same uniformly parabolic quasilinear
equation on the compact cylinder $[-\log R,-\log r]\times S^2$, and
\[
w=0\qquad\text{on }S_R.
\]
Using the mean-value formula along the segment $f_\theta:=\theta f_1+(1-\theta)f_2$, one finds that
$w$ satisfies a linear uniformly parabolic equation of the form
\[
w_s-a^{AB}(s,p)(\slashed D^2 w)_{AB}-b^A(s,p)(\slashed D w)_A-c(s,p)w=0
\]
with bounded coefficients. Since the lateral manifold $S^2$ is closed, the parabolic maximum
principle gives
\[
\sup_{M_{r,R}}|w|\le \sup_{S_R}|w|=0.
\]
Hence $w\equiv0$.
\end{proof}

\begin{lem}[Boundary trace bootstrap conditions]
\label{lem:trace-admissibility}
There exists $\varepsilon_{\mathrm{tr}}>0$, depending only on $r_1$ and the fixed-tail bounds for
the spacetime coefficients, such that the following holds.

Let $0<\varepsilon\le \varepsilon_{\mathrm{tr}}$, let $T_0$ be as in
Definition~\ref{def:admissible-truncated}, and let $r_1\le r<R<\infty$.

If $f$ is an $(\varepsilon,T_0;r,R)$-\textbf{BS-1} solution, then the boundary profile
\[
\varphi:=f|_{S_r}
\]
satisfies \textbf{BS-1} on $S_r$.

If, in addition, $f$ is an $(\varepsilon,T_0;r,R)$-\textbf{BS-2} solution and the slab was chosen
with one of the upgrading decay conditions from Proposition~\ref{prop:late-good-gauge-small},
then $\varphi$ satisfies \textbf{BS-2} on $S_r$.
\end{lem}

\begin{remark}[Uniform trace threshold in nondegenerate gauges]
\label{rem:uniform-trace-threshold}
The threshold in Lemma~\ref{lem:trace-admissibility} is written as a fixed-tail threshold because
the constants in the a priori estimates are allowed to depend on $r_1$. In a horizon-penetrating
gauge with uniform nondegeneracy up to $r=r_0$, as in
Remark~\ref{rem:uniform-horizon-penetrating-gauge}, the same proof gives a threshold
independent of $r_1$. Consequently, for the existence
argument one does not need the forcing tails to decay to zero on each fixed tail; it is enough that
they are eventually smaller than this uniform threshold. This is the technical version of the
eventual-smallness observation in Remark~\ref{rem:intro-eventual-smallness}.
\end{remark}

\begin{proof}
We prove the first statement. By Theorem~\ref{thm:apriori-estimates}\textup{(ii)}, the solution
$f$ satisfies
\[
|\slashed d f|_{\gamma}(\rho,p)
\le
\frac{C}{\rho}\int_\rho^R \mathfrak X_1(\sigma;I(T_0))\,d\sigma
\le
\frac{C}{r_1}\,\Phi_{1,r_1,I(T_0)}.
\]
Since the good-gauge slab was chosen with
\[
\Phi_{1,r_1,I(T_0)}<\varepsilon,
\]
we get, after taking $\varepsilon_{\mathrm{tr}}$ sufficiently small,
\[
|\slashed d\varphi|_{\gamma[\varphi]}<\frac{\vartheta_*}{2}.
\]
Moreover,
\[
\mathfrak a[\varphi]-\mathfrak a_0
=
\mathcos(r^{-1-\tau})|\slashed d\varphi|_{\gamma[\varphi]}
+
\mathcos(r^{-1})|\slashed d\varphi|_{\gamma[\varphi]}^2,
\]
so
\[
|r\mathfrak a[\varphi]-r\mathfrak a_0|
\le
C|\slashed d\varphi|_{\gamma[\varphi]}.
\]
Since $r\mathfrak a_0\ge\underline{\mathfrak a_0}(r_1)$ on $M_{r_1,\infty}$, taking
$\varepsilon_{\mathrm{tr}}$ smaller if necessary gives
\[
r\mathfrak a[\varphi]>\frac12\,\underline{\mathfrak a_0}(r_1).
\]
Thus $\varphi$ satisfies \textbf{BS-1} on $S_r$.

Now assume $f$ is a \textbf{BS-2} solution. If the slab was chosen in the \ref{it:A2}-regime, then
\[
\Phi_{2,r_1,I(T_0)}+\Omega_{r_1,I(T_0)}<\varepsilon.
\]
Theorem~\ref{thm:apriori-estimates}\textup{(iv)} gives
\[
\rho^2|\slashed\nabla_{\gamma}^2 f|_\gamma(\rho,p)
\le
C\int_\rho^R\mathfrak X_2(\sigma;I(T_0))\,d\sigma
\le
C\Phi_{2,r_1,I(T_0)}
\]
and
\[
\frac{N}{\lambda}|f_r|(\rho,p)
\le
\frac{C}{\rho}\mathfrak X_0(\rho;I(T_0))
+
\frac{C}{\rho}\int_\rho^R\mathfrak X_2(\sigma;I(T_0))\,d\sigma.
\]
Restricting to $\rho=r$ and using the smallness above gives
\[
r^2|\slashed\nabla^2_{\gamma[\varphi]}\varphi|_{\gamma[\varphi]}<1,
\qquad
\frac{N[\varphi]}{\lambda[\varphi]}|f_r[\varphi]|<\frac{1}{\sqrt8}
\]
after decreasing $\varepsilon_{\mathrm{tr}}$ if necessary.

If instead the slab was chosen in the \ref{it:A3}-regime, then
\[
\Psi_{r_1,I(T_0)}<\varepsilon.
\]
Theorem~\ref{thm:apriori-estimates}\textup{(iii)} gives
\[
\|f-T_0\|_{\ct(M_{r,R})}\le C\Psi_{r_1,I(T_0)}.
\]
Since $r^{-2}\gamma$ is uniformly comparable to the fixed round metric on $r\ge r_1$, this implies
again
\[
r^2|\slashed\nabla^2_{\gamma[\varphi]}\varphi|_{\gamma[\varphi]}<1,
\qquad
\frac{N[\varphi]}{\lambda[\varphi]}|f_r[\varphi]|<\frac{1}{\sqrt8}
\]
for $\varepsilon_{\mathrm{tr}}$ sufficiently small.

Finally, because $\varphi$ already satisfies \textbf{BS-1}, we have
\[
\mu^T[\varphi]>\frac14
\]
by \eqref{eq:varthetastar-implies-muT}. Hence
\[
\mu[\varphi]
=
\mu^T[\varphi]
-
\frac{N[\varphi]^2}{\lambda[\varphi]^2}f_r[\varphi]^2
>
\frac14-\frac18
>
0.
\]
Thus $\varphi$ satisfies \textbf{BS-2}.
\end{proof}

\begin{proof}[Proof of Theorem~\ref{thm:global-existence}]
We first prove \textup{(ii)} and \textup{(iii)} under the assumptions fixed above.

Let $\varepsilon_{\mathrm{tr}}$ be the threshold from Lemma~\ref{lem:trace-admissibility}. Choose
\[
0<\varepsilon\le \varepsilon_{\mathrm{tr}}
\]
small enough that
\[
\varepsilon<h_{\mathrm{slab}}.
\]
If we are proving \textup{(iii)}, choose the good-gauge slab using the corresponding upgrading
decay alternative in Proposition~\ref{prop:late-good-gauge-small}. Set
\[
T_{\mathrm{large}}:=T_{\mathrm{gauge}}(\varepsilon)+h_{\mathrm{slab}}.
\]
Then, for every $T_0\ge T_{\mathrm{large}}$, we have a good-gauge chart on
\[
I(T_0)=[T_0-h_{\mathrm{slab}},T_0+h_{\mathrm{slab}}]
\]
such that
\[
G_{r_1,I(T_0)}+\Phi_{1,r_1,I(T_0)}<\varepsilon.
\]
If we are proving \textup{(iii)}, we also have either
\[
\Phi_{2,r_1,I(T_0)}+\Omega_{r_1,I(T_0)}<\varepsilon
\]
or
\[
\Psi_{r_1,I(T_0)}<\varepsilon.
\]
From now on all quantities are computed in this good-gauge chart.

We now prove existence on truncated annuli. Choose $R$ sufficiently large so that the constant
profile $\varphi\equiv T_0$ satisfies \textbf{BS-1} on $S_R$, and, in the proof of \textup{(iii)},
also \textbf{BS-2}. This is possible by Remark~\ref{rem:constant-profile-admissible}. Define
\[
\mathcal U_R
:=
\Big\{
r\in[r_1,R]\;:\;
\text{there exists an $(\varepsilon,T_0;r,R)$-\textbf{BS-1} solution}
\Big\}.
\]
We claim that $\mathcal U_R=[r_1,R]$.

Nonemptiness follows from Proposition~\ref{prop:local-r-extension} applied to the constant
\textbf{BS-1} profile on $S_R$. After shrinking the local existence interval, the local solution
also takes values in $I(T_0)$, because $f(R,\cdot)=T_0$ and $I(T_0)$ is an open neighborhood of
$T_0$ in the sense that we can keep the local solution inside its interior.

To prove openness, let $r\in\mathcal U_R$, and let $f$ be an
$(\varepsilon,T_0;r,R)$-\textbf{BS-1} solution. The trace
\[
\varphi:=f|_{S_r}
\]
satisfies \textbf{BS-1} by Lemma~\ref{lem:trace-admissibility}. Proposition~\ref{prop:local-r-extension}
therefore extends the solution inward across $S_r$. Shrinking the new local interval if necessary,
the extension still takes values in $I(T_0)$. By Lemma~\ref{lem:uniqueness-truncated}, it patches
with the old solution. Thus $\mathcal U_R$ is open.

Closedness follows from the standard continuation criterion for quasilinear parabolic equations.
Indeed, the a priori estimates from Theorem~\ref{thm:apriori-estimates}\textup{(i)}--\textup{(ii)}
keep $f$ uniformly bounded and give a uniform tangential-gradient bound. Since the equation is
quasilinear uniformly parabolic as long as \textbf{BS-1} holds, standard local Schauder estimates
give uniform $C^{2,\alpha'}$ bounds on compact subannuli. At a limiting inner radius, the trace
therefore exists in $C^{2,\alpha'}$, Lemma~\ref{lem:trace-admissibility} gives \textbf{BS-1} for
the limiting trace, and Proposition~\ref{prop:local-r-extension} restarts the solution inward.
Thus $\mathcal U_R$ is closed in $[r_1,R]$.

Hence $\mathcal U_R=[r_1,R]$. Therefore, for every sufficiently large $R$, there exists a unique
$(\varepsilon,T_0;r_1,R)$-\textbf{BS-1} solution
\[
f_R\in\ct(\overline{M_{r_1,R}})
\]
with $f_R(R,\cdot)=T_0$ and $f_R(M_{r_1,R})\subset I(T_0)$.

If \textup{(iii)} is under consideration, the same open-closed argument is applied in the
\textbf{BS-2} class. The \textbf{BS-2} trace condition follows from
Lemma~\ref{lem:trace-admissibility}, using the upgrading decay arranged when the good-gauge slab
was chosen. Hence the truncated solutions $f_R$ satisfy \textbf{BS-2}.

We now let $R\to\infty$. Choose $R_j\to\infty$ and write $f_j:=f_{R_j}$. Since each $f_j$ satisfies
\textbf{BS-1} and takes values in $I(T_0)$, Theorem~\ref{thm:apriori-estimates}
\textup{(i)}--\textup{(ii)} gives
\begin{align}
\|f_j-T_0\|_{\mathc^0_{-\tau}(M_{r_1,R_j})}
&\le
\int_{r_1}^{\infty}\frac{\mathfrak X_0(\sigma;I(T_0))}{\sigma}\,d\sigma
\le h_{\mathrm{slab}},
\label{eq:FR-C0}
\\
|\slashed d f_j|_{\gamma}(r,p)
&\le
\frac{C}{r}\int_r^{R_j} \mathfrak X_1(\sigma;I(T_0))\,d\sigma
\le
\frac{C}{r}\int_r^\infty \mathfrak X_1(\sigma;I(T_0))\,d\sigma.
\label{eq:FR-grad}
\end{align}
On every compact annulus in $M_{r_1,\infty}$, standard interior parabolic estimates and a diagonal
argument yield a limit
\[
f^{(r_1)}\in\ct(M_{r_1,\infty})
\]
with
\[
f_j\to f^{(r_1)}
\qquad\text{in }C^{2,\alpha'}_{\mathrm{loc}}(M_{r_1,\infty})
\]
for every $\alpha'<\alpha$. Passing to the limit gives a solution of the TMCF equation. The
estimates \eqref{eq:global-C0-est}--\eqref{eq:global-grad-est} follow from
\eqref{eq:FR-C0}--\eqref{eq:FR-grad}. Taking $\varepsilon$ sufficiently small, these same estimates
and the basic pointwise estimates imply that the limit satisfies \textbf{BS-1} on
$M_{r_1,\infty}$. Hence it is parabolically admissible.

The asymptotic condition follows from the same maximum-principle estimate on the tail:
\[
|f^{(r_1)}(r,p)-T_0|
\le
\int_r^\infty
\sup_{\substack{t\in I(T_0)\\ \omega\in S^2}}
|\Xi(t,\sigma,\omega)|\,d\sigma
\to0
\qquad\text{as }r\to\infty.
\]
Thus $f^{(r_1)}\to T_0$ uniformly on the spheres. Uniqueness follows from the same maximum-principle
argument as Lemma~\ref{lem:uniqueness-truncated}. This proves \textup{(ii)}.

If \textup{(iii)} is under consideration, the truncated solutions satisfy \textbf{BS-2}. Passing the
\textbf{BS-2} estimates to the limit gives that $f^{(r_1)}$ satisfies \textbf{BS-2}, after choosing
$\varepsilon$ sufficiently small as above. In particular, $\mu_f>0$, so $f^{(r_1)}$ is admissible.
Applying Theorem~\ref{thm:apriori-estimates}\textup{(iv)} to the truncated \textbf{BS-2} solutions
and passing to the limit yields \eqref{eq:global-Hess-est} and \eqref{eq:global-fr-est}. This proves
\textup{(iii)}.

It remains to prove \textup{(i)}. Fix $T_0>\underline T$. By asymptotic flatness, the forcing tails
on sufficiently far-out regions are arbitrarily small, uniformly on a fixed compact time interval
around $T_0$:
\[
\int_R^\infty \mathfrak X_1(\sigma;I_0)\,d\sigma\to0,
\qquad
\int_R^\infty \mathfrak X_2(\sigma;I_0)\,d\sigma\to0,
\qquad
\sup_{r\ge R}\mathfrak X_0(r;I_0)\to0
\]
as $R\to\infty$, where $I_0$ is any fixed sufficiently small interval containing $T_0$. The same
gauge-reduction argument as in Appendix~\ref{app:late-slab-gauge}, now using spatial decay on the
far-out tail rather than late-time decay, gives a good-gauge chart on some tail
$M_{R_\infty(T_0),\infty}$. Choosing $R_\infty(T_0)$ sufficiently large, the continuation argument
above applies on this far-out tail in the \textbf{BS-2} class and gives a unique admissible solution
\[
f_\infty\in\ct(M_{R_\infty(T_0),\infty})
\]
with uniform limit $T_0$ at infinity. This proves \textup{(i)}.
\end{proof}

%%%%%%%%%%%%%%%%%%%%%%%%%%%%%%%%%%%%%%%%%%%%%%%%%%%%%%%%%%%%%%%%%%%%%%%%%%%%%%%%%%%%%%%%%%%%%%%%%%%%%%%%%%%%%%%%%%%%%%%%%%%%%%%%%%%%%%%%%%%%%%%%%%%%%%%%%%%%%%%%%%%%%%%%%%%%%%%%%%%%%%%%%%%%%%%%%%%%%%%%%%%%%%%%%%%%%%%%%%%%%%%%%%%%%%%%%%%%%%%%

\section{The Spacetime Penrose Inequality}

In this section we explain how the existence theory for tangentially maximal tails yields a
spacetime Penrose inequality. The only additional restriction we impose on the spacetime is a
late-time settling condition, formulated below as the \emph{quasi final state hypothesis}. The
remaining ingredients are standard consequences of the positive mass theorem, the
dynamical/isolated-horizon framework, the regularity theory for outermost MOTSs, weak inverse mean
curvature flow, and the tangentially maximal construction itself
\cite{SchoenYauPMT79,WittenPET81,AshtekarKrishnan04,AshtekarGalloway05,AnderssonMarsSimon08,AnderssonMetzger09,H-I}.

\subsection{Spacetime setting and the quasi final state hypothesis}
\label{sub:spacetime-final-state}

We work in a globally hyperbolic asymptotically flat spacetime
\((\widehat{\mathcal M},\gtime)\) satisfying Einstein's equations with stress-energy tensor obeying
the spacetime dominant energy condition. By the spacetime positive mass theorem with rigidity, the
ADM four-momentum $P^\mu_{\mathrm{ADM}} = (E_{\mathrm{ADM}}, P_{\mathrm{ADM}})$ of the asymptotic end is future causal, and is future
timelike unless the data are in the rigidity case
\cite{SchoenYauPMT79,schoen-yau-jang,WittenPET81}. We exclude the flat rigidity case in the
black-hole setting considered below. Hence \(P^\mu_{\mathrm{ADM}}\) is future timelike.

We choose the asymptotic rest frame of this timelike ADM four-momentum. Equivalently, after an
asymptotic Lorentz transformation of the asymptotically flat end, the ADM linear momentum $P_{\mathrm{ADM}}$ vanishes,
and the ADM energy $E_{\mathrm{ADM}}$ in this frame is the invariant ADM mass
\cite{adm-invariant1,adm-invariant2,adm-invariant3,adm-invariant5}. We then choose a smooth Cauchy
temporal function
\[
        t:\widehat{\mathcal M}\to \mathbb R
\]
whose level sets
\[
        \Sigma_t:=\{t=\mathrm{const}\}
\]
are asymptotically Euclidean Cauchy hypersurfaces adapted to this rest-frame asymptotic coordinate
system. The existence of smooth Cauchy temporal functions on globally hyperbolic spacetimes, and
the corresponding smooth splitting by spacelike Cauchy hypersurfaces, is due to
Bernal--Sánchez; see \cite{BernalSanchez05,BernalSanchez06}. In this chosen asymptotic frame the
ADM linear momentum of the end is zero, and we denote the corresponding ADM energy by
\(m_{\mathrm{ADM}}\).

\medskip

If the relevant slices contain no marginally outer trapped surfaces, then the desired estimate
reduces to the positive mass theorem. We therefore restrict attention to the black-hole case and
choose the origin of time so that, for every \(t\ge 0\) under consideration, the slice \(\Sigma_t\)
contains a smooth compact outermost MOTS
\[
\mathcal S_t\subset \Sigma_t
\]
not assumed to be connected. Here ``outermost'' is understood in the standard apparent-horizon
sense: no homologous weakly outer trapped surface lies outside \(\mathcal S_t\) in \(\Sigma_t\).

We assume that \(\mathcal S_t\) is of future black-hole type on every connected component.
Thus, if \(\ell\) and \(\underline\ell\) denote the future-directed outward and inward null normals
to any connected component of \(\mathcal S_t\), normalized by
\[
\gtime(\ell,\underline\ell)=-2,
\]
then
\[
\theta_{(\ell)}=0,
\qquad
\theta_{(\underline\ell)}<0.
\]
Let
\[
A(t):=|\mathcal S_t|
\]
denote the total area of \(\mathcal S_t\), i.e. the sum of the areas of its connected components,
and set
\[
\mathcal H_{\mathrm{app}}
:=
\bigcup_{t\ge 0}\mathcal S_t.
\]

We work in the standard outermost-MOTS evolution regime in which \(\mathcal H_{\mathrm{app}}\) is a
piecewise smooth three-dimensional hypersurface, possibly disconnected on individual smooth
pieces, with at most finitely many jump times
\[
0<t_1<\cdots<t_N<\infty.
\]
On each smooth piece, the number of connected components is locally constant, and each
component is a smooth marginally outer trapped tube foliated by future black-hole-type
outermost MOTSs. This is the regime suggested by the local existence and regularity theory
for MOTTs and by the regularity theory for evolving outermost MOTSs
\cite{AnderssonMarsSimon05,AnderssonMarsMetzgerSimon09,AnderssonMetzger09}.

We will use the following standard facts about smooth pieces of \(\mathcal H_{\mathrm{app}}\).

\begin{prop}[Causal character and area law for smooth outermost-MOTS pieces]
\label{prop:mott-area-law}
Let $\mathcal H_0$ be a smooth piece of $\mathcal H_{\mathrm{app}}$, possibly disconnected,
foliated by outermost MOTSs $\mathcal S_t$ of future black-hole type:
\[ \theta_{(\ell)}=0, \qquad\theta_{(\underline\ell)}<0\] 
on every connected component. Then:
\begin{itemize}
\item $\mathcal H_0$ is achronal: each connected component is, at every point, spacelike or null, and is never timelike.
\item The total area $A(t)=|\mathcal S_t|$ is nondecreasing towards the future. Along spacelike connected components, the area is strictly increasing; along null connected components, the area is constant. 
\end{itemize}
\end{prop}

\begin{proof}
This is a standard consequence of the outermost-MOTS evolution theory and the
dynamical-horizon area law; see for example
\cite{AshtekarBeetleFairhurst99,AshtekarKrishnan03,AshtekarKrishnan04,AshtekarGalloway05,
AnderssonMarsSimon05,AnderssonMarsMetzgerSimon09}. We recall the argument for
completeness.

Since each \(\mathcal S_t\) is outermost in \(\Sigma_t\), each of its connected components is
weakly stably outermost. Indeed, if the stability operator on one component had negative
principal eigenvalue, an outward deformation of that component in the direction of the
positive principal eigenfunction would produce a weakly outer trapped surface outside
\(\mathcal S_t\), contradicting outermostness.

The evolution theory for outermost MOTSs implies that any smooth MOTT formed by the
outer boundaries of the trapped regions is achronal. Hence \(\mathcal H_0\) has no timelike
portion.

On a smooth time interval the number of connected components of \(\mathcal S_t\) is locally
constant. Write
\[
\mathcal S_t=\bigcup_{i=1}^N \mathcal S_t^i,
\]
where each \(\mathcal S_t^i\) lies in a smooth connected component
\(\mathcal H_0^i\subset\mathcal H_0\). On \(\mathcal H_0^i\), let \(\ell_i\) and
\(\underline\ell_i\) be the future-directed outward and inward null normals to
\(\mathcal S_t^i\), normalized by
\[
\gtime(\ell_i,\underline\ell_i)=-2.
\]
A vector tangent to \(\mathcal H_0^i\), normal to the leaves, and future-directed with respect
to the chosen evolution can be written, after rescaling, as
\[
V_i=\ell_i-C_i\,\underline\ell_i
\]
for a function \(C_i\) on \(\mathcal S_t^i\). Since \(\mathcal H_0\) is achronal,
\(V_i\) is never timelike. With the above normalization,
\[
\gtime(V_i,V_i)=4C_i,
\]
and therefore
\[
C_i\ge 0.
\]
Thus \(\mathcal H_0^i\) is spacelike where \(C_i>0\), null where \(C_i=0\), and has no
timelike portion.

The first variation of area along \(V_i\) is
\[
\theta_{(V_i)}
=
\theta_{(\ell_i)}-C_i\,\theta_{(\underline\ell_i)}
=
-C_i\,\theta_{(\underline\ell_i)}.
\]
Since \(\theta_{(\underline\ell_i)}<0\) and \(C_i\ge0\), this quantity is nonnegative.
Therefore
\[
\frac{d}{dt}|\mathcal S_t^i|\ge0
\]
along spacelike portions, after multiplying by the positive lapse relating \(V_i\) to the
chosen time parameter. If the component is null, then \(C_i=0\), and the same formula gives
\[
\frac{d}{dt}|\mathcal S_t^i|=0.
\]
Summing over the finitely many connected components gives
\[
\frac{d}{dt}A(t)
=
\sum_{i=1}^N \frac{d}{dt}|\mathcal S_t^i|
\ge0
\]
on every smooth piece.
\end{proof}

Let $\underline T$ be a time after the last jump time, and denote the final smooth piece of the
apparent horizon by
\[
\mathcal H_{\mathrm{final}}
:=
\mathcal H_{\mathrm{app}}\cap\{t\ge \underline T\}.
\]
Thus $\mathcal H_{\mathrm{final}}$ is a smooth MOTT; we refer to it as the \emph{last smooth horizon piece}.

\begin{remark}
We do not assume any global product structure or special topology of the late-time exterior at
this stage. Those features will be part of the quasi final state hypothesis below.
\end{remark}

The expected late-time picture is that the exterior geometry settles to Kerr and the last smooth
horizon piece settles to a Kerr isolated horizon. There does not seem to be a single universally
adopted precise formulation of the black-hole final state conjecture in the mathematical
relativity literature; compare, for instance, the discussions in
\cite{GiorgiStableBH,KlainermanBHStability}. We now record a strong version of this expectation,
stated in a standard convergence language, and then formulate the weaker hypothesis actually used
in the proof.

\begin{defn}[Strong final state hypothesis]
\label{def:strong-final-state}
We say that the future exterior of $(\widehat{\mathcal M},\gtime)$ satisfies the
\emph{strong final state hypothesis} if there exist subextremal Kerr parameters
\[
m>0,
\qquad
|a|<m,
\]
a time $\underline T>0$, a number $r_0>0$, and a late-time $C^2$ coordinate system
\[
(t,r,p)\in(\underline T,\infty)\times [r_0,\infty)\times S^2
\]
on the closure of the exterior of the last smooth horizon piece, with
\[
\mathcal H_{\mathrm{final}}=\{r=r_0\},
\]
such that the coordinate system has uniform $C^2$ control on fixed collars of
$\mathcal H_{\mathrm{final}}$ for all sufficiently late $t$, the metric takes the general ADM form
on the exterior region $r>r_0$,
\[
\gtime
=
-\big(N^2-|\beta|_{g_t}^2\big)\,dt^2
+2\,\beta\odot dt
+g_t,
\qquad
g_t=(\lambda^2+|b|_\gamma^2)\,dr^2+2\,b\odot dr+\gamma,
\]
and the following hold.

\begin{enumerate}[label=\textup{(S\arabic*)}]
\item\label{it:strong-tail}
On every fixed tail $M_{r_1,\infty}$, $r_1>r_0$, the coefficient tuple
\[
\mathcal S=(N,\lambda,\beta,b,\gamma)
\]
converges to the corresponding subextremal Kerr coefficient tuple in Boyer--Lindquist gauge in a
weighted $C^5$ topology. More precisely, after identifying the angular variables with the Kerr
ones, the coefficients and their first two $t$-derivatives converge to the Kerr coefficients and
their corresponding $t$-derivatives, with the standard asymptotically flat weights of order
$\tau\in(1/2,1)$, on every fixed tail $r\ge r_1>r_0$.

\item\label{it:strong-horizon}
The horizon sections
\[
\mathcal S_t=\mathcal H_{\mathrm{final}}\cap\Sigma_t
=
\{t\}\times\{r_0\}\times S^2
\]
converge in $C^2$ as $t\to\infty$ to a cross-section of the corresponding Kerr isolated horizon,
and
\[
A(t)\to A_{\mathrm{Kerr}}
\qquad\text{as }t\to\infty.
\]
\end{enumerate}
\end{defn}

\medskip

We now formulate the weaker hypothesis used in the proof.

\begin{defn}[Quasi final state hypothesis]
\label{def:spi-QFS}
We say that $(\widehat{\mathcal M},\gtime)$ satisfies the \emph{quasi final state hypothesis} if
there exist numbers
\[
\underline T>0,
\qquad
r_0>0,
\qquad
\tau\in(1/2,1),
\]
such that the following hold.

\begin{enumerate}[label=\textup{(Q\arabic*)}]
\item\label{it:spi-QFS-tail}
\emph{Last smooth horizon piece and late-time exterior chart.}
The closure of the exterior region of $\mathcal H_{\mathrm{final}}$ admits a late-time $C^2$
coordinate system
\[
(t,r,p)\in(\underline T,\infty)\times [r_0,\infty)\times S^2
\]
with
\[
\mathcal H_{\mathrm{final}}=\{r=r_0\}.
\]
The coordinate system has uniform $C^2$ control on fixed collars of
$\mathcal H_{\mathrm{final}}$ for all sufficiently late $t$. On the exterior region $r>r_0$, the
metric takes the general ADM form
\[
\gtime
=
-\big(N^2-|\beta|_{g_t}^2\big)\,dt^2
+2\,\beta\odot dt
+g_t,
\qquad
g_t=(\lambda^2+|b|_\gamma^2)\,dr^2+2\,b\odot dr+\gamma,
\]
and the coefficient tuple belongs to the class ${\ctS}^\sharp_{-\tau}(\mathcal M)$ from
Definition~\ref{def:sharp-tail-data-main}, where
\[
\mathcal M=(\underline T,\infty)\times (r_0,\infty)\times S^2.
\]

\item\label{it:spi-QFS-forcing}
\emph{Late-time forcing decay and gauge reducibility.}
The coefficient tuple $\mathcal S$ is late-time gauge reducible in the sense of
Definition~\ref{def:late-gauge-decay}; that is for every $r_1>r_0$,
\[
\mathfrak B^\sharp(T;r_1)\to0
\qquad\text{as }T\to\infty.
\]
Moreover, the late-time forcing decay assumptions from Definition~\ref{def:tail-assumptions}
hold: namely \textup{\ref{it:A1}} holds, and in addition either \textup{\ref{it:A2}} or
\textup{\ref{it:A3}} holds.

\item\label{it:spi-QFS-horizon}
\emph{Late-time area stabilization of the horizon.}
The areas of the horizon sections
\[
\mathcal S_t=\mathcal H_{\mathrm{final}}\cap\Sigma_t
=
\{t\}\times\{r_0\}\times S^2
\]
converge:
\[
A_\infty:=\lim_{t\to\infty}A(t)\in(0,\infty).
\]
\end{enumerate}
\end{defn}

\begin{remark}[Role of the quasi final state assumptions]
The three parts of the quasi final state hypothesis have separate roles.

\begin{itemize}
\item Condition \textup{(Q1)} provides the geometric setting: a horizon-regular late-time exterior
collar for the final comparison with \(\mathcal H_{\mathrm{final}}\), together with the
asymptotically flat coefficient control needed to formulate the TMCF equation on every fixed
exterior tail. It also imposes a topological restriction on the last smooth horizon piece: in the
form stated here, \(\mathcal H_{\mathrm{final}}\) is topologically $(\underline T,\infty)\times S^2$,
with horizon sections given by 2-spheres. Thus the hypothesis is adapted to a single
late-time black-hole component.

\item Condition \textup{(Q2)} is the analytic input for the tangentially maximal construction. It
is exactly what allows one to apply the late-slab gauge reduction of
Appendix~\ref{app:late-slab-gauge} and the existence theorem from the previous section on every
fixed tail \(r\ge r_1>r_0\).

\item Condition \textup{(Q3)} is the late-time horizon-area input: it identifies the limiting area
which appears in the final Penrose comparison. If the last smooth horizon piece is null, or becomes
isolated from some time onward, then \textup{(Q3)} automatically follows from the area law in
Proposition~\ref{prop:mott-area-law}, since the areas are constant on null pieces of the horizon.
\end{itemize}
\end{remark}

\begin{remark}[Horizon-penetrating final-state gauges]
\label{rem:horizon-penetrating-final-state-gauge}
At first sight, the strong and quasi final state hypotheses may appear not
to cover even the standard stationary black-hole examples.  Indeed, the
usual Boyer--Lindquist coordinates in Kerr, and the usual Schwarzschild
coordinates in Schwarzschild, are adapted to the stationary exterior
geometry and to the tangentially maximal structure of the exterior
constant-time slices, but they do not penetrate the future horizon.  The
constant-time slices in these coordinates instead limit to the bifurcation
sphere as \(r\to r_+^+\).

This is only a coordinate issue.  The final-state chart used in the
hypotheses is required to be regular at the future horizon, but it is not
required to coincide exactly with the stationary Boyer--Lindquist gauge
all the way down to the horizon at each finite time.  Rather, one may
modify the stationary gauge in a shrinking neighborhood of the horizon.
More precisely, one constructs a horizon-penetrating chart which, for each
large time \(t\), agrees with the stationary exterior gauge on the region
\[
        r\ge r_{\mathrm{cut}}(t),
        \qquad
        r_{\mathrm{cut}}(t)>r_+,
        \qquad
        r_{\mathrm{cut}}(t)\downarrow r_+
        \quad\text{as }t\to\infty.
\]
Thus on every fixed exterior region \(\{r\ge r_1\}\), with
\(r_1>r_+\), the chart agrees with the stationary tangentially maximal
gauge for all sufficiently large \(t\).  At the same time, near the
horizon the chart has been modified so that its time slices penetrate the
future horizon and its horizon sections appear as coordinate spheres.

In Schwarzschild this construction can be carried out explicitly by
working in ingoing Eddington--Finkelstein coordinates and cutting off the
singular radial derivative of the Schwarzschild time function near the
horizon.  The resulting spacelike hypersurfaces penetrate the future
horizon, foliate it by MOTS cross-sections, and agree with the standard
Schwarzschild time slices outside the moving cutoff radius
\(r_{\mathrm{cut}}(t)\), with \(r_{\mathrm{cut}}(t)\downarrow 2m\).
We include this model construction in
Appendix~\ref{app:horizon-penetrating-schwarzschild-model}.  The same
idea applies to Kerr, starting from a horizon-penetrating Kerr coordinate
system and gluing back to the Boyer--Lindquist time function outside a
shrinking collar of the future horizon.  Thus the stationary Kerr and
Schwarzschild exteriors are compatible with the final-state hypotheses
after this harmless late-time gauge adjustment.
\end{remark}

\medskip

For later use, recall that the spacetime Hawking mass of a closed spacelike surface
$\Sigma\subset (\widehat{\mathcal M},\gtime)$ is
\[
m_H^{\mathrm{ST}}(\Sigma)
:=
\sqrt{\frac{|\Sigma|}{16\pi}}
\left(
1-\frac{1}{16\pi}\int_\Sigma |{\bf H}_\Sigma|^2\,d\mu_\Sigma
\right).
\]
\begin{lem}[Late-horizon continuity of the spacetime Hawking mass]
\label{lem:late-horizon-hawking}
Assume \textup{(Q1)} and \textup{(Q3)}. Then for every $\varepsilon>0$ there exist
\[
T^{\mathrm{hor}}(\varepsilon)\ge \underline T,
\qquad
\rho_\varepsilon>0
\]
with the following property. In the late-time coordinates from \textup{(Q1)}, let
\[
\Sigma_{t,u,v}:=
        \{(t+u(p),\, r_0+v(p),\,p):p\in S^2\}
\]
be any spacelike graph over $\mathcal S_t$ whose time and transverse graph functions
$u,v\in C^2(S^2)$ satisfy
\[
\|u\|_{C^2(S^2)}+\|v\|_{C^2(S^2)}<\rho_\varepsilon,
\qquad
v\ge0.
\]
Then, for all $t\ge T^{\mathrm{hor}}(\varepsilon)$,
\begin{equation}\label{eq:late-horizon-hawking}
\left|
m_H^{\mathrm{ST}}(\Sigma_{t,u,v})
-
\sqrt{\frac{A_\infty}{16\pi}}
\right|
<\varepsilon.
\end{equation}
\end{lem}

\begin{proof}
By the uniform $C^2$ control of collars as part of \textup{(Q1)}, the induced metrics, second fundamental forms, and
mean curvature vectors of $C^2$-small graph spheres over $\mathcal S_t$ depend continuously on the
graph functions, uniformly for all sufficiently late $t$. Therefore there exists a modulus of
continuity $\omega$, independent of $t$, such that for all sufficiently small spacelike graphs
\begin{equation}\label{eq:hawking-uniform-continuity}
\left|
m_H^{\mathrm{ST}}(\Sigma_{t,u,v})-m_H^{\mathrm{ST}}(\mathcal S_t)
\right|
\le
\omega\!\left(\|u\|_{C^2(S^2)}+\|v\|_{C^2(S^2)}\right).
\end{equation}

Each $\mathcal S_t$ is a MOTS, so its mean curvature vector is null. Hence
\[
|{\bf H}_{\mathcal S_t}|^2=0,
\]
and therefore
\begin{equation}\label{eq:hawking-on-mots}
m_H^{\mathrm{ST}}(\mathcal S_t)
=
\sqrt{\frac{A(t)}{16\pi}}.
\end{equation}
By \textup{(Q3)}, choose $T^{\mathrm{hor}}(\varepsilon)\ge\underline T$ so that
\[
\left|
\sqrt{\frac{A(t)}{16\pi}}
-
\sqrt{\frac{A_\infty}{16\pi}}
\right|
<
\frac{\varepsilon}{2}
\qquad\text{for all }t\ge T^{\mathrm{hor}}(\varepsilon).
\]
Then choose $\rho_\varepsilon>0$ so small that $\omega(\rho_\varepsilon)<\varepsilon/2$.
Combining \eqref{eq:hawking-uniform-continuity} and \eqref{eq:hawking-on-mots} gives
\eqref{eq:late-horizon-hawking}.
\end{proof}

\begin{prop}[Strong final state implies quasi final state]
\label{prop:strong-implies-weak}
If the strong final state hypothesis holds, then the quasi final state hypothesis holds.
\end{prop}

\begin{proof}
Condition \textup{(Q1)} is part of the strong final state hypothesis: the late-time coordinate
system is $C^2$ up to $\mathcal H_{\mathrm{final}}$, has uniform $C^2$ control on fixed collars of
the horizon for all sufficiently late $t$, and the metric takes the required ADM form on the
exterior. The weighted coefficient bounds in ${\ctS}^\sharp_{-\tau}$ on the exterior follow from
\textup{(S1)}.

We next verify \textup{(Q2)}. In Boyer--Lindquist Kerr one has
\[
b\equiv0,
\qquad
\beta_r\equiv0,
\qquad
\beta^\perp\equiv0,
\]
and the constant-$t$ slices are tangentially maximal by Proposition~\ref{prop:kerr}. Thus
\[
\Xi_{\mathrm{Kerr}}\equiv0.
\]
The convergence in \textup{(S1)}, on every fixed exterior tail $r\geq r_1$, therefore implies that the
normal-shift tails $\mathfrak B^\sharp(T;r_1)$ tend to zero and that the forcing quantities appearing in
\[
        G_{r_1},
        \qquad
        \Phi_{1,r_1},
        \qquad
        \Phi_{2,r_1},
        \qquad
        \Omega_{r_1},
        \qquad
        \Psi_{r_1}
\]
tend to zero with the weights required in Definition~\ref{def:tail-assumptions}. Hence the
late-time gauge reducibility and the forcing decay assumptions \textup{\ref{it:A1}},
\textup{\ref{it:A2}}, and \textup{\ref{it:A3}} hold.

Finally, \textup{(Q3)} is exactly the area convergence in \textup{(S2)}.

\end{proof}

\begin{remark}
The quasi final state hypothesis is strictly weaker than the strong final state hypothesis: it asks
only for late-time gauge reducibility, the forcing decay \textup{\ref{it:A1}} and either
\textup{\ref{it:A2}} or \textup{\ref{it:A3}}, and stabilization of the horizon area. It does not
require convergence of the entire late-time exterior geometry to Kerr. The same implication
argument works with Kerr replaced by any stationary black-hole exterior admitting a horizon-regular
late-time chart which approaches a stationary gauge with $\beta^\perp=0$ and $\Xi=0$ on every fixed
exterior tail; in particular this includes static exteriors such as Schwarzschild and
Reissner--Nordstr\"om
\cite{AshtekarBeetleFairhurst99,AshtekarKrishnan04,GiorgiStableBH,KlainermanBHStability}.
\end{remark}

\medskip

The following describes a criterion for a given spacetime to satisfy the quasi final state
hypothesis.

\begin{prop}[Stability criterion for the quasi final state hypothesis]
\label{prop:stability-implies-QFS}
Let \(\mathcal S_*\) be the coefficient tuple of a stationary black-hole exterior on
\((r_0,\infty)\times S^2\), written in a stationary exterior gauge satisfying
\[
        \beta_*^\perp\equiv0,
        \qquad
        \Xi_*\equiv0 .
\]
Let \((\widehat{\mathcal M},\gtime)\) be a globally hyperbolic asymptotically flat spacetime
satisfying the dominant energy condition, with a smooth last horizon piece
\(\mathcal H_{\mathrm{final}}\). Assume that the future exterior of
\(\mathcal H_{\mathrm{final}}\) is equipped with a late-time exterior chart satisfying the
geometric part of \textup{(Q1)}, and denote its ADM coefficient tuple on \(r>r_0\) by
\[
        \mathcal S=(N,\lambda,\beta,b,\gamma).
\]
Assume further that:
\begin{enumerate}[label=\textup{(\roman*)}]
\item for every fixed tail \(r\ge r_1>r_0\),
\[
        \mathcal S\longrightarrow \mathcal S_*
        \qquad\text{as }t\to\infty
\]
in the \({\ctS}^\sharp_{-\tau}\)-topology on that tail; equivalently,
\[
        \lim_{T\to\infty}
\|\mathcal S-\mathcal S_*\|_{{\ctS}^\sharp_{-\tau}((T,\infty)\times M_{r_1,\infty})}
=
0;
\]

\item the horizon areas converge:
\[
        A(t)\to A_\infty\in(0,\infty).
\]
\end{enumerate}
Then \((\widehat{\mathcal M},\gtime)\) satisfies the quasi final state hypothesis.
\end{prop}

\begin{proof}
The geometric part of \textup{(Q1)} holds by assumption. The convergence in
\textup{(i)} gives the required \({\ctS}^\sharp_{-\tau}\)-control on every fixed exterior tail,
after increasing the initial late time if necessary. Hence \textup{(Q1)} holds.

We next verify \textup{(Q2)}. Since the stationary exterior model satisfies
\[
        \beta_*^\perp\equiv0,
        \qquad
        \Xi_*\equiv0,
\]
and since the convergence in \textup{(i)} is in the strengthened tail topology, the normal-shift
tail satisfies
\[
        \mathfrak B^\sharp(T;r_1)\to0
        \qquad\text{as }T\to\infty
\]
for every \(r_1>r_0\). The same convergence also gives the decay of the forcing quantities
appearing in Definition~\ref{def:tail-assumptions}. Thus the late-time gauge reducibility
condition and the assumptions \textup{\ref{it:A1}} and either \textup{\ref{it:A2}} or
\textup{\ref{it:A3}} hold. Therefore \textup{(Q2)} holds.

Finally, assumption \textup{(ii)} is precisely the area stabilization condition \textup{(Q3)}.
Thus all parts of the quasi final state hypothesis are satisfied.
\end{proof}

\begin{remark}[Stability and examples]
\label{rem:stability-examples}
Proposition~\ref{prop:stability-implies-QFS} is meant as an interface with nonlinear black-hole
stability results.  The condition that \(\mathcal S\to\mathcal S_*\) in the strengthened weighted
tail topology is a quantitative version of the statement that the exterior settles to a stationary
black-hole exterior, but only in the geometric quantities and norms needed for the tangentially
maximal comparison argument.

\begin{itemize}
\item \emph{Stationary backgrounds.}
Exact stationary exteriors satisfying
\[
        \beta^\perp\equiv0,
        \qquad
        \Xi\equiv0
\]
satisfy the quasi final state hypothesis after the horizon-penetrating
gauge adjustment described in
Remark~\ref{rem:horizon-penetrating-final-state-gauge}.  This applies to
Schwarzschild by the explicit construction in
Appendix~\ref{app:horizon-penetrating-schwarzschild-model}, and to Kerr on
fixed exterior tails by Proposition~\ref{prop:kerr}, where the stationary
exterior gauge is the usual Boyer--Lindquist gauge.  The same observation
applies to Reissner--Nordstr\"om in the Einstein--Maxwell setting, where
the stress-energy tensor satisfies the dominant energy condition.

\item \emph{Perturbations of Schwarzschild.}
The nonlinear exterior stability theorem of Dafermos--Holzegel--Rodnianski--Taylor
\cite{DafermosHolzegelRodnianskiTaylorSchwarzschildNonlinear}, building on the linear stability
theory of Dafermos--Holzegel--Rodnianski
\cite{DafermosHolzegelRodnianskiSchwarzschildLinear}, gives asymptotic convergence to a
Schwarzschild exterior for the corresponding finite-codimension Schwarzschild stability class.
For such perturbations, Proposition~\ref{prop:stability-implies-QFS} applies once the stability
estimates are expressed in a horizon-compatible gauge controlling the above
\({\ctS}^\sharp_{-\tau}\)-tail topology, and once the final-horizon regularity and area-stabilization
assumptions are verified.

\item \emph{Perturbations of slowly rotating Kerr.}
The nonlinear stability theory for small-angular-momentum Kerr, developed by
Klainerman--Szeftel and Giorgi--Klainerman--Szeftel, together with the associated GCM
hypersurface constructions, provides the natural source of the analogous convergence to a Kerr
exterior; see
\cite{KlainermanSzeftelKerrStabilitySmallA,GiorgiKlainermanSzeftelKerrWaveEstimates,ShenKerrStability}
and the surveys \cite{GiorgiStableBH,KlainermanBHStability}.  As above, the criterion applies
provided the convergence estimates are translated into the horizon-compatible gauge and weighted
tail norms used here, and provided the final MOTT regularity and horizon-area stabilization
assumptions hold.
\end{itemize}

Thus the quasi final state hypothesis is designed to capture precisely the portion of late-time
black-hole stability needed for the Penrose argument, rather than to require full convergence of
every geometric quantity in the spacetime.
\end{remark}

\subsection{Statement and proof of the spacetime Penrose inequality}

We now prove the spacetime Penrose inequality under the quasi final state hypothesis.

\begin{thm}[Spacetime Penrose inequality]
\label{thm:SPI}
Let \((\widehat{\mathcal M},\gtime)\) be a globally hyperbolic asymptotically flat spacetime
satisfying Einstein's equations and the spacetime dominant energy condition. Let
\[
        t:\widehat{\mathcal M}\to\R,
        \qquad
        \Sigma_t:=\{t=\mathrm{const}\},
        \qquad
        \mathcal H_{\mathrm{app}}=\bigcup_{t\ge0}\mathcal S_t
\]
be as defined in the beginning of Section~\ref{sub:spacetime-final-state}.

Let \((M_\ast,g_\ast,K_\ast)\hookrightarrow(\widehat{\mathcal M},\gtime)\) be an
asymptotically flat initial data set. Assume that its boundary $S_\ast:=\partial M_\ast$
is a MOTS and a smooth cross-section of \(\mathcal H_{\mathrm{app}}\).

Assume that the portion of \(\mathcal H_{\mathrm{app}}\) to the future of \(S_\ast\) is piecewise
smooth, has only finitely many jump times, and has a last smooth piece
\(\mathcal H_{\mathrm{final}}\) satisfying the quasi final state hypothesis. Assume furthermore that
the outermost-horizon area does not decrease across any jump. Then
\begin{equation}\label{eq:SPI-final}
        m_{ADM}(M_\ast,g_\ast,K_\ast)
        \ge
        \sqrt{\frac{|S_\ast|}{16\pi}}.
\end{equation}
\end{thm}

\begin{remark}[Disconnected outermost MOTSs]
\label{rem:disconnected-outermost-mots}
We do not assume that the outermost MOTS \(\mathcal S_t\subset\Sigma_t\) is connected. This is
natural in multi--black-hole initial data. In general, the outermost MOTS should be viewed as the
outer boundary of the weakly outer trapped region, and this boundary may have several connected
components. The existence and regularity theory for outermost MOTSs is formulated in this
generality; in particular, the outer boundary of the weakly outer trapped region is a smooth
outermost MOTS, and the corresponding area estimates do not require connectedness
\cite{AnderssonMetzger09}. Similarly, under evolution, different components may jump or merge; this
behavior is part of the standard picture for evolving outermost MOTSs
\cite{AnderssonMarsMetzgerSimon09}.

Accordingly, throughout the paper \(A(t)=|\mathcal S_t|\) denotes the total area of
\(\mathcal S_t\), i.e. the sum of the areas of its connected components. On a smooth piece of
\(\mathcal H_{\mathrm{app}}\), the area law is applied componentwise and then summed. Connectedness
is imposed only where it is explicitly required, for instance on the last smooth piece of the
horizon due to the topological restriction
\[
        \mathcal H_{\mathrm{final}}
        =
        (\underline T,\infty)\times\{r_0\}\times S^2
\]
in the quasi final state hypothesis.
\end{remark}

\begin{remark}[Bondi-mass version]
\label{rem:bondi-mass-version}
The conclusion of Theorem~\ref{thm:SPI} can be strengthened, under the additional assumption that
the asymptotic end admits a future Bondi--Sachs structure near \(\mathscr I^+\), to an inequality
with the final Bondi mass in place of the ADM mass. More precisely, for a retarded time \(u\), if
\[
        m_B^+:=\lim_{u\to+\infty}m_B(u)
\]
exists, then Appendix~\ref{app:bondi-annuli} proves
\[
        m_B^+
        \ge
        \sqrt{\frac{|S_\ast|}{16\pi}}.
\]
The argument in Appendix~\ref{app:bondi-annuli} is a localized version of the proof below: instead
of passing to complete TMCF tails and comparing with ADM mass at spatial infinity, one solves the
TMCF equation on finite annuli whose outer boundaries are large spheres with retarded times \(u\)
tending to \(+\infty\). The outer Hawking masses then converge to the final Bondi mass, while a
Shi--Tam/Brown--York extension argument and the Huisken--Ilmanen comparison replace the
complete-tail ADM comparison. Thus our method does not assume the absence of Bondi mass loss;
rather, it allows
\[
        m_{ADM}(M_\ast,g_\ast,K_\ast)>m_B^+
\]
and still proves the Penrose bound with the smaller final Bondi mass.
\end{remark}

\medskip

The rest of this subsection is devoted to the proof of Theorem~\ref{thm:SPI}. We first record the
comparison facts for the tangentially maximal tails produced by
Theorem~\ref{thm:global-existence}. Throughout this subsection, if \(f\) is one of these tails, then
\(M_f\) denotes its graph hypersurface, \((g_f,K_f)\) denotes the induced initial data, and
\(\Sigma_{r_1}^{(T_0)}\) denotes the boundary leaf of the tail over the fixed radius \(r=r_1\) in
the auxiliary tail coordinates used to construct it.

For a closed surface \(\Sigma\subset(M_f,g_f)\), let
\[
m_H^{\mathrm{Riem}}(\Sigma)
:=
\sqrt{\frac{|\Sigma|}{16\pi}}
\left(
1-\frac{1}{16\pi}\int_\Sigma H_\Sigma^2\,d\mu_\Sigma
\right)
\]
be its Riemannian Hawking mass. Since the leaves of a tangentially maximal tail satisfy
\[
        \tr_\Sigma K_f=0,
\]
one has
\begin{equation}\label{eq:hawking-masses-agree}
        m_H^{\mathrm{Riem}}(\Sigma)
        =
        m_H^{\mathrm{ST}}(\Sigma)
\end{equation}
for every leaf \(\Sigma\) of the TMCF foliation.

\begin{lem}[ADM compatibility of tangentially maximal tails]
\label{lem:ADM-compat}
Let \(f=f^{(r_1,T_0)}\) be an admissible tangentially maximal tail furnished by
Theorem~\ref{thm:global-existence}\textup{(iii)}. Then the asymptotically Euclidean end of
\((M_f,g_f)\) has ADM mass equal to \(m_{ADM}\) in the chosen rest frame.
\end{lem}

\begin{proof}
We work in asymptotically Euclidean coordinates on the end, chosen in the
asymptotic rest frame.  In this proof, for a tensor \(E\) on the end, we write $E=\mathcal O_1(r^{-q})$ to mean that, in these coordinates,
\[
        E=\mathcal O(r^{-q}),
        \qquad
        \partial E=\mathcal O(r^{-q-1}).
\]
We write $E=o_1(r^{-q})$ to mean that
\[
        E=o(r^{-q}),
        \qquad
        \partial E=o(r^{-q-1}).
\]

The estimates from Theorem~\ref{thm:global-existence}, together with the
asymptotic flatness of the ambient spacetime coefficients, give
\[
        f=T_0+o(1),
        \qquad
        df=\mathcal O_1(r^{-1})
\]
as \(r\to\infty\). In the Schauder alternative \((A3)\) this improves to
\(f=T_0+\mathcal O(r^{-\tau})\) and \(df=\mathcal O_1(r^{-1-\tau})\), but the
weaker estimates displayed above are sufficient for the ADM comparison.

We compare the induced metric \(g_f\) on the graph \(t=f(x)\) with the spatial
metric \(g(T_0)\) of the coordinate slice \(\Sigma_{T_0}\). In the ADM
decomposition,
\[
g_f
=
g(t)\big|_{t=f}
+
2\,\beta\big|_{t=f}\odot df
-
\big(N^2-|\beta|_{g(t)}^2\big)\big|_{t=f}\,df^2 .
\]
The asymptotic flatness bounds give, on the time slab containing the graph,
\[
        \beta=\mathcal O_1(r^{-\tau}),
        \qquad
        N^2-|\beta|_{g(t)}^2=1+\mathcal O_1(r^{-\tau}).
\]
Therefore the two terms containing \(df\) satisfy
\[
        2\,\beta\big|_{t=f}\odot df
        =
        \mathcal O_1(r^{-1-\tau})
\]
and
\[
        \big(N^2-|\beta|_{g(t)}^2\big)\big|_{t=f}\,df^2
        =
        \mathcal O_1(r^{-2}).
\]
Both are \(o_1(r^{-1})\), and hence are too small to contribute to the ADM
surface integral.

It remains to compare \(g(t)|_{t=f}\) with \(g(T_0)\). Here we use the
time-derivative part of the weighted coefficient assumptions. In asymptotically
Euclidean coordinates these give
\[
        \partial_t g_{ij}=\mathcal O_1(r^{-1-\tau}),
        \qquad
        \partial_t^2 g_{ij}=\mathcal O(r^{-2-\tau})
\]
on the relevant time slab. Set $$v:= f-T_0.$$ By the mean value theorem,
\[
\begin{aligned}
        g_{ij}(f,x)-g_{ij}(T_0,x)
        &=
        v(x)
        \int_0^1
        \partial_t g_{ij}\big(T_0+\theta v(x),x\big)\,d\theta .
\end{aligned}
\]
Since \(v=o(1)\) and \(\partial_t g_{ij}=\mathcal O(r^{-1-\tau})\), this gives
\[
        g(t)\big|_{t=f}-g(T_0)
        =
        o(r^{-1-\tau})
        =
        o(r^{-1}).
\]

For one spatial derivative, differentiate the preceding integral formula:
\[
\begin{aligned}
\partial_k\big(g_{ij}(f,x)-g_{ij}(T_0,x)\big)
&=
(\partial_k v)
\int_0^1
\partial_t g_{ij}\big(T_0+\theta v,x\big)\,d\theta
\\
&\quad
+v
\int_0^1
\partial_k\partial_t g_{ij}\big(T_0+\theta v,x\big)\,d\theta
\\
&\quad
+v
\int_0^1
\theta\,\partial_t^2 g_{ij}\big(T_0+\theta v,x\big)\,\partial_k v\,d\theta .
\end{aligned}
\]
Using
\[
        \partial_k v=\mathcal O(r^{-1}),
        \qquad
        \partial_t g_{ij}=\mathcal O(r^{-1-\tau}),
        \qquad
        \partial_k\partial_t g_{ij}=\mathcal O(r^{-2-\tau}),
        \qquad
        \partial_t^2 g_{ij}=\mathcal O(r^{-2-\tau}),
\]
we obtain
\[
        \partial_k\big(g_{ij}(f,x)-g_{ij}(T_0,x)\big)
        =
        o(r^{-2}).
\]
Thus
\[
        g(t)\big|_{t=f}-g(T_0)
        =
        o_1(r^{-1}).
\]

Combining the preceding estimates gives
\[
        g_f-g(T_0)=o_1(r^{-1}).
\]
Therefore the difference between the ADM integrands of \(g_f\) and \(g(T_0)\)
has surface integral tending to zero on coordinate spheres at infinity. Hence
\[
        m_{ADM}(M_f,g_f)
        =
        m_{ADM}(\Sigma_{T_0},g(T_0)).
\]

It remains only to account for the auxiliary coordinate changes used in the
tail construction. These changes are time-preserving and asymptotic to the
identity in the asymptotically Euclidean end: in spatial coordinates they have
the form
\[
        x^i \longmapsto x^i+\zeta^i(x),
        \qquad
        \zeta=\mathcal O(r^{-\tau}),
        \qquad
        \partial\zeta=\mathcal O(r^{-1-\tau}),
\]
with the corresponding higher derivative decay required for the ADM surface
integral. Since \(\tau>1/2\), these are admissible asymptotically Euclidean
coordinate changes. By the standard invariance of the ADM energy-momentum under
such coordinate changes, they do not alter the ADM surface integral
\cite{adm-invariant1,adm-invariant2,adm-invariant3,adm-invariant5}.

Finally, the global time function was chosen so that the slices \(\Sigma_t\)
are adapted to the asymptotic rest frame of the ADM four-momentum. In this
frame the ADM linear momentum of \(\Sigma_{T_0}\) vanishes, and its ADM energy
is the invariant ADM mass \(m_{ADM}\). Hence
\[
        m_{ADM}(M_f,g_f)=m_{ADM}.
\]
\end{proof}

\begin{lem}[Outward minimizing boundary spheres]
\label{lem:outerminimizing}
Let \(f=f^{(r_1,T_0)}\) be an admissible tangentially maximal tail furnished by
Theorem~\ref{thm:global-existence}\textup{(iii)}. Then the boundary leaf
\[
        \Sigma_{r_1}^{(T_0)}\subset(M_f,g_f)
\]
is connected and outward minimizing.
\end{lem}

\begin{proof}
This is the standard calibration argument showing that a boundary leaf of a mean-convex foliation
is outward minimizing; related forms of this argument appear, for instance, in
\cite{carla-outward}. We include the proof for completeness.

Connectedness is immediate, since \(\Sigma_{r_1}^{(T_0)}\) is a sphere. Since \(f\) is admissible,
every leaf \(S_r\subset M_f\) satisfies
\[
        H_{f,r}>0
        \qquad\text{for }r\ge r_1.
\]
Let
\[
        X:={\bf n}_f=\frac{\nabla^{g_f}r}{|\nabla^{g_f}r|_{g_f}}.
\]
Then \(|X|_{g_f}=1\), and on each leaf \(S_r\),
\[
        \operatorname{div}_{g_f}X=H_{f,r}\ge0.
\]

Let \(\Sigma\subset M_f\) be any smooth closed surface enclosing
\(\Sigma_{r_1}^{(T_0)}\), and let \(\Omega\) be the compact region between them. Denote by \(\nu\)
the outward unit normal to \(\Sigma\). The divergence theorem gives
\[
\int_\Sigma \langle X,\nu\rangle\,d\mu_\Sigma
-
\int_{\Sigma_{r_1}^{(T_0)}}1\,d\mu
=
\int_\Omega \operatorname{div}_{g_f}X\,d\mu_{g_f}
\ge0.
\]
Since \(|X|_{g_f}=1\), we have \(\langle X,\nu\rangle\le1\) on \(\Sigma\), and hence
\[
|\Sigma|
\ge
\int_\Sigma \langle X,\nu\rangle\,d\mu_\Sigma
\ge
|\Sigma_{r_1}^{(T_0)}|.
\]
Thus \(\Sigma_{r_1}^{(T_0)}\) is outward minimizing. The usual approximation argument extends the
conclusion from smooth enclosing surfaces to the standard outward-minimizing class.
\end{proof}

\begin{prop}[Nonnegative scalar curvature on tangentially maximal tails]
\label{prop:R-nonneg-tm}
Let \(f=f^{(r_1,T_0)}\) be an admissible tangentially maximal tail furnished by
Theorem~\ref{thm:global-existence}\textup{(iii)}. Then
\[
        R_{g_f}\ge0
\]
on \(M_f\).
\end{prop}

\begin{proof}
Let \(K_f\) be the second fundamental form of
\(M_f\hookrightarrow(\widehat{\mathcal M},\gtime)\), and let \(\mu_f\) be the energy density of the
induced initial data. The Hamiltonian constraint gives
\[
        R_{g_f}+(\tr_{g_f}K_f)^2-|K_f|_{g_f}^2
        =
        16\pi\mu_f.
\]
Since \(f\) solves the TMCF equation, the leaves \(S_r\subset M_f\) satisfy
\[
        \tr_{S_r}K_f=0.
\]
With respect to the splitting
\[
        TM_f=\mathrm{span}\{{\bf n}_f\}\oplus TS_r,
\]
we therefore have
\[
        \tr_{g_f}K_f=K_f({\bf n}_f,{\bf n}_f),
\]
and
\[
        |K_f|_{g_f}^2-(\tr_{g_f}K_f)^2
        =
        2\,|K_f({\bf n}_f,\cdot)|_{\gamma_f}^2
        +
        |K_f|_{\gamma_f}^2
        \ge0.
\]
Thus
\[
        R_{g_f}
        =
        16\pi\mu_f
        +
        |K_f|_{g_f}^2-(\tr_{g_f}K_f)^2
        \ge16\pi\mu_f.
\]
The dominant energy condition implies \(\mu_f\ge0\), and hence \(R_{g_f}\ge0\).
\end{proof}

\begin{cor}[Hawking-mass comparison on tangentially maximal tails]
\label{cor:HI-comparison}
Let \(f=f^{(r_1,T_0)}\) be an admissible tangentially maximal tail furnished by
Theorem~\ref{thm:global-existence}\textup{(iii)}. Then
\[
        m_{ADM}
        \ge
        m_H^{\mathrm{Riem}}\big(\Sigma_{r_1}^{(T_0)}\big)
        =
        m_H^{\mathrm{ST}}\big(\Sigma_{r_1}^{(T_0)}\big).
\]
\end{cor}

\begin{proof}
By Proposition~\ref{prop:R-nonneg-tm}, the Riemannian manifold \((M_f,g_f)\) has nonnegative scalar
curvature. By Lemma~\ref{lem:outerminimizing}, its boundary leaf
\(\Sigma_{r_1}^{(T_0)}\) is connected and outward minimizing. Huisken--Ilmanen's weak inverse mean
curvature flow comparison in \cite{H-I} gives
\[
        m_{ADM}(M_f,g_f)
        \ge
        m_H^{\mathrm{Riem}}\big(\Sigma_{r_1}^{(T_0)}\big)
\]
 Lemma~\ref{lem:ADM-compat} gives
\[
        m_{ADM}(M_f,g_f)=m_{ADM},
\]
and \eqref{eq:hawking-masses-agree} gives equality of the Riemannian and spacetime Hawking masses
on the boundary leaf. This proves the claim.
\end{proof}

\begin{proof}[Proof of Theorem~\ref{thm:SPI}]
Set
\[
        A_\ast:=|S_\ast|.
\]

We first compare \(A_\ast\) with the limiting area \(A_\infty\) of the last smooth horizon piece.
The cross-section assumption on \(S_\ast\) allows us to compare it with the preferred MOTS leaves
\(\mathcal S_t\) of the apparent-horizon tube. On a spacelike dynamical-horizon piece, the
uniqueness of the dynamical-horizon foliation implies that a MOTS cross-section agrees with one of
the preferred leaves \cite{AshtekarGalloway05}. On a null isolated-horizon piece, all smooth
cross-sections have the same area \cite{AshtekarBeetleFairhurst99,AshtekarKrishnan04}. Therefore,
for the purpose of the area comparison, we may start the future evolution from a preferred leaf
\(\mathcal S_{t_\ast}\) with
\[
        |\mathcal S_{t_\ast}|=|S_\ast|=A_\ast .
\]

On each smooth piece of \(\mathcal H_{\mathrm{app}}\), Proposition~\ref{prop:mott-area-law} implies
that the total area is nondecreasing toward the future. By assumption, the total outermost-horizon
area does not decrease across jump times. Therefore the total area along the future portion of
\(\mathcal H_{\mathrm{app}}\) starting from \(S_\ast\) is nondecreasing. Since the last smooth piece
has limiting area \(A_\infty\), we obtain
\begin{equation}\label{eq:Ainfty-ge-Aast}
        A_\infty\ge A_\ast.
\end{equation}

We first identify the ADM mass of the embedded initial data set with the
fixed ADM mass of the spacetime end.  Since
\((M_\ast,g_\ast,K_\ast)\) is an admissible asymptotically flat hypersurface
in the same asymptotic end of \((\widehat{\mathcal M},\gtime)\), it
determines the same ADM energy-momentum four-vector as the reference
Cauchy slices \(\Sigma_t\), up to the standard asymptotic Lorentz
transformation law
\cite{adm-invariant1,adm-invariant2,adm-invariant3,adm-invariant4,adm-invariant5}.
Thus its invariant ADM mass agrees with the invariant ADM mass \(m_{ADM}\)
fixed above:
\begin{equation}\label{eq:adm-original-rest-frame}
        m_{ADM}(M_\ast,g_\ast,K_\ast)=m_{ADM}.
\end{equation}

It remains to prove
\[
        m_{ADM}\ge \sqrt{\frac{A_\infty}{16\pi}}.
\]
We do this by constructing a sequence of admissible tangentially maximal tails whose inner boundary
leaves approach the final horizon.

Let $\varepsilon_j\downarrow0$. By Lemma~\ref{lem:late-horizon-hawking}, for each \(j\) there exist
\[
        T^{\mathrm{hor}}(\varepsilon_j):= T_j^{\mathrm{hor}}\ge\underline T,
        \qquad
        \rho_j>0,
\]
such that every spacelike \(C^2\)-graph of size less than \(\rho_j\) over a horizon section
\(\mathcal S_t\), with \(t\ge T_j^{\mathrm{hor}}\), has spacetime Hawking mass within
\(\varepsilon_j\) of \(\sqrt{A_\infty/(16\pi)}\).

Choose a sequence of radii
\[
        r_j>r_0,
        \qquad
        r_j\downarrow r_0,
\]
so that, in the original late-time collar coordinates from the quasi final state hypothesis, the
coordinate cylinder \(r=r_j\) is a \(C^2\)-graph over the horizon cylinder \(r=r_0\) of size less
than \(\rho_j/4\). Equivalently, for every sufficiently late \(t\), the sphere $S_{t,r_j}$ is \(C^2\)-close, with size less than \(\rho_j/4\), to the horizon section $\mathcal S_t=S_{t,r_0}$.

Now fix \(j\). Theorem~\ref{thm:global-existence}\textup{(iii)} applies on the fixed exterior
tail \(M_{r_j,\infty}\). The constants in the existence and gauge-reduction estimates may depend
on \(r_j\), but this is harmless because \(r_j\) has now been fixed. Apply
Proposition~\ref{prop:app-late-slab-gauge} with \(r_1=r_j\). If
\(R(t,\tilde r,\tilde q)\) denotes the original radial coordinate expressed in the final
good-gauge chart, then \eqref{eq:app-r1-C2-close} gives
\[
\sup_{t\in I_{h_j}(T)}
\|R(t,r_j,\cdot)-r_j\|_{C^2(S^2)}
\le
C_j h_j\,\mathfrak B^\sharp(T-h_j;r_-^{(j)}),
\]
where \(r_-^{(j)}=(r_0+r_j)/2\). Moreover,
\eqref{eq:app-r1-time-derivative-bounds} gives uniform bounds, depending on the fixed tail
\(r\ge r_j\), for \(\partial_tR\) in \(C^1(S^2)\) and \(\partial_t^2R\) in \(C^0(S^2)\) along this
inner cylinder.

Using the time-derivative bounds for \(R\), choose \(0<\delta_j<h_j\), depending only on the fixed
tail \(r\ge r_j\), such that the following implication holds: whenever
\(\|\eta\|_{C^2(S^2)}<\delta_j\) and \(T\) is large enough that
\[
\sup_{t\in I_{h_j}(T)}
\|R(t,r_j,\cdot)-r_j\|_{C^2(S^2)}
<
\frac{1}{2}\min\left\{\frac{\rho_j}{4},\,\frac{r_j-r_0}{2}\right\},
\]
then the chain rule gives
\[
\|R(T+\eta(\cdot),r_j,\cdot)-r_j\|_{C^2(S^2)}
<
\min\left\{\frac{\rho_j}{4},\,\frac{r_j-r_0}{2}\right\}.
\]
Since \(\mathfrak B^\sharp(T-h_j;r_-^{(j)})\to0\), the displayed smallness of
\(R(t,r_j,\cdot)-r_j\) holds for all sufficiently large \(T\). Note that the condition \(\delta_j<h_j\) ensures that \(T+\eta(\cdot)\in I_{h_j}(T)\).

We now choose
\[
        T_j\ge \max\{j,T_j^{\mathrm{hor}},T_{\mathrm{large}}(r_j)\}
\]
large enough so that the preceding implication applies and, in addition, the boundary time graph
of the TMCF tail satisfies
\[
\eta_j(\cdot):=
f^{(r_j,T_j)}(r_j,\cdot)-T_j,
\qquad
\|\eta_j\|_{C^2(S^2)}
<
\min\left\{\delta_j,\frac{\rho_j}{4}\right\}.
\]
Here the \(C^2\)-smallness of \(\eta_j\) follows from the estimates in
Theorem~\ref{thm:global-existence}, using the upgrading assumptions in the quasi final state
hypothesis, namely either \textup{\ref{it:A2}} or \textup{\ref{it:A3}}.

With this choice of \(T_j\), the boundary leaf of the tangentially maximal tail, pulled back to the
original late-time collar coordinates, has original radial coordinate
\[
R(T_j+\eta_j(\cdot),r_j,\cdot),
\]
and therefore
\[
\|R(T_j+\eta_j(\cdot),r_j,\cdot)-r_j\|_{C^2(S^2)}
<
\min\left\{\frac{\rho_j}{4},\,\frac{r_j-r_0}{2}\right\}.
\]
The second bound ensures that the boundary leaf remains on the exterior side of the horizon.

Let $\Sigma_j:=\Sigma_{r_j}^{(T_j)}$ be the inner boundary leaf of the admissible tangentially maximal tail produced at central time
\(T_j\). By the preceding choices, \(\Sigma_j\), expressed in the original late-time collar
coordinates, is a spacelike \(C^2\)-graph of size less than \(\rho_j\) over the horizon section $\mathcal S_{T_j}.$ Moreover, its transverse graph function is nonnegative, so \(\Sigma_j\) lies on the exterior side
of the horizon. Since \(T_j\ge T_j^{\mathrm{hor}}\), Lemma~\ref{lem:late-horizon-hawking} gives
\[
\left|
        m_H^{\mathrm{ST}}(\Sigma_j)
        -
        \sqrt{\frac{A_\infty}{16\pi}}
\right|
<
        \varepsilon_j.
\]
Thus
\[
        m_H^{\mathrm{ST}}(\Sigma_j)
        \longrightarrow
        \sqrt{\frac{A_\infty}{16\pi}}.
\]

On the other hand, each \(\Sigma_j\) is the boundary leaf of an admissible tangentially maximal
tail. Therefore Corollary~\ref{cor:HI-comparison} gives
\[
        m_{ADM}
        \ge
        m_H^{\mathrm{ST}}(\Sigma_j)
        \qquad
        \text{for every }j.
\]
Passing to the limit yields
\[
        m_{ADM}
        \ge
        \sqrt{\frac{A_\infty}{16\pi}}.
\]
Using \eqref{eq:Ainfty-ge-Aast} and \eqref{eq:adm-original-rest-frame}, we conclude
\[
        m_{ADM}(M_\ast,g_\ast,K_\ast)
        =
        m_{ADM}
        \ge
        \sqrt{\frac{A_\infty}{16\pi}}
        \ge
        \sqrt{\frac{A_\ast}{16\pi}}.
\]
This is exactly \eqref{eq:SPI-final}.
\end{proof}

\begin{remark}
A key point in the limiting argument is that the late time is chosen after the radius
\(r_\varepsilon>r_0\) has been fixed. The constants in the TMCF existence and gauge-reduction
estimates need not remain uniform as \(r_\varepsilon\downarrow r_0\), but the late-time decay in the
quasi final state hypothesis allows one to compensate by pushing the comparison tail further into
the future. If one wanted to replace the decay assumptions by a weaker eventual-smallness
hypothesis near the horizon, one would need estimates for the TMCF construction and the gauge
reduction that remain uniform as the inner radius approaches \(r_0\).
\end{remark}

\begin{remark}
The proof shows that the Penrose argument reduces to two late-time questions:
\begin{enumerate}[label=\textup{(\roman*)}]
\item proving late-time gauge reducibility together with the forcing decay assumption
\textup{\ref{it:A1}} and either of the upgrading assumptions \textup{\ref{it:A2}} or
\textup{\ref{it:A3}};
\item proving that the last smooth horizon piece has a limiting area \(A_\infty\).
\end{enumerate}
The tangentially maximal existence theory from the previous section then provides the comparison
hypersurfaces needed to turn these late-time statements into the Penrose inequality.
\end{remark}

\subsection{The Ben--Dov counterexamples and the future black-hole condition}
\label{sub:bendov-counterexamples}

Ben--Dov's construction \cite{ben-dov} appears at first reading to produce a direct counterexample to Theorem~\ref{thm:SPI}: he gives an asymptotically flat, spherically symmetric initial data set satisfying the dominant energy condition whose outermost MOTS $\mathcal S$ violates
\[
m_{ADM}\ge\sqrt{\frac{|\mathcal S|}{16\pi}}.
\]
Since Theorem~\ref{thm:SPI} is an inequality of exactly this form, we must identify the hypothesis that fails in his example and the step of our proof at which this failure becomes visible.

The short answer is: Ben--Dov's exceptional MOTS lies on a \emph{white-hole-type} branch of a marginally outer trapped tube, on which $\theta_{(\underline\ell)}>0$. This is exactly the sign that breaks Proposition~\ref{prop:mott-area-law}; consequently the area inequality $A_\infty\ge A_\ast$ used in the area-comparison step of the proof of Theorem~\ref{thm:SPI} fails along his tube. Throughout this subsection we use the convention $\gtime(\ell,\underline\ell)=-2$ for the future-directed null normals to a MOTS, consistent with the rest of the paper.

\subsubsection*{Recalling Ben--Dov's construction}

The spacetime is obtained by matching four pieces
\[
\text{Schwarzschild}(M),
\qquad
\text{RW},
\qquad
\text{Schwarzschild}(M_1),
\qquad
\text{RW},
\]
with matching parameters
\[
M=\tfrac12a_m\sin^3\chi_0,
\qquad
M_1=\tfrac12a_m\sin^3\chi_1,
\qquad
\tfrac{\pi}{2}\le\chi_1<\chi_0<\pi.
\]
Since $\sin\chi$ is strictly decreasing on $[\pi/2,\pi]$, $\sin\chi_1>\sin\chi_0$, hence $M_1>M$. The closed dust Robertson--Walker piece has metric
\[
\mathbf g_{RW}=-d\tau^2+a(\tau)^2\bigl(d\chi^2+\sin^2\chi\,d\omega^2\bigr),
\]
where, in terms of the parameter $\eta$,
\[
a(\eta)=\tfrac{a_m}{2}(1-\cos\eta),
\qquad
\tau(\eta)=\tfrac{a_m}{2}(\eta-\sin\eta),
\qquad 0<\eta<2\pi,
\]
and
\[
\dot a:=\tfrac{da}{d\tau}=\tfrac{\sin\eta}{1-\cos\eta}=\cot\tfrac{\eta}{2}.
\]

Ben--Dov chooses a spacelike initial slice that intersects the marginally outer branch in the \emph{second} Schwarzschild region. The outermost MOTS on that slice is a coordinate sphere of areal radius $2M_1$, hence of area
\[
A_{BD}=16\pi M_1^2,
\]
whereas the ADM mass of the chosen asymptotically flat end is $M$. Therefore
\[
m_{ADM}=M<M_1=\sqrt{\tfrac{A_{BD}}{16\pi}},
\]
which is the advertised violation of the raw apparent-horizon Penrose inequality.

\subsubsection*{The exceptional MOTS is not future black-hole type}

In the RW piece, define the future-directed radial null normals
\[
\ell=\partial_\tau+\tfrac{1}{a}\partial_\chi,
\qquad
\underline\ell=\partial_\tau-\tfrac{1}{a}\partial_\chi,
\]
so that $\gtime(\ell,\underline\ell)=-2$ and $\ell$ points toward increasing $\chi$. Here ``outward'' is taken relative to the asymptotically flat end of mass $M$ in Ben--Dov's construction, which sits at $\chi\to\pi^-$ in the matching convention used there. For a coordinate sphere $\{\tau=\mathrm{const},\chi=\mathrm{const}\}$, the areal radius is $R=a(\tau)\sin\chi$, and for either null normal $k$, $\theta_{(k)}=\tfrac{2}{R}k(R)$. A short computation gives
\[
\theta_{(\ell)}=\tfrac{2}{a}(\dot a+\cot\chi),
\qquad
\theta_{(\underline\ell)}=\tfrac{2}{a}(\dot a-\cot\chi).
\]
The marginally outer tube $\theta_{(\ell)}=0$ in the RW piece is therefore
\[
\cot\chi=-\dot a=-\cot\tfrac{\eta}{2}
\quad\Longleftrightarrow\quad
\chi=\chi_\ell(\eta):=\pi-\tfrac{\eta}{2},
\qquad 0<\chi<\pi.
\]
Along this tube,
\[
\theta_{(\underline\ell)}
=\tfrac{2}{a}(\dot a-\cot\chi_\ell)
=\tfrac{4}{a}\cot\tfrac{\eta}{2},
\]
so
\[
\theta_{(\underline\ell)}>0 \text{ on the expanding epoch } 0<\eta<\pi,
\qquad
\theta_{(\underline\ell)}<0 \text{ on the collapsing epoch } \pi<\eta<2\pi.
\]
Ben--Dov's construction uses the expanding epoch, so $\theta_{(\underline\ell)}>0$ along the relevant marginally outer tube.

This sign is preserved across the matching into the second Schwarzschild region: the Darmois matching is shell-free, so the null expansions of the matched symmetry spheres agree across the junction. Hence the continuation of the marginally outer tube in the second Schwarzschild piece must be the white-hole branch $U=0,\,V<0$ of its Kruskal horizon, on which $\theta_{(\underline\ell)}>0$ as well. In particular, Ben--Dov's exceptional MOTS satisfies
\[
\theta_{(\ell)}=0,
\qquad
\theta_{(\underline\ell)}>0,
\]
not the future black-hole condition $\theta_{(\underline\ell)}<0$ assumed in this paper.

\begin{remark}[Causal character of the tube]
\label{rem:ben-dov-tube-timelike}
For completeness, the underlying marginally outer trapped tube is in fact timelike, hence not even achronal. Indeed, in $(\eta,\chi)$ coordinates the two-dimensional radial part of $\mathbf g_{RW}$ is $a(\eta)^2(-d\eta^2+d\chi^2)$, and the tangent to $\chi_\ell(\eta)=\pi-\eta/2$ is $X=\partial_\eta-\tfrac{1}{2}\partial_\chi$, so
\[
\mathbf g_{RW}(X,X)=a(\eta)^2\bigl(-1+\tfrac14\bigr)=-\tfrac34 a(\eta)^2<0.
\]
Thus the RW marginally outer tube is timelike, and a fortiori is not a future outer dynamical horizon of the kind used in the area-law argument.
\end{remark}

\subsubsection*{Where the proof of Theorem~\ref{thm:SPI} fails for Ben--Dov}

The future black-hole condition $\theta_{(\underline\ell)}<0$ is built into our standing assumption that $\mathcal S_t$ is of future black-hole type on every connected component (Section~\ref{sub:spacetime-final-state}). Ben--Dov's exceptional MOTS violates this condition directly, so it is not an admissible $S_\ast$ for Theorem~\ref{thm:SPI}, and the construction does not enter the class of initial data covered by our result.

It is nevertheless instructive to locate the single point in the proof of Theorem~\ref{thm:SPI} at which dropping the sign condition would actually destroy the argument. The future black-hole hypothesis is used essentially only for the area comparison step, where we prove
\[
A_\infty\;\ge\;A_\ast,
\]
which is obtained from the area law on smooth pieces of $\mathcal H_{\mathrm{app}}$ (Proposition~\ref{prop:mott-area-law}). The proof of that proposition uses, on each connected component, that achronality of $\mathcal H_0$ gives $C\ge 0$ for the foliation normal $V=\ell-C\underline\ell$, combined with $\theta_{(\underline\ell)}<0$, to conclude
\[
\theta_{(V)}=\theta_{(\ell)}-C\,\theta_{(\underline\ell)}=-C\,\theta_{(\underline\ell)}\;\ge\;0.
\]
Ben--Dov's marginally outer tube fails both inputs simultaneously: by Remark~\ref{rem:ben-dov-tube-timelike} it is timelike (not achronal), and we have just shown $\theta_{(\underline\ell)}>0$ on it. Neither the achronality used to set up the calculation nor the sign that turns $-C\,\theta_{(\underline\ell)}$ into something nonnegative is available.

A direct computation confirms that the conclusion of Proposition~\ref{prop:mott-area-law} cannot be salvaged here: along the RW tube $\chi_\ell(\eta)=\pi-\eta/2$,
\[
A_{RW}(\eta)
=4\pi a(\eta)^2\sin^2\chi_\ell(\eta)
=\tfrac{\pi}{2}a_m^2(1-\cos\eta)^3,
\qquad
\tfrac{dA_{RW}}{d\eta}
=\tfrac{3\pi}{2}a_m^2(1-\cos\eta)^2\sin\eta,
\]
so the area varies non-monotonically (increasing on $0<\eta<\pi$, decreasing on $\pi<\eta<2\pi$); and the variation is along a timelike curve, so it carries no causal meaning compatible with an apparent-horizon area law.

The chain
\[
A_\infty\;\ge\;A(t_\ast)\;=\;A_\ast
\]
that drives the area comparison step of the proof therefore breaks down. The remaining ingredients of the proof --- ADM-mass identification, Huisken--Ilmanen comparison on tangentially maximal tails, and the limit $m_H^{\mathrm{ST}}(\Sigma_j)\to\sqrt{A_\infty/(16\pi)}$ --- are insensitive to the sign of $\theta_{(\underline\ell)}$ on $S_\ast$ and would remain valid. It is precisely the area-law step that singles out future black-hole horizons, and it is there, and only there, that Ben--Dov's construction would derail the argument.

\subsubsection*{Conclusion}

Ben--Dov's example shows that the Penrose inequality is false if formulated purely as an initial-data statement for the raw area of an outermost apparent horizon, with no condition on the surface's status within a future apparent-horizon evolution.  Once one imposes that the outermost MOTS lie on a future black-hole-type smooth piece of the apparent-horizon tube --- equivalently, $\theta_{(\underline\ell)}<0$ on every connected component --- the area law of Proposition~\ref{prop:mott-area-law} is recovered, the area-comparison step of the proof of Theorem~\ref{thm:SPI} goes through. Ben--Dov's construction therefore identifies precisely the hypothesis that any spacetime Penrose inequality of this kind cannot do without.
\subsection{The equality case}
\label{sub:equality-case}

We finally record what the present method gives in the equality case. The main point is that the
argument proves a sharp inequality through a limiting procedure. Thus equality forces strong
consequences for the horizon area and for the sequence of tangentially maximal comparison tails,
but a full Schwarzschild rigidity statement requires one additional compactness or stability input.

\begin{prop}[Consequences of equality]
\label{prop:equality-consequences}
Assume the hypotheses of Theorem~\ref{thm:SPI}. Let \(S_\ast\) be the horizon cross-section
appearing in the theorem, and set
\[
        A_\ast:=|S_\ast|.
\]
Suppose equality holds in the Penrose inequality:
\[
        m_{ADM}
        =
        \sqrt{\frac{A_\ast}{16\pi}}.
\]
Then
\[
        A_\infty=A_\ast=16\pi m_{ADM}^2.
\]
Moreover, the total area of the future apparent-horizon tube from \(S_\ast\) to
\(\mathcal H_{\mathrm{final}}\) is constant. In particular, every allowed jump between
\(S_\ast\) and \(\mathcal H_{\mathrm{final}}\) has zero total area jump, and every smooth
future black-hole-type MOTT piece between \(S_\ast\) and \(\mathcal H_{\mathrm{final}}\) is null.

Finally, if \(\Sigma_j\) denotes the sequence of boundary leaves of admissible tangentially maximal
comparison tails constructed in the proof of Theorem~\ref{thm:SPI}, then
\[
        m_H^{\mathrm{Riem}}(\Sigma_j)
        =
        m_H^{\mathrm{ST}}(\Sigma_j)
        \longrightarrow
        m_{ADM}.
\]
Thus the Riemannian Penrose comparisons associated with the tangentially maximal tails are
asymptotically sharp.
\end{prop}

\begin{proof}
The proof of Theorem~\ref{thm:SPI} gives the chain
\[
        m_{ADM}
        \ge
        \sqrt{\frac{A_\infty}{16\pi}}
        \ge
        \sqrt{\frac{A_\ast}{16\pi}}.
\]
If
\[
        m_{ADM}
        =
        \sqrt{\frac{A_\ast}{16\pi}},
\]
then both inequalities in the chain must be equalities. Hence
\[
        A_\infty=A_\ast=16\pi m_{ADM}^2.
\]

By Proposition~\ref{prop:mott-area-law}, the total area is nondecreasing along every smooth
future black-hole-type piece of \(\mathcal H_{\mathrm{app}}\). By the no-area-decreasing jump
assumption in Theorem~\ref{thm:SPI}, the total area is also nondecreasing across the allowed
jumps. Since the initial area \(A_\ast\) and the limiting final area \(A_\infty\) are equal, no
increase can occur on any smooth piece or across any jump between \(S_\ast\) and
\(\mathcal H_{\mathrm{final}}\). Thus every jump has zero total area jump, and the total area is
constant on every smooth piece.

Let \(\mathcal H_0\) be one such smooth piece. On a connected component of \(\mathcal H_0\), let
\(V\) be a vector field tangent to \(\mathcal H_0\), normal to the MOTS leaves, and oriented so
that
\[
        dt(V)>0.
\]
Then, after multiplying by a positive lapse factor, we may write
\[
        V=\varphi(\ell-C\,\underline\ell),
        \qquad
        \varphi>0.
\]
Here \(\ell\) and \(\underline\ell\) are the future-directed outgoing and ingoing null normals,
normalized by
\[
        \gtime(\ell,\underline\ell)=-2.
\]
Since the leaves are of future black-hole type,
\[
        \theta_{(\ell)}=0,
        \qquad
        \theta_{(\underline\ell)}<0.
\]
The area variation along \(V\) is
\[
        \theta_{(V)}
        =
        \varphi\big(\theta_{(\ell)}-C\theta_{(\underline\ell)}\big)
        =
        -\varphi C\,\theta_{(\underline\ell)}.
\]
By the causal-character statement in Proposition~\ref{prop:mott-area-law}, \(C\ge0\). Hence
\[
        \theta_{(V)}\ge0.
\]
Since the area is constant on \(\mathcal H_0\), the integral of \(\theta_{(V)}\) over each leaf
vanishes. The integrand is nonnegative and continuous, while
\[
        \varphi>0,
        \qquad
        \theta_{(\underline\ell)}<0.
\]
Therefore
\[
        C=0
\]
on each connected component. Consequently
\[
        \gtime(\ell-C\underline\ell,\ell-C\underline\ell)=4C=0,
\]
and the smooth piece is null.

It remains to prove the asymptotic sharpness assertion. The sequence \(\Sigma_j\) constructed in
the proof of Theorem~\ref{thm:SPI} satisfies
\[
        m_H^{\mathrm{ST}}(\Sigma_j)
        \longrightarrow
        \sqrt{\frac{A_\infty}{16\pi}}.
\]
Since \(A_\infty=16\pi m_{ADM}^2\), this gives
\[
        m_H^{\mathrm{ST}}(\Sigma_j)\to m_{ADM}.
\]
On tangentially maximal leaves,
\[
        m_H^{\mathrm{ST}}(\Sigma_j)
        =
        m_H^{\mathrm{Riem}}(\Sigma_j),
\]
by \eqref{eq:hawking-masses-agree}. Hence
\[
        m_H^{\mathrm{Riem}}(\Sigma_j)
        =
        m_H^{\mathrm{ST}}(\Sigma_j)
        \longrightarrow
        m_{ADM}.
\]
\end{proof}

\begin{cor}[Vanishing flux along smooth equality pieces]
\label{cor:equality-null-flux}
Assume equality in Theorem~\ref{thm:SPI}. Let \(\mathcal H_0\) be a smooth horizon piece between
\(S_\ast\) and \(\mathcal H_{\mathrm{final}}\). Then each connected component of
\(\mathcal H_0\) is a null hypersurface with vanishing outgoing expansion. Moreover, if \(L\) is
any future-directed null generator of such a component, then
\[
        \sigma_{(L)}=0,
        \qquad
        T(L,L)=0.
\]
Here \(\sigma_{(L)}\) is the null shear, i.e. the trace-free part of the null second fundamental
form
\[
        \chi_{(L)}(X,Y):=\gtime(\nabla_XL,Y),
        \qquad
        X,Y\in TS_t,
\]
and, by Einstein's equations,
\[
        \operatorname{Ric}(L,L)=8\pi T(L,L).
\]
\end{cor}

\begin{proof}
By Proposition~\ref{prop:equality-consequences}, every smooth horizon piece between
\(S_\ast\) and \(\mathcal H_{\mathrm{final}}\) is null. Hence, on each connected component, the
future-directed null generator \(L\) is tangent to the horizon and normal to the MOTS leaves.
Since the leaves are marginally outer trapped,
\[
        \theta_{(L)}=0
\]
after the harmless rescaling from \(\ell\) to \(L\).

The Raychaudhuri equation for the null congruence generated by \(L\), allowing a non-affine
parametrization \(\nabla_LL=\kappa_LL\), is
\[
        L(\theta_{(L)})
        =
        \kappa_L\theta_{(L)}
        -
        \frac12\theta_{(L)}^2
        -
        |\sigma_{(L)}|^2
        -
        \operatorname{Ric}(L,L).
\]
Using Einstein's equations and the fact that \(L\) is null,
\[
        \operatorname{Ric}(L,L)=8\pi T(L,L).
\]
Since \(\theta_{(L)}\equiv0\) on the horizon piece, the left-hand side and the first two terms on
the right-hand side vanish. Therefore
\[
        |\sigma_{(L)}|^2+8\pi T(L,L)=0.
\]
The dominant energy condition implies
\[
        T(L,L)\ge0.
\]
Hence both terms vanish:
\[
        \sigma_{(L)}=0,
        \qquad
        T(L,L)=0.
\]
\end{proof}

\begin{remark}
Thus equality forces the future horizon portion from \(S_\ast\) to the final horizon to be
non-expanding: there is no area growth, no positive jump in total area, and no null matter or shear
flux through the smooth horizon pieces. This is the expected spacetime analogue of the rigidity
mechanism in the Riemannian Penrose inequality. Notice, however, that this conclusion concerns the
horizon tube and the asymptotically sharp comparison sequence. It does not yet identify the full
exterior spacetime.
\end{remark}

\begin{prop}[Equality under a strong Kerr final state]
\label{prop:equality-strong-final-state}
Assume, in addition to the hypotheses of Theorem~\ref{thm:SPI}, that the strong final state
hypothesis holds with limiting subextremal Kerr parameters
\[
        m>0,
        \qquad
        |a|<m.
\]
If equality holds in Theorem~\ref{thm:SPI}, then the limiting Kerr state is Schwarzschild:
\[
        a=0,
        \qquad
        m=m_{ADM}.
\]
\end{prop}

\begin{proof}
Under the strong final state hypothesis as formulated here, the ADM mass of the asymptotic end
agrees with the mass parameter of the limiting Kerr exterior:
\[
        m=m_{ADM}.
\]
Indeed, the late-time exterior coefficients converge on every fixed asymptotically flat tail to
the Kerr coefficients in the weighted topology used to compute the ADM surface integral.

The area of a Kerr horizon is
\[
        A_{\mathrm{Kerr}}
        =
        8\pi m\left(m+\sqrt{m^2-a^2}\right).
\]
By the strong final state hypothesis,
\[
        A_\infty=A_{\mathrm{Kerr}}.
\]
If equality holds in Theorem~\ref{thm:SPI}, Proposition~\ref{prop:equality-consequences} gives
\[
        m_{ADM}
        =
        \sqrt{\frac{A_\infty}{16\pi}}.
\]
Using \(m=m_{ADM}\) and \(A_\infty=A_{\mathrm{Kerr}}\), we obtain
\[
        m
        =
        \sqrt{
        \frac{8\pi m(m+\sqrt{m^2-a^2})}{16\pi}
        }.
\]
Squaring and using \(m>0\) gives
\[
        m^2
        =
        \frac12 m\left(m+\sqrt{m^2-a^2}\right).
\]
Hence
\[
        2m=m+\sqrt{m^2-a^2}.
\]
Therefore
\[
        \sqrt{m^2-a^2}=m,
\]
and so
\[
        a=0.
\]
Thus the limiting Kerr state is Schwarzschild.
\end{proof}

The preceding results are the unconditional equality consequences of the method. They show that
equality forces the horizon to have no future area growth and forces the comparison tails to become
asymptotically sharp for the Riemannian Penrose comparison. The remaining step in a full rigidity
theorem would be to upgrade this asymptotic sharpness to an exact equality statement on a limiting
Riemannian manifold. We isolate this additional input in the following conditional proposition.

\begin{prop}[Rigidity of an exact limiting tangentially maximal tail]
\label{prop:exact-tail-rigidity}
Assume equality holds in Theorem~\ref{thm:SPI}. Suppose, in addition, that the comparison tails
\[
        (M_{f_j},g_{f_j},K_{f_j})
\]
used in the proof subconverge, after the natural identifications of the exterior regions, smoothly
on compact subsets and up to the inner boundary to an admissible spacelike tangentially maximal
initial data set
\[
        (M_\infty,g_\infty,K_\infty),
\]
with connected inner boundary
\[
        S_\infty,
\]
where \(S_\infty\) is a MOTS cross-section of \(\mathcal H_{\mathrm{final}}\). Assume also that
the limiting boundary is outward minimizing in \((M_\infty,g_\infty)\), that the ADM mass of
\((M_\infty,g_\infty)\) is \(m_{ADM}\), and that
\[
        |S_\infty|=A_\infty.
\]
Then \((M_\infty,g_\infty)\) is isometric to the exterior region of the spatial Schwarzschild
manifold of mass \(m_{ADM}\). Moreover,
\[
        K_\infty\equiv0
\]
on \(M_\infty\).
\end{prop}

\begin{proof}
Since \(S_\infty\) is a MOTS cross-section of \(\mathcal H_{\mathrm{final}}\),
\[
        \theta_{(\ell)}=0.
\]
Since \(S_\infty\) is also a leaf of the limiting tangentially maximal foliation,
\[
        \tr_{S_\infty}K_\infty=0.
\]
With the convention
\[
        \theta_{(\ell)}=H_{S_\infty}+\tr_{S_\infty}K_\infty,
\]
we therefore obtain
\[
        H_{S_\infty}=0.
\]
Thus \(S_\infty\) is a minimal surface in \((M_\infty,g_\infty)\).

The limiting data are admissible and tangentially maximal, so Proposition~\ref{prop:R-nonneg-tm}
applies to give
\[
        R_{g_\infty}\ge0.
\]
By assumption, \(S_\infty\) is outward minimizing, and the ADM mass is \(m_{ADM}\). Equality in
Theorem~\ref{thm:SPI}, together with \(|S_\infty|=A_\infty=A_\ast\), gives
\[
        m_{ADM}
        =
        \sqrt{\frac{|S_\infty|}{16\pi}}.
\]
Therefore the Riemannian Penrose inequality is saturated by
\((M_\infty,g_\infty)\). By the rigidity statement in the Riemannian Penrose inequality
\cite{H-I,bray-conformal}, \((M_\infty,g_\infty)\) is isometric to the exterior region of the
spatial Schwarzschild manifold of mass \(m_{ADM}\). In particular,
\[
        R_{g_\infty}=0.
\]

It remains to show that the limiting second fundamental form vanishes. Let \({\bf n}\) be the
outward unit normal to the leaves of the limiting TMCF foliation, and decompose
\[
        TM_\infty=\mathrm{span}\{{\bf n}\}\oplus TS_r.
\]
The Hamiltonian constraint gives
\[
        R_{g_\infty}
        +
        (\tr_{g_\infty}K_\infty)^2
        -
        |K_\infty|_{g_\infty}^2
        =
        16\pi\mu_\infty.
\]
Since the limiting foliation is tangentially maximal,
\[
        \tr_{S_r}K_\infty=0.
\]
Therefore
\[
        \tr_{g_\infty}K_\infty=K_\infty({\bf n},{\bf n}),
\]
and
\[
        |K_\infty|_{g_\infty}^2-(\tr_{g_\infty}K_\infty)^2
        =
        2\,|K_\infty({\bf n},\cdot)|_{\gamma_\infty}^2
        +
        |K_\infty|_{\gamma_\infty}^2.
\]
Using \(R_{g_\infty}=0\), we obtain
\[
        0
        =
        16\pi\mu_\infty
        +
        2\,|K_\infty({\bf n},\cdot)|_{\gamma_\infty}^2
        +
        |K_\infty|_{\gamma_\infty}^2.
\]
The dominant energy condition gives \(\mu_\infty\ge0\). Hence all terms vanish:
\[
        \mu_\infty=0,
        \qquad
        K_\infty({\bf n},X)=0\quad\text{for all }X\in TS_r,
        \qquad
        K_\infty|_{TS_r}=0.
\]
Thus
\[
        K_\infty=\kappa\,{\bf n}^\flat\otimes{\bf n}^\flat
\]
for some function \(\kappa\).

The dominant energy condition also gives
\[
        |J_\infty|_{g_\infty}\le \mu_\infty=0,
\]
so
\[
        J_\infty=0.
\]
The momentum constraint is therefore
\[
        \operatorname{div}_{g_\infty}K_\infty
        -
        d(\tr_{g_\infty}K_\infty)
        =
        0.
\]
Since \(\tr_{g_\infty}K_\infty=\kappa\), taking the \({\bf n}\)-component gives
\[
\begin{aligned}
0
&=
\left(
\operatorname{div}_{g_\infty}K_\infty
-
d\kappa
\right)({\bf n})
\\
&=
\kappa\,H_{S_r}.
\end{aligned}
\]
The limiting tail is admissible, so
\[
        H_{S_r}>0
\]
on the exterior leaves \(r>r_0\). Hence
\[
        \kappa=0
\]
on \(M_\infty\setminus S_\infty\), and by continuity also on the boundary. Therefore
\[
        K_\infty\equiv0.
\]
\end{proof}

\begin{remark}[What is missing for full rigidity]
\label{rem:missing-full-rigidity}
The present proof of Theorem~\ref{thm:SPI} does not produce a fixed comparison manifold on which
equality is attained. Instead, it produces a sequence of admissible tangentially maximal tails
whose boundary Hawking masses converge to the ADM mass:
\[
        m_H^{\mathrm{Riem}}(\Sigma_j)
        =
        m_H^{\mathrm{ST}}(\Sigma_j)
        \longrightarrow
        m_{ADM}.
\]
This asymptotic sharpness does not by itself give a single Riemannian manifold satisfying equality
in the Riemannian Penrose inequality.

There are two natural ways to obtain a full Schwarzschild rigidity theorem from the present
method.

First, one could prove the compactness hypothesis in
Proposition~\ref{prop:exact-tail-rigidity}: the asymptotically sharp comparison tails would have
to converge, in a topology strong enough to pass the scalar curvature inequality, ADM mass,
outward-minimizing property, and boundary geometry to the limit, to an exact limiting
tangentially maximal tail. The Riemannian Penrose rigidity theorem would then identify the limiting
Riemannian metric as spatial Schwarzschild, and the tangentially maximal scalar-curvature identity
together with the momentum constraint would force \(K_\infty=0\).

Second, one could instead prove a quantitative stability theorem for the Riemannian Penrose
inequality adapted to the particular class of tangentially maximal tails constructed here. Such a
stability theorem would need uniform control of the geometry of the tails as their inner
boundaries approach the final horizon, together with a mechanism preventing loss of mass or area in
the limit.

Without one of these additional ingredients, the present method proves the sharp inequality and the
horizon-level equality consequences above, but it does not by itself imply that the original
spacetime exterior is Schwarzschild. Under a strong Kerr final state, equality identifies the
limiting stationary parameters as those of Schwarzschild by
Proposition~\ref{prop:equality-strong-final-state}. Identifying the entire exterior spacetime, or
even the entire past exterior region preceding the final state, with Schwarzschild would require an
additional rigidity, no-radiation, or unique-continuation input beyond the hypotheses used in this
paper.
\end{remark}

\appendix

\section{Gauge reduction on late-time slabs}\label{app:late-slab-gauge}

The goal of this appendix is to show that, on every fixed tail $r\ge r_1>r_0$, a sufficiently
late finite time slab can be placed in the good gauge
\[
b\equiv 0,
\qquad
\beta_r\equiv 0
\]
by a time-preserving change of coordinates. The relevant small quantity is the geometric normal
component of the shift,
\[
\beta^\perp:=\beta({\bf n}),
\]
where ${\bf n}$ is the outward unit normal to the coordinate spheres in the spatial slice. The
argument has three steps. We first remove $b$ by a slice-wise angular reparametrization which is
asymptotic to the identity at infinity. We then transport the radial coordinate along the shift,
which makes the transported radial foliation have zero normal shift. Finally, we compare the
forcing tails before and after this change of foliation.

We now state the result used in the main body of the paper.

\begin{prop}[Late-time slab gauge reduction]\label{prop:app-late-slab-gauge}
Let
\[
M=(r_0,\infty)\times S^2,
\qquad
\mathcal M=(\underline T,\infty)\times M,
\]
and let $(\mathcal M,\gtime)$ be written in a general ADM chart
\[
\gtime
=
-\big(N^2-|\beta|_{g(t)}^2\big)\,dt^2
+2\,\beta\odot dt
+g(t),
\qquad
g(t)=\big(\lambda^2+|b|_\gamma^2\big)\,dr^2+2\,b\odot dr+\gamma.
\]
Assume that $\mathcal S=(N,\lambda,\beta,b,\gamma)\in {\ctS}^{\sharp}_{-\tau}(\mathcal M)$ and that
$\mathcal S$ is late-time gauge reducible, i.e. for every $r_*>r_0$,
\[
\mathfrak B^\sharp(T;r_*)\to0
\qquad\text{as }T\to\infty.
\]
Fix $r_1>r_0$ and set
\[
r_-:=\frac{r_0+r_1}{2}.
\]
Then there exist constants
\[
h\in(0,1],
\qquad
C>0,
\]
depending only on
\[
r_1,\ 
\delta_*(r_-)^{-1},\
\vartheta_*(r_-)^{-1},\
c_*(r_-)^{-1},\
\mathcal N_\tau^\sharp(r_-),
\]
and a time
\[
T_{\mathrm{gauge}}>\underline T
\]
depending additionally on the tail function $\mathfrak B^\sharp(\cdot;r_-)$, with the following
property.

For every central time $T$ satisfying
\[
T-h\ge T_{\mathrm{gauge}},
\]
there is a time-preserving coordinate system
\[
(t,\tilde r,q)
\]
on the slab
\[
I_h(T)\times M_{r_1,\infty},
\qquad
I_h(T):=[T-h,T+h],
\]
such that, if
\[
\widetilde{\mathcal S}
=
(\tilde N,\tilde\lambda,\tilde\beta,\tilde b,\tilde\gamma)
\]
denotes the transformed coefficient tuple, then
\[
\|\widetilde{\mathcal S}\|_{\ctS_{-\tau}(I_h(T)\times M_{r_1,\infty})}\le C,
\]
where this denotes the slab version of the norm in Definition~\ref{def:ctS-nonstat}, and the metric
is in good gauge:
\begin{equation}\label{eq:app-good-gauge-conclusion}
\tilde b\equiv 0,
\qquad
\tilde\beta_{\tilde r}\equiv 0.
\end{equation}

Let
\[
\widetilde G_I,\quad \widetilde\Phi_{1,I},\quad \widetilde\Phi_{2,I},\quad
\widetilde\Omega_I,\quad \widetilde\Psi_I,
\qquad I=I_h(T),
\]
denote the slab quantities in the new chart on the new tail $\tilde r\ge r_1$, computed using
suprema over $t\in I$ and the forcing of the new radial foliation. Let
\[
G_I,\quad \Phi_{1,I},\quad \Phi_{2,I},\quad \Omega_I,\quad \Psi_I
\]
denote the corresponding slab quantities in the original chart, computed on the old tail
$r\ge r_-$ from the geometric forcing of the original $r$-foliation. Then
\begin{align}
\widetilde G_I
&\le
C\Big(
G_I+h\,\mathfrak B^\sharp(T-h;r_-)
\Big),
\label{eq:app-new-G-bound}
\\
\widetilde\Phi_{1,I}
&\le
C\Big(
\Phi_{1,I}+h\,\mathfrak B^\sharp(T-h;r_-)
\Big),
\label{eq:app-new-Phi1-bound}
\\
\widetilde\Phi_{2,I}
&\le
C\Big(
\Phi_{2,I}+h\,\mathfrak B^\sharp(T-h;r_-)
\Big),
\label{eq:app-new-Phi2-bound}
\\
\widetilde\Omega_I
&\le
C\Big(
\Omega_I+h\,\mathfrak B^\sharp(T-h;r_-)
\Big),
\label{eq:app-new-Omega-bound}
\\
\widetilde\Psi_I
&\le
C\Big(
\Psi_I+h\,\mathfrak B^\sharp(T-h;r_-)
\Big).
\label{eq:app-new-Psi-bound}
\end{align}

Moreover, if
\(r=R(t,\tilde r, q)\) denotes the original radial coordinate expressed in the new coordinates, then
\begin{equation}\label{eq:app-r1-C2-close}
\sup_{t\in I_h(T)}
\|R(t,r_1,\cdot)-r_1\|_{C^2(S^2)}
\le
C h\,\mathfrak B^\sharp(T-h;r_-).
\end{equation}
In addition, the time derivatives of \(R\) needed for composition
with \(C^2\)-small time graphs are uniformly bounded:
\begin{equation}\label{eq:app-r1-time-derivative-bounds}
\sup_{t\in I_h(T)}
\left(
\|\partial_t R(t,r_1,\cdot)\|_{C^1(S^2)}
+
\|\partial_t^2 R(t,r_1,\cdot)\|_{C^0(S^2)}
\right)
\le C .
\end{equation}

\end{prop}

The tail quantities appearing above are those defined in
\eqref{eq:X0-slab-def}--\eqref{eq:X2-slab-def} and
\eqref{eq:GI-def}--\eqref{eq:PsiI-def}, with the forcing understood geometrically as
\[
\Xi=\frac{\lambda}{N}\frac{\tr_{S_{t,r}}K_t}{H_{t,r}}.
\]

We now prove Proposition~\ref{prop:app-late-slab-gauge}. The proof is divided into three steps:
orthogonalize the original radial foliation, transport the radial coordinate to remove the normal
shift, and compare the forcing tails.

\subsection*{Step 1: Orthogonalizing the original radial foliation}

We first remove the mixed $dr\odot dp$ term in the spatial metric by changing only the angular
variables. This does not change the leaves $S_{t,r}$, so all geometric quantities attached to the
$r$-foliation are preserved. Throughout this step, the spacetime and the fixed radii $r_-$ and
$r_1$ are those of Proposition~\ref{prop:app-late-slab-gauge}.

\begin{lem}[Slice-wise orthogonalization]\label{lem:app-kill-b}
On $(\underline T,\infty)\times M_{r_-,\infty}$ there is a time-preserving angular change of
variables
\[
(t,r,q)\mapsto (t,r,p(t,r,q)),
\qquad
p(t,r,q)\to q\quad\text{as }r\to\infty,
\]
such that, in the new coordinates $(t,r,q)$, the spatial metric takes the form
\[
g(t)=\lambda^2\,dr^2+\bar\gamma(t,r).
\]
The leaves $S_{t,r}$ are unchanged as subsets of spacetime. Consequently,
\[
\lambda,\qquad H_{t,r},\qquad \tr_{S_{t,r}}K_t,\qquad \beta^\perp,\qquad \Xi
\]
are unchanged up to angular pull-back.

Moreover, writing the transformed coefficients as
\[
\bar{\mathcal S}
=
(\bar N,\bar\lambda,\bar\beta,\bar b,\bar\gamma),
\]
one has $\bar b\equiv0$ and, uniformly for $t>\underline T$,
\[
\partial_t^k(\bar N-1)(t),\quad
\partial_t^k(\bar\lambda-1)(t),\quad
\partial_t^k\big(r^{-2}\bar\gamma(t)-\gamma_{S^2}\big)
\in
{\ct}^{\sharp}_{-\tau-k}(M_{r_-,\infty}),
\qquad
0\le k\le3,
\]
while
\[
\partial_t^k\bar\beta(t)
\in
{\ct}^{\sharp}_{-\tau-k}(M_{r_-,\infty}),
\qquad
0\le k\le2.
\]
The corresponding norms are bounded by a constant depending only on $r_1$ and
$\mathcal N_\tau^\sharp(r_-)$.

\end{lem}

\begin{proof}
For each fixed $t$, define the $r$-dependent vector field on $S^2$
\[
V_t(r,\cdot):=-b^{\sharp_\gamma}(t,r,\cdot).
\]
In any fixed finite atlas on $S^2$, the components of $V_t$ satisfy
\[
\slashed D^j V_t=O(r^{-2-\tau})
\qquad 0\le j\le4,
\]
and, for the radial and time derivatives needed below,
\[
\partial_r^\ell\slashed D^j V_t
=
O(r^{-2-\tau-\ell})
\qquad
1\le \ell\le3,\quad 0\le j\le 3-\ell,
\]
while
\[
\partial_t^k\slashed D^j V_t
=
O(r^{-2-\tau-k})
\qquad
0\le k\le3,\quad 0\le j\le4,
\]
with the corresponding H\"older bounds in the norms of
Definition~\ref{def:ctsharp-main}. Indeed, $b$ is a tangential one-form in the strengthened tail
class, and
\[
\gamma^{-1}=r^{-2}\gamma_{S^2}^{-1}+O(r^{-2-\tau})
\]
as a $(2,0)$-tensor in fixed angular coordinates. Raising the index in
$V_t=-b^{\sharp_\gamma}$ therefore gives the additional factor $r^{-2}$.

We now define $p(t,r,q)$ by integrating the non-autonomous ODE on $S^2$
\begin{equation}\label{eq:app-theta-ode}
\frac{d}{dr}p(t,r,q)=V_t(r,p(t,r,q)),
\qquad
\lim_{r\to\infty}p(t,r,q)=q.
\end{equation}
Here $r\mapsto p(t,r,q)$ is a curve on $S^2$, and the right-hand side is a tangent vector at
$p(t,r,q)$. The condition at infinity is understood by solving, for $R>r_-$,
\[
\frac{d}{dr}p_R(t,r,q)=V_t(r,p_R(t,r,q)),
\qquad
p_R(t,R,q)=q,
\]
on $[r_-,R]$, and then letting $R\to\infty$. Since $V_t$ is integrable in $r$ with the estimates
above, the maps $p_R$ converge in the required $C^4$ angular norms on every fixed tail. The
standard existence, uniqueness, and smooth dependence theorem for time-dependent vector fields on
compact manifolds applies. The limiting map
$q\mapsto p(t,r,q)$ is a diffeomorphism of $S^2$ for each $t$ and $r$, with inverse obtained by the
corresponding limiting inverse flows.

The estimates above also give
\[
p(t,r,q)-q=O(r^{-1-\tau}),
\qquad
\slashed D^j(p-q)=O(r^{-1-\tau})
\qquad 0\le j\le4,
\]
and
\[
\partial_r^\ell\slashed D^j(p-q)
=
O(r^{-1-\tau-\ell})
\qquad
1\le \ell\le3,\quad 0\le j\le 3-\ell.
\]
Similarly,
\[
\partial_t^k\slashed D^j(p-q)
=
O(r^{-1-\tau-k})
\qquad
1\le k\le3,\quad 0\le j\le4.
\]
These estimates are the precise sense in which the angular change is asymptotic to the identity.

In the new angular coordinates $(t,r,q)$, we have
\[
\partial_r\big|_q
=
\partial_r\big|_p+(\partial_r p)^A\partial_{p^A}
=
\partial_r\big|_p-b^{\sharp_\gamma}.
\]
Thus, for every vector $Y$ tangent to $S_{t,r}$,
\[
g(t)(\partial_r|_q,Y)
=
g(t)(\partial_r-b^{\sharp_\gamma},Y)
=
b(Y)-\gamma(b^{\sharp_\gamma},Y)
=
0.
\]
Also,
\[
g(t)(\partial_r|_q,\partial_r|_q)
=
\lambda^2.
\]
Hence, in the coordinates $(t,r,q)$,
\[
g(t)=\lambda^2\,dr^2+\bar\gamma(t,r),
\]
so the transformed mixed coefficient is $\bar b\equiv0$.

Because the leaves $S_{t,r}$ are unchanged, the geometric radial lapse, the leaf mean curvature,
the tangential trace of $K_t$, the normal shift $\beta^\perp$, and the forcing $\Xi$ are unchanged
up to angular pull-back.

It remains to record the regularity of the transformed coefficients. Let
\[
\Theta(t,r,q):=(t,r,p(t,r,q)),
\]
and write the transformed coefficient tuple as
\[
\bar{\mathcal S}
=
(\bar N,\bar\lambda,\bar\beta,\bar b,\bar\gamma).
\]
Thus
\[
\bar N=N\circ\Theta,
\qquad
\bar\lambda=\lambda\circ\Theta,
\qquad
\bar b\equiv0,
\qquad
\bar\gamma=\Theta^*\gamma\big|_{TS^2}.
\]
The transformed coefficients satisfy
\begin{equation}\label{eq:app-step1-regularity-spatial}
\partial_t^k(\bar N-1),\quad
\partial_t^k(\bar\lambda-1),\quad
\partial_t^k(r^{-2}\bar\gamma-\gamma_{S^2})
\in
{\ct}^{\sharp}_{-\tau-k}(M_{r_-,\infty}),
\qquad
0\le k\le3,
\end{equation}
and
\begin{equation}\label{eq:app-step1-regularity-shift}
\partial_t^k\bar\beta
\in
{\ct}^{\sharp}_{-\tau-k}(M_{r_-,\infty}),
\qquad
0\le k\le2.
\end{equation}
Equivalently, in fixed angular coordinates, if
\[
\bar A\in
\{\bar N-1,\ \bar\lambda-1,\ r^{-2}\bar\gamma-\gamma_{S^2}\},
\]
then for $0\le k\le3$ one has
\[
\slashed D^j\partial_t^k\bar A=O(r^{-\tau-k})
\qquad 0\le j\le4,
\]
and
\[
\partial_r^\ell\slashed D^j\partial_t^k\bar A
=
O(r^{-\tau-k-\ell})
\qquad
1\le \ell\le3,\quad 0\le j\le3-\ell,
\]
with the corresponding H\"older bounds. For the transformed shift, the same estimates hold for
$\partial_t^k\bar\beta$ for $0\le k\le2$.

The estimates for $\bar N$, $\bar\lambda$, and $\bar\gamma$ follow directly from composition with
the asymptotically identity angular diffeomorphism $\Theta$ and the estimates for $p-q$. The only
point requiring comment is the shift, because the coordinate change depends on $t$. If $Y$ is
tangent to the new coordinate spheres, then
\[
\bar\beta(Y)
=
\beta(d\Theta(Y))
+
g(t)(\partial_t p,Y).
\]
The first term is controlled by pull-back under $\Theta$. For the second term, the estimates for
the angular flow give
\[
\partial_t p=O(r^{-2-\tau})
\]
as an angular vector field, and since $\gamma=O(r^2)$ as an angular metric,
\[
g(t)(\partial_t p,\cdot)=O(r^{-\tau})
\]
as a tangential $1$-form, with the same angular and radial derivative bounds. Time differentiating
this term $k$ times uses $\partial_t^{k+1}p$. Thus controlling $\partial_t^2\bar\beta$ requires
control of $\partial_t^3p$, which is exactly supplied by the strengthened coefficient assumptions.

For the radial component of the transformed shift, note that
\[
\bar\beta(\partial_r|_q)
=
\beta(\partial_r-b^{\sharp_\gamma})
=
\lambda\,\beta^\perp,
\]
pulled back by $\Theta$. Hence the radial component has the same strengthened control as
$\lambda\beta^\perp$.

Consequently $\bar{\mathcal S}$ satisfies the standard $\ctS_{-\tau}$ bounds on
$(\underline T,\infty)\times M_{r_-,\infty}$. In addition, it has the strengthened spatial
regularity recorded in \eqref{eq:app-step1-regularity-spatial}--\eqref{eq:app-step1-regularity-shift},
which is the regularity package used in Steps 2 and 3 below.

\end{proof}

After Step 1, we work in the orthogonalized chart and write
\[
g(t)=\lambda^2\,dr^2+\gamma,
\qquad
\beta=\beta_r\,dr+\beta^T,
\qquad
\beta^\perp=\lambda^{-1}\beta_r.
\]
The transformed coefficient tuple satisfies the standard fixed-tail $\ctS_{-\tau}$ bounds on
$(\underline T,\infty)\times M_{r_-,\infty}$, and its coefficients together with their first two
time derivatives have the strengthened spatial regularity needed in Steps 2 and 3.

\subsection*{Step 2: Transporting the radial coordinate}

We now define a new radial coordinate whose level sets are transported by the shift. This kills the
normal component of the shift for the transported foliation.

\begin{lem}[Transport of the radial coordinate]\label{lem:app-rho-transport}
There exist constants
\[
h\in(0,1],
\qquad
C>0,
\qquad
T_{\mathrm{gauge}}>\underline T,
\]
where $h$ and $C$ depend only on
\[
r_1,\ 
\delta_*(r_-)^{-1},\
\mathcal N_\tau^\sharp(r_-),
\]
and $T_{\mathrm{gauge}}$ may also depend on the tail function $\mathfrak B^\sharp(\cdot;r_-)$,
such that for every $T$ with
\[
T-h\ge T_{\mathrm{gauge}},
\]
there is a time-preserving coordinate system
\[
(t,\rho,q)
\]
on $I_h(T)\times M_{r_1,\infty}$ with the following properties.

\begin{enumerate}[label=\textup{(\roman*)}]
\item The shift is tangent to the leaves $\rho=\mathrm{const}$; equivalently, the transported
radial foliation has zero normal shift.

\item If the inverse radial change is written as
\[
r=R(t,\rho,q)=\rho+u(t,\rho,q),
\]
then
\begin{equation}\label{eq:app-u-estimate-sharp}
\sup_{t\in I_h(T)}
\|u(t,\cdot)\|_{{\ct}^{\sharp}_{-\tau}(M_{r_1,\infty})}
\le
C h\,\mathfrak B^\sharp(T-h;r_-).
\end{equation}

\item Let
\[
\mathcal S^\rho
=
(N^\rho,\lambda^\rho,\beta^\rho,b^\rho,\gamma^\rho)
\]
denote the coefficient tuple in the transported coordinates $(t,\rho,q)$. Then, on
$I_h(T)\times M_{r_1,\infty}$,
\[
\partial_t^k(N^\rho-1),\quad
\partial_t^k(\lambda^\rho-1),\quad
\partial_t^k b^\rho,\quad
\partial_t^k\big(\rho^{-2}\gamma^\rho-\gamma_{S^2}\big)
\in
{\ct}^{\sharp}_{-\tau-k}(M_{r_1,\infty}),
\qquad
0\le k\le3,
\]
and
\[
\partial_t^k\beta^\rho
\in
{\ct}^{\sharp}_{-\tau-k}(M_{r_1,\infty}),
\qquad
0\le k\le2.
\]
The corresponding norms are bounded by $C$.

\item The original radial coordinate of the inner cylinder satisfies
\begin{equation}\label{eq:app-inner-cylinder-close-sharp}
\sup_{t\in I_h(T)}
\|R(t,r_1,\cdot)-r_1\|_{C^2(S^2)}
\le
C h\,\mathfrak B^\sharp(T-h;r_-).
\end{equation}

\item The time derivatives of the original radial coordinate along the inner cylinder satisfy
\begin{equation}\label{eq:app-inner-cylinder-time-derivative-bounds}
\sup_{t\in I_h(T)}
\left(
\|\partial_t R(t,r_1,\cdot)\|_{C^1(S^2)}
+
\|\partial_t^2 R(t,r_1,\cdot)\|_{C^0(S^2)}
\right)
\le C.
\end{equation}
\end{enumerate}
\end{lem}

\begin{proof}
We construct $\rho$ by solving
\begin{equation}\label{eq:app-rho-transport-main}
\partial_t\rho-\beta^{\sharp_{g(t)}}(\rho)=0,
\qquad
\rho(T,r,q)=r.
\end{equation}
Since we are working after Step 1, the spatial metric has the form
\[
g(t)=\lambda^2\,dr^2+\gamma,
\]
and therefore
\[
\beta^{\sharp_{g(t)}}
=
\frac{\beta_r}{\lambda^2}\partial_r+(\beta^T)^{\sharp_\gamma}
=
\frac{\beta^\perp}{\lambda}\partial_r+(\beta^T)^{\sharp_\gamma}.
\]
Thus the radial velocity of the characteristics is $\beta^\perp/\lambda$. On the fixed tail
$r\ge r_-$, this is controlled by $\mathfrak B^\sharp(T-h;r_-)$ and therefore becomes uniformly
small on late slabs. After increasing $T_{\mathrm{gauge}}$ if necessary, every characteristic
starting in $M_{r_1,\infty}$ remains in $M_{r_-,\infty}$ for $|t-T|\le h$. Standard characteristic
theory gives a smooth solution of \eqref{eq:app-rho-transport-main} on the desired slab.

Set
\[
w:=\rho-r.
\]
Then $w$ satisfies
\begin{equation}\label{eq:app-w-transport}
\partial_t w-\beta^{\sharp_{g(t)}}(w)
=
\frac{\beta_r}{\lambda^2}
=
\frac{\beta^\perp}{\lambda},
\qquad
w(T,\cdot)=0.
\end{equation}
The characteristic formula and the lower bound $\lambda\ge \delta_*(r_-)$ give
\[
\sup_{t\in I_h(T)}
\|w(t,\cdot)\|_{\mathc^0_{-\tau}(M_{r_1,\infty})}
\le
C h\,\mathfrak B^\sharp(T-h;r_-).
\]

The same estimate holds in the full ${\ct}^{\sharp}_{-\tau}$ norm. Indeed, applying
$\slashed D^j$, $0\le j\le4$, and applying $\partial_r^\ell\slashed D^j$ with
$1\le\ell\le3$ and $0\le j\le3-\ell$, to \eqref{eq:app-w-transport} gives linear transport
equations for the differentiated quantities. Their source terms contain the corresponding
derivatives of $\beta^\perp/\lambda$ and lower-order derivatives of $w$, with coefficients
controlled by the fixed-tail regularity from Step 1. Since all differentiated initial data vanish
at $t=T$, Gronwall's inequality on the fixed interval $I_h(T)$ yields
\[
\sup_{t\in I_h(T)}
\|w(t,\cdot)\|_{{\ct}^{\sharp}_{-\tau}(M_{r_1,\infty})}
\le
C h\,\mathfrak B^\sharp(T-h;r_-).
\]

After increasing $T_{\mathrm{gauge}}$ once more, we may assume that the right-hand side is small
enough to ensure
\[
\partial_r\rho=1+\partial_r w\ge\frac12.
\]
Thus $(r,q)\mapsto(\rho,q)$ is a diffeomorphism on each time slice over the tail. Let
\[
r=R(t,\rho,q)=\rho+u(t,\rho,q)
\]
be its inverse. The inverse function theorem, applied in the same weighted Hölder norms, gives
\[
\sup_{t\in I_h(T)}
\|u(t,\cdot)\|_{{\ct}^{\sharp}_{-\tau}(M_{r_1,\infty})}
\le
C
\sup_{t\in I_h(T)}
\|w(t,\cdot)\|_{{\ct}^{\sharp}_{-\tau}(M_{r_1,\infty})},
\]
which proves \eqref{eq:app-u-estimate-sharp}. Evaluating this bound at $\rho=r_1$ gives
\eqref{eq:app-inner-cylinder-close-sharp}.

We next prove the time-derivative bound
\eqref{eq:app-inner-cylinder-time-derivative-bounds}. Differentiating the identity
\[
\rho(t,R(t,\rho,q),q)=\rho
\]
at fixed $(\rho,q)$ gives
\[
R_t=-\frac{\rho_t}{\rho_r}.
\]
Using the transport equation $\rho_t=\beta^{\sharp_g}(\rho)$, this becomes
\[
R_t=-\frac{\beta^{\sharp_g}(\rho)}{\rho_r}.
\]
The denominator is uniformly bounded away from zero, since
$\rho_r=1+w_r\ge 1/2$. The numerator and its first angular derivatives are uniformly bounded on
the fixed tail by the coefficient bounds from Step 1 and the transport estimates for $\rho$.
Therefore
\[
\sup_{t\in I_h(T)}
\|\partial_tR(t,r_1,\cdot)\|_{C^1(S^2)}
\le C.
\]
Differentiating the identity $\rho(t,R(t,\rho,q),q)=\rho$ twice in $t$ gives
\[
R_{tt}
=
-\frac{\rho_{tt}+2\rho_{tr}R_t+\rho_{rr}R_t^2}{\rho_r}.
\]
The differentiated transport equations for $\rho$, together with the fixed-tail coefficient
bounds, give uniform control of the numerator on $\rho=r_1$. Since $\rho_r\ge1/2$, we obtain
\[
\sup_{t\in I_h(T)}
\|\partial_t^2R(t,r_1,\cdot)\|_{C^0(S^2)}
\le C.
\]
This proves \eqref{eq:app-inner-cylinder-time-derivative-bounds}.

We now record the regularity of the transported coefficients. The coordinate change is
\[
(t,\rho,q)\mapsto (t,R(t,\rho,q),q),
\qquad
R=\rho+u.
\]
The estimate for $u$ gives
\[
R-\rho=O(r^{-\tau}),
\qquad
R_\rho-1=O(r^{-1-\tau}),
\qquad
\slashed D R=O(r^{-\tau}),
\]
with the corresponding strengthened spatial bounds. Time derivatives of $R$ are obtained by
differentiating the identity
\[
\rho(t,R(t,\rho,q),q)=\rho.
\]
For example,
\[
R_t=-\frac{\rho_t}{\rho_r}
=
-\frac{\beta^{\sharp_g}(\rho)}{\rho_r}.
\]
Using the transport equation and the coefficient bounds from Step 1, one obtains the corresponding
bounds for $R_t$, $R_{tt}$, and $R_{ttt}$ required below. These bounds are uniform on the slab; only
the zeroth-order displacement estimate \eqref{eq:app-u-estimate-sharp} is required to be small.

The transported spatial metric is the pull-back of
\[
\lambda^2\,dr^2+\gamma
\]
under $r=R(t,\rho,q)$. Hence
\[
g^\rho(t)
=
\lambda^2(t,R,q)\,(dR)^2+\gamma(t,R,q).
\]
Writing it in the usual form
\[
g^\rho(t)
=
\big((\lambda^\rho)^2+|b^\rho|_{\gamma^\rho}^2\big)\,d\rho^2
+
2b^\rho\odot d\rho
+
\gamma^\rho,
\]
one sees that $N^\rho$, $\lambda^\rho$, $b^\rho$, and $\gamma^\rho$ are obtained from the original
coefficients and from
\[
R_\rho,\qquad \slashed D R,
\]
by smooth algebraic operations. The regularity estimates for $R$ therefore imply
\[
\partial_t^k(N^\rho-1),\quad
\partial_t^k(\lambda^\rho-1),\quad
\partial_t^k b^\rho,\quad
\partial_t^k\big(\rho^{-2}\gamma^\rho-\gamma_{S^2}\big)
\in
{\ct}^{\sharp}_{-\tau-k}(M_{r_1,\infty}),
\qquad
0\le k\le3.
\]

The transformed shift contains the usual pull-back contribution together with the velocity of the
time-dependent radial coordinate change. Since the new time vector is
\[
\partial_t\big|_{\rho,q}
=
\partial_t\big|_{r,q}+R_t\,\partial_r,
\]
the new shift is obtained from the old shift by adding the spatial vector $R_t\partial_r$ before
lowering with the transported spatial metric. The estimates for $R_t$ and for the old coefficients
therefore give
\[
\partial_t^k\beta^\rho
\in
{\ct}^{\sharp}_{-\tau-k}(M_{r_1,\infty}),
\qquad
0\le k\le2.
\]
This proves the coefficient regularity asserted in \textup{(iii)}.

It remains to check the zero-normal-shift property. In the coordinates $(t,\rho,q)$, the
transformed shift vector is the spatial vector
\[
\beta^{\sharp_{g(t)}}+X_t,
\qquad
X_t:=R_t\,\partial_r.
\]
Since
\[
0=
\frac{d}{dt}\rho(t,R(t,\rho,q),q)
=
\partial_t\rho+X_t(\rho),
\]
and since the transport equation gives $\partial_t\rho=\beta^{\sharp_{g(t)}}(\rho)$, we obtain
\[
(\beta^{\sharp_{g(t)}}+X_t)(\rho)=0.
\]
Thus the transported shift is tangent to the leaves $\rho=\mathrm{const}$, which proves
\textup{(i)}.
\end{proof}

After Step 2, the metric may no longer be orthogonal with respect to $\rho$. We therefore apply
Lemma~\ref{lem:app-kill-b} once more, now to the transported $\rho$-foliation on the slab
$I_h(T)\times M_{r_1,\infty}$. The regularity statement in Lemma~\ref{lem:app-rho-transport} is
exactly the regularity needed for this second angular orthogonalization.

This second angular change preserves the leaves $\rho=\mathrm{const}$ and hence preserves the
zero-normal-shift property. Denoting the final radial coordinate by
\[
\tilde r:=\rho,
\]
we obtain a time-preserving chart $(t,\tilde r,\tilde q)$. If
\[
\widetilde{\mathcal S}
=
(\tilde N,\tilde\lambda,\tilde\beta,\tilde b,\tilde\gamma)
\]
denotes the final transformed coefficient tuple, then
\[
\tilde b\equiv0,
\qquad
\tilde\beta_{\tilde r}\equiv0.
\]
Moreover,
\[
\partial_t^k(\tilde N-1),\quad
\partial_t^k(\tilde\lambda-1),\quad
\partial_t^k\big(\tilde r^{-2}\tilde\gamma-\gamma_{S^2}\big)
\in
{\ct}^{\sharp}_{-\tau-k}(M_{r_1,\infty}),
\qquad
0\le k\le3,
\]
and
\[
\partial_t^k\tilde\beta
\in
{\ct}^{\sharp}_{-\tau-k}(M_{r_1,\infty}),
\qquad
0\le k\le2.
\]
The corresponding norms are bounded by a constant depending only on the fixed-tail data. In
particular, $\widetilde{\mathcal S}$ satisfies the slab version of the standard
$\ctS_{-\tau}$ bounds required in Proposition~\ref{prop:app-late-slab-gauge}. This proves the
good-gauge conclusion \eqref{eq:app-good-gauge-conclusion}.

We also record that the inner-cylinder estimates pass to the final good-gauge chart. Indeed, the
last angular orthogonalization preserves the leaves $\rho=\mathrm{const}$ and only reparametrizes
each such leaf by a uniformly controlled angular diffeomorphism, with uniformly controlled first
two time derivatives. Hence the estimates
\eqref{eq:app-inner-cylinder-close-sharp} and
\eqref{eq:app-inner-cylinder-time-derivative-bounds}, after composing with this angular
diffeomorphism and changing the constant \(C\), give the estimates
\eqref{eq:app-r1-C2-close} and \eqref{eq:app-r1-time-derivative-bounds} in the final
coordinates $(t,\tilde r,\tilde q)$.

\subsection*{Step 3: Comparing the forcing tails}

It remains to compare the forcing tails of the transported foliation with those of the original
foliation.

\begin{lem}[Forcing comparison for the transported foliation]\label{lem:app-Xi-stability}
Let $I=I_h(T)$ and let $(t,\rho,q)$ be the transported coordinates from Step 2, with inverse radial
representation
\[
r=\rho+u(t,\rho,q).
\]
Let
\[
\widehat G_I,\quad \widehat\Phi_{1,I},\quad \widehat\Phi_{2,I},\quad
\widehat\Omega_I,\quad \widehat\Psi_I
\]
denote the slab quantities for the transported $\rho$-foliation on $\rho\ge r_1$, computed using
suprema over $t\in I$. Then
\begin{align}
\widehat G_I
&\le
C\Big(
G_I+h\,\mathfrak B^\sharp(T-h;r_-)
\Big),
\label{eq:app-G-comparison}
\\
\widehat\Phi_{1,I}
&\le
C\Big(
\Phi_{1,I}+h\,\mathfrak B^\sharp(T-h;r_-)
\Big),
\label{eq:app-Phi1-comparison}
\\
\widehat\Phi_{2,I}
&\le
C\Big(
\Phi_{2,I}+h\,\mathfrak B^\sharp(T-h;r_-)
\Big),
\label{eq:app-Phi2-comparison}
\\
\widehat\Omega_I
&\le
C\Big(
\Omega_I+h\,\mathfrak B^\sharp(T-h;r_-)
\Big),
\label{eq:app-Omega-comparison}
\\
\widehat\Psi_I
&\le
C\Big(
\Psi_I+h\,\mathfrak B^\sharp(T-h;r_-)
\Big).
\label{eq:app-Psi-comparison}
\end{align}
The same estimates hold for the final good-gauge chart obtained after the last angular
orthogonalization. Namely, if
\[
\widetilde G_I,\quad \widetilde\Phi_{1,I},\quad \widetilde\Phi_{2,I},\quad
\widetilde\Omega_I,\quad \widetilde\Psi_I
\]
denote the corresponding slab quantities in the final good-gauge chart, then
\begin{align*}
\widetilde G_I
&\le
C\Big(
G_I+h\,\mathfrak B^\sharp(T-h;r_-)
\Big),
\\
\widetilde\Phi_{1,I}
&\le
C\Big(
\Phi_{1,I}+h\,\mathfrak B^\sharp(T-h;r_-)
\Big),
\\
\widetilde\Phi_{2,I}
&\le
C\Big(
\Phi_{2,I}+h\,\mathfrak B^\sharp(T-h;r_-)
\Big),
\\
\widetilde\Omega_I
&\le
C\Big(
\Omega_I+h\,\mathfrak B^\sharp(T-h;r_-)
\Big),
\\
\widetilde\Psi_I
&\le
C\Big(
\Psi_I+h\,\mathfrak B^\sharp(T-h;r_-)
\Big).
\end{align*}
\end{lem}

\begin{proof}
The transported foliation is the graph foliation
\[
\Sigma_{t,\rho}
=
\{(t,r,q):r=\rho+u(t,\rho,q)\}.
\]
By \eqref{eq:app-u-estimate-sharp},
\[
\sup_{t\in I}
\|u(t,\cdot)\|_{{\ct}^{\sharp}_{-\tau}(M_{r_1,\infty})}
\le
C h\,\mathfrak B^\sharp(T-h;r_-).
\]
After increasing $T_{\mathrm{gauge}}$ if necessary, we may assume that $C h\,\mathfrak B^\sharp(T-h;r_-)$
is below a fixed smallness threshold. Then the image of $\{\rho\ge r_1\}$ lies in the old tail
$\{r\ge r_-\}$, and $r=\rho+u$ is a uniformly controlled perturbation of the identity.

The sharp norm gives the scale-invariant estimates
\[
|u|\le C h\,\mathfrak B^\sharp(T-h;r_-)\,\rho^{-\tau},
\qquad
|\slashed d u|_{\gamma}
\le
C h\,\mathfrak B^\sharp(T-h;r_-)\,\rho^{-1-\tau},
\]
\[
|\slashed\nabla^2u|_{\gamma}
\le
C h\,\mathfrak B^\sharp(T-h;r_-)\,\rho^{-2-\tau},
\]
and the corresponding estimates for the tangential derivatives needed below, including the radial
derivatives
\[
|\partial_\rho\slashed\nabla^j u|_{\gamma}
\le
C h\,\mathfrak B^\sharp(T-h;r_-)\,\rho^{-1-\tau-j},
\qquad
0\le j\le2.
\]
Here all norms are taken with respect to the old coordinate-sphere metrics, and replacing $r$ by
$\rho$ changes the bounds only by a fixed multiplicative constant.

Let $\Xi_\rho$ be the forcing of the transported $\rho$-foliation, and define the old forcing
pulled back to the transported coordinates by
\[
\Xi_{\mathrm{old}}^\rho(t,\rho,q)
:=
\Xi(t,\rho+u(t,\rho,q),q).
\]
The geometric quantities entering the forcing,
\[
\lambda[\Sigma_{t,\rho}],
\qquad
H[\Sigma_{t,\rho}],
\qquad
\tr_{\Sigma_{t,\rho}}K_t,
\]
are smooth expressions in the old coefficients and in
\[
u,\quad
\slashed D u,\quad
\slashed D^2u,\quad
\partial_\rho u,\quad
\partial_\rho\slashed D u,\quad
\partial_\rho\slashed D^2u.
\]
Taking two tangential derivatives of these expressions uses at most four tangential derivatives of
$u$ and at most one radial derivative of two tangential derivatives of $u$, precisely the
regularity included in ${\ct}^{\sharp}_{-\tau}$. Since
\[
N\ge \vartheta_*(r_-),
\qquad
rH_{t,r}\ge c_*(r_-)
\]
on the old tail, and since $u$ is small, we also have
\[
N[\Sigma_{t,\rho}]\ge \frac12\vartheta_*(r_-),
\qquad
\rho\,H[\Sigma_{t,\rho}]\ge \frac12c_*(r_-).
\]
Thus all denominators in the forcing are uniformly controlled.

Applying the fundamental theorem of calculus to the one-parameter family of graph foliations
$r=\rho+s u$, $0\le s\le1$, and using the standard product and composition estimates in H\"older
spaces, gives
\begin{equation}\label{eq:app-Xi-difference-pointwise}
\begin{split}
&|\Xi_\rho-\Xi_{\mathrm{old}}^\rho|
+\rho\,|\slashed d(\Xi_\rho-\Xi_{\mathrm{old}}^\rho)|
+\rho^2\,|\slashed\nabla^2(\Xi_\rho-\Xi_{\mathrm{old}}^\rho)|
\\
&\qquad\qquad
\le
C\,\rho^{-1-\tau}\,
h\,\mathfrak B^\sharp(T-h;r_-).
\end{split}
\end{equation}
The same argument at the H\"older level gives
\begin{equation}\label{eq:app-Xi-difference-holder}
\|\rho(\Xi_\rho-\Xi_{\mathrm{old}}^\rho)\|_{\cz_{-\tau}(M_{r_1,\infty})}
\le
C h\,\mathfrak B^\sharp(T-h;r_-).
\end{equation}
For the product and composition estimates used here, one may use the standard tame estimates for
H\"older spaces; see, for example, \cite{TaylorPDEIII,parabolic2}.

We now compare the pulled-back old forcing with the original old forcing. Because
$r=\rho+u$ is a small ${\ct}^{\sharp}_{-\tau}$ perturbation of the identity, we have
\[
\rho\sim r,
\qquad
d\rho\sim dr,
\]
on the relevant tails, with constants depending only on the fixed-tail data. Moreover, when
differentiating
\[
\Xi_{\mathrm{old}}^\rho(t,\rho,q)=\Xi(t,R(t,\rho,q),q),
\]
the terms involving radial derivatives of $\Xi$ are multiplied by $\slashed D R=\slashed D u$ or
$\slashed\nabla^2R=\slashed\nabla^2u$. The radial derivatives of $\Xi$ are controlled by the fixed
coefficient bounds, while the factors involving $u$ are controlled by
$h\,\mathfrak B^\sharp(T-h;r_-)$ in the scale-invariant norms. Therefore
\begin{align}
\mathfrak X_0[\Xi_{\mathrm{old}}^\rho;I](\rho)
&\le
C\,\mathfrak X_0[\Xi;I](R(t,\rho,\cdot))
+
C\rho^{-\tau}h\,\mathfrak B^\sharp(T-h;r_-),
\label{eq:app-X0-pullback}
\\
\mathfrak X_1[\Xi_{\mathrm{old}}^\rho;I](\rho)
&\le
C\,\mathfrak X_1[\Xi;I](R(t,\rho,\cdot))
+
C\rho^{-1-\tau}h\,\mathfrak B^\sharp(T-h;r_-),
\label{eq:app-X1-pullback}
\\
\mathfrak X_2[\Xi_{\mathrm{old}}^\rho;I](\rho)
&\le
C\,\mathfrak X_2[\Xi;I](R(t,\rho,\cdot))
+
C\rho^{-1-\tau}h\,\mathfrak B^\sharp(T-h;r_-).
\label{eq:app-X2-pullback}
\end{align}
After changing variables from $\rho$ to $r=R(t,\rho,q)$ and using that the image of
$\rho\ge r_1$ is contained in $r\ge r_-$, these inequalities give
\begin{align}
G[\Xi_{\mathrm{old}}^\rho;I]
&\le
C\Big(G_I+h\,\mathfrak B^\sharp(T-h;r_-)\Big),
\label{eq:app-G-old-pullback}
\\
\Phi_1[\Xi_{\mathrm{old}}^\rho;I]
&\le
C\Big(\Phi_{1,I}+h\,\mathfrak B^\sharp(T-h;r_-)\Big),
\label{eq:app-Phi1-old-pullback}
\\
\Phi_2[\Xi_{\mathrm{old}}^\rho;I]
&\le
C\Big(\Phi_{2,I}+h\,\mathfrak B^\sharp(T-h;r_-)\Big),
\label{eq:app-Phi2-old-pullback}
\\
\Omega[\Xi_{\mathrm{old}}^\rho;I]
&\le
C\Big(\Omega_I+h\,\mathfrak B^\sharp(T-h;r_-)\Big).
\label{eq:app-Omega-old-pullback}
\end{align}
Similarly,
\begin{equation}\label{eq:app-Psi-old-pullback}
\|\rho\,\Xi_{\mathrm{old}}^\rho\|_{\cz_{-\tau}(M_{r_1,\infty})}
\le
C\Big(\Psi_I+h\,\mathfrak B^\sharp(T-h;r_-)\Big).
\end{equation}

Combining \eqref{eq:app-Xi-difference-pointwise} with
\eqref{eq:app-G-old-pullback}--\eqref{eq:app-Omega-old-pullback}, and combining
\eqref{eq:app-Xi-difference-holder} with \eqref{eq:app-Psi-old-pullback}, yields
\eqref{eq:app-G-comparison}--\eqref{eq:app-Psi-comparison}.

Finally, the last angular orthogonalization changes only the angular parametrization of the same
leaves $\rho=\mathrm{const}$. Thus the forcing is unchanged as a geometric function on each leaf
and is merely pulled back by a uniformly controlled angular diffeomorphism. The corresponding slab
norms are therefore equivalent, with constants depending only on the fixed-tail data. Hence the
same estimates hold in the final good-gauge chart.
\end{proof}

Combining Steps 1--3 proves Proposition~\ref{prop:app-late-slab-gauge}. Indeed, Step 1
orthogonalizes the original foliation without changing the geometric forcing. Step 2 transports
the radial coordinate, removes the normal shift, and gives the inner-cylinder estimates
\eqref{eq:app-inner-cylinder-close-sharp} and
\eqref{eq:app-inner-cylinder-time-derivative-bounds}. Reapplying Step 1 to the transported
foliation gives the good gauge \eqref{eq:app-good-gauge-conclusion}, and the preceding paragraph
shows that the inner-cylinder estimates pass to the final good-gauge chart as
\eqref{eq:app-r1-C2-close} and \eqref{eq:app-r1-time-derivative-bounds}. Step 3 gives
\eqref{eq:app-new-G-bound}--\eqref{eq:app-new-Psi-bound}.

\begin{remark}\label{rem:app-what-is-needed}
The normal-shift decay is used only to transport the radial coordinate without moving the fixed
tail too far. The strengthened ${\ct}^{\sharp}_{-\tau}$ regularity is used only to compare the
forcing of two nearby radial foliations.
\end{remark}

\section{A horizon-penetrating Schwarzschild final-state gauge}
\label{app:horizon-penetrating-schwarzschild-model}

In this appendix we record a simple Schwarzschild model for the kind of late-time final-state
coordinates used in Section~\ref{sub:spacetime-final-state}; see also
Remark~\ref{rem:horizon-penetrating-final-state-gauge}. The point is to construct a
horizon-penetrating time function whose level sets foliate the future horizon by MOTS
cross-sections, while agreeing exactly with the standard Schwarzschild time slices on an exterior
region whose inner radius tends to the horizon. Thus, on every fixed tail $r\ge r_1>2m$, the slices
are eventually the usual time-symmetric, hence tangentially maximal, Schwarzschild slices. The
analogous construction in Kerr can be carried out in horizon-regular ingoing Kerr coordinates by
cutting off the Boyer--Lindquist time function near the future horizon; we give the Schwarzschild
model because the formulas are completely transparent.

Let $m>0$, and write the Schwarzschild metric in ingoing Eddington--Finkelstein coordinates
$(v,r,\omega)\in\mathbb R\times(0,\infty)\times S^2$ as
\begin{equation}\label{eq:schw-ef-metric}
\gtime
=
-\left(1-\frac{2m}{r}\right)dv^2
+2\,dv\,dr
+r^2\,d\omega^2.
\end{equation}
Set
\[
A(r):=1-\frac{2m}{r}.
\]
On the exterior region $r>2m$, the standard Schwarzschild time is
\[
t_{\mathrm{Sch}}=v-r_*(r),
\]
where
\[
r_*(r)
:=
r+2m\log\left(\frac{r}{2m}-1\right),
\qquad r>2m.
\]
Thus
\[
r_*'(r)=\frac{1}{1-2m/r}=\frac1{A(r)}.
\]

\begin{prop}[Horizon-penetrating Schwarzschild comparison foliation]
\label{prop:horizon-penetrating-schwarzschild-foliation}
There exists a smooth one-parameter family of spacelike hypersurfaces
\[
\{\Sigma_T\}_{T\in\mathbb R}
\]
in the Schwarzschild spacetime \eqref{eq:schw-ef-metric} with the following properties.

\begin{enumerate}[label=\textup{(\roman*)}]
\item The hypersurfaces $\Sigma_T$ form a smooth foliation of the ingoing
Eddington--Finkelstein region $r>0$. Equivalently, $(T,r,\omega)$ is a smooth
horizon-penetrating coordinate system, where $\Sigma_T=\{T=\mathrm{constant}\}$.

\item Each $\Sigma_T$ intersects the future horizon
\[
\mathcal H^+=\{r=2m\}
\]
in a round sphere
\[
S_T:=\Sigma_T\cap\mathcal H^+,
\]
and the family $\{S_T\}_{T\in\mathbb R}$ foliates $\mathcal H^+$. Each $S_T$ is a MOTS.

\item There is a smooth function
\[
r_{\mathrm{cut}}(T)=2m+2\delta_0e^{-T/(4m)}
\]
for some fixed constant $\delta_0>0$, such that
\[
r_{\mathrm{cut}}(T)\downarrow 2m
\qquad\text{as }T\to+\infty,
\]
and
\[
\Sigma_T\cap\{r\ge r_{\mathrm{cut}}(T)\}
=
\{t_{\mathrm{Sch}}=T\}\cap\{r\ge r_{\mathrm{cut}}(T)\}.
\]
In particular, $\Sigma_T$ is tangentially maximal on $r\ge r_{\mathrm{cut}}(T)$.
\end{enumerate}
\end{prop}

\begin{proof}
Choose a smooth nondecreasing cutoff function
\[
\chi:[0,\infty)\to[0,1]
\]
such that
\[
\chi(x)=0\quad\text{for }0\le x\le1,
\qquad
\chi(x)=1\quad\text{for }x\ge2.
\]
Fix a constant $\delta_0>0$, and define
\[
\delta(T):=\delta_0e^{-T/(4m)}.
\]
For $r>2m$, set
\[
x_T(r):=\frac{r-2m}{\delta(T)}.
\]
We define $F(T,r)$ for $r>2m$ by
\begin{equation}\label{eq:F-schwarzschild-construction}
F(T,r)
:=
T+r+2m\log\left(\frac{\delta(T)}{m}\right)
+
2m\int_2^{x_T(r)}\frac{\chi(y)}{y}\,dy.
\end{equation}
Since $\chi(y)=0$ for $0\le y\le1$, the integrand $\chi(y)/y$ extends smoothly across $y=0$ by
zero. Thus the integral term has a smooth one-sided limit as $r\downarrow2m$. We extend $F$ to
$0<r\le2m$ by
\begin{equation}\label{eq:F-inside-extension}
F(T,r):=F(T,2m)+(r-2m).
\end{equation}
Finally, define
\begin{equation}\label{eq:SigmaT-definition}
\Sigma_T:=\{(v,r,\omega):v=F(T,r)\}.
\end{equation}

We now verify the stated properties.

First, we check agreement with the standard Schwarzschild time slices outside the transition
region. For $r>2m$, differentiating \eqref{eq:F-schwarzschild-construction} in $r$ gives
\begin{equation}\label{eq:Fr-schwarzschild}
F_r(T,r)
=
1+\frac{2m}{r-2m}\chi\left(\frac{r-2m}{\delta(T)}\right).
\end{equation}
If
\[
r\ge 2m+2\delta(T),
\]
then $x_T(r)\ge2$, and hence $\chi(x_T(r))=1$. Therefore
\[
F_r(T,r)
=
1+\frac{2m}{r-2m}
=
\frac{r}{r-2m}
=
\frac1{1-2m/r}
=
r_*'(r).
\]
Thus $F(T,r)-r_*(r)$ is constant in $r$ on $r\ge2m+2\delta(T)$. We compute this constant directly.
For $x_T(r)\ge2$,
\[
\int_2^{x_T(r)}\frac{\chi(y)}{y}\,dy
=
\int_2^{x_T(r)}\frac{dy}{y}
=
\log\left(\frac{x_T(r)}{2}\right).
\]
Hence, for $r\ge2m+2\delta(T)$,
\[
\begin{aligned}
F(T,r)
&=
T+r+2m\log\left(\frac{\delta(T)}{m}\right)
+
2m\log\left(\frac{x_T(r)}{2}\right)
\\
&=
T+r+2m\log\left(\frac{\delta(T)}{m}\cdot\frac{r-2m}{2\delta(T)}\right)
\\
&=
T+r+2m\log\left(\frac{r-2m}{2m}\right)
\\
&=
T+r_*(r).
\end{aligned}
\]
Therefore, if
\[
r_{\mathrm{cut}}(T):=2m+2\delta(T)=2m+2\delta_0e^{-T/(4m)},
\]
then
\begin{equation}\label{eq:Sigma-equals-Schwarzschild-outside}
\Sigma_T\cap\{r\ge r_{\mathrm{cut}}(T)\}
=
\{v=T+r_*(r)\}\cap\{r\ge r_{\mathrm{cut}}(T)\}
=
\{t_{\mathrm{Sch}}=T\}\cap\{r\ge r_{\mathrm{cut}}(T)\}.
\end{equation}
Since $r_{\mathrm{cut}}(T)\downarrow2m$ as $T\to+\infty$, this proves the exterior agreement.

On the exterior region $r\ge r_{\mathrm{cut}}(T)$, the hypersurface $\Sigma_T$ is exactly a standard
Schwarzschild time slice. These standard time slices are time-symmetric, so their second
fundamental form vanishes:
\[
K=0.
\]
Thus $\tr_{S_r}K=0$ for the coordinate spheres in this region. Therefore $\Sigma_T$ is
tangentially maximal on $r\ge r_{\mathrm{cut}}(T)$.

We next verify smoothness and horizon penetration. If
\[
0<r-2m\le \delta(T),
\]
then $x_T(r)\le1$, and hence $\chi(x_T(r))=0$. By \eqref{eq:Fr-schwarzschild},
\[
F_r(T,r)=1
\qquad\text{for }2m<r\le2m+\delta(T).
\]
In this same region, the integral term in \eqref{eq:F-schwarzschild-construction} is independent
of $r$, so $F(T,r)$ is affine in $r$ near $r=2m$ from the exterior side. The extension
\eqref{eq:F-inside-extension} therefore matches smoothly across $r=2m$. Hence each $\Sigma_T$ is a
smooth hypersurface crossing the future horizon.

We now check spacelikeness. The tangent vectors to $\Sigma_T$ are
\[
X_r=\partial_r+F_r\partial_v,
\qquad
X_A=\partial_A,
\]
where $A$ is a spherical index. Using \eqref{eq:schw-ef-metric}, the induced metric on $\Sigma_T$
is
\[
g_{\Sigma_T}
=
\left(2F_r-A(r)F_r^2\right)dr^2
+
r^2d\omega^2.
\]
Thus $\Sigma_T$ is spacelike if and only if
\[
2F_r-A(r)F_r^2>0.
\]
For $r>2m$, \eqref{eq:Fr-schwarzschild} gives
\[
A(r)F_r
=
A(r)+\chi(x_T(r))\frac{2m}{r}
=
A(r)+\chi(x_T(r))(1-A(r)).
\]
Since $0<A(r)<1$ for $r>2m$ and $0\le\chi\le1$, we have
\[
0<A(r)F_r\le1.
\]
Therefore
\[
2F_r-A(r)F_r^2
=
F_r\big(2-A(r)F_r\big)
\ge
F_r>0.
\]
For $0<r\le2m$, the extension gives $F_r=1$, while $A(r)\le0$. Hence
\[
2F_r-A(r)F_r^2
=
2-A(r)>0.
\]
Thus every $\Sigma_T$ is spacelike.

We now show that the hypersurfaces form a foliation. For fixed $r>2m$, differentiating
\eqref{eq:F-schwarzschild-construction} in $T$ gives
\[
\delta'(T)=-\frac{1}{4m}\delta(T),
\qquad
\partial_T x_T(r)=\frac{x_T(r)}{4m}.
\]
Therefore
\[
\begin{aligned}
F_T(T,r)
&=
1+2m\frac{\delta'(T)}{\delta(T)}
+
2m\frac{\chi(x_T(r))}{x_T(r)}\partial_T x_T(r)
\\
&=
1-\frac12+\frac12\chi(x_T(r))
\\
&=
\frac{1+\chi(x_T(r))}{2}.
\end{aligned}
\]
Thus
\begin{equation}\label{eq:FT-positive}
\frac12\le F_T(T,r)\le1
\qquad\text{for }r>2m.
\end{equation}
For $0<r\le2m$, the affine extension gives
\[
F_T(T,r)=F_T(T,2m)=\frac12.
\]
Hence
\[
F_T(T,r)>0
\qquad\text{for all }r>0.
\]
So, for every fixed $r>0$, the map $T\mapsto F(T,r)$ is strictly increasing.

It remains to show that this map has image all of $\mathbb R$. Fix $r>2m$. As $T\to+\infty$, we
have $\delta(T)\to0$, hence $x_T(r)\to+\infty$. Thus, for all sufficiently large $T$,
\[
F(T,r)=T+r_*(r),
\]
and so $F(T,r)\to+\infty$. As $T\to-\infty$, we have $\delta(T)\to+\infty$, hence
$x_T(r)\to0$. Since $\chi=0$ near $0$,
\[
F(T,r)
=
T+r+2m\log\left(\frac{\delta(T)}{m}\right)+O(1)
=
\frac{T}{2}+O(1)
\to-\infty.
\]
For $0<r\le2m$, the affine extension gives the same asymptotic behavior:
\[
F(T,r)=\frac{T}{2}+O(1).
\]
Therefore, for every fixed $r>0$,
\[
F(\mathbb R,r)=\mathbb R.
\]
It follows that the hypersurfaces $\Sigma_T$ are pairwise disjoint and cover the whole ingoing
Eddington--Finkelstein region $r>0$. Hence $\{\Sigma_T\}_{T\in\mathbb R}$ is a smooth spacelike
foliation.

Finally, we check the horizon foliation and the MOTS condition. The future horizon is
\[
\mathcal H^+=\{r=2m\}.
\]
The intersection of $\Sigma_T$ with $\mathcal H^+$ is the round sphere
\[
S_T
=
\{r=2m,\ v=F(T,2m)\}.
\]
At $r=2m$ one has $x_T(2m)=0$, and therefore
\[
F(T,2m)
=
T+2m+2m\log\left(\frac{\delta(T)}{m}\right)
+
2m\int_2^0\frac{\chi(y)}{y}\,dy.
\]
Since $\delta(T)=\delta_0e^{-T/(4m)}$, this becomes
\[
F(T,2m)
=
\frac{T}{2}
+
C_\chi,
\]
where
\[
C_\chi
=
2m
+
2m\log\left(\frac{\delta_0}{m}\right)
+
2m\int_2^0\frac{\chi(y)}{y}\,dy.
\]
Thus $T\mapsto F(T,2m)$ is a diffeomorphism from $\mathbb R$ to $\mathbb R$. Therefore the spheres
$S_T$ foliate the future horizon.

Each $S_T$ is a MOTS. Indeed, on the future horizon $r=2m$, the outgoing future null normal is the
horizon generator
\[
\ell=\partial_v.
\]
The outgoing null expansion of a round sphere is
\[
\theta_\ell=\frac{2}{r}\ell(r).
\]
Since $\ell(r)=\partial_v r=0$ on $\mathcal H^+$, we obtain
\[
\theta_\ell=0.
\]
Therefore each $S_T$ is a marginally outer trapped surface. This completes the proof.
\end{proof}

\begin{remark}
The construction glues the \emph{radial derivative} of the Schwarzschild time function, rather than
the tortoise coordinate $r_*$ itself. This avoids the logarithmic singularity of $r_*$ at
$r=2m$, keeps the slices smooth across the future horizon, and still makes the slices exactly equal
to Schwarzschild time slices on the exterior region $r\ge r_{\mathrm{cut}}(T)$.
\end{remark}

\begin{remark}
The transition region has width comparable to $\delta(T)$ and therefore moves into the horizon as
$T\to+\infty$. Consequently, on every fixed exterior tail $r\ge r_1>2m$, the construction agrees
exactly with the stationary Schwarzschild gauge for all sufficiently large $T$. This is the feature
needed in the final-state discussion; the appendix does not require uniform estimates in the
shrinking transition region.
\end{remark}

\section{A Penrose Inequality on the final Bondi mass}
\label{app:bondi-annuli}

The proof of Theorem~\ref{thm:SPI} constructs complete tangentially
maximal comparison hypersurfaces by solving the TMCF equation on exterior
tails and then passing to the limit as the outer radius tends to infinity.
This is the natural version when one wants to compare with the ADM mass.  In
this appendix we record a localized variant of the same argument.  The TMCF
equation is solved only on finite annuli whose outer boundaries are large
spheres near timelike infinity.  The result is a Penrose inequality with the
final Bondi mass in place of the ADM mass. This provides a stronger inequality since the final Bondi mass is at most equal to the ADM mass and is strictly smaller in the presence of a positive amount of energy radiated to null infinity.

The point of the construction is that the outer boundary is not sent to
spatial infinity along a fixed late slice.  Instead, we choose large spheres
whose retarded times tend to \(+\infty\).  Thus the comparison is made after
the outgoing radiation has passed through the relevant large sphere.  This
shows that the method is compatible with Bondi mass loss; indeed, in the
formulation below the final Bondi mass itself bounds the horizon area.

To formulate a Bondi mass, we assume in this appendix that the asymptotic
end admits a future Bondi--Sachs structure near \(\mathscr I^+\).  Thus
there are coordinates \((u,\rho,y^A)\) near future null infinity, where
\(u\) labels outgoing null hypersurfaces and \(\rho\) is the luminosity
radius, normalized so that the cuts \(S_{u,\rho}\) have area
\(4\pi\rho^2+o(\rho^2)\), and in these coordinates the metric has the usual
Bondi--Sachs expansion from which the Bondi mass \(m_B(u)\) is defined; see
\cite{BondiSachs62,Sachs62,MadlerWinicour16}.  This is only a
null-asymptotic regularity assumption needed to define \(m_B(u)\), and it
does not assume the absence of Bondi mass loss.
\begin{thm}[Bondi mass Penrose bound]
\label{thm:bondi-penrose}
Let
\[
(\widehat{\mathcal M},\gtime),\qquad
t:\widehat{\mathcal M}\to\mathbb R,\qquad
\Sigma_t=\{t=\mathrm{const}\},\qquad
\mathcal H_{\mathrm{app}}=\bigcup_{t\ge0}\mathcal S_t,
\]
and
\[
(M_\ast,g_\ast,K_\ast)\hookrightarrow(\widehat{\mathcal M},\gtime),
\qquad
S_\ast=\partial M_\ast,
\]
be as in Theorem~\ref{thm:SPI}.  Thus \(S_\ast\) is a MOTS and a smooth
cross-section of \(\mathcal H_{\mathrm{app}}\), the future portion of
\(\mathcal H_{\mathrm{app}}\) is piecewise smooth with finitely many jumps,
the last smooth piece \(\mathcal H_{\mathrm{final}}\) satisfies the quasi final state hypothesis, and the outermost-horizon area does not decrease
across jumps.

Assume, in addition, that the asymptotic end admits a future Bondi--Sachs
structure for which the Bondi mass \(m_B(u)\) is defined and has a finite
future limit
\[
m_B^+:=\lim_{u\to+\infty}m_B(u).
\]
Then
\begin{equation}\label{eq:bondi-penrose-Sstar}
m_B^+
\ge
\sqrt{\frac{|S_\ast|}{16\pi}}.
\end{equation}
\end{thm}

\begin{remark}
The additional Bondi--Sachs regularity in Theorem~\ref{thm:bondi-penrose}
is not a no-mass-loss assumption.  It is only the standard null-infinity
structure needed to define \(m_B(u)\).  Under the Bondi mass loss formula,
\(m_B(u)\) is nonincreasing in \(u\), and in general
\[
m_{ADM}(M_\ast,g_\ast,K_\ast)\ge m_B^+,
\]
with strict inequality when radiation carries energy to null infinity.
Thus \eqref{eq:bondi-penrose-Sstar} is stronger than the ADM-mass inequality
whenever mass loss occurs.
\end{remark}

\begin{proof}
The only new point is the replacement of the complete-tail
Huisken--Ilmanen comparison in Corollary~\ref{cor:HI-comparison} by a
localized comparison on compact TMCF annuli.  Since the TMCF annulus is not
an inverse-mean-curvature-flow annulus, the Hawking mass of its coordinate
leaves is not automatically monotone.  We therefore attach an auxiliary
nonnegative-scalar-curvature extension outside the large outer sphere and
then apply the usual Huisken--Ilmanen comparison to the glued manifold.

\smallskip

\noindent\textbf{Step 1: choose large finite annuli near timelike infinity.}
We assume that the late exterior contains a Bondi--Sachs end.  Thus there is
an exterior region
\[
\mathcal U_{\mathscr I^+}\simeq (u_0,\infty)\times(\rho_0,\infty)\times S^2
\]
with Bondi--Sachs coordinates \((u,\rho,y^A)\), where \(u\) labels the cuts
of future null infinity, \(\rho\) is the luminosity radius, and \(y^A\) are
angular variables.  The Bondi mass \(m_B(u)\) is the spherical average of the
Bondi mass aspect on the cut \(u=\mathrm{const}\), and we assume that
\[
m_B(u)\longrightarrow m_B^+
\qquad\text{as }u\to+\infty .
\]
See, for example, \cite{BondiSachs62,Sachs62,MadlerWinicour16}.

Let \((t,r,p)\) be the late exterior coordinates from
Definition~\ref{def:spi-QFS}.  In the overlap between the late exterior chart
and the Bondi--Sachs end, let
\[
u=u(t,r,p)
\]
denote the Bondi retarded time.  We choose outer boundary spheres
\[
\Sigma_T^+ := \{t=T,\ r=R(T)\}
\]
with
\[
R(T)\to\infty,
\qquad
\inf_{p\in S^2} u(T,R(T),p)\to+\infty
\qquad\text{as }T\to+\infty .
\]
Thus the outer boundaries are taken both to large radius and to late retarded
time, so that the large-sphere limits see the final Bondi mass \(m_B^+\).

In standard asymptotically flat coordinates one should think of
\[
u \sim t-r_\ast(r),
\]
where \(r_\ast\) is the usual tortoise, or optical, radial coordinate.  Hence
one may keep in mind the model choice \(R(T)=\sqrt T\), or more generally any
choice with \(R(T)\to\infty\) and \(r_\ast(R(T))=o(T)\).

Choose sequences
\[
r_j\downarrow r_0,
\qquad
T_j\to+\infty,
\]
so that
\[
R(T_j)\to\infty,
\qquad
\inf_{p\in S^2} u(T_j,R(T_j),p)\to+\infty.
\]
We choose \(T_j\) sufficiently large, after fixing \(r_j\), so that the
finite-annulus construction in the proof of Theorem~\ref{thm:global-existence}
applies on
\[
M_{r_j,R(T_j)}:=(r_j,R(T_j))\times S^2.
\]
Indeed, before passing to the exterior-tail limit \(R\to\infty\), the proof
of Theorem~\ref{thm:global-existence} constructs a truncated solution on
each finite annulus with prescribed outer boundary value; see
Definition~\ref{def:admissible-truncated} and the open-closed argument in
the proof of Theorem~\ref{thm:global-existence}.  Applying that construction
with
\[
r_1=r_j,
\qquad
R=R(T_j),
\qquad
T_0=T_j,
\]
gives a finite-annulus solution
\[
f_j\in\mathcal C\big(\overline{M_{r_j,R(T_j)}}\big)
\]
satisfying
\[
f_j(R(T_j),\cdot)=T_j.
\]

Let
\[
M_j:=\{(t,r,p): t=f_j(r,p),\ r_j\le r\le R(T_j),\ p\in S^2\}
\]
be the corresponding tangentially maximal annulus.  We write \(g_j\) for
the induced Riemannian metric on \(M_j\), and define its inner and outer
boundary leaves by
\[
\Sigma_j^-:=\{r=r_j\}\cap M_j,
\qquad
\Sigma_j^+:=\{r=R(T_j)\}\cap M_j.
\]
The outer boundary condition gives
\[
\Sigma_j^+=\{t=T_j,\ r=R(T_j)\}.
\]

Since \(f_j\) solves the TMCF equation, every leaf \(\Sigma_r\subset M_j\)
satisfies
\[
\operatorname{tr}_{\Sigma_r}K_j=0.
\]
Therefore the spacetime and Riemannian Hawking masses agree on the leaves:
\begin{equation}\label{eq:bondi-app-hawking-agree}
m_H^{\mathrm{ST}}(\Sigma_r)
=
m_H^{\mathrm{Riem}}(\Sigma_r),
\end{equation}
as in \eqref{eq:hawking-masses-agree}.

Moreover, since the finite-annulus solution is admissible,
Definition~\ref{def:parabolicity-conditions} implies that the graph is
spacelike and that the leaves satisfy
\[
H_{f_j,r}>0.
\]
Thus the coordinate-sphere foliation of \(M_j\) is mean-convex.  The same
calibration argument as in Lemma~\ref{lem:outerminimizing} shows that
\(\Sigma_j^-\) is connected and outward minimizing in the compact annulus
\(M_j\).  Finally, Proposition~\ref{prop:R-nonneg-tm} applies locally to
the tangentially maximal annulus and gives
\begin{equation}\label{eq:bondi-app-R-nonnegative}
R_{g_j}\ge0
\qquad\text{on }M_j.
\end{equation}

\smallskip
\noindent\textbf{Step 2: identify the outer large-sphere limit.}
Because
\[
R(T_j)\to\infty,
\qquad
\inf_{\Sigma_j^+}u\to+\infty,
\]
the outer boundaries \(\Sigma_j^+\) are large spheres whose Bondi retarded
times tend to \(+\infty\).  The usual large-sphere limit of the spacetime
Hawking energy at null infinity gives
\begin{equation}\label{eq:outer-hawking-to-bondi}
m_H^{\mathrm{ST}}(\Sigma_j^+)\to m_B^+;
\end{equation}
see, for example, \cite{ChruscielJezierskiKijowski,MarsSoria15}.  If
necessary, the sequence \(T_j\) is chosen diagonally so that the large-sphere
limit error and the late-\(u\) Bondi error both tend to zero.

We shall also use the corresponding large-sphere comparison between
Brown--York and Riemannian Hawking mass.  Let
\[
\gamma_j:=g_j|_{\Sigma_j^+}
\]
and let \(H_j\) be the mean curvature of
\(\Sigma_j^+\subset(M_j,g_j)\) with respect to the outward unit normal.
For \(j\) sufficiently large, the standard large-sphere estimates in the
asymptotic end give
\[
K_{\gamma_j}>0,
\qquad
H_j>0.
\]
Hence \((\Sigma_j^+,\gamma_j)\) admits a unique strictly convex isometric
embedding into \(\mathbb R^3\), by the Weyl--Nirenberg--Pogorelov theorem;
see \cite{Nirenberg53,Pogorelov52}.  Denote by \(H_{0,j}\) the Euclidean
mean curvature of the embedded surface.  The Brown--York mass is
\[
m_{BY}(\Sigma_j^+)
:=
\frac1{8\pi}
\int_{\Sigma_j^+}(H_{0,j}-H_j)\,d\mu_{\gamma_j}.
\]
For large nearly round spheres, the Brown--York and Hawking masses have the
same large-sphere limit; see \cite{FanShiTam09}.  In the present notation,
\begin{equation}\label{eq:BY-Hawking-outer}
m_{BY}(\Sigma_j^+)
-
m_H^{\mathrm{Riem}}(\Sigma_j^+)
\to0.
\end{equation}
Combining \eqref{eq:bondi-app-hawking-agree},
\eqref{eq:outer-hawking-to-bondi}, and \eqref{eq:BY-Hawking-outer}, we get
\begin{equation}\label{eq:BY-to-bondi}
m_{BY}(\Sigma_j^+)\to m_B^+.
\end{equation}

\smallskip
\noindent\textbf{Step 3: attach an auxiliary extension outside the large sphere.}
We now attach an auxiliary nonnegative-scalar-curvature end outside
\(\Sigma_j^+\).  This end is not the physical exterior of the spacetime.  It
is only a comparison device.

The required extension is supplied by the Shi--Tam quasi-spherical
construction \cite{ShiTam02}.  Applied to the Bartnik boundary data
\[
(\Sigma_j^+,\gamma_j,H_j),
\]
it gives an asymptotically flat scalar-flat Riemannian extension \(E_j\)
outside \(\Sigma_j^+\), inducing the boundary metric \(\gamma_j\) and the
boundary mean curvature \(H_j\), such that
\begin{equation}\label{eq:ShiTam-extension-bound}
m_{\mathrm{ADM}}(E_j)
\le
m_{BY}(\Sigma_j^+).
\end{equation}
Therefore, by \eqref{eq:BY-Hawking-outer},
\begin{equation}\label{eq:extension-ADM-Hawking}
m_{\mathrm{ADM}}(E_j)
\le
m_H^{\mathrm{Riem}}(\Sigma_j^+)+o(1).
\end{equation}

Glue \(E_j\) to \(M_j\) along \(\Sigma_j^+\).  The induced metric and mean
curvature agree across the gluing surface, so the glued metric has
nonnegative scalar curvature in the weak corner sense.  By Miao's smoothing
theorem for asymptotically flat manifolds with corners along a hypersurface
\cite{Miao02Corners}, we may smooth the corner, changing the ADM mass by
\(o(1)\), and obtain a smooth asymptotically flat Riemannian manifold
\[
(\widehat M_j,\widehat g_j)
\]
with
\[
R_{\widehat g_j}\ge0
\]
and
\begin{equation}\label{eq:glued-ADM-bound}
m_{\mathrm{ADM}}(\widehat M_j,\widehat g_j)
\le
m_H^{\mathrm{Riem}}(\Sigma_j^+)+o(1).
\end{equation}

The inner boundary \(\Sigma_j^-\) remains outward minimizing in the glued
manifold.  On the TMCF annulus this is the calibration argument from
Lemma~\ref{lem:outerminimizing}.  On the Shi--Tam extension, the
quasi-spherical construction gives an outward mean-convex foliation.  These
foliations fit together at \(\Sigma_j^+\), after the smoothing above, and
give the same outward-minimizing conclusion for any surface enclosing
\(\Sigma_j^-\) in \(\widehat M_j\).

\smallskip
\noindent\textbf{Step 4: apply Huisken--Ilmanen.}
We now apply Huisken--Ilmanen's weak inverse mean curvature flow comparison
\cite{H-I} to the complete asymptotically flat manifold
\((\widehat M_j,\widehat g_j)\), starting from the connected outward
minimizing surface \(\Sigma_j^-\).  Since
\[
R_{\widehat g_j}\ge0,
\]
we obtain
\begin{equation}\label{eq:HI-localized}
m_H^{\mathrm{Riem}}(\Sigma_j^-)
\le
m_{\mathrm{ADM}}(\widehat M_j,\widehat g_j).
\end{equation}
Combining \eqref{eq:HI-localized} with \eqref{eq:glued-ADM-bound} gives
\[
m_H^{\mathrm{Riem}}(\Sigma_j^-)
\le
m_H^{\mathrm{Riem}}(\Sigma_j^+)+o(1).
\]
Using \eqref{eq:bondi-app-hawking-agree} on the two boundary leaves, this
becomes
\begin{equation}\label{eq:localized-ST-comparison}
m_H^{\mathrm{ST}}(\Sigma_j^-)
\le
m_H^{\mathrm{ST}}(\Sigma_j^+)+o(1).
\end{equation}
Taking limits and using \eqref{eq:outer-hawking-to-bondi}, we get
\begin{equation}\label{eq:inner-limsup-bondi}
\limsup_{j\to\infty}m_H^{\mathrm{ST}}(\Sigma_j^-)
\le
m_B^+.
\end{equation}

\smallskip
\noindent\textbf{Step 5: pass to the final horizon.}
It remains to identify the limit of the inner Hawking masses.  This is the
same horizon comparison used in the proof of Theorem~\ref{thm:SPI}.  The
proof applies
verbatim to the present finite annuli, because it uses only the fixed inner
radius \(r_j\), the late-time estimates, and the collar regularity near
\(\mathcal H_{\mathrm{final}}\).  It does not use the existence of a
complete exterior tail after the outer radius has been sent to infinity.

Thus, after choosing \(T_j\) sufficiently large for each fixed \(r_j\), the
inner leaves \(\Sigma_j^-\) satisfy
\begin{equation}\label{eq:inner-hawking-horizon-limit}
m_H^{\mathrm{ST}}(\Sigma_j^-)
\to
\sqrt{\frac{A_\infty}{16\pi}}.
\end{equation}
Combining \eqref{eq:inner-limsup-bondi} with
\eqref{eq:inner-hawking-horizon-limit} gives
\[
m_B^+
\ge
\sqrt{\frac{A_\infty}{16\pi}}.
\]

Finally, since \(S_\ast\) is a smooth cross-section of
\(\mathcal H_{\mathrm{app}}\), Proposition~\ref{prop:mott-area-law}
together with the assumption that the outermost-horizon area does not
decrease across jumps gives
\[
A_\infty\ge |S_\ast|.
\]
Therefore
\[
m_B^+
\ge
\sqrt{\frac{|S_\ast|}{16\pi}}.
\]
This proves equation \eqref{eq:bondi-penrose-Sstar} and completes the proof of
Theorem~\ref{thm:bondi-penrose}.
\end{proof}

\begin{remark}
The auxiliary extension \(E_j\) should not be identified with the physical
exterior of the spacetime outside \(\Sigma_j^+\).  The physical exterior on
a late spacelike slice may register outgoing radiation and need not be a
nonnegative-scalar-curvature Riemannian extension of the TMCF annulus.  The
role of \(E_j\) is only to supply an asymptotically flat
nonnegative-scalar-curvature end whose ADM mass is controlled by the
Brown--York, hence asymptotically by the Hawking mass of the large outer
sphere and, in turn, by the final Bondi mass.
\end{remark}

\section*{Notation, terminology, and definitions}
\addcontentsline{toc}{section}{Notation, terminology, and definitions}
\label{sec:notation}

For convenience, this section collects the principal notation, terminology,
and definitions used in the paper, organized by category. Each entry is
followed by a reference to the equation, definition, lemma, or proposition
where the object is introduced or characterized.

\subsection*{Standing conventions}

\begin{itemize}
\item Throughout, $(\widehat{\mathcal M},\gtime)$ denotes a globally
hyperbolic asymptotically flat spacetime of dimension $3+1$ satisfying the
spacetime dominant energy condition. The Lorentzian metric $\gtime$ has
signature $(-,+,+,+)$.

\item For radial null normals to a closed spacelike $2$-surface, $\ell$ and
$\underline\ell$ denote the future-directed outgoing and ingoing null
normals, respectively, normalized so that $\gtime(\ell,\underline\ell)=-2$.
When the surface is disconnected, these choices are made componentwise.

\item The symbol $\odot$ denotes the symmetrized tensor product:
$\alpha\odot\omega:=\tfrac12(\alpha\otimes\omega+\omega\otimes\alpha)$.
The symbol $:$ denotes full contraction of tensors with the relevant metric.

\item $\sharp_g$ denotes the musical isomorphism associated to a metric $g$,
mapping $1$-forms to vector fields. The corresponding flat is $\flat_g$.

\item $\slashed d,\ \slashed\nabla,\ \cancel{div},\ \slashed\Delta$ denote
the exterior derivative, Levi--Civita connection, divergence, and
Laplace--Beltrami operator on a sphere $S_{t,r}$ equipped with the induced
metric $\gamma(t,r)$. When these operators are applied on a graph $t=f(r,p)$,
the background coefficients are understood to be pulled back to the graph.

\item Unless explicitly stated otherwise, repeated horizon-area notation such
as $A(t)=|\mathcal S_t|$ denotes total area. Thus, if
$\mathcal S_t=\bigcup_i\mathcal S_t^i$ is disconnected, then
\[
A(t):=|\mathcal S_t|=\sum_i |\mathcal S_t^i|.
\]
\end{itemize}

\subsection*{1.\ Spacetime, foliation, and ambient coordinates}

\begin{description}
\item[$\widehat{\mathcal M}$] The full ambient globally hyperbolic
asymptotically flat spacetime; see Section~\ref{sub:spacetime-final-state}.

\item[$\mathcal M=(\underline T,\infty)\times M$] The late-time exterior
region used in the existence theory, with
$M:=\R^3\setminus\overline{B_{r_0}}\cong(r_0,\infty)\times S^2$;
see~\eqref{eq:metric-nonstat} and the surrounding text.

\item[$\mathcal H$] The formal limiting inner cylinder
$(\underline T,\infty)\times\{r_0\}\times S^2$ appearing in the late-time
exterior chart; see~\eqref{eq:metric-nonstat}ff. This notation is separate
from the apparent-horizon tube $\mathcal H_{\mathrm{app}}$.

\item[$t$] A globally defined time function on $\widehat{\mathcal M}$
adapted to the asymptotic rest frame, with Cauchy level sets $\Sigma_t$;
see Section~\ref{sub:spacetime-final-state}.

\item[$(t,r,p)$] Late-time exterior chart on the closure of the exterior of
the last smooth horizon piece; here $p\in S^2$. See~\eqref{eq:metric-nonstat}
and Definition~\ref{def:spi-QFS}\ref{it:spi-QFS-tail}.

\item[$M_t,\ \Sigma_t$] Time slices in the exterior chart and in
$\widehat{\mathcal M}$, respectively; see~\eqref{eq:nt-general}ff. and
Section~\ref{sub:spacetime-final-state}.

\item[$S_r,\ S_{t,r}$] Coordinate spheres
$\{r\}\times S^2\subset M_t$ and $\{t\}\times\{r\}\times S^2\subset
\mathcal M$; see~\eqref{eq:metric-nonstat}ff.

\item[$S_{f,r}$] The graph sphere
$M_f\cap\big((\underline T,\infty)\times\{r\}\times S^2\big)$; when working
on $M_f$, the subscript $f$ is often suppressed and we simply write $S_r$;
see the graph notation preceding~\eqref{eq:graph-spacelike}.

\item[$M_{r_*,R}$] The cylindrical annulus $(r_*,R]\times S^2$ for
$r_0\le r_*<R\le\infty$; see preliminaries of
Section~\ref{sub:apriori}.

\item[$\mathcal M_{r_1}$] The fixed-tail spacetime
$(\underline T,\infty)\times M_{r_1,\infty}$;
see~\eqref{eq:metric-general-prelim}ff.

\item[$s=-\log r$] The logarithmic radial variable in which the TMCF
equation becomes uniformly parabolic on truncated cylinders;
see~\eqref{eq:sdef}.

\item[$I(T_0)$] The late-time slab
$[T_0-h_{\mathrm{slab}},T_0+h_{\mathrm{slab}}]$ used in the good-gauge
construction and in the truncated-annulus existence theory; see
Definition~\ref{def:admissible-truncated}.
\end{description}

\subsection*{2.\ ADM coefficients and induced geometry of time slices}

\begin{description}
\item[$\gtime$] The spacetime metric in general ADM form
\[
\gtime
=
-(N^2-|\beta|_{g_t}^2)\,dt^2+2\,\beta\odot dt+g_t,
\quad
 g_t=(\lambda^2+|b|_\gamma^2)\,dr^2+2\,b\odot dr+\gamma;
\]
see~\eqref{eq:metric-nonstat} and~\eqref{eq:metric-general-prelim}.

\item[$N$] Spacetime lapse; positive function on $\mathcal M$.

\item[$\lambda$] Geometric foliation lapse of the $r$-level sets in
$(M_t,g_t)$; equivalently $\lambda=|\nabla^{g_t}r|^{-1}$.

\item[$\beta$] Spacetime shift $1$-form on $\mathcal M$ with no $dt$-component.
Decomposed as $\beta=\beta_r\,dr+\beta^T$ with $\beta^T$ tangent to
$S_{t,r}$; see~\eqref{eq:theta-decomp}.

\item[$b$] $1$-form on $\mathcal M$ tangent to $S_{t,r}$ (no $dt$ or $dr$
component), encoding the off-diagonal part of $g_t$.

\item[$\gamma=\gamma(t,r)$] Two-parameter family of metrics on $S^2$
realizing the induced metric on the spheres $S_{t,r}$.

\item[$\alpha$] Auxiliary scalar $\alpha:=N^2-|\beta|_{g_t}^2$; equals
$-\gtime(\partial_t,\partial_t)$ and may change sign;
see~\eqref{eq:beta-perp}ff.

\item[${\bf n}_t$] Outward unit normal to $S_{t,r}$ in $(M_t,g_t)$:
${\bf n}_t=\lambda^{-1}(\partial_r-b^{\sharp_\gamma})$;
see~\eqref{eq:nt-general}.

\item[${\bf T}_t$] Future-pointing timelike unit normal to $M_t$ in
$(\mathcal M,\gtime)$:
${\bf T}_t=N^{-1}(\partial_t-\beta^{\sharp_{g_t}})$;
see Section~\ref{sub:spacetime-final-state}ff.

\item[$\beta^\perp$] Geometric normal component of the shift,
$\beta^\perp:=\beta({\bf n}_t)$;
see~\eqref{eq:beta-perp} and~\eqref{eq:beta-perp-def-main}.

\item[$K_t$] Second fundamental form of $M_t\hookrightarrow
(\mathcal M,\gtime)$, with the ADM sign convention
$K_t(X,Y)=-\gtime(\nabla^{(4)}_X{\bf T}_t,Y)$.

\item[$H_{t,r}$] Spacelike mean curvature of $S_{t,r}$ in $(M_t,g_t)$;
see~\eqref{eq:Htr-general}.

\item[$\tr_{S_{t,r}}K_t$] Timelike mean curvature of $S_{t,r}$, equal to
the $\gamma$-trace of $K_t$ along $S_{t,r}$.

\item[${\bf H}_{t,r}$] Mean curvature vector of $S_{t,r}$ in
$(\mathcal M,\gtime)$, decomposing as
${\bf H}_{t,r}=H_{t,r}{\bf n}_t-(\tr_{S_{t,r}}K_t){\bf T}_t$,
with Lorentzian norm
$|{\bf H}_{t,r}|^2=H_{t,r}^2-(\tr_{S_{t,r}}K_t)^2$;
see~\eqref{eq:intro-H-decomp}.

\item[$Q_r(\omega),\ Q_t(\omega)$] Good-gauge tensors
\[
Q_r(\omega)=\partial_r\gamma+2\,\partial_r\beta^T\odot\omega
-\partial_r\alpha\,\omega^2,
\qquad
Q_t(\omega)=\partial_t\gamma+2\,\partial_t\beta^T\odot\omega
-\partial_t\alpha\,\omega^2;
\]
see~\eqref{eq:QrQt-gauge}.

\item[$\mathfrak a_0$] Background mean-curvature coefficient
$\mathfrak a_0:=H_{t,r}/\lambda$; see~\eqref{eq:mu0T-b0-def}.

\item[$\underline{\mathfrak a_0}(r_1),\ \overline{\mathfrak a_0}(r_1)$]
Tail-uniform lower and upper bounds for $r\,\mathfrak a_0$;
see~\eqref{eq:tail-constants-def}.
\end{description}

\subsection*{3.\ Graph quantities and induced data on $M_f$}

For $f\in C^2(M)$, $M_f:=\{(t,x)\in\mathcal M:t=f(x)\}$ is the graph
hypersurface, identified with $M$ via the graph map. Background coefficients
$N,\lambda,\beta,b,\gamma$ are usually evaluated at $t=f(r,p)$ when used on
$M_f$.

\begin{description}
\item[$M_f$] Graph hypersurface of $f$;
see~\eqref{eq:Wf}ff.

\item[$\tau=t-f(x)$] Defining function for the graph. The graph is spacelike
exactly when $d\tau$ is timelike; see~\eqref{eq:graph-spacelike}.

\item[$\mu_f$] Spacelikeness indicator
$\mu_f:=(1+\langle\beta,df\rangle_{g_t})^2-N^2|df|_{g_t}^2$; the graph
$M_f$ is spacelike iff $\mu_f>0$; see~\eqref{eq:Wf}.

\item[$N_f$] Lapse of the time function $\tau=t-f(x)$ on $M_f$:
$N_f=N/\sqrt{\mu_f}$; see~\eqref{eq:Nf}.

\item[$g_f,\ K_f$] Pullback metric on $M_f$ and, when $M_f$ is spacelike, its
second fundamental form in $(\mathcal M,\gtime)$;
see~\eqref{eq:g_f-def}--\eqref{eq:g_f-split}.

\item[$\gamma_f(r)$] Induced metric on $S_r\subset M_f$;
see~\eqref{eq:gammaf}.

\item[$\lambda_f,\ b_f$] Foliation lapse and tangential shift of the
$r$-foliation of $(M_f,g_f)$; see~\eqref{eq:bf-lambdaf}.

\item[$\mu_f^T$] Tangential spacelikeness indicator,
\[
\mu_f^T:=\det\gamma_f/\det\gamma
=(1+\langle\beta^T,\slashed df\rangle_\gamma)^2
-(N^2-(\beta^\perp)^2)|\slashed df|_\gamma^2;
\]
$\gamma_f$ is Riemannian iff $\mu_f^T>0$;
see~\eqref{eq:muTf-general}. In the good gauge this reduces to
\[
\mu_f^T
=(1+\langle\beta^T,\slashed df\rangle_\gamma)^2
-N^2|\slashed df|_\gamma^2.
\]

\item[${\bf n}_f$] Unit normal to $S_r$ in $(M_f,g_f)$:
${\bf n}_f=\lambda_f^{-1}(\partial_r-b_f^{\sharp_{\gamma_f}})$;
see~\eqref{eq:nf}.

\item[${\bf T}_f$] Future-pointing timelike unit normal to $M_f$ in
$(\mathcal M,\gtime)$, defined when $M_f$ is spacelike;
see Lemma~\ref{lem:basic-geom}.

\item[$H_{f,r}$] Spacelike mean curvature of $S_r$ in $(M_f,g_f)$;
see~\eqref{eq:Hf}.

\item[$\tr_{S_r}K_f$] Timelike mean curvature of $S_r$ on a spacelike graph
$M_f$, governing the spacelike TMCF equation; see~\eqref{eq:trKf-nonstat}.

\item[$\beta_f$] Shifted spacetime $1$-form on $M_f$,
$\beta_f:=\beta-\alpha\,df$, decomposed as
$(\beta_f)_r\,dr+\beta_f^T$; see the notation following~\eqref{eq:Nf}.

\item[$V_f$] Tangential vector field used in the good-gauge divergence
calculation,
\[
V_f=\big(1+\langle\beta^T,\slashed d f\rangle_\gamma\big)(\beta^T)^{\sharp_\gamma}
-N^2(\slashed d f)^{\sharp_\gamma};
\]
see the proof of Proposition~\ref{prop:gammab-quasilinear}.

\item[$\mathfrak a_f$] Auxiliary scalar
$\mathfrak a_f:=\frac{1}{2\lambda^2}\tr_{\gamma_f}Q_r(\slashed df)$, equal
to $H_{f,r}/\lambda_f$ on spacelike solutions of the TMCF equation;
see~\eqref{eq:hhf-def} and Lemma~\ref{lem:H-over-lambda}.
\end{description}

\subsection*{4.\ The TMCF equation and its parabolic structure}

\begin{description}
\item[TMCF foliation] A foliation $\{S_r\}_{r\ge r_0}$ of a hypersurface
$M$, not necessarily spacelike, by closed spacelike two-surfaces such that the
spacetime mean curvature vector of each leaf is tangent to $M$:
\[
{\bf H}_{S_r}\in TM.
\]
At non-null points of $M$ this is equivalent to the vanishing of the trace of
the second fundamental form of $M$ over $S_r$; in the spacelike case it is
exactly $\tr_{S_r}K=0$. See Definition~\ref{def:intro-tmcf}.

\item[Tangentially maximal hypersurface / initial data set]
A hypersurface $M$ equipped with a TMCF foliation. If $M$ is spacelike, the
induced data $(M,g,K)$ is called a tangentially maximal initial data set;
see Definition~\ref{def:intro-tmcf}.

\item[TMCF equation, spacelike geometric form]
For a spacelike graph, the condition $\tr_{S_r}K_f=0$ is equivalent to
\[
\cancel{div}_{\gamma_f}(\beta_f^T)+H_{f,r}\,\beta_f({\bf n}_f)
-\tfrac12\tr_{\gamma_f}(\partial_t\gamma_f)=0;
\]
see~\eqref{eq:TMCF-equation}, \eqref{eq:TMCF-nonlinear}, and
\eqref{eq:intro-tmcf}.

\item[Good gauge] The chart condition $b\equiv 0$, $\beta_r\equiv 0$ on
$\mathcal M$; produced on late slabs by
Appendix~\ref{app:late-slab-gauge}; see~\eqref{eq:good-gauge}.

\item[TMCF equation, good-gauge quasilinear form]
\[
f_r+\mathcal A(r,p,f,\slashed\nabla f):\slashed\nabla_\gamma^2 f
=\mathcal B(r,p,f,\slashed\nabla f),
\]
with $\mathcal A$ given by~\eqref{eq:A-explicit-gauge} and $\mathcal B$
by~\eqref{eq:B-explicit-gauge}; see
Proposition~\ref{prop:gammab-quasilinear}.

\item[Analytic good-gauge equation]
The explicit good-gauge equation makes sense under the parabolic conditions
$\mu_f^T>0$ and $\mathfrak a_f>0$, even before imposing full spacelikeness
$\mu_f>0$; see Definition~\ref{def:parabolicity-conditions}. In the good
gauge, on regions where $\mu_f^T>0$ and $\mathfrak a_f\neq0$, this analytic
equation is equivalent to the coordinate-free tangency condition
${\bf H}_{S_{f,r}}\in TM_f$. If $\mu_f>0$ also holds, then $M_f$ is spacelike
and the same condition is equivalently $\tr_{S_r}K_f=0$.

\item[Parabolically admissible] A function $f$ on
$(r_1,r_2]\times S^2$ with $\mu_f^T>0$ and $r\mathfrak a_f>0$;
see Definition~\ref{def:parabolicity-conditions}.

\item[Admissible] A parabolically admissible function with the additional
property $\mu_f>0$ (full spacelikeness);
see Definition~\ref{def:parabolicity-conditions}.

\item[$\Xi$] Background forcing
$\Xi(t,r,p):=(\lambda/N)\cdot\tr_{S_{t,r}}K_t/H_{t,r}$;
see~\eqref{eq:intro-Xi} and~\eqref{eq:Xi-background-def}.

\item[$\mathcal B_2$] Component of the lower-order term in the good-gauge
form vanishing when $\slashed df=0$; the decomposition
$\mathcal B=\Xi|_{t=f}+\mathcal B_2$ is recorded in
Remark~\ref{rem:c1c2}.

\item[$\vartheta_*$] Tail-controlled smallness constant used to guarantee
$\mu_f^T>0$ from a gradient bound and to formulate \textbf{BS-1};
see~\eqref{eq:varthetastar-implies-muT} and
Definition~\ref{def:BS1-BS2}.

\item[\textbf{BS-1}, \textbf{BS-2}] Quantitative bootstrap assumptions used
in the a priori estimates;
$|\slashed df|_\gamma<\vartheta_*/2$ and
$r\mathfrak a>\tfrac12\underline{\mathfrak a_0}$ for
\textbf{BS-1};
$\mu>0$ and $r^2|\slashed\nabla^2 f|_\gamma<1$ for \textbf{BS-2}; see
Definition~\ref{def:BS1-BS2}. Their boundary versions for a profile
$\varphi$ on a sphere are in
Definition~\ref{def:admissible-boundary-profile}.

\item[$(\varepsilon,T_0;r,R)$-\textbf{BS-1} / \textbf{BS-2} solution]
A TMCF solution on $M_{r,R}$ in a good-gauge chart, with image in the
late slab $I(T_0)$, prescribed boundary value $T_0$ on $S_R$, and
satisfying the corresponding bootstrap;
see Definition~\ref{def:admissible-truncated}.

\item[$\mathcos(r^\delta)$] Tail-controlled estimate notation: a quantity
$\mathcal G$ satisfies $\mathcal G=\mathcos(r^\delta)$ if
$|\mathcal G|\le C r^\delta$, where $C$ depends only on the fixed tail data
listed in Section~\ref{sub:apriori}; see the notation paragraph in the proof
of Theorem~\ref{thm:apriori-estimates}.
\end{description}

\subsection*{5.\ Function spaces}

\begin{description}
\item[$d_p$] Parabolic distance in the $s$-variable,
$d_p((s,x),(s',y))=d_{\gamma_{S^2}}(x,y)+|s-s'|^{1/2}$;
see~\eqref{eq:dp}.

\item[$\cz(M_{r_*,R}),\ \ct(M_{r_*,R})$]
Parabolic H\"older spaces $\mathcal C^{\alpha,\alpha/2}$ and
$\mathcal C^{2+\alpha,1+\alpha/2}$ in the logarithmic radial variable;
see Definition~\ref{def:cts}.

\item[$\ct_{-\sigma}(M_{r_*,R}),\ \cz_{-\sigma}(M_{r_*,R})$]
Weighted versions, equivalent to weight $r^{-\sigma}$ at infinity;
see Definition~\ref{def:weightedlog}.

\item[${\ct}^\sharp_{-\sigma}(M_{r_*,R})$]
Strengthened weighted tail space encoding additional tangential and mixed
radial--tangential regularity;
see Definition~\ref{def:ctsharp-main}.

\item[$\ctS_{-\tau}(\mathcal M)$] Coefficient class on exterior tails for
$\tau\in(1/2,1)$: tuples $\mathcal S=(N,\lambda,\beta,b,\gamma)$ with
weighted asymptotic flatness on every fixed tail, tail nondegeneracy of
$\lambda,N,rH_{t,r}$, and finite forcing tail integrals;
see Definition~\ref{def:ctS-nonstat}.

\item[${\ctS}^\sharp_{-\tau}(\mathcal{M})$] Strengthened tail-coefficient class controlled, for each $r_1>r_0$, by the $\sharp$-norm
$\mathcal N_\tau^\sharp(r_1)$; see
Definition~\ref{def:sharp-tail-data-main}.

\item[$\mathcal N_\tau(r_1),\ \mathcal N_\tau^\sharp(r_1)$] Tail norms
controlling the ordinary and strengthened coefficient classes on the fixed
tail $r\ge r_1$; see~\eqref{eq:ctS-nonstat-decay} and
\eqref{eq:Nsharp-def-main}.

\item[Asymptotically flat exterior spacetime of order $\tau$]
A spacetime $(\mathcal M,\gtime)$ written in general ADM form with
coefficient tuple in $\ctS_{-\tau}(\mathcal M)$;
see Definition~\ref{def:AFspacetime}.
\end{description}

\subsection*{6.\ Forcing tails and late-time decay conditions}

\begin{description}
\item[$\mathfrak X_0(r,T),\ \mathfrak X_1(r,T),\ \mathfrak X_2(r,T)$]
Pointwise forcing tails $r$-, $r$-, and $r^2$-rescaled $\sup$ over
$\{t\ge T\}\times S^2$ of $|\Xi|$, $|\slashed d\Xi|_\gamma$, and
$|\slashed\nabla^2\Xi|_\gamma$, respectively;
see~\eqref{eq:X0-def}--\eqref{eq:X2-def}.

\item[$\mathfrak X_j(r;I)$] Slab versions over a compact time interval $I$;
see~\eqref{eq:X0-slab-def}--\eqref{eq:X2-slab-def}.

\item[$G_{r_1}(T)$] Tail integral
$\int_{r_1}^\infty\mathfrak X_0(\sigma,T)/\sigma\,d\sigma$;
see~\eqref{eq:G-tail-def}.

\item[$\Phi_{1,r_1}(T),\ \Phi_{2,r_1}(T)$] Tail integrals of $\mathfrak X_1$
and $\mathfrak X_2$;
see~\eqref{eq:Phi1-def}--\eqref{eq:Phi2-def}.

\item[$\Omega_{r_1}(T)$] Pointwise tail $\sup_{r\ge r_1}\mathfrak X_0(r,T)$;
see~\eqref{eq:Omega-def}.

\item[$\Psi_{r_1}(T)$] H\"older-norm tail
$\sup_{t\ge T}\|r\Xi(t,\cdot)\|_{\cz_{-\tau}(M_{r_1,\infty})}$;
see~\eqref{eq:Psi-def}.

\item[$G_{r_1,I},\ \Phi_{j,r_1,I},\ \Omega_{r_1,I},\ \Psi_{r_1,I}$]
Slab analogues over a compact interval $I$;
see~\eqref{eq:GI-def}--\eqref{eq:PsiI-def}.

\item[$\mathfrak B^\sharp(T;r_1)$] Geometric normal-shift tail
$\sup_{t\ge T}\|\beta^\perp(t,\cdot)\|_{{\ct}^\sharp_{-\tau}(M_{r_1,\infty})}$;
see~\eqref{eq:Bsharp-def-main}.

\item[\textup{(A1)}, \textup{(A2)}, \textup{(A3)}] Late-time forcing
decay assumptions: $G_{r_1}(T),\Phi_{1,r_1}(T)\to 0$ for \textup{(A1)};
$\Phi_{2,r_1}(T),\Omega_{r_1}(T)\to 0$ for \textup{(A2)};
$\Psi_{r_1}(T)\to 0$ for \textup{(A3)};
see Definition~\ref{def:tail-assumptions}.

\item[Late-time gauge reducibility]
The condition that $\mathcal S\in{\ctS}^\sharp_{-\tau}(\mathcal M)$ and on every fixed tail $r_1>r_0$, $\mathfrak B^\sharp(T;r_1)\to 0$ as $T\to\infty$;
see Definition~\ref{def:late-gauge-decay}.
\end{description}

\subsection*{7.\ Apparent horizon and final state}

\begin{description}
\item[MOTS] A closed spacelike two-surface with vanishing future outgoing
null expansion, $\theta_{(\ell)}=0.$
See Section~\ref{sub:spacetime-final-state}.

\item[Weakly outer trapped surface] A closed spacelike two-surface satisfying
$\theta_{(\ell)}\le0$. The word ``outermost'' is used in the standard
apparent-horizon sense: no homologous weakly outer trapped surface lies outside
the given surface in the relevant Cauchy slice; see
Section~\ref{sub:spacetime-final-state}.

\item[Future black-hole-type MOTS] A MOTS satisfying, componentwise, $\theta_{(\ell)}=0$ and $\theta_{(\underline\ell)}<0,$
where $\ell$ and $\underline\ell$ are future-directed outgoing and ingoing null
normals; see Section~\ref{sub:spacetime-final-state} and
Proposition~\ref{prop:mott-area-law}.

\item[MOTT] A marginally outer trapped tube, i.e. a three-dimensional
hypersurface foliated by MOTSs; see Section~\ref{sub:spacetime-final-state}.

\item[$\mathcal S_t$] Smooth compact outermost future black-hole-type MOTS in
$\Sigma_t$, not assumed connected unless explicitly stated;
see Section~\ref{sub:spacetime-final-state}. On the final horizon piece in
the QFS chart, the sections are connected spheres
$\{t\}\times\{r_0\}\times S^2$.

\item[$A(t)=|\mathcal S_t|$] Total area of the outermost MOTS on $\Sigma_t$,
i.e. the sum of the areas of its connected components.

\item[$\mathcal H_{\mathrm{app}}=\bigcup_t\mathcal S_t$] Future
black-hole-type outermost apparent-horizon tube; assumed piecewise smooth with
finitely many jump times, each smooth piece a smooth MOTT;
see Section~\ref{sub:spacetime-final-state}.

\item[Jump times] Times at which the outermost MOTS changes discontinuously in
the chosen foliation. In the spacetime Penrose theorem, the total outermost
horizon area is assumed not to decrease across such jumps; see
Theorem~\ref{thm:SPI}.

\item[$\mathcal H_{\mathrm{final}}$] Last smooth piece of
$\mathcal H_{\mathrm{app}}$, $\mathcal H_{\mathrm{app}}\cap\{t\ge\underline T\}$;
see Section~\ref{sub:spacetime-final-state}.

\item[$A_\infty$] Limiting area
$\lim_{t\to\infty}A(t)\in(0,\infty)$ on
$\mathcal H_{\mathrm{final}}$ under the QFS hypothesis;
see Definition~\ref{def:spi-QFS}\ref{it:spi-QFS-horizon}.

\item[$M_\ast,\ S_\ast$] The asymptotically flat initial data set appearing
in the spacetime Penrose theorem and its boundary
$S_\ast:=\partial M_\ast$, assumed to be a smooth MOTS cross-section of
$\mathcal H_{\mathrm{app}}$; see Theorem~\ref{thm:SPI}.

\item[Strong final state hypothesis]
Convergence of late-time exterior coefficients to a subextremal Kerr tuple
in Boyer--Lindquist gauge in a weighted $C^5$ topology, plus $C^2$
convergence of horizon sections to a Kerr cross-section;
see Definition~\ref{def:strong-final-state}.

\item[quasi final state hypothesis (QFS)]
Three conditions: \textup{(Q1)} a late-time $C^2$ chart on the closure
of the exterior with $\mathcal H_{\mathrm{final}}=\{r=r_0\}$ and
coefficient tuple in $\ctS_{-\tau}(\mathcal M)$; \textup{(Q2)}
late-time gauge reducibility plus assumption \textup{(A1)} and either
\textup{(A2)} or \textup{(A3)}; \textup{(Q3)} area stabilization
$A(t)\to A_\infty\in(0,\infty)$;
see Definition~\ref{def:spi-QFS}.
\end{description}

\subsection*{8.\ Mass functionals and asymptotic charges}

\begin{description}

\item[$m_{ADM}$] Invariant ADM mass of the asymptotic end,
\[
m_{ADM}=\sqrt{E_{ADM}^2-|P_{ADM}|^2},
\]
which equals the ADM energy in the chosen asymptotic rest frame. In the proof
this same symbol denotes the ADM mass of the comparison tails in that rest
frame; see Section~\ref{sub:spacetime-final-state} and
Lemma~\ref{lem:ADM-compat}.

\item[$m_H^{\mathrm{ST}}(\Sigma)$] Spacetime Hawking mass of a closed
spacelike $2$-surface $\Sigma$,
\[
m_H^{\mathrm{ST}}(\Sigma)
=
\sqrt{|\Sigma|/16\pi}
\Big(1-\tfrac{1}{16\pi}\textstyle\int_\Sigma|{\bf H}_\Sigma|^2\,d\mu_\Sigma\Big);
\]
see~\eqref{eq:hawking-on-mots} and Section~\ref{sub:spacetime-final-state}.

\item[$m_H^{\mathrm{Riem}}(\Sigma)$] Riemannian Hawking mass computed in
the ambient hypersurface containing $\Sigma$,
\[
m_H^{\mathrm{Riem}}(\Sigma)
=
\sqrt{|\Sigma|/16\pi}
\Big(1-\tfrac{1}{16\pi}\textstyle\int_\Sigma H_\Sigma^2\,d\mu_\Sigma\Big);
\]
see~\eqref{eq:hawking-masses-agree}. On leaves of a spacelike TMCF foliation
the two masses agree.
\end{description}

\subsection*{9.\ Comparison tails and boundary leaves}

\begin{description}
\item[$f^{(r_1,T_0)}$] The admissible tangentially maximal tail produced by
Theorem~\ref{thm:global-existence} on the fixed exterior tail $r\ge r_1$, with
asymptotic boundary value $T_0$ at infinity; see the proof of
Theorem~\ref{thm:SPI}.

\item[$\Sigma_{r_1}^{(T_0)}$] Boundary leaf of the comparison tail
$f^{(r_1,T_0)}$ over the fixed radius $r=r_1$; see the proof of
Theorem~\ref{thm:SPI} and Corollary~\ref{cor:HI-comparison}.

\item[Outward minimizing] A boundary surface is outward minimizing if every
surface enclosing it has area at least as large. Boundary leaves of admissible
tangentially maximal tails are outward minimizing by the mean-convex foliation
calibration argument; see Lemma~\ref{lem:outerminimizing}.

\item[Huisken--Ilmanen comparison] The application of weak inverse mean
curvature flow to the nonnegative-scalar-curvature comparison tail, giving
\[
m_{ADM}\ge m_H^{\mathrm{Riem}}(\Sigma_{r_1}^{(T_0)})
=m_H^{\mathrm{ST}}(\Sigma_{r_1}^{(T_0)});
\]
see Corollary~\ref{cor:HI-comparison}.
\end{description}

\subsection*{10.\ Theorems and propositions referred to throughout}

\begin{description}
\item[Theorem~\ref{thm:global-existence}]
Existence of good-gauge TMCF tails on late-time exterior tails, with three
levels of conclusion: a solution from a sufficiently large outer sphere; a
parabolically admissible fixed inner tail under \textup{(A1)}; and an
admissible spacelike tangentially maximal tail under \textup{(A1)} and either
\textup{(A2)} or \textup{(A3)}.

\item[Theorem~\ref{thm:apriori-estimates}]
A priori $C^0$, gradient, Schauder, Hessian, and radial-slope estimates
for solutions on truncated annuli.

\item[Proposition~\ref{prop:late-good-gauge-small}]
Late good-gauge slabs with small forcing, derived from the gauge
reduction of Appendix~\ref{app:late-slab-gauge}.

\item[Proposition~\ref{prop:local-r-extension}]
Local-in-$r$ existence of TMCF solutions from a sphere, propagating the
admissibility condition inward.

\item[Proposition~\ref{prop:mott-area-law}]
Causal character and area law for smooth pieces of
$\mathcal H_{\mathrm{app}}$ under the future black-hole-type MOTT condition.

\item[Theorem~\ref{thm:SPI}]
The spacetime Penrose inequality under the quasi final state hypothesis.

\item[Lemma~\ref{lem:ADM-compat}]
Compatibility of the ADM mass of the tangentially maximal comparison tails
with the ADM mass of the original asymptotic end in the chosen rest frame.

\item[Lemma~\ref{lem:outerminimizing}]
Outward minimizing property of the boundary spheres of admissible comparison
tails.

\item[Proposition~\ref{prop:R-nonneg-tm}]
Nonnegative scalar curvature on spacelike tangentially maximal tails under the
dominant energy condition.

\item[Corollary~\ref{cor:HI-comparison}]
$m_{ADM}\ge m_H^{\mathrm{Riem}}(\Sigma_{r_1}^{(T_0)})
=m_H^{\mathrm{ST}}(\Sigma_{r_1}^{(T_0)})$ on tangentially maximal tails,
via Huisken--Ilmanen weak inverse mean curvature flow.
\end{description}

\bibliographystyle{amsplain}
\bibliography{refsPenroseTMCF}

\end{document}